\documentclass[a4paper,11pt]{article}

\usepackage{amsthm}
\usepackage{fullpage}%
\usepackage{graphicx,tikz,enumerate}
\usepackage[colorlinks=true]{hyperref}
\usepackage[title]{appendix}
\usepackage{amssymb, amsfonts,mathtools,stmaryrd}
\usepackage{xcolor}
\usepackage{enumitem}
\usepackage{etoolbox}

\usepackage{algorithm}
\usepackage{algpseudocode}

\usepackage{dsfont}
\usepackage{bbm}
\newtheorem{definition}{\textbf{Definition}}

\newtheorem{theorem}[definition]{\textbf{Theorem}}

\newtheorem{lemma}[definition]{\textbf{Lemma}}

\newtheorem{corollary}[definition]{\textbf{Corollary}}

\newtheorem{proposition}[definition]{\textbf{Proposition}}

\newtheorem{example}[definition]{\textbf{Example}}

\newtheorem{observation}[definition]{\textbf{Observation}}

\newtheorem{notation}[definition]{\textbf{Notation}}

\newtheorem*{proof*}{\textbf{Proof}}

\newcommand{\Succc}{\mathsf{Succ}}

\newcommand{\Ag}{\msf{Ag}}
\newcommand{\fanBr}[1]{\langle\langle#1\rangle\rangle}

\newcommand{\Tr}{\ensuremath{\msf{True}}}
\newcommand{\Fl}{\ensuremath{\msf{False}}}
\newcommand{\bool}{\ensuremath{\mathbb{B}}}
\newcommand{\Op}[2]{\mathsf{Op}_{\msf{#1}}^{\msf{#2}}}
\newcommand{\Ut}{\mathsf{U}^\mathsf{t}}
\newcommand{\Bt}{\mathsf{B}^\mathsf{t}}
\newcommand{\Bl}{\mathsf{B}^\mathsf{l}}

\newcommand{\Size}[1]{|#1|}

\newcommand{\N}{\mathbb{N}}
\newcommand{\Z}{\mathbb{Z}}

\newcommand{\msf}[1]{\mathsf{#1}}

\newcommand{\LTL}{\msf{LTL}}

\newcommand{\CTL}{\msf{CTL}}
\newcommand{\ATL}{\msf{ATL}}

\newcommand{\CL}{\msf{C}\text{-}\msf{L}}
\newcommand{\LC}{\msf{L}\text{-}\msf{C}}

\newcommand{\prop}{\mathsf{Prop}}

\DeclareMathOperator{\lF}{\mathbf{F}}
\DeclareMathOperator{\lG}{\mathbf{G}}
\DeclareMathOperator{\lU}{\mathbf{U}}
\DeclareMathOperator{\lX}{\mathbf{X}}
\DeclareMathOperator{\lR}{\mathbf{R}}
\DeclareMathOperator{\lW}{\mathbf{W}}
\DeclareMathOperator{\lM}{\mathbf{M}}

\usepackage{pifont}
\usepackage{authblk}
\usepackage{graphicx}
\usepackage{multirow}
\author[1,2]{Benjamin Bordais}
\author[1,2]{Daniel Neider}
\author[3]{Rajarshi Roy}
\affil[1]{TU Dortmund University, Dortmund, Germany}
\affil[2]{Center for Trustworthy Data Science and Security, University Alliance Ruhr, Dortmund, Germany}
\affil[3]{Department of Computer Science, University of Oxford, Oxford, UK}

\date{}
\usepackage{array,booktabs}
\newcolumntype{M}[1]{>{\centering\arraybackslash}m{#1}}

\begin{document}
	\title{The Complexity of Learning Temporal Properties}
	\maketitle	
	
	\begin{abstract}
		We consider the problem of learning temporal logic formulas from examples of system behavior. Learning temporal properties has crystallized as an effective mean to explain complex temporal behaviors. Several efficient algorithms have been designed for learning temporal formulas. However, the theoretical understanding of the complexity of the learning decision problems remains largely unexplored. To address this, we study the complexity of the passive learning problems of three prominent temporal logics, Linear Temporal Logic (LTL), Computation Tree Logic (CTL) and Alternating-time Temporal Logic (ATL) and several of their fragments. 
		We show that learning formulas using an unbounded amount of occurrences of binary operators is NP-complete for all of these logics. On the other hand, when investigating the complexity of learning formulas with bounded amount of occurrences of binary operators, we exhibit discrepancies between the complexity of learning LTL, CTL and ATL formulas (with a varying number of agents). 
	\end{abstract}

\section{Introduction}

Temporal logics are the de-facto standard for expressing temporal properties for software and cyber-physical systems.
Originally introduced in the context of program verification~\cite{pnueli77,ClarkeE81}, temporal logics are now well-established in numerous areas, including reinforcement learning~\cite{SadighKCSS14,LiVB17,CamachoIKVM19}, motion planning~\cite{FainekosKP05,DBLP:conf/aaai/CamachoTMBM17}, process mining~\cite{CecconiGCMM22}, and countless others.
The popularity of temporal logics can be attributed to their unique blend of mathematical rigor and resemblance to natural language.

Until recently, formulating properties in temporal logics has been a manual task, requiring human intuition and expertise~\cite{BjornerH14,Rozier16}.
To circumvent this step, in the past ten years, there have been numerous works to automatically learn (i.e., generate) properties in temporal logic.
Among them, a substantial number of works~\cite{flie,CamachoM19,scarlet,DBLP:conf/aaai/LuoLDWPZ22,DBLP:journals/corr/abs-2402-12373} target Linear Temporal Logic (LTL)~\cite{pnueli77}.
There is now a growing interest~\cite{DBLP:journals/corr/abs-2310-13778,DBLP:journals/corr/abs-2402-06366} in learning formulas in Computation Tree Logic (CTL)~\cite{ClarkeE81} and Alternating-time Temporal Logic (ATL)~\cite{AlurATL} due to their ability to express branching-time properties of multi-agent systems.

While existing approaches for learning temporal properties demonstrate impressive empirical performance, their computational complexity remains largely unexplored.
The only noteworthy related works are~\cite{FijalkowL21} and the follow-up work~\cite{arXivFijalkow}.
They present $\msf{NP}$-completeness results for learning formulas in LTL and several of its fragments.

In this work, we extend the existing results to encompass a wider range of LTL operators
.
Moreover, we extend the study 
to learning CTL and ATL formulas.


To elaborate on our contributions, we describe the precise problem that we consider, the fundamental \emph{passive learning} problem~\cite{Gold78}.
Its decision version asks the following question: given two sets $\mathcal{P}$, $\mathcal{N}$ of positive and negative examples of system behavior 
and a size bound $B$, does there exist a formula of size at most $B$ satisfied by the positive examples and violated by the negative examples.

Our instantiation of the above problem varies slightly depending on the considered logic
. Indeed, LTL-formulas 
express linear-time properties, CTL-formulas express branching-time properties, and ATL-formulas express properties involving on multi-agent systems
.
Accordingly, the input examples for learning LTL, CTL and ATL are linear structures (equivalently infinite words), Kripke structures and concurrent game structures, respectively. 
We refer to Section~\ref{sec:Def} for formal definitions and other prerequisites, 

\textbf{We summarize our contributions in Table~\ref{tab:summary}.}
Our first result, illustrated in the left column, shows that without any restriction on the use of binary operators, the learning problem for any logic is $\msf{NP}$-complete, regardless of the binary operators allowed. The $\msf{NP}$-hardness results are (heavily) inspired by the proofs by~\cite{arXivFijalkow} and (mostly) use reductions from the hitting set problem---one of Karp's 21 $\msf{NP}$-complete problem. The details of the proofs are given in Section~\ref{sec:hard_binary_op}.

\begin{table}[t]
	\centering
	\resizebox{0.8\linewidth}{!}{
		\begin{tabular}{|c|c|c|c|c|}
			\hline
			\multirow{ 3}{*}{} & Unbounded 
			& \multicolumn{3}{|c|}{Bounded use of binary operators} \\ \cline{3-5}
			& use of binary & \multirow{ 2}{*}{$\lX \in \Ut$} & \multicolumn{2}{|c|}{$\lX \notin \Ut$} \\ \cline{4-5}
			& operators
			& & $\{\lF,\lG\} \subseteq \Ut$ & $\Ut = \{\lF\},\{\lG\}$ \\ \hline
			$\LTL$ & \multirow{5}{*}{$\msf{NP}$-c} & \multicolumn{3}{|c|}{$\msf{L}$} \\ \cline{1-0}\cline{3-5}
			$\CTL$ &  & $\msf{NP}$-c & \multicolumn{2}{|c|}{$\msf{NL}$-c} \\ \cline{1-0} \cline{3-5} 
			$\ATL(2)$ &  & \multicolumn{2}{|c|}{$\msf{NP}$-c} & \multicolumn{1}{|c|}{$\msf{P}$-c} \\ \cline{1-0}\cline{3-5}
			$\ATL(p)$ &  & \multicolumn{3}{|c|}{$\msf{NP}$-c} \\ \cline{1-5}
	\end{tabular}}
	\caption{Summary of the complexity results for learning LTL, CTL and ATL. $\ATL(k)$ corresponds to learning ATL with $k$ agents, while $\ATL$ refers to learning ATL with the set of agents as input. $\Ut$ refers to the set of unary operators allowed.}
	\label{tab:summary}
\end{table} 

In the search for logic fragments with lower complexities, we turn to formulas using only a bounded amount of binary operators, and unary operators in a set $\Ut \subseteq \{ \neg,\lX,\lF,\lG\}$, depicted in the middle column. Note that the bound on the number of binary operators is fixed beforehand (i.e. it is part of the learning problem itself, not of the input). In this case, the complexity of the learning problems varies between different logics and unary operators. Importantly, we exhibit fragments where the learning problem is decidable in polynomial time. This is handled in Section~\ref{sec:unary}. 







Note that this work largely extends a preliminary version \cite{DBLP:journals/corr/abs-2312-11403}. 

\paragraph{Related Works.}
The closest related works are~\cite{FijalkowL21} and~\cite{arXivFijalkow}.
Both works consider learning problems in several fragments of LTL, especially involving boolean operators such as $\lor$ and $\land$, and temporal operators such as $\lX$, $\lF$ and $\lG$ and prove their $\msf{NP}$-completeness.
We extend part of their work by categorizing fragments based on the arity of the operators and studying which type of operators contribute to the hardness.
Moreover, there are several differences in the parameters considered for the learning problem.
For instance, the above works consider the size upper bound $B$ to be in binary, while we assume $B$ given in unary.
Considering size bound in unary is often justified since one may want to output a unary-sized formula anyway.
We discuss more thoroughly such differences in Section~\ref{subsubsec:discussion}.
Nonetheless, in addition to LTL, we widen the scope of the complexity results to CTL and ATL.

In the past, complexity analysis of passive learning has been studied for formalisms other than temporal logics.
For instance,~\cite{DBLP:journals/iandc/Gold78} and~\cite{DBLP:journals/iandc/Angluin78} proved $\msf{NP}$-completeness of the passive learning problems of deterministic finite automata (DFAs) and regular expressions (REs).


When considering temporal logics, most related works focus on devising efficient algorithms for learning temporal logic.
Several works learn LTL (or its important fragments) by either exploiting constraint solving~\cite{flie,CamachoM19,Riener19} or efficient enumerative search~\cite{scarlet,DBLP:journals/corr/abs-2402-12373}.
Some recent works rely on neuro-symbolic approaches to learn LTL formulas from noisy data~\cite{DBLP:conf/aaai/LuoLDWPZ22,DBLP:conf/aaai/WanLDLYP24}.
For CTL, many works resort to handcrafted templates~\cite{Chan00,WasylkowskiZ11} for simple enumerative search, while others learn formulas of aribtrary structure through constraint solving~\cite{DBLP:journals/corr/abs-2310-13778,DBLP:journals/corr/abs-2402-06366}.

There are also works on learning 
other logics such as Signal Temporal Logic~\cite{dtbombara,MohammadinejadD20}, Metric Temporal Logic~\cite{DBLP:conf/vmcai/RahaRFNP24}, Past LTL~\cite{ArifLERCT20}, Property Specification Language~\cite{0002FN20}, etc.
	
	\section*{Acknowledgement} Rajarshi Roy acknowledges partial funding by the ERC under the European
	Union's Horizon 2020 research and innovation programme
	(grant agreement No.834115, FUN2MODEL).
	
	\setcounter{tocdepth}{3}
	\tableofcontents
	
	\section{Definitions}
	\label{sec:Def}
	\paragraph{Complexity classes}
	In this paper, we are going to show several completeness results. As can be seen in Table~\ref{tab:summary}, the complexity classes that we will consider are $\msf{L}$ (logspace), $\msf{NL}$ (non-deterministic logspace), $\msf{P}$ (polynomial time), and $\msf{NP}$ (non-deterministic polynomial time). Note that all the reductions that we will define are logspace reductions.
	
	\paragraph{Some notations}
	We let $\N$ denote the set of all integers and $\N_1$ denote the set of all positive integers. Furthermore, for all $i \leq j \in \N$, we let $[i,\ldots,j] \subseteq \N$ denote the set of integers $\{ i,i+1,\ldots,j \}$.
	
	Given any non-empty set $Q$, we let $Q^*,Q^+$ and $Q^\omega$ denote the sets of finite, non-empty finite and infinite sequences of elements in $Q$, respectively. For all $\rho \in Q^+$, we denote by $|\rho| \in \N$ the length of $\rho$, i.e. its number of elements.
	
	Furthermore, for all $\bullet \in \{ +,\omega \}$, $\rho \in Q^\bullet$ and $i \in \N_1$, if $\rho$ has at least $i$ elements, we let:
	\begin{itemize}
		\item $\rho[i] \in Q$ denote the $i$-th element in $\rho$, in particular $\rho[1] \in Q$ denotes the first element of $\rho$;
		\item $\rho[:i] \in Q^+$ denote the non-empty finite sequence $\rho_1 \cdots \rho_i \in Q^+$;
		\item $\rho[i:] \in Q^\bullet$ denote the non-empty sequence $\rho_i \cdot \rho_{i+1} \cdots \in Q^\bullet$, in particular we have $\rho[1:] = \rho$.
	\end{itemize}
	
	For the remainder of this section, we fix a non-empty set of propositions $\prop$.
	
	\subsection{$\LTL$: syntax and semantics}
	Before introducing $\LTL$-formulas \cite{DBLP:conf/focs/Pnueli77}, let us first introduce the objects on which these formulas will be interpreted: infinite words. 
	
	\paragraph{Infinite and ultimately periodic words.}
	Given a set of propositions $\prop$, an infinite word is an element of the set $(2^\prop)^\omega$. Furthermore, we will be particularly interested in ultimately periodic words, and in size-1 ultimately periodic words. They are formally defined below.
	\begin{definition}
		Consider a set of propositions $\prop$. An ultimately periodic word $w \in (2^\prop)^\omega$ is such that $w = u \cdot v^\omega$ for some finite words $u \in (2^\prop)^*,v \in (2^\prop)^+$. In that case, we set the size $|w|$ of $w$ to be equal to $|w| := |u| + |v|$. Then, for all sets $S$ of ultimately periodic words, we set $|S| := \sum_{w \in S} |w|$.
		
		An ultimately periodic word $w$ is of size-1 if $|w| = 1$. That is, $w = \alpha^\omega$ for some $\alpha \in 2^\prop$. 
	\end{definition}

	\paragraph{Set of operators}
	The $\LTL$, $\CTL$ and $\ATL$-formulas that we will consider in the following will use the following temporal operators: $\lX$ (neXt), $\lF$ (Future), $\lG$ (Globally), $\lU$ (Until), $\lR$ (Release), $\lW$ (Weak until), $\lM$ (Mighty release).
	
	We let $\Op{Un}{} := \{ \neg,\lX,\lF,\lG\}$ and $\Op{Bin}{tp} := \{\lU,\lR,\lW,\lM\}$ denote the sets of unary and binary operators respectively. We also let $\Op{Bin}{lg}$ denote the set of all logical binary operators, i.e. the classical operators along with their negations: $\Op{Bin}{lg} := \{ \lor,\wedge,\Rightarrow,\Leftarrow,\Leftrightarrow,\prescript{\neg}{}{\lor},\prescript{\neg}{}{\wedge},\prescript{\neg}{}{\Rightarrow},\prescript{\neg}{}{\Leftarrow},\prescript{\neg}{}{\Leftrightarrow} \}$. 
	
	\paragraph{Syntax.}
	For all $\Ut \subseteq \Op{Un}{}$, $\Bt \subseteq \Op{Bin}{tp}$ and $\Bl \subseteq \Op{Bin}{lg}$, we denote by $\LTL(\prop,\Ut,\Bt,\Bl)$ the set of $\LTL$-formulas defined inductively as follows:
	\begin{align*}
		\varphi \Coloneqq p \mid \neg \varphi \mid *_1 \varphi \mid \varphi *_2 \varphi
	\end{align*}
	where $p\in \prop$, $*_1 \: \in \Ut$ and $*_2 \in \Bt \cup \Bl$. 
	
	
	We define the size $\Size{\varphi}$ of an $\LTL$-formula $\varphi$ to be its number of sub-formulas. That set of sub-formulas $\msf{SubF}(\varphi)$ is defined inductively as follows:
	\begin{itemize}
		\item $\msf{SubF}(p) := \{ p \}$ for all $p \in \prop$;
		\item $\msf{SubF}(\bullet \varphi) := \msf{SubF}(\varphi) \cup \{ \bullet \varphi \}$ for all unary operators $\bullet \in \Op{Un}{}$;
		\item $\msf{SubF}(\varphi_1 \bullet \varphi_2) := \msf{SubF}(\varphi_1) \cup  \msf{SubF}(\varphi_2) \cup \{ \varphi_1 \bullet \varphi_2 \}$ for all binary operators $\bullet \in \Op{Bin}{tp} \cup \Op{Bin}{lg}$.
	\end{itemize}
	We also denote by $\prop(\varphi)$ the set $\prop \cap \msf{SubF}(\varphi)$. Furthermore, we say that a set of propositions $Y \subseteq \prop$ \emph{occurs} in a formula $\varphi$ if $Y \cap \prop(\varphi) \neq \emptyset$. (This notation will also be used with $\CTL$ and $\ATL$-formulas.)
	
	Finally, we let $|\varphi|_{\msf{Bin}} \in \N$ denote the number of usage of binary operators in $\varphi$, i.e.:
	\begin{equation*}
		|\varphi|_{\msf{Bin}} := |\msf{SubBin}(\varphi)|
	\end{equation*}
	with 
	\begin{equation*}
		\msf{SubBin}(\varphi) := \{ \varphi_1 \bullet \varphi_2 \in \msf{SubF}(\varphi) \mid \varphi_1,\varphi_2 \in \msf{SubF}(\varphi), \bullet \in \Op{Bin}{tp} \cup \Op{Bin}{lg} \}
	\end{equation*}
	
	\paragraph{Semantics.}
		We define the semantics of $\LTL$-formulas. That is, given an $\LTL$-formula $\varphi$ and an infinite word $w \in (2^\prop)^\omega$, we define when $w \models \varphi$, i.e. when $\varphi$ accepts $w$, (otherwise it rejects it). In order to give a semantics to all binary logical operators, when $w \models \varphi$ is true, it is seen as the boolean value $\Tr$, and when $w \models \varphi$ is false, it is seen as the boolean value $\Fl$. Then, for all $w \in (2^\prop)^\omega$, we have (for $* \in \Bl$):
	\begin{align*}
		w \models p & \text{ iif } p \in w[0], \\
		w \models \neg \varphi & \text{ iif } w \not\models \varphi; \\
		w \models \varphi_1 * \varphi_2 & \text{ iif } (w \models \varphi_1) * (w \models \varphi_2) = \Tr; \\
		w \models \lX \varphi & \text{ iif } w[2:] \models \varphi; \\
		w \models \lF \varphi & \text{ iif } \exists i \in \N_1,\; w[i:] \models \varphi; \\
		w \models \lG \varphi & \text{ iif } \forall i \in \N_1,\; w[i:] \models \varphi; \\
		w \models \varphi_1 \lU \varphi_2 & \text{ iif } \exists i \in \N_1,\; w[i:] \models \varphi_2 \text{ and } \forall 1 \leq j \leq i-1,\; w[j:] \models \varphi_1; \\
		w \models \varphi_1 \lR \varphi_2 & \text{ iif } w \models \neg (\neg \varphi_1 \lU \neg \varphi_2) \\
		w \models \varphi_1 \lW \varphi_2 & \text{ iif } w \models (\varphi_1 \lU \varphi_2) \vee \lG\varphi_1; \\
		w \models \varphi_1 \lM \varphi_2 & \text{ iif } w \models (\varphi_1 \lR \varphi_2) \wedge \lF \varphi_1
	\end{align*}
	
	Given two $\LTL$-formulas $\varphi,\varphi'$, we write $\varphi \implies \varphi'$ when, for all ultimately periodic words $w$, we have that if $w \models \varphi$, then $v \models \varphi'$. We write $\varphi \equiv \varphi'$ if $\varphi \implies \varphi'$ and $\varphi' \implies \varphi$.
	
	Given any set $S \subseteq (2^\prop)^\omega$ of infinite words, we say that an $\LTL$-formula $\varphi$ accepts $S$ if it accepts all words in $S$, and we say that $\varphi$ rejects $S$ if it rejects all words in $S$. Furthermore, such an $\LTL$-formula $\varphi$ distinguishes two sets of infinite words $S \subseteq (2^\prop)^\omega$ and $S' \subseteq (2^\prop)^\omega$ if it accepts $S$ and rejects $S'$, or if it accepts $S'$ and rejects $S$. 
	
	\subsection{$\ATL$ and $\CTL$: syntax and semantics}
	Let us first introduce the notion of concurrent game structures (CGS), on which $\ATL$-formulas are evaluated. We then introduce Kripke structures, a special kind of concurrent game structure, on which $\CTL$-formulas are evaluated.
	\begin{definition}
		A concurrent game structure (CGS for short) is a the tuple $C = \langle Q,I,k,\prop, \pi, d,\delta \rangle$ where,
		\begin{itemize}
			\item $Q$ is the finite set of states;
			\item $I \subseteq Q$ is the set of initial states;
			\item $k \in \N$ denotes the number of agents, we denote by $\Ag := [1,\ldots,k]$ the set of $k$ agents;
			\item $\pi: Q\mapsto 2^{\prop}$ maps each state $s\in Q$ to the set of propositions that hold in $s$;
			\item $d: Q \times \Ag \rightarrow \mathbb{N}^+$ maps each state and agent to the number of actions available to that agent at that state;
			\item $\delta: Q_\msf{Act} \rightarrow Q$ is the function mapping every state and tuple of one action per agent to the next state, where $Q_\msf{Act} := \{ (q,\alpha_1,\ldots,\alpha_k) \mid q \in Q,\; \forall a \in \Ag,\; \alpha_a \in [1,\ldots,d(q,a)] \}$. 
		\end{itemize}
		For all states $q \in Q$ and coalitions of agents $A \subseteq \Ag$, we let $\msf{Act}_A(q) := \{ \alpha = (\alpha_a)_{a \in A} \mid \forall a \in \Ag,\; \alpha_a \in [1,\ldots,d(q,a)] \}$. Then, for all tuple $\alpha = (\alpha_a)_{a \in A} \in \msf{Act}_A(q)$ of one action per agent in $A$, we let: 
		\begin{equation*}
			\msf{Succ}(q,\alpha) := \{ q' \in Q \mid \exists \alpha' = (\alpha'_a)_{a \in \Ag \setminus A} \in \msf{Act}_{\Ag \setminus A}(q),\; \delta(q,(\alpha,\alpha')) = q' \}
		\end{equation*}
		Finally, the size $|C|$ of the concurrent structure $C$ is equal to: $|C| = |Q_\msf{Act}| + |\prop|$.
	\end{definition}
	Unless otherwise stated, a concurrent game structure $C$ will always refer to the tuple $\langle Q,I,k,\prop, \pi, d,\delta \rangle$. 

	A Kripke structure is then a concurrent game structure where there is only one agent.
	\begin{definition}
		A \emph{Kripke structure} is a concurrent game structure $C$ where $k = 1$ and $d$ and $\delta$ are replaced by subsets of successor states $\emptyset \neq \msf{Succ}(q) \subseteq Q$ for all states $q \in Q$.
	\end{definition}
	Unless otherwise state, a Krikpe structure $C$ will will always refer to the tuple $\langle Q,I,k,\prop, \pi, \msf{Succ} \rangle$. 
	
	In a concurrent game structure, a strategy for an agent is a function that prescribes to the agent what to do as a function of the history of the game, i.e. of the finite sequence of states seen so far. Furthermore, given a coalition of agents and a tuple of one strategy per agent in the coalition, we define the set of infinite sequences of states that can occur with this tuple of strategies, from any state. This is formally defined below.
	\begin{definition}
		Consider a concurrent game structure $C$ and a agent $a \in [1,\ldots,k]$. A \emph{strategy} for agent $a$ is a function $s_a: Q^+ \rightarrow \N_1$ such that, for all $\rho = \rho_0 \dots \rho_n \in Q^+$, we have $s_a(\rho) \leq d(\rho_n,a)$. We denote by $\msf{S}_a$ the set of strategies available to Agent $a \in [1,\ldots,k]$.
		
		Given any coalition (or subset) of agents $A \subseteq [1,\ldots,k]$, a \emph{strategy profile} for the coalition $A$ is a tuple $s = (s_a)_{a \in A}$ of one strategy per agent in $A$. We denote by $\msf{S}_A$ the set of strategy profiles for the coalition $A$. Given any such strategy profile $s$, for all states $q \in Q$, we let $\msf{Out}^Q(q,s) \subseteq Q^\omega$ denote the set of infinite paths $\rho$ that are compatible with the strategy profile $s$ from $q$, i.e.:
		\begin{equation*}
			\msf{Out}^Q(q,s) := \{ \rho \in Q^\omega \mid \rho[1] = q,\; \forall i \in \N_1, \rho[i+1] \in \msf{Succ}(\rho[i],(s_a(\rho[:i]))_{a \in A}) \}
		\end{equation*}
	
		In a Kripke structure, we simply consider the set of all infinite paths $\msf{Out}^Q(q)$ that can occur from a specific state $q \in Q$, regardless of strategies: 
		\begin{equation*}
			\msf{Out}^Q(q) := \{ \rho \in Q^\omega \mid \rho[1] = q,\; \forall i \in \N_1, \rho[i+1] \in \msf{Succ}(\rho[i]) \}
		\end{equation*}
	\end{definition}
	
	\paragraph{Syntax.}
	To define the syntax of $\ATL$ and $\CTL$-formulas, we introduce two types of formulas: state formulas and path formulas. Intuitively, state formulas express properties of states, where the strategic quantifier occurs, whereas path formulas express temporal properties of paths. For ease of notation, we denote state formulas and path formulas with the Greek letter $\phi$ and the Greek letter $\psi$, respectively. Consider some $\Ut \subseteq \Op{Un}{}$, $\Bt \subseteq \Op{Bin}{tp}$ and $\Bl \subseteq \Op{Bin}{lg}$ and $k \in \N_1$. Then, we denote by $\ATL^k(\prop,\Ut,\Bt,\Bl)$ the set of $\ATL$-state formulas defined by the grammar:
	\begin{align*}
		\phi \Coloneqq p \mid \neg \phi \mid \phi * \phi \mid \fanBr{A} \psi
	\end{align*}
	where $p\in \prop$, $* \in \Bl$, $A \subseteq \Ag := \{ 1,\ldots,k\}$ is a subset of agents, and $\psi$ is an $\ATL^k(\prop,\Ut,\Bt,\Bl)$-path formula. Note that, for $\CTL$-formulas, we have $k = 1$. Hence, there are only two possible subsets: $\emptyset$ and $\Ag$ itself. Usually, in the $\CTL$ syntax, $\fanBr{\emptyset} \psi$ is denoted $\forall \psi$ whereas $\fanBr{\Ag} \psi$ is denoted $\exists \psi$. 
		
	Next, $\ATL^k(\prop,\Ut,\Bt,\Bl)$-path formulas are given by the grammar
	\begin{align*}
		\psi \Coloneqq *_1 \phi \mid \phi *_2 \phi 
	\end{align*}
	where $*_1 \in \Ut \setminus \{\neg\}$ and $*_2 \in \Bt$. We denote by $\ATL^k$ the set of all $\ATL^k$-formulas, with $\CTL$ referring to $\ATL^1$. The set of sub-formulas $\msf{SubF}(\phi)$ of a formula $\phi$ is then defined inductively as follows:
	\begin{itemize}
		\item $\msf{SubF}(p) := \{ p \}$ for all $p \in \prop$;
		\item $\msf{SubF}(\neg \phi) := \{ \neg \phi \} \cup \msf{SubF}(\phi)$;
		\item $\msf{SubF}(\fanBr{A} (\bullet \phi)) := \{ \fanBr{A} (\bullet \phi) \} \cup \msf{SubF}(\phi)$ for all $\bullet \in \Op{Un}{}$ and $A \subseteq \Ag$;
		\item $\msf{SubF}(\phi_1 \bullet \phi_2) := \{ \phi_1 \bullet \phi_2 \} \cup \msf{SubF}(\phi_1) \cup \msf{SubF}(\phi_2)$ for all $\bullet \in \Op{Bin}{lg}$;
		\item $\msf{SubF}(\fanBr{A} (\phi_1 \bullet \phi_2)) := \{ \fanBr{A} (\phi_1 \bullet \phi_2) \} \cup \msf{SubF}(\phi_1) \cup \msf{SubF}(\phi_2)$ for all $\bullet \in \Op{Bin}{tp}$  and $A \subseteq \Ag$.
	\end{itemize}
	
	The size $|\phi|$ of a $\ATL$-formula is defined as its number of sub-formulas: $|\phi| := |\msf{SubF}(\phi)|$. Finally, we also let $\msf{NbBin}(\phi) \in \N$ denote the number of usage of binary operators in $\phi$, i.e.:
	\begin{equation*}
		|\phi|_{\msf{Bin}} := |\msf{SubBin}(\phi)|
	\end{equation*} 
	with 
	\begin{equation*}
		\msf{SubBin}(\phi) := \{ \phi_1 \bullet \phi_2 \in \msf{SubF}(\phi) \mid \phi_1,\phi_2 \in \msf{SubF}(\phi), \bullet \in \Op{Bin}{tp} \cup \Op{Bin}{lg} \}
	\end{equation*}
	
	\paragraph{Semantics.}
	As mentioned above, we interpret 
	$\ATL^k$-formulas over CGS with the set of agents $\Ag := \{1,\ldots,k\}$ using the standard definitions 
	\cite{AlurATL}. Given a state $s$ and a state formula~$\phi$, we define when $\phi$ holds in state $s$, denoted using $s\models \phi$, inductively as follows:
	\begin{align*}
		s \models p & \text{ iif } p \in \pi(s), \\
		s \models \neg \phi & \text{ iif } s \not\models \phi, \\
		s\models \phi_1 * \phi_2 & \text{ iif } (s\models\phi_1) * (s\models\phi_2) = \Tr, \\
		s\models \fanBr{A} \psi & \text{ iif }
		\exists s \in \msf{S}_A,\; \forall \pi\in \msf{Out}^Q(q,s),\; \pi\models\psi
	\end{align*}
	where $* \in \Op{Bin}{lg}$. In a Kripke structure, this last line cam be rewritten as follows (where the first line corresponds to $A = \emptyset$ and the second to $A = \Ag$):
	\begin{align*}
		s\models \exists \psi & \text{ iif }
		\exists \pi\in \msf{Out}^Q(q),\; \pi\models\psi \\
		s\models \forall \psi & \text{ iif }
		\forall \pi\in \msf{Out}^Q(q),\; \pi\models\psi
	\end{align*}

	Furthermore, given a path $\pi$ and a path formula $\psi$, we define when $\psi$ holds for the path $\pi$, also denoted using $\pi\models \phi$, inductively as follows:
	\begin{align*}
		\pi \models \lX \phi & \text{ iif } \pi[2:] \models \phi; \\
		\pi \models \lF \phi & \text{ iif } \exists i \in \N_1,\; \pi[i:] \models \phi; \\
		\pi \models \lG \phi & \text{ iif } \forall i \in \N_1,\; \pi[i:] \models \phi; \\
		\pi \models \phi_1 \lU \phi_2 & \text{ iif } \exists i \in \N_1,\; \pi[i:] \models \phi_2 \text{ and } \forall 1 \leq j \leq i-1,\; \pi[j:] \models \phi_1; \\
		\pi \models \phi_1 \lR \phi_2 & \text{ iif } \pi \models \neg (\neg \phi_1 \lU \neg \phi_2) \\
		\pi \models \phi_1 \lW \phi_2 & \text{ iif } \pi \models (\phi_1 \lU \phi_2) \vee \lG\phi_1; \\
		\pi \models \phi_1 \lM \phi_2 & \text{ iif } \pi \models (\phi_1 \lR \phi_2) \wedge \lF \phi_1
	\end{align*}
	We now say that an $\ATL$-formula $\phi$ holds on a CGS $C$, denoted by $C\models \phi$, if $s\models\phi$ for all initial states $s\in I$ of $C$.
	
	We use the notations $\implies$ and $\equiv$ as for $\LTL$-formulas.
	
	\subsection{Decision problems}
	We define the $\LTL$, $\CTL$ and $\ATL^k$ learning problems below (for $k \in \N_1$), where a model for $\LTL$ is an ultimately periodic word, a model for $\CTL$ is a Kripke structure and a model for $\ATL^1$ is a concurrent game structure on the set of agents $\Ag := \{1,\ldots,k\}$.
	\begin{definition}
		\label{def:learning_problems}
		Let $\msf{L} \in \{ \LTL,\CTL,\ATL^k \mid k \in \N_1 \}$ and consider some sets of operators $\Ut \subseteq \Op{Un}{}$, $\Bt \subseteq \Op{Bin}{tp}$ and $\Bl \subseteq \Op{Bin}{lg}$. For $n \in \N \cup \{ \infty \}$, we denote by $\msf{L}_\msf{Learn}(\Ut,\Bt,\Bl,n)$ the following decision problem:
		\begin{itemize}
			\item Input: $(\prop,\mathcal{P},\mathcal{N},B)$ where $\prop$ is a set of propositions, $\mathcal{P},\mathcal{N}$ are two finite sets of models for $\msf{L}$, and $B \in \N$.
			\item Output: yes iff there exists an $\msf{L}$-formula $\varphi \in \msf{L}(\prop,\Ut,\Bt,\Bl)$ such that $|\varphi| \leq B$, $|\varphi|_{\msf{Bin}} \leq n$ and $\varphi$ separates $\mathcal{P}$ and $\mathcal{N}$, i.e. such that:
			\begin{itemize}
				\item for all $X \in \mathcal{P}$, we have $X \models \varphi$;
				\item for all $X \in \mathcal{N}$, we have $X \not\models \varphi$.
			\end{itemize}
		\end{itemize}
		The size of the input is equal to $|\prop| + |\mathcal{P}| + |\mathcal{N}| + B$ (i.e. $B$ is written in unary).

	\end{definition}

	As mentioned in the introduction), since the model checking problems for $\LTL$, $\CTL$ and $\ATL$ can be decided in polynomial time \cite{AlurATL}
	, 
	the problems $\LTL_\msf{Learn}$, $\CTL_\msf{Learn}$ and $\ATL^k_\msf{Learn}$ are all in $\msf{NP}$, using a straightforward guess-and-check subroutine. 
	
	\begin{proposition}
		\label{prop:easy}
		For all $\Ut \subseteq \Op{Un}{}$, $\Bt \subseteq \Op{Bin}{tp}$, $\Bl \subseteq \Op{Bin}{lg}$ and $n \in \N \cup \{\infty\}$, the decision problems $\LTL_\msf{Learn}(\Ut,\Bt,\Bl)$,  $\CTL_\msf{Learn}(\Ut,\Bt,\Bl)$, and
		$\ATL^k_\msf{Learn}(\Ut,\Bt,\Bl,k)$ for all $k \in \N_1$ are in $\msf{NP}$.
	\end{proposition}
	
	\section{Learning with non-unary binary operators is $\msf{NP}$-hard}
	\label{sec:hard_binary_op}
	In this section, we study the complexity of the learning decision problems in the case where the number of occurrence of binary operators is unbounded. 
	We first show that, in this setting, $\LTL$ learning is $\msf{NP}$-hard. We then show that $\CTL$ learning is at least as hard as $\LTL$ learning, which, in turn, implies that $\CTL$ learning with unbounded occurrence of binary operators is $\msf{NP}$-hard. This holds as well for $\ATL^k$ learning, for all $k \in \N$. 
	
	\subsection{$\LTL$ learning}
	Let us consider $\LTL$ learning. Before we proceed to our contributions, let us discuss the very important related work \cite{arXivFijalkow}.
	
	\subsubsection{What is done in \cite{arXivFijalkow}}
	\label{subsubsec:discussion}
	In \cite{arXivFijalkow}, the authors study $\LTL$ learning. However, the setting that they consider differs in several ways with the setting that we consider in this paper. We list the main differences below.
	\begin{itemize}
		\item The letters of the words considered in \cite{arXivFijalkow} are propositions (i.e. elements of $\prop$) whereas the letters that we consider are subsets of propositions (i.e. elements of $2^\prop$).
		\item The words that we consider are infinite whereas the words considered in \cite{arXivFijalkow} are finite.
		\item Crucially for complexity questions, the alphabet (as it is referred to in \cite{arXivFijalkow}, i.e. the set of propositions $\prop$) that we consider is part of the input, it is not fixed beforehand.
		\item The bound $B$ that we consider is written in unary instead of binary.
	\end{itemize}
	We discuss below the motivations behind our choices, and the implications they have on the complexity of the learning problems.
	
	\paragraph{Letters.}
	Using subsets of proposition as letters is not unusual, this is done for instance in the seminal book \cite{DBLP:books/daglib/0020348}. Furthermore, in this paper, we focus on comparing $\LTL$, $\CTL$ and $\ATL$ learning. In that regard, we aim at having settings for all of these logics as close as possible. Hence, we use subsets of propositions as letters, since the states of Kripke structures are usually labeled with subsets of propositions (and the same goes for concurrent game structures). Note that this choice has significant impact on the complexity of the learning problem considered. Indeed, in our setting, letters are subsets of propositions. In combination with having the set of propositions as part of the input --- which we discuss below --- we show, in particular, that $\LTL$ learning with only operator $\lor$ is $\msf{NP}$-hard. On the other hand, in the setting of \cite{arXivFijalkow}, the authors show that with only the operator $\lor$, the $\LTL$ learning problem can de decided in polynomial time (Proposition 7). 
	
	\paragraph{Word length.}
	In this paper, we consider infinite words. For representation issues, we focus on ultimately periodic ones. As above, this is more closely related to $\CTL$ and $\ATL$ semantics. It has little influence on complexity questions, although some differences may arise. For instance, 
	as stated in \cite[Proposition 8]{arXivFijalkow}, an $\LTL$-formula of the shape $\varphi := \lX^k \varphi'$ for some $k \in \N$ is always false when evaluated on words of size at most $k-1$, which is irrelevant for us since all words have infinite size. 
	
	\paragraph{Alphabet.}
	Let us now consider this more complicated issue: is it better to have the alphabet (or set of propositions) fixed a priori or to have it part of the input? As is mentioned in \cite{arXivFijalkow}, it is much more usual to have the alphabet fixed a priori, as in the classical examples of automata learning (
	\cite{Gold78}). However, we believe that in a learning setting, it makes sense not to know a priori the propositions occurring in the model. That way, the set of propositions could be learned by looking at the models, positive or negative, that we need to separate\footnote{Note that in our setting the set of propositions is given as an explicit part of the input. It could alternatively be given implicitly in the input as in the can read in the models.}. 
	
	This choice is important in terms of complexity. Indeed, consider the very technical results of \cite{arXivFijalkow}, in Section 7 and especially in Section 8, which deal with $\msf{NP}$-hardness results (and hard-to-approximate results) with and without the next operator $\lX$. These results are especially hard to prove because the set of propositions is fixed a priori (and is of size 3). On the other hand, our proofs of $\msf{NP}$-hardness for $\LTL$ learning are significantly easier than, for instance, the proof of Theorem 9 from \cite{arXivFijalkow}. However, note that the distinction between having the set of propositions part of the input or not becomes less relevant when we restrict ourselves to formulas (for $\LTL$, $\CTL$ and $\ATL$ learning) with bounded occurrence of binary operators, since in that case our $\msf{NP}$-hardness proofs hold even for formula that do not use at all binary operators, and therefor use exactly one proposition. 
	
	
	\paragraph{Bound.}
	Finally, there is the issue of the representation of the integer $B$ that bounds the size of the formulas that we consider. In this paper, we consider the case where it is given in unary. Indeed, although we consider only decision problems where we only answer whether or not there exists a formula, it would also be interesting to explicitly synthesize formulas separating negative and positive instances. If the bound were written in binary, explicitly writing the formula could be exponential in the size of the input. Additionally, having the bound written in unary allows us to have completeness results, and not only hardness ones. 
	
	\paragraph{Results from \cite{arXivFijalkow}}
	Let us now briefly discuss the $\msf{NP}$-hardness results of \cite{arXivFijalkow}. Letting $\msf{Op}$ denote the operators allowed in $\LTL$-formulas, they show that the following $\LTL$ learning problems are $\msf{NP}$-hard:
	\begin{itemize}
		\item With alphabet part of the input:
		\begin{itemize}
			\item When $\msf{Op} = \{ \lF,\lor \}$ (Theorem 2)
			\item When $\msf{Op} = \{ \lF,\wedge \}$ (Theorem 8), and the learning problem is hard to approximate;
		\end{itemize}
		\item With alphabet not part of the input:
		\begin{itemize}
			\item When $\{\lX,\wedge,\lor\} \subseteq \msf{Op} \subseteq \{\wedge,\lor,\lX,\lF,\lG\}$ (Theorem 6, Proposition 11, Proposition 12);
			\item When $\{\lF,\wedge\} \subseteq \msf{Op} \subseteq \{\wedge,\lor,\lF,\lG,\neg\}$ (Theorem 9);
			\item When $\{\lG,\lor\} \subseteq \msf{Op} \subseteq \{\wedge,\lor,\lF,\lG,\neg\}$ (Theorem 10).
		\end{itemize}
	\end{itemize}
	
	Note that, in a previous version of this paper, independently of \cite{arXivFijalkow}, we have shown that the $\LTL$ learning problem is $\msf{NP}$-hard for $\{\lor\} \subseteq \msf{Op} \subseteq \{\wedge,\lor,\Rightarrow,\Leftrightarrow,\lX,\lF,\lG,\lU,\lR,\lW,\lM\}$\footnote{We present this way our result to mimic \cite{arXivFijalkow}, however in this paper we distinguish between the logical and temporal operators that we consider.} (with the setting used in this paper). The reduction was established from the satisfiability problem $\msf{SAT}$. This made the proof quite convoluted as, from a positive instance of the learning problem, we had to be able to extract to a satisfying valuation of the variables. On the other hand, in \cite{arXivFijalkow}, the authors often use the hitting set problem (one of Karp's 21 $\msf{NP}$-complete problem) for their reductions. We define it below. 
	\begin{definition}[Hitting set problem]
		We denote by $\msf{Hit}$ the following decision problem:
		\begin{itemize}
			\item Input: $(l,C,k)$ where $l \in \N_1$, $C = C_1,\ldots,C_n$ are non-empty subsets of $[1,\dots,l]$ such that $\cup_{1 \leq i \leq n} C_i = [1,\ldots,l]$ and $1 \leq k \leq l$.
			\item Output: yes iff there is a subset $H \subseteq [1,\ldots,l]$ of size at most $k$ such that, for all $1 \leq i \leq n$, we have $H \cap C_i \neq \emptyset$.
		\end{itemize}
	\end{definition}
	In the following, unless otherwise stated, when we use $(l,C,k)$ as an instance of the hitting set problem, $C$ refers to $C = C_1,\ldots,C_n$, for some $n \in \N_1$.
	
	\begin{theorem}[\cite{DBLP:conf/coco/Karp72}]
		\label{thm:hitting_set_problem}
		The hitting set problem is $\msf{NP}$-hard.
	\end{theorem}
	\begin{observation}
		The above theorem holds even if $k$ is given in unary. This comes from the fact that $k \leq l$ and $\cup_{1 \leq i \leq n} C_i = [1,\ldots,l]$.
	\end{observation}
	
	This problem is much better suited for establishing the $\msf{NP}$-hardness of the learning problem decision that we consider. Indeed, it is sufficient to be able to exhibit a set of integers, not a valuation on variables. In particular, it makes it easier to handle various logical operators. In this paper, all the $\msf{NP}$-hardness proofs that we exhibit in this paper but one are established from the hitting set problem\footnote{Except for the proof that $\CTL$ and $\ATL$ learning with any non-unary binary operator is $\msf{NP}$-hard, which is proved via a reduction to the $\LTL$ case}. Note that some of the $\msf{NP}$-hardness proofs for the $\LTL$ learning case are very close to some proof from \cite{arXivFijalkow}. We discuss it in details below. 
	
	\subsection{Our results}

	The goal of this subsection is to show the theorem below:
	\begin{theorem}
		\label{thm:non_unary_binary_NP_hard}
		Consider some $\Bt \subseteq \Op{Bin}{tp}$, and $\Bl \subseteq \Op{Bin}{lg}$ and assume that $\Bl \cup \Bt \neq \emptyset$. Then, for all $\Ut \subseteq \Op{Un}{}$, the decision problem $\LTL_\msf{Learn}(\Ut,\Bt,\Bl,\infty)$ is $\msf{NP}$-hard.
	\end{theorem}

	Overall, there are fourteen different binary operators that we will handle, the ten logical operators $\{ \lor,\Rightarrow,\Leftarrow,\wedge,\prescript{\neg}{}{\Rightarrow},\prescript{\neg}{}{\Leftarrow},\prescript{\neg}{}{\lor},\prescript{\neg}{}{\wedge},\Leftrightarrow,\prescript{\neg}{}{\Leftrightarrow} \}$ and the four temporal operators $\{ \lU,\lR,\lW,\lM \}$. We give below a bird's eye view of how the proof of Theorem~\ref{thm:non_unary_binary_NP_hard} is structured.
	\begin{itemize}
		\item We start with the operators $\lor,\Rightarrow,\Leftarrow$, i.e. we assume that $\Bl \cap \{ \lor,\Rightarrow,\Leftarrow \} \neq \emptyset$, and we show that for all $\Bt \subseteq \Op{Bin}{tp}$ and $\Ut \subseteq \Op{Un}{}$, the decision problem $\LTL_\msf{Learn}(\Ut,\Bt,\Bl,\infty)$ is $\msf{NP}$-hard. This is stated in Corollary~\ref{coro:or}. The reduction for this case is actually a straightforward adaptation of the proof of \cite[Theorem 2]{arXivFijalkow} (that additionally makes use of the fact that our letters are subsets of propositions). In fact, the operators $\prescript{\neg}{}{\Rightarrow},\prescript{\neg}{}{\Leftarrow},\prescript{\neg}{}{\lor}$ are handled at the same time. The reduction for these operators is obtained from the previous one by reversing the positive and negative sets of words.
		\item We then handle the operators $\prescript{\neg}{}{\lor},\prescript{\neg}{}{\wedge}$ (Corollary~\ref{coro:not_or}). The reduction used for the previous item cannot be used as is, because when the operator $\prescript{\neg}{}{\wedge}$ (or the operator $\prescript{\neg}{}{\lor}$) is used successively, the formula obtained is semantically equivalent to an alternation of conjunction and disjunction. An illustrating example is given in Example~\ref{ex:successively_using_not_lor}. To circumvent this difficulty, we define a reduction that is both slightly more subtle and an adaptation of the previous one. 
		\item Before considering the last two logical operators $\Leftrightarrow,\prescript{\neg}{}{\Leftrightarrow}$, we handle the temporal operators $\lW,\lM$ (Corollary~\ref{coro:W_M}). This is actually quite straightforward. Indeed, the two previous reductions only use size-1 infinite words (i.e. a subset of propositions is repeated indefinitely). On such words, the temporal operators $\lW,\lM$ actually behave like $\lor,\wedge$ respectively. Hence, we can use the reduction of the first item.
		\item We then handle the final two logical operators $\Leftrightarrow,\prescript{\neg}{}{\Leftrightarrow}$ (Corollary~\ref{coro:equiv}). These operators behave quite differently from all other operators. In fact,  in this case the reduction is not established from the hitting set problem, but from an $\msf{NP}$-complete problem dealing with modulo-2 calculus (see Definition~\ref{def:coset_weight}), although it still uses only size-1 infinite words. Contrary to the other reductions, here we explicitly assume that the only non-binary operators considered are $\Leftrightarrow$ and $\prescript{\neg}{}{\Leftrightarrow}$.
		\item Finally, we handle the last two operators: the temporal operators $\lU$ and $\lR$ (Corollary~\ref{coro:U_R}). Contrary to the temporal operators $\lW$ and $\lM$, on size-1 words, $\lU$ and $\lR$ are equivalent to unary binary operators. Hence, the reduction that we consider does not use only size-1 infinite words. It is once again established from the hitting set problem, however the construction is more involved than the reductions of the two first items.
	\end{itemize}

	Overall, note that we provide in this subsection a reduction for all possible binary operators to properly justify the statement: \textquotedblleft if the occurrence of any binary operator is unbounded, the $\LTL$ learning decision problem is $\msf{NP}$-hard\textquotedblright{}. However, binary operators are not all equivalently relevant, and it may be that some binary operator are not relevant at all when considered alone. This seems to be particularly the case for the operators $\Leftrightarrow,\prescript{\neg}{}{\Leftrightarrow}$.

	\subsubsection{Two useful lemmas}	
	Before we handle all possible operators as described above, we first state and prove two lemmas that we will use extensively in this subsection. 
	
	First, we state a lemma that establishes a condition on the set of propositions occurring in an $\LTL$-formula that distinguishes a pair of infinite words.
	\begin{lemma}
		\label{lem:distinguish}
		Consider a subset of propositions $Y \subseteq \prop$ and two infinite words $w_1,w_2 \in (2^\prop)^\omega$. Assume that, for all $i \in \N_1$ and $x \in \prop \setminus Y$, we have $x \in w_1[i]$ if and only if $x \in w_2[i]$. Then, if the formula $\varphi$ distinguishes the words $w_1$ and $w_2$, the set $Y$ occurs in $\varphi$.
	\end{lemma}
	\begin{proof}
		Let us prove this property $\mathcal{P}(\varphi)$ on $\LTL$-formulas $\varphi$ by induction. Consider such an $\LTL$-formula $\varphi$:
		\begin{itemize}
			\item Assume that $\varphi = x$ for some $x \in \prop$. Then, if $\varphi$ distinguishes $w_1$ and $w_2$, it must be that $x \in Y$, by assumption.
			\item Assume that $\varphi = \neg \varphi'$ and $\mathcal{P}(\varphi')$ holds. Then, if $\varphi$ distinguishes $w_1$ and $w_1$, so does $\varphi'$. Hence, $\mathcal{P}(\varphi)$ follows.
			\item Assume that $\varphi = \varphi_1 * \varphi_2$ for some $* \in \Op{Bin}{lg}$ and that $\mathcal{P}(\varphi_1)$ and $\mathcal{P}(\varphi_2)$ hold. Then, if neither $\varphi_1$ nor $\varphi_2$ distinguish $w_1$ and $w_2$, neither does $\varphi$. Hence, assuming that $\varphi$ distinguishes $w_1$ and $w_2$, then $Y$ occurs in $\varphi_1$ or $\varphi_2$, and it therefore also occurs in $\varphi$. Hence, $\mathcal{P}(\varphi)$ holds.
			\item Assume that $\varphi = \bullet \varphi'$ for some $\bullet \in \Op{Un}{} \setminus \{\neg\}$ and $\mathcal{P}(\varphi')$ holds. Assume that for all $i \in \N_1$, $\varphi'$ does not distinguish $w_1[i:]$ and $w_2[i:]$. Then, $\varphi'$ does not distinguish $w_1[2:]$ and $w_2[2:]$, and $\lX \varphi'$ does not distinguish $w_1$ and $w_2$. Furthermore, there is some $i \in \N_1$ such that $\varphi'$ accepts $w_1[i:]$ iff there is some $i \in \N_1$ such that $\varphi'$ accepts $w_2[i:]$. That is, $\lF \varphi'$ does not distinguish $w_1$ and $w_2$. This is similar for $\lG \varphi'$. Hence, if $\varphi$ distinguishes $w_1$ and $w_2$, $Y$ occurs in $\varphi'$, and also in $\varphi$.
			\item The case of binary temporal operators is similar.
		\end{itemize}
	\end{proof}

	In addition, we show the straightforward relation between the number of propositions occurring in an $\LTL$-formula and the size of that $\LTL$-formula. 
	\begin{lemma}
		\label{lem:at_least_that_many_subformulas}
		For all $Y \subseteq \prop$, we let $\msf{Occ}(Y),\mathsf{NbSubF}(Y)$ denote the properties on $\LTL$-formulas such that an $\LTL$-formula $\varphi$ satisfies:
		\begin{itemize}
			\item $\msf{Occ}(Y)$ if all variables in $Y$ occur in $\varphi$, i.e. $Y \subseteq \prop(\varphi)$;
			\item $\mathsf{NbSubF}(Y)$ if there are at least $2|Y|-1$ different sub-formulas of $\varphi$ where $Y$ occurs.
		\end{itemize}
		
		Then, an $\LTL$-formula $\varphi$ on $\prop$ that satisfies $\msf{Occ}(Y)$ also satisfies $\mathsf{NbSubF}(Y)$.
	\end{lemma}
	\begin{proof}
		Let us show this lemma by induction on $\LTL$-formulas $\varphi$:
		\begin{itemize}
			\item Assume that $\varphi = x$. Then, $\varphi$ satisfies $\msf{Occ}(Y)$ only for $Y = \emptyset$ and $Y = \{ x \}$, and it also satisfies $\mathsf{NbSubF}(\emptyset)$ and $\mathsf{NbSubF}(\{x\})$;
			\item for all $\bullet \in \Op{Un}{}$, assume that $\varphi = \bullet \varphi'$. Consider any $Y \subseteq \prop$. If $\varphi$ satisfies $\msf{Occ}(Y)$, then so does $\varphi'$. Hence, $\varphi'$ satisfies $\mathsf{NbSubF}(Y)$, and therefore so does $\varphi$;
			\item for all $\bullet \in \Op{Bin}{lg} \cup \Op{Bin}{tp}$, assume that $\varphi = \varphi_1 \bullet \varphi_2$. Consider any $Y \subseteq \prop$ and assume that $\varphi$ satisfies $\msf{Occ}(Y)$. Let $Y_1 := \prop(\varphi_1) \cap Y \subseteq Y$ denote the set of different propositions in $Y$ occurring in $\varphi_1$ (they may also occur in $\varphi_2$). Let also $Y_2' := (\prop(\varphi_2) \setminus \prop(\varphi_1)) \cap Y \subseteq Y$ denote the set of different propositions in $Y$  occurring in $\varphi_2$ and not in $\varphi_1$. We have $Y = Y_1 \cup Y_2'$. Furthermore:
			\begin{itemize}
				\item by our induction hypothesis on $\varphi_1$, there are at least $k_1 := 2|Y_1|-1$ sub-formulas in $\msf{SubF}(\varphi_1) \subseteq \msf{SubF}(\varphi)$ where $Y_1$ occurs;
				\item by our induction hypothesis on $\varphi_2$, there are also at least $k_2' := 2|Y_2'|-1$ sub-formulas in $\msf{SubF}(\varphi_2) \subseteq \msf{SubF}(\varphi)$ where $Y'_2$ occurs. By definition of $Y'_2$, it follows that all these sub-formulas are not sub-formulas of $\varphi_1$;
				\item Finally, the sub-formula $\varphi$ itself is a sub-formula of $\varphi$ that is neither a sub-formula of $\varphi_1$ nor of $\varphi_2$ and where $Y$ occurs.
			\end{itemize}
			Therefore, there are at least $k_1 + k_2' + 1 = 2|Y_1|-1 + 2|Y'_2|-1 + 1 \geq 2|Y|-1$ different sub-formulas of $\varphi$ where $Y$ occurs. That is, $\varphi$ satisfies $\msf{NbSubF}(Y)$.
		\end{itemize}
	\end{proof}

	\subsubsection{Proof of Theorem~\ref{thm:non_unary_binary_NP_hard}: when $\Bl \cap \{ \lor,\Rightarrow,\Leftarrow,\wedge,\prescript{\neg}{}{\Rightarrow},\prescript{\neg}{}{\Leftarrow} \} \neq \emptyset$}
	We present a first reduction that we will consider for the operators $\lor,\Rightarrow,\Leftarrow$, along with the dual reduction for the operators $\wedge,\prescript{\neg}{}{\Rightarrow},\prescript{\neg}{}{\Leftarrow}$. This is obtained via a slight modification of the reduction presented in \cite{arXivFijalkow} to establish Theorem 2. We make use of the fact that the letters we consider are subsets of propositions instead of being a single proposition.

	\begin{definition}
		\label{def:reduction_or_easy}
		Consider an instance $(l,C,k)$ of the hitting set problem $\msf{Hit}$. We define:
		\begin{itemize}
			\item $\prop := \{ a_j,b_j \mid 1 \leq j \leq l \}$ to be the set of propositions;
			\item $\msf{Set} := \{ \alpha_i \mid 1 \leq i \leq n \}$ where for all $1 \leq i \leq n$: we let $\alpha_i := \{ c_i^1,\ldots,c_i^l \}^\omega \in (2^\prop)$ with, for all $1 \leq j \leq l$:
			\begin{align*}
			c_i^j := \begin{cases}
			a_j & \text{ if }j \in C_i \\
			b_j & \text{ if }j \notin C_i \\
			\end{cases}
			\end{align*}
			\item $\msf{EmptySet} := \{ \beta \}$ with $\beta := \{ b_1,\dots,b_l \}^\omega \in (2^\prop)$;
			\item $B = 2 k - 1$.
		\end{itemize}
		Then, we define the inputs $\msf{In}^\lor_{(l,C,k)} := (\prop,\msf{Set},\msf{EmptySet},B)$ and $\msf{In}^\wedge_{(l,C,k)} := (\prop,\msf{EmptySet},\msf{Set},B)$.
	\end{definition}

	The positive and negative words that we have defined above satisfy the observation below.
	\begin{observation}
		\label{obs:equiv_on_words_reduc}
		For all $w \in \msf{EmptySet} \cup \msf{Set}$, we have:
		\begin{equation*}
			\label{eqn:equiv_on_words}
			\forall 1 \leq j \leq l,\; a_j \models w \Leftrightarrow \neg b_j \models w
		\end{equation*}
	\end{observation}

	Let us describe on an example below what this reduction amounts to on a specific instance of the hitting set problem.
	\begin{example}
		\label{ex:1}
		Assume that $l = 4$, $C = C_1,C_2$ with $C_1 := \{ 1,3 \}$ and $C_2 := \{ 1,2,4 \}$ and $k=1$. Then, we have: $\prop = \{ a_1,b_1,a_2,b_2,a_3,b_3,a_4,b_4 \}$, $\alpha_1 = (\{ a_1,b_2,a_3,b_4 \})^\omega$, $\alpha_2 = (\{ a_1,a_2,b_3,a_4 \})^\omega$ and $\beta = (\{ b_1,b_2,b_3,b_4 \})^\omega$. Finally, $B = 1$.
	\end{example}


	The above definition satisfies the lemma below.
	\begin{lemma}
		\label{lem:reduc_or_and_and}
		Consider an instance $(l,C,k)$ is a positive instance of the hitting set problem $\msf{Hit}$ and sets of operators $\Ut \subseteq \Op{Un}{}$, $\Bt \subseteq \Op{Bin}{tp}$, and $\Bl \subseteq \Op{Bin}{}$. If $\msf{In}^\lor_{(l,C,k)}$ or $\msf{In}^\wedge_{(l,C,k)}$ is a positive instance of the decision problem $\LTL_\msf{Learn}(\Ut,\Bt,\Bl,\infty)$, then  $(l,C,k)$ is a positive instance of the hitting set problem $\msf{Hit}$.
		
		On the other hand, if $(l,C,k)$ is a positive instance of the hitting set problem $\msf{Hit}$ and $\{ \lor,\Rightarrow,\Leftarrow \} \cap \Bl \neq \emptyset$ (resp. $\{ \wedge,\prescript{\neg}{}{\Rightarrow},\prescript{\neg}{}{\Leftarrow} \} \cap \Bl \neq \emptyset$), then $\msf{In}^\lor_{(l,C,k)}$ (resp. $\msf{In}^\wedge_{(l,C,k)}$) is a positive instance of the decision problem $\LTL_\msf{Learn}(\Ut,\Bt,\Bl,\infty)$.
	\end{lemma}
	
	\begin{proof}
		Assume that $\msf{In}^\lor_{(l,C,k)}$ or $\msf{In}^\wedge_{(l,C,k)}$ is a positive instance of $\LTL_\msf{Learn}(\Ut,\Bt,\Bl,\infty)$. Consider an $\LTL$-formula $\varphi$ of size at most $B = 2k-1$ that distinguishes the sets of infinite words $\msf{Set}$ and $\msf{EmptySet}$. Let $H = \{ i \in [1,\ldots,l] \mid \{ a_i,b_i \} \cap \prop(\varphi) \neq \emptyset \}$ denote the set of integers $i$ for which at least one of the corresponding variables $a_i$ or $b_i$ occurs in $\varphi$. By Lemma~\ref{lem:at_least_that_many_subformulas}, it must be that $|H| \leq k$ since $\Size{\varphi} \leq 2k-1$. Let us show that $H$ is a hitting set. Let $1 \leq i \leq n$. The formula $\varphi$ distinguishes the infinite words $\alpha_i$ and $\beta$. 
		Furthermore, for all $j \in [1,\ldots,l] \setminus C_i$, we have $b_j \in \alpha_i[1]$ and $b_j \in \beta[1]$ (also, recall Observation~\ref{obs:equiv_on_words_reduc}). Hence, by Lemma~\ref{lem:distinguish}, it must be that $C_i \cap H \neq \emptyset$. Since this holds for all $1 \leq i \leq n$, we obtain that $H$ is indeed a hitting set. Note that the arguments that we have given here hold regardless of the operators used in the formula $\varphi$. Hence, $(l,C,k)$ is a positive instance of the hitting set problem $\msf{Hit}$.
		
		Assume now that $(l,C,k)$ is a positive instance of the hitting set problem $\msf{Hit}$. Consider a hitting set $H \subseteq [1,\ldots,l]$ of size at most $k$. We denote $H := \{ j_1,\ldots,j_r \}$ with $|H| = r \leq k$. We define $\LTL$-formulas indexed by the operator that we consider.
		\begin{itemize}
			\item $\varphi_\lor := a_{j_1} \lor a_{j_2} \lor \ldots \lor a_{j_r}$
			\item $\varphi_\Rightarrow := b_{j_1} \Rightarrow (b_{j_2} \Rightarrow ( \ldots \Rightarrow \; a_{j_r} ))$ 
			\item $\varphi_\Leftarrow := ((a_{j_1} \Leftarrow b_{j_2}) \Leftarrow \ldots ) \Leftarrow b_{j_r}$
		\end{itemize}
	 	Recall that for all $x_1,x_2 \in \bool$, we have $x_1 \Rightarrow x_2 = \neg x_1 \lor x_2$ and $x_1 \Leftarrow x_2 = x_1 \lor \neg x_2$, hence, by Observation~\ref{obs:equiv_on_words_reduc}, for all $w \in \msf{EmptySet} \cup \msf{Set}$, we have $w \models \varphi_\lor$ iff $w \models \varphi_\Rightarrow$ iff $w \models \varphi_\Leftarrow$.
	 	
	 	We also define the $\LTL$-formulas below.
		\begin{itemize}
			\item $\varphi_\wedge := b_{j_1} \wedge b_{j_2} \wedge \ldots \wedge b_{j_r}$
			\item $\varphi_{\prescript{\neg}{}{\Leftarrow}} := a_{j_1} \prescript{\neg}{}{\Leftarrow} \; (a_{j_2} \prescript{\neg}{}{\Leftarrow} \; ( \ldots \prescript{\neg}{}{\Leftarrow} \; b_{j_r} ))$ 
			\item $\varphi_{\prescript{\neg}{}{\Rightarrow}} := ((b_{j_1} \prescript{\neg}{}{\Rightarrow} \; a_{j_2}) \prescript{\neg}{}{\Rightarrow} \; \ldots ) \prescript{\neg}{}{\Rightarrow} \; a_{j_r}$
		\end{itemize}
		Recall that for all $x_1,x_2 \in \bool$, we have $x_1 \prescript{\neg}{}{\Rightarrow} \; x_2 = x_1 \wedge \neg x_2$ and $x_1 \prescript{\neg}{}{\Leftarrow} \; x_2 = \neg x_1 \wedge x_2$. Hence, by Observation~\ref{obs:equiv_on_words_reduc}, for all $w \in \msf{EmptySet} \cup \msf{Set}$, we have $w \models \varphi_\wedge$ iff $w \models \varphi_{\prescript{\neg}{}{\Rightarrow}}$ iff $w \models \varphi_{\prescript{\neg}{}{\Leftarrow}}$ and $w \models \varphi_\wedge$ iff $w \not\models \varphi_\lor$.
	
		Clearly, we have $\Size{\varphi_\lor} = \Size{\varphi_\Rightarrow} = \Size{\varphi_\Leftarrow} = \Size{\varphi_\wedge} = \Size{\varphi_{\prescript{\neg}{}{\Rightarrow}}} = \Size{\varphi_{\prescript{\neg}{}{\Leftarrow}}} = 2r-1 \leq B$.
		
		Furthermore, consider any $1 \leq i \leq n$. Let $j \in H \cap C_i \neq \emptyset$. We have $\alpha_i \models a_j$, hence $\alpha_i \models \varphi_\lor$. This holds for all  $1 \leq i \leq n$.  Furthermore, we also have $\beta \not\models \varphi_\lor$. Therefore, $\varphi_\lor$ accepts $\msf{Set}$ and rejects $\msf{EmptySet}$. It is also the case for $\varphi_\Rightarrow$ and $\varphi_\Leftarrow$. It is the opposite for the formulas $\varphi_\wedge,\varphi_{\prescript{\neg}{}{\Rightarrow}},\varphi_{\prescript{\neg}{}{\Leftarrow}}$ (i.e. they accept $\msf{EmptySet}$ and reject $\msf{Set}$). Hence, if $\{ \lor,\Rightarrow,\Leftarrow \} \cap \Bl \neq \emptyset$ then $\msf{In}^\lor_{(l,C,k)}$ is a positive instance of $\LTL_\msf{Learn}(\Ut,\Bt,\Bl,\infty)$ and if $\{ \wedge,\prescript{\neg}{}{\Rightarrow},\prescript{\neg}{}{\Leftarrow} \} \cap \Bl \neq \emptyset$, then $\msf{In}^\wedge_{(l,C,k)}$ is a positive instance of the $\LTL_\msf{Learn}(\Ut,\Bt,\Bl,\infty)$.
	\end{proof}
	
	We obtain the corollary below.
	\begin{corollary}
		\label{coro:or}
		Consider a set $\Bl \subseteq \Op{Bin}{lg}$ of binary logical operators and assume that $\Bt \cap \{ \lor,\Rightarrow,\Leftarrow,\wedge,\prescript{\neg}{}{\Rightarrow},\prescript{\neg}{}{\Leftarrow} \} \neq \emptyset$. For all $\Ut \subseteq \Op{Un}{}$ and $\Bt \subseteq \Op{Bin}{tp}$, the $\LTL_\msf{Learn}(\Ut,\Bt,\Bl,\infty)$ decision problem is $\msf{NP}$-hard.
	\end{corollary}
	\begin{proof}
		This is a direct consequence of Lemmas
		~\ref{lem:reduc_or_and_and} and the fact that the instances $\msf{In}^{\prescript{\neg}{}{\lor}}_{(l,C,k)}$ and $\msf{In}^{\prescript{\neg}{}{\wedge}}_{(l,C,k)}$ can be computed in logarithmic space from $(l,C,k)$.
	\end{proof}
	
	\subsubsection{Proof of Theorem~\ref{thm:non_unary_binary_NP_hard}: when $\Bl \cap \{ \prescript{\neg}{}{\lor},\prescript{\neg}{}{\wedge} \} \neq \emptyset$.}
	The case of the operators $\prescript{\neg}{}{\lor}$ and $\prescript{\neg}{}{\wedge}$ is slightly different. Indeed, contrary to the above operators, when successively using one of these operators, we obtain (semantically) an alternation of conjunctions and disjunctions. We describe it on an example below.
	\begin{example}
		\label{ex:successively_using_not_lor}
		Consider six variables $x_1,x_2,x_3,x_4,x_5,x_6$. Assume that we want to use them in a single $\LTL$-formula. If we can use the $\lor$ operator, we can consider the formula $x_1 \lor x_2 \lor x_3 \lor x_4 \lor x_5 \lor x_6$. If we can use the $\Rightarrow$ operator, we can consider the formula $x_1 \Rightarrow (x_2 \Rightarrow (x_3 \Rightarrow (x_4 \Rightarrow (x_5 \Rightarrow x_6))))$. Note that, up to some negation on the variables, this amounts semantically to only using the $\lor$ operator. However, assume now that we can only use the $\prescript{\neg}{}{\wedge}$ operator. For instance, consider:
		\begin{equation*}
			\varphi := x_1 \prescript{\neg}{}{\wedge} \; (x_2 \prescript{\neg}{}{\wedge} \; (x_3 \prescript{\neg}{}{\wedge} \; (x_4 \prescript{\neg}{}{\wedge} \; (x_5 \prescript{\neg}{}{\wedge} \; x_6))))
		\end{equation*}
		It is semantically equivalent to:
		\begin{equation*}
			\varphi \equiv \neg x_1 \lor (x_2 \wedge (\neg x_3 \lor (x_4 \wedge (\neg x_5 \lor \neg x_6)))
		\end{equation*}
		Here, we have both operators $\lor$ and $\wedge$. 
	\end{example}

	To circumvent this difficulty, we are going to change the reduction by adding propositions that will always hold on the words of interest. We can then place these propositions where $x_2$ and $x_4$ were in the above formula. That way, on the infinite words where these propositions hold, we obtain a disjunction, as in the formula above. 
	
	We start with the reduction for the $\prescript{\neg}{}{\wedge}$ operator.
	\begin{definition}
		\label{def:reduction_or_not_so_easy}
		Consider an instance $(l,C,k)$ of the hitting set problem $\msf{Hit}$. If $k \geq l$, $(l,C,k)$ is obviously a positive instance of the hitting set problem $\msf{Hit}$, and we define $\msf{In}^{\prescript{\neg}{}{\wedge}}_{(l,C,k)}$ to be an arbitrary positive instance of the $\LTL$ learning decision problem. Otherwise, we define:
		\begin{itemize}
			\item $\prop := \{ a_i,b_i \mid 1 \leq j \leq l \} \cup \{ x_i \mid 1 \leq i \leq k \}$ to be the set of propositions;
			\item $\mathcal{P} := \{ \alpha_1,\ldots,\alpha_n,\beta \}$ where for all $1 \leq i \leq n$: we let $\alpha_i := \{ x_1,\ldots,x_{k},c_i^1,\ldots,c_i^l \}^\omega \in (2^\prop)$ with, for all $1 \leq j \leq l$:
			\begin{align*}
				c_i^j := \begin{cases}
					a_j & \text{ if }j \in C_i \\
					b_j & \text{ if }j \notin C_i \\
				\end{cases}
			\end{align*}
		 	and $\beta := \{ x_1,\ldots,x_{k-1},b_1,\ldots,b_l \}^\omega$;
			\item $\mathcal{N} := \{ \alpha,\beta_1,\ldots,\beta_k \}$ with $\alpha := \{ x_1,\ldots,x_{k},b_1,\dots,b_l \}^\omega \in (2^\prop)$ and, for all $1 \leq i \leq k-1$, we have $\beta_i := (\{ x_1,\ldots,x_{k-1},b_1,\ldots,b_l \} \setminus \{ x_i \})^\omega$ and $\beta_k := \{ x_1,\ldots,x_{k},b_1,\ldots,b_l \}^\omega$;
			\item $B = 4 k - 1$.
		\end{itemize}
		Then, we define the input $\msf{In}^{\prescript{\neg}{}{\wedge}}_{(l,C,k)} := (\prop,\mathcal{P},\mathcal{N},B)$ of the $\LTL_\msf{Learn}$ decision problem.
	\end{definition}

	Similarly to the previous reduction, we have the following lemma. 
	\begin{lemma}
		\label{lem:reduc_not_and}
		Consider a set $\Bl$ of binary logical operators and assume that $\prescript{\neg}{}{\wedge} \in \Bl$. Then, for all $\Ut \subseteq \Op{Un}{}$ and $\Bt \subseteq \Op{Bin}{tp}$, $(l,C,k)$ is a positive instance of the hitting set problem $\msf{Hit}$ if and only if $\msf{In}^{\prescript{\neg}{}{\wedge}}_{(l,C,k)}$ is a positive instance of the the $\LTL_\msf{Learn}(\Ut,\Bt,\Bl,\infty)$ decision problem.
	\end{lemma}
	This proof of this lemma is quite close to the proof of Lemma~\ref{lem:reduc_or_and_and}. 
	\begin{proof}
		If $k \geq l$, the equivalence is straightforward. We assume in the following that $k \leq l$. 
		
		Assume that $(l,C,k)$ is a positive instance of the hitting set problem $\msf{Hit}$. Consider a hitting set $H \subseteq [1,\ldots,l]$ of size at most $k$. Consider any set $H' \subseteq [1,\ldots,l]$ of size exactly $k$ such that $H \subseteq H'$. We let $H' := \{ j_1,\ldots,j_k \}$. We 
		let:
		\begin{itemize}
			\item $\varphi_{\prescript{\neg}{}{\wedge}} := b_{j_1} \prescript{\neg}{}{\wedge} \; (x_1 \prescript{\neg}{}{\wedge} \; (b_{j_2} \prescript{\neg}{}{\wedge} \; (\ldots x_{k-1} \prescript{\neg}{}{\wedge} \; (b_{j_k} \prescript{\neg}{}{\wedge} \; x_k))))$
			\item $\varphi_{\prescript{\neg}{}{\wedge}}^{\msf{expl}} := \neg b_{j_1} \lor (x_1 \wedge (\neg b_{j_2} \lor (\ldots x_{k-1} \wedge (\neg b_{j_k} \vee \neg x_k))))$
		\end{itemize}
		The formula $\varphi_{\prescript{\neg}{}{\wedge}}^{\msf{expl}}$ is written to make more explicit what $\varphi_{\prescript{\neg}{}{\wedge}}$ is equal to. Indeed, recall that for all $x_1,x_2 \in \bool$, we have $x_1 \prescript{\neg}{}{\wedge} \; x_2 = \neg x_1 \lor \neg x_2$, hence, for all $w \in \mathcal{P} \cup \mathcal{N}$, we have $w \models \varphi_{\prescript{\neg}{}{\wedge}}$ iff $w \models \varphi_{\prescript{\neg}{}{\wedge}}^{\msf{expl}}$. Furthermore, we have $\Size{\varphi_{\prescript{\neg}{}{\wedge}}} = 4k - 3 = B$.
		
		Let $\gamma \in \{ \alpha,\alpha_1,\ldots,\alpha_n\}$. We have $\gamma \models x_1 \wedge \ldots \wedge x_{k}$. Hence, $\gamma \models \varphi_{\prescript{\neg}{}{\wedge}}$ if and only if $\gamma \models \neg b_{j_1} \lor \ldots \lor \neg b_{j_k}$. Therefore, $\alpha \not \models \varphi_{\prescript{\neg}{}{\wedge}}$. However, for any $1 \leq i \leq n$, there is $j \in H \cap C_i \neq \emptyset$ such that $\alpha_i \models \neg b_j$, and thus $\alpha_i \models \varphi_{\prescript{\neg}{}{\wedge}}$. 
		In addition, consider any $\delta \in \{ \beta,\beta_1,\ldots,\beta_k\}$. 
		For all $j \in [1,\ldots,l]$, we have $\delta \models b_j$. Hence, $\delta \models \varphi_{\prescript{\neg}{}{\wedge}}$ if and only if $\delta \models x_1 \wedge \ldots \wedge x_{k-1} \wedge \neg x_k$. Hence,  $\beta \models \varphi_{\prescript{\neg}{}{\wedge}}$ whereas, for all $1 \leq i \leq k$, we have $\beta_i \not\models \varphi_{\prescript{\neg}{}{\wedge}}$. 
		Overall, the formula $\varphi_{\prescript{\neg}{}{\wedge}}$ accepts $\mathcal{P}$ and rejects $\mathcal{N}$. Hence, the decision problem $\msf{In}^{\prescript{\neg}{}{\wedge}}_{(l,C,k)}$ is a positive instance of $\LTL_\msf{Learn}(\Ut,\Bt,\Bl,\infty)$. 
		
		Assume now the decision problem $\msf{In}^{\prescript{\neg}{}{\wedge}}_{(l,C,k)}$ is a positive instance of $\LTL_\msf{Learn}(\Ut,\Bt,\Bl,\infty)$. Consider an $\LTL$-formula $\varphi$ of size at most $B = 4k-1$ that distinguishes the sets of infinite words $\mathcal{P}$ and $\mathcal{N}$. Let $H = \{ i \in [1,\ldots,l] \mid \{ a_i,b_i \} \cap \prop(\varphi) \neq \emptyset \}$ denote the set of integers $i$ for which at least one of the corresponding variables $a_i$ or $b_i$ occurs in $\varphi$. Let us show that this set is of size at most $k$ and intersects all sets $C_i$. 
		
		By Lemma~\ref{lem:at_least_that_many_subformulas}, we have $\prop(\varphi) \leq 2k$. Furthermore, by Lemma~\ref{lem:distinguish}, since for all $1 \leq i \leq k$, the formula $\varphi$ distinguishes the words $\beta$ and $\beta_i$, it follows that $x_i \in \prop(\varphi)$. Hence, $|H| \leq 2k - k = k$.
		
		Furthermore, let $1 \leq i \leq n$. The formula $\varphi$ distinguishes the infinite words $\alpha_i$ and $\alpha$. In addition, for all $j \in [1,\ldots,l] \setminus C_i$, we have $b_j \in \alpha_i[1] \cap \alpha[1]$, $a_j \notin \alpha_i[1] \cup \alpha[1]$ and $\alpha_i[1] \cap \{ x_1,\ldots,x_k \} = \alpha[1] \cap \{ x_1,\ldots,x_k \}$. Hence, by Lemma~\ref{lem:distinguish} and by Definition of $\alpha_i$, it must be that $C_i \cap H \neq \emptyset$. Since this holds for all $1 \leq i \leq n$, we obtain that $H$ is indeed a hitting set. Hence, $(l,C,k)$ is a positive instance of the hitting set problem $\msf{Hit}$.
	\end{proof}

	Contrary to the reductions we defined in Definition~\ref{def:reduction_or_easy}, the reduction for the operator $\prescript{\neg}{}{\lor}$ is not obtained from the reduction for $\prescript{\neg}{}{\wedge}$ by reversing the positive and negative sets of words, though it is quite similar. We give it below.
	\begin{definition}
		\label{def:reduction_and_not_so_easy}
		Consider an instance $(l,C,k)$ of the hitting set problem $\msf{Hit}$. If $k \geq l$, $(l,C,k)$ is obviously a positive instance of the hitting set problem $\msf{Hit}$, and we define $\msf{In}^{\prescript{\neg}{}{\lor}}_{(l,C,k)}$ to be an arbitrary positive instance of the $\LTL$ learning decision problem. Otherwise, we define:
		\begin{itemize}
			\item $\prop := \{ a_i,b_i \mid 1 \leq j \leq l \} \cup \{ x_i \mid 1 \leq i \leq k \}$ to be the set of propositions;
			\item $\mathcal{P} := \{ \alpha,\beta_1,\ldots,\beta_k \}$ with $\alpha := \{ b_1,\dots,b_l \}^\omega \in (2^\prop)$ and, for all $1 \leq i \leq k-1$, we have $\beta_i := \{b_1,\dots,b_l,x_i,x_k\}^\omega$ and $\beta_k = \{b_1,\dots,b_l\}^\omega$;
			\item $\mathcal{N} := \{ \alpha_1,\ldots,\alpha_n,\beta \}$ where for all $1 \leq i \leq n$: we let $\alpha_i := \{ c_i^1,\ldots,c_i^l \}^\omega \in (2^\prop)$ with, for all $1 \leq j \leq l$:
			\begin{align*}
				c_i^j := \begin{cases}
					a_j & \text{ if }j \in C_i \\
					b_j & \text{ if }j \notin C_i \\
				\end{cases}
			\end{align*}
			and $\beta := \{b_1,\dots,b_l,x_k\}^\omega$;
			\item $B = 4 k - 1$.
		\end{itemize}
		Then, we define the inputs $\msf{In}^{\prescript{\neg}{}{\lor}}_{(l,C,k)} := (\prop,\mathcal{P},\mathcal{N},B)$ of the $\LTL_\msf{Learn}$ decision problem.
	\end{definition}
	
	Similarly to the previous reduction, we have the following lemma. 
	\begin{lemma}
		\label{lem:reduc_not_or}
		Consider a set $\Bl$ of binary logical operators and assume that $\prescript{\neg}{}{\lor} \in \Bl$. Then, for all $\Ut \subseteq \Op{Un}{}$ and $\Bt \subseteq \Op{Bin}{tp}$, $(l,C,k)$ is a positive instance of the hitting set problem $\msf{Hit}$ if and only if $\msf{In}^{\prescript{\neg}{}{\lor}}_{(l,C,k)}$ is a positive instance of the the $\LTL_\msf{Learn}(\Ut,\Bt,\Bl,\infty)$ decision problem.
	\end{lemma}
	The proof of this lemma is very close to the proof of Lemma~\ref{lem:reduc_not_and}. Hence, we only give the formula using only the $\prescript{\neg}{}{\lor}$ operator built from a hitting set that we consider.
	\begin{proof}[Proof sketch]
		If $k \geq l$, the equivalence is straightforward. Let us now assume that $k \leq l$. 
		
		Assume that $(l,C,k)$ is a positive instance of the hitting set problem $\msf{Hit}$. Consider a hitting set $H \subseteq [1,\ldots,l]$ of size at most $k$. Consider any set $H' \subseteq [1,\ldots,l]$ of size exactly $k$ such that $H \subseteq H'$. We let $H' := \{ j_1,\ldots,j_k \}$ and:
		\begin{itemize}
			\item $\varphi_{\prescript{\neg}{}{\lor}} := a_{j_1} \prescript{\neg}{}{\lor} \; (x_1 \prescript{\neg}{}{\lor} \; (a_{j_2} \prescript{\neg}{}{\lor} \; (\ldots x_{k-1} \prescript{\neg}{}{\lor} \; (a_{j_k} \prescript{\neg}{}{\lor} \; x_k))))$
			\item $\varphi_{\prescript{\neg}{}{\lor}}^{\msf{expl}} := \neg a_{j_1} \wedge (x_1 \lor (\neg a_{j_2} \wedge (\ldots x_{k-1} \lor (\neg a_{j_k} \wedge \neg x_k))))$
		\end{itemize}
		The formula $\varphi_{\prescript{\neg}{}{\lor}}^{\msf{expl}}$ is written to make more explicit what $\varphi_{\prescript{\neg}{}{\lor}}$ is equal to. Indeed, recall that for all $x_1,x_2 \in \bool$, we have $x_1 \prescript{\neg}{}{\lor} \; x_2 = \neg x_1 \wedge \neg x_2$, hence, for all $w \in \mathcal{P} \cup \mathcal{N}$, we have $w \models \varphi_{\prescript{\neg}{}{\lor}}$ iff $w \models \varphi_{\prescript{\neg}{}{\lor}}^{\msf{expl}}$. One can then check that the formula $\varphi_{\prescript{\neg}{}{\lor}}$ accepts $\mathcal{P}$ and rejects $\mathcal{N}$.  
	\end{proof}

	We deduce the corollary below.
	\begin{corollary}
		\label{coro:not_or}
		Consider a set $\Bl \subseteq \Op{Bin}{lg}$ of binary logical operators and assume that $\Bt \cap \{ \prescript{\neg}{}{\wedge},\prescript{\neg}{}{\lor} \} \neq \emptyset$. For all $\Ut \subseteq \Op{Un}{}$ and $\Bt \subseteq \Op{Bin}{tp}$, the $\LTL_\msf{Learn}(\Ut,\Bt,\Bl,\infty)$ decision problem is $\msf{NP}$-hard.
	\end{corollary}
	\begin{proof}
		This is a direct consequence of Lemmas~\ref{lem:reduc_not_and} and~\ref{lem:reduc_not_or} and the fact that the instances $\msf{In}^{\prescript{\neg}{}{\lor}}_{(l,C,k)}$ and $\msf{In}^{\prescript{\neg}{}{\lor}}_{(l,C,k)}$ can be computed in logarithmic space from $(l,C,k)$.
	\end{proof}

	\subsubsection{Proof of Theorem~\ref{thm:non_unary_binary_NP_hard}: with the temporal operators $\lW$ and $\lM$}
	The last two logical binary operators $\Leftrightarrow$ and $\prescript{\neg}{}{\Leftrightarrow}$ will be handled by considering a reduction that is completely different from what we have presented so far. However, to prove the correctness of that reduction we first show a way to handle (i.e. to remove them from the formulas that we consider) the temporal operators. This can be done because, in all the reductions that we have presented so far, and with the reduction for the operators $\Leftrightarrow$ and $\prescript{\neg}{}{\Leftrightarrow}$ that will follow, all the infinite words that we consider are size-1 infinite words. We define a notion of equivalence on $\LTL$-formulas over these infinite words.
	\begin{definition}
		Consider a set of propositions $\prop$ and two $\LTL$-formulas $\varphi_1$ and $\varphi_2$. We write $\varphi_1 \equiv_1 \varphi_2$ when, for all size-1 infinite words $w \in (2^{\prop})^\omega$, we have $w \models \varphi_1$ if and only if $w \models \varphi_2$.
	\end{definition}
	Then, we have the following lemma.
	\begin{lemma}
		\label{lem:equiv_temp}
		Consider three $\LTL$-formulas $\varphi,\varphi_1,\varphi_2 \in \LTL$. We have:
		\begin{enumerate}
			\item for all $\bullet \in \{ \lX,\lF,\lG \}$, we have $\varphi \equiv_1 \bullet \varphi$;
			\item for all $\bullet \in \{ \lU,\lR \}$, we have $\varphi_2 \equiv_1 \varphi_1 \bullet \varphi_2$;
			\item $\varphi_1 \lor \varphi_2 \equiv_1  \varphi_1 \lW \varphi_2$;
			\item $\varphi_1 \wedge \varphi_2 \equiv_1  \varphi_1 \lM \varphi_2$.
		\end{enumerate}
	\end{lemma} 
	\begin{proof}
		Consider any $\alpha \subseteq \prop$. Note that for all $i \in \N_1$, we have $\alpha^\omega[i:] = \alpha^\omega$. Therefore:
		\begin{enumerate}
			\item \begin{itemize}
				\item $\alpha^\omega \models \varphi \text{ iff } \alpha^\omega[2:] \models \varphi \text{ iff } \alpha^\omega \models \lX \varphi$;
				\item $\alpha^\omega \models \varphi \text{ iff } \exists i\in \N_1,\; \alpha^\omega[i:] \models \varphi \text{ iff } \alpha^\omega \models \lF \varphi$;
				\item $\alpha^\omega \models \varphi \text{ iff } \forall i\in \N_1,\; \alpha^\omega[i:] \models \varphi \text{ iff } \alpha^\omega \models \lG \varphi;$ 
			\end{itemize}
			\item If $\alpha^\omega \models \varphi_2$, then $\alpha^\omega \models \varphi_1 \lU \varphi_2$.  Furthermore, if $\alpha^\omega \models \varphi_1 \lU \varphi_2$, then there is some $i \in \N_1$ such that $\alpha^\omega[i:] \models \varphi_2$. Hence, $\alpha^\omega \models \varphi_2$. In fact,  $\alpha^\omega \models \varphi_1 \lU \varphi_2 \text{ iff } \alpha^\omega \models \varphi_2$.
			
			In addition, we have:
			\begin{equation*}
				\alpha^\omega \models \varphi_1 \lR \varphi_2 \text{ iff } \alpha^\omega \not\models \neg\varphi_1 \lU \neg\varphi_2 \text{ iff } \alpha^\omega \not\models \neg\varphi_2 \text{ iff } \alpha^\omega \models \varphi_2
			\end{equation*}
			\item $\alpha^\omega \models \varphi_1 \lW \varphi_2 \text{ iff } \alpha^\omega \models (\varphi_1 \lU \varphi_2) \lor \lG \varphi_1 \text{ iff } \alpha^\omega \models \varphi_2 \lor \varphi_1$;
			\item $\alpha^\omega \models \varphi_1 \lM \varphi_2 \text{ iff } \alpha^\omega \models (\varphi_1 \lR \varphi_2) \wedge \lF \varphi_1 \text{ iff } \alpha^\omega \models \varphi_2 \wedge \varphi_2$.
		\end{enumerate}
	\end{proof}
	In particular, this lemma tells us that, on size-1 infinite words, the $\lW$ temporal operator behaves like the $\lor$ logical operator and the $\lM$ temporal operator behaves like the $\wedge$ logical operator. Hence, we can reuse the reduction from Definition~\ref{def:reduction_or_easy} to establish that the $\LTL$ learning problem is $\msf{NP}$-hard with at least one of the two temporal operators $\lW$ or $\lM$, as stated in the lemma below.
	\begin{lemma}
		\label{lem:reduc_W_M}
		Consider a set $\Bt$ of binary temporal operators and assume that $\lW \in \Bt$ (resp. $\lM \in \Bt$). Then, for all $\Ut \subseteq \Op{Un}{}$ and $\Bl \subseteq \Op{Bin}{lg}$, $(l,C,k)$ is a positive instance of the hitting set problem $\msf{Hit}$ if and only if $\msf{In}^\lor_{(l,C,k)}$ is a positive instance of the the $\LTL_\msf{Learn}(\Ut,\Bt,\Bl,\infty)$ decision problem (resp. $\msf{In}^\wedge_{(l,C,k)}$ is a positive instance of the the $\LTL_\msf{Learn}(\Ut,\Bt,\Bl,\infty)$).
	\end{lemma}
	\begin{proof}
		The bottom to top implication is already given by Lemma~\ref{lem:reduc_or_and_and}. Furthermore, let us assume that $(l,C,k)$ is a positive instance of the hitting set problem $\msf{Hit}$. Then, from a hitting set $H \subseteq [1,\ldots,l]$ of size at most $k$ with $H := \{ j_1,\ldots,j_r \}$, we consider the $\LTL$-formulas: $\varphi_{\lW} := a_{j_1} \lW (a_{j_2} \lW ( \ldots \lW a_{j_r}))$ and $\varphi_{\lM} := b_{j_1} \lM (b_{j_2} \lM ( \ldots \lM b_{j_r}))$. We have $\Size{\varphi_{\lW}} = \Size{\varphi_{\lM}} = 2r-1 \leq 2k-1$.
		In addition, by Lemma~\ref{lem:equiv_temp}, $\varphi_{\lW} \equiv_1 a_{j_1} \lor \ldots \lor a_{j_r} = \varphi_\lor$ and $\varphi_{\lM} := b_{j_1} \wedge \ldots \wedge b_{j_r} = \varphi_\wedge$, with $\varphi_\lor$ and $\varphi_\wedge$ from the proof of Lemma~\ref{lem:reduc_or_and_and}, where we have shown that $\varphi_{\lor}$ accepts $\msf{Set}$ and rejects $\msf{EmptySet}$ and $\varphi_{\wedge}$ accepts $\msf{EmptySet}$ and rejects $\msf{Set}$. The lemma follows.
		%
	\end{proof}

	We deduce the corollary below.
	\begin{corollary}
		\label{coro:W_M}
		Consider a set $\Bt \subseteq \Op{Bin}{tp}$ of binary temporal operators with $\Bt \cap \{ \lW,\lM \} \neq \emptyset$. For all $\Ut \subseteq \Op{Un}{}$ and $\Bl \subseteq \Op{Bin}{lg}$, the $\LTL_\msf{Learn}(\Ut,\Bt,\Bl,\infty)$ decision problem is $\msf{NP}$-hard.
	\end{corollary}
	\begin{proof}
		This is a direct consequence of Lemma~\ref{lem:reduc_W_M} and the fact that the instances $\msf{In}^\lor_{(l,C,k)}$ and $\msf{In}^\wedge_{(l,C,k)}$ can be computed in logarithmic space from $(l,C,k)$.
	\end{proof}

	\subsubsection{Proof of Theorem~\ref{thm:non_unary_binary_NP_hard}: when $\Bl \cap \{ \Leftrightarrow,\prescript{\neg}{}{\Leftrightarrow} \} \neq \emptyset$.}
	We now consider the case of the operators $\Leftrightarrow$ and $\prescript{\neg}{}{\Leftrightarrow}$. To handle this case, we are going to restrict ourselves to 
	$\Leftrightarrow$-formula, defined below, i.e. $\LTL$-formulas that only use operators $\{ \neg,\Leftrightarrow,\prescript{\neg}{}{\Leftrightarrow} \}$.
	\begin{definition}
		We say that an $\LTL$-formula is a $\Leftrightarrow$-formula if it belongs to the fragment $\LTL(\{\neg\},\{\Leftrightarrow,\prescript{\neg}{}{\Leftrightarrow}\},\emptyset)$.
	\end{definition}
	
	We exhibit a function that removes all the operators that we do not want to consider here.
	\begin{lemma}
		\label{lem:temp_into_logic}
		There exists a function $\msf{tr}: \LTL(\Op{Un}{},\{\Leftrightarrow,\prescript{\neg}{}{\Leftrightarrow}\},\{\lU,\lR\}) \rightarrow \LTL(\{\neg\},\{\Leftrightarrow,\prescript{\neg}{}{\Leftrightarrow}\},\emptyset)$ such that, for all $\LTL$-formulas $\varphi \in \LTL$:
		\begin{enumerate}
			\item $\msf{SubF}(\msf{tr}(\varphi)) \subseteq \msf{tr}[\msf{SubF}(\varphi)] := \{ \msf{tr}(\varphi') \mid \varphi' \in \msf{SubF}(\varphi) \}$, hence $\Size{\msf{tr}(\varphi)} \leq \Size{\varphi}$;
			\item $\varphi \equiv_1 \msf{tr}(\varphi)$.
		\end{enumerate}
	\end{lemma}
	\begin{proof}
		We define $\msf{tr}: \LTL(\Op{Un}{},\{\Leftrightarrow,\prescript{\neg}{}{\Leftrightarrow}\},\{\lU,\lR\}) \rightarrow \LTL(\{\neg\},\{\Leftrightarrow,\prescript{\neg}{}{\Leftrightarrow}\},\emptyset)$ by induction on $\LTL$-formulas. Consider a formula $\varphi \in \LTL(\Op{Un}{},\{\Leftrightarrow,\prescript{\neg}{}{\Leftrightarrow}\},\{\lU,\lR\})$:
		\begin{itemize}
			\item Assume that $\varphi = x$ for any $x \in \prop$. We set $\msf{tr}(\varphi) := \varphi$ which satisfies conditions 1.-2.;
			\item Assume that $\varphi = \neg \varphi'$ and that $\msf{tr}$ is defined on all sub-formulas of $\varphi'$ and satisfies conditions 1.-2. on $\varphi'$. We set $\msf{tr}(\varphi) := \lnot \msf{tr}(\varphi')$, which ensures that $\msf{SubF}(\msf{tr}(\varphi)) = \{ \lnot \msf{tr}(\varphi') \} \cup \msf{SubF}(\msf{tr}(\varphi')) \subseteq \msf{tr}[\{ \lnot \varphi' \}] \cup \msf{tr}[\msf{SubF}(\varphi')] = \msf{tr}[\msf{SubF}(\varphi)]$. Furthermore, for all $\alpha \subseteq \prop$, we have $\alpha^\omega \models \msf{tr}(\varphi) \text{ iff } \alpha^\omega \not\models \msf{tr}(\varphi') \text{ iff } \alpha^\omega \not\models \varphi' \text{ iff } \alpha^\omega \models \varphi$. Thus, $\varphi \equiv_1\msf{tr}(\varphi)$;
			\item For all $\bullet \in \{ \lX,\lF,\lG \}$, assume that $\varphi = \bullet \varphi'$, that $\msf{tr}$ is defined on all sub-formulas of $\varphi'$ and satisfies conditions 1.-2. on $\varphi'$. Then, we set $\msf{tr}(\varphi) := \msf{tr}(\varphi')$, which ensures that $\msf{SubF}(\msf{tr}(\varphi)) = \msf{SubF}(\msf{tr}(\varphi')) \subseteq \msf{tr}[\msf{SubF}(\varphi')] \subseteq \msf{tr}[\msf{SubF}(\varphi)]$. Furthermore, $\varphi \equiv_1 \msf{tr}(\varphi)$ by Lemma~\ref{lem:equiv_temp} and since $\varphi' \equiv_1 \msf{tr}(\varphi')$.
			\item For any $\bullet \in \{\Leftrightarrow,\prescript{\neg}{}{\Leftrightarrow},\lU,\lR \}$, assume that $\varphi = \varphi_1 \bullet \varphi_2$ and that $\msf{tr}$ is defined on all sub-formulas of $\varphi_1$ and $\varphi_2$ and satisfies conditions 1.-2. on $\varphi_1$ and $\varphi_2$.
			\begin{itemize}
			\item If $\bullet \in \{\Leftrightarrow,\prescript{\neg}{}{\Leftrightarrow}\}$, we set $\msf{tr}(\varphi) := \msf{tr}(\varphi_1) \bullet\msf{tr}(\varphi_2)$, which ensures that $\msf{SubF}(\msf{tr}(\varphi)) = \{ \msf{tr}(\varphi) \} \cup \msf{SubF}(\msf{tr}(\varphi_1)) \cup \msf{SubF}(\msf{tr}(\varphi_2)) \subseteq \msf{tr}[\{ \varphi \}] \cup \msf{tr}[\msf{SubF}(\varphi_1)] \cup \msf{tr}[\msf{SubF}(\varphi_2)] \subseteq \msf{tr}[\msf{SubF}(\varphi)]$. Furthermore, since $\varphi_1 \equiv_1 \msf{tr}(\varphi_1)$ and $\varphi_2 \equiv_1 \msf{tr}(\varphi_2)$, we have $\varphi \equiv_1 \msf{tr}(\varphi)$;
			\item If $\bullet \in \{ \lU,\lR \}$, we set $\msf{tr}(\varphi) := \msf{tr}(\varphi_2)$, which ensures that $\msf{SubF}(\msf{tr}(\varphi)) = \msf{SubF}(\msf{tr}(\varphi_2)) \subseteq \msf{tr}[\msf{SubF}(\varphi_2)] \subseteq \msf{tr}[\msf{SubF}(\varphi)]$. Furthermore, $\varphi \equiv_1 \msf{tr}(\varphi)$ by Lemma~\ref{lem:equiv_temp} and since $\varphi_1 \equiv_1 \msf{tr}(\varphi_1)$ and $\varphi_2 \equiv_1 \msf{tr}(\varphi_2)$
			\end{itemize}
		\end{itemize}
	\end{proof}

	Now, in order to gain an intuition on the reduction that we will consider (and especially the problem from which we make that reduction), let us give the central property satisfied by $\Leftrightarrow$-formulas.
	\begin{lemma}
		\label{lem:equiv_satis_prop}
		Consider a set of propositions $\prop$. Given an $\Leftrightarrow$-formula $\varphi$, we let $\msf{Neg}(\varphi) \in \N$ denote the number of occurrences of the $\neg$ and $\prescript{\neg}{}{\Leftrightarrow}$ operators in $\varphi$. Furthermore, given any subset $\alpha \subseteq \prop$, we also denote by $\msf{NbPr}_{\bar{\alpha}}(\varphi) \in \N$ the number of occurrences of propositions in $\prop \setminus \alpha$ occurring in $\prop$ (hence, some propositions may be counted several times if they appear more than once in $\varphi$). Both of these numbers are defined inductively on the tree structure of the $\Leftrightarrow$-formula $\varphi$, without the notion of sub-formulas.	Then, for all $\alpha \subseteq \prop$, we have:
		\begin{equation*}
			\alpha^\omega \models \varphi \text{ if and only if }
			\msf{Neg}(\varphi) \text{ and } \msf{NbPr}_{\bar{\alpha}}(\varphi) \text{ have the same parity}
		\end{equation*}
		Stated algebraically, we have $\alpha^\omega \models \varphi \text{ if and only if }
		\msf{Neg}(\varphi) + \msf{NbPr}_{\bar{\alpha}}(\varphi) \mod 2 = 0$.
	\end{lemma} 
	\begin{proof}
		We show by induction on $\varphi$ the property $\mathcal{P}(\varphi)$: $\alpha^\omega \models \varphi$ iff
		$\msf{Neg}(\varphi) + \msf{NbPr}_{\bar{\alpha}}(\varphi) \text{ is even}$.
		\begin{itemize}
			\item If $\varphi = x$ for some $x \in \prop$, then $\alpha^\omega \models \varphi$ if and only if $x \in \alpha$, i.e. $\msf{NbPr}_{\bar{\alpha}}(\varphi) = \emptyset$. Since $\msf{Neg}(\varphi) = 0$ in any case, $\mathcal{P}(\varphi)$ follows.
			\item Assume that $\varphi = \neg \varphi'$ and that $\mathcal{P}(\varphi')$ holds. Then, $\msf{Neg}(\varphi) = \msf{Neg}(\varphi') +1$ and $\prop(\varphi) = \prop(\varphi')$, thus $\msf{NbPr}_{\bar{\alpha}}(\varphi') = \msf{NbPr}_{\bar{\alpha}}(\varphi)$. Hence, $\msf{Neg}(\varphi')$ and $\msf{NbPr}_{\bar{\alpha}}(\varphi')$ have the same parity if and only if $\msf{Neg}(\varphi)$ and $\msf{NbPr}_{\bar{\alpha}}(\varphi)$ do not, and $\mathcal{P}(\varphi)$ follows.
			\item Assume that $\varphi = \varphi_1 \bullet \varphi_2$ for some $\bullet \in \{ \Leftrightarrow,\prescript{\neg}{}{\Leftrightarrow} \}$ and assume that both $\mathcal{P}(\varphi_1)$ and $\mathcal{P}(\varphi_2)$ hold. We have $\msf{NbPr}_{\bar{\alpha}}(\varphi) = \msf{NbPr}_{\bar{\alpha}}(\varphi_1) + \msf{NbPr}_{\bar{\alpha}}(\varphi_2)$. Then:
			\begin{itemize}
				\item If $\bullet = \Leftrightarrow$, we have $\msf{Neg}(\varphi) = \msf{Neg}(\varphi_1) + \msf{Neg}(\varphi_2)$. Furthermore, $\alpha^\omega \models \varphi$ iff we have $\alpha^\omega \models \varphi_1$ and $\alpha^\omega \models \varphi_2$ or $\alpha^\omega \not\models \varphi_2$ and $\alpha^\omega \not\models \varphi_2$. That is, by $\mathcal{P}(\varphi_1)$ and $\mathcal{P}(\varphi_2)$, we have $\alpha^\omega \models \varphi$ iff
				\begin{equation*}
					\msf{Neg}(\varphi_1) + \msf{NbPr}_{\bar{\alpha}}(\varphi_1) \mod 2 = \msf{Neg}(\varphi_2) + \msf{NbPr}_{\bar{\alpha}}(\varphi_2) \mod 2
				\end{equation*}
				which is equivalent to 
				\begin{equation*}
					\msf{Neg}(\varphi_1) + \msf{Neg}(\varphi_2) + \msf{NbPr}_{\bar{\alpha}}(\varphi_1) + \msf{NbPr}_{\bar{\alpha}}(\varphi_2) = 0 \mod 2
				\end{equation*}
				That is:
				\begin{equation*}
					\msf{Neg}(\varphi) + \msf{NbPr}_{\bar{\alpha}}(\varphi) = 0 \mod 2
				\end{equation*}
				\item If $\bullet = \prescript{\neg}{}{\Leftrightarrow}$, we have $\msf{Neg}(\varphi) = \msf{Neg}(\varphi_1) + \msf{Neg}(\varphi_2) + 1$. Furthermore,  $\alpha^\omega \models \varphi$ iff we have $\alpha^\omega \models \varphi_1$ and $\alpha^\omega \not\models \varphi_2$ or $\alpha^\omega \models \varphi_2$ and $\alpha^\omega \not\models \varphi_2$. That is, by $\mathcal{P}(\varphi_1)$ and $\mathcal{P}(\varphi_2)$, we have $\alpha^\omega \models \varphi$ iff
				\begin{equation*}
					\msf{Neg}(\varphi_1) + \msf{NbPr}_{\bar{\alpha}}(\varphi_1) \mod 2 = \msf{Neg}(\varphi_2) + \msf{NbPr}_{\bar{\alpha}}(\varphi_2) + 1 \mod 2
				\end{equation*}
				which is equivalent to 
				\begin{equation*}
					\msf{Neg}(\varphi_1) + \msf{Neg}(\varphi_2) + 1 + \msf{NbPr}_{\bar{\alpha}}(\varphi_1) + \msf{NbPr}_{\bar{\alpha}}(\varphi_2) = 0 \mod 2
				\end{equation*}
				That is:
				\begin{equation*}
					\msf{Neg}(\varphi) + \msf{NbPr}_{\bar{\alpha}}(\varphi) = 0 \mod 2
				\end{equation*}
			\end{itemize}
			The property $\mathcal{P}(\varphi)$ follows.
		\end{itemize}
	\end{proof}
	
	We state below, as a corollary, the lemma above in a form that is easier to use for us.
	\begin{corollary}
		\label{coro:equiv_satis_prop_useful}
		Consider a set of propositions $\prop$. Given an $\Leftrightarrow$-formula $\varphi$ and any subset $\alpha \subseteq \prop$, we let $\msf{TrueNbPr}_{\bar{\alpha}}(\varphi) \in \N$ denote the number of propositions in $\prop \setminus \alpha$ that occur oddly many times in $\varphi$. Then, for all $\alpha \subseteq \prop$, we have:
		\begin{equation*}
			\alpha^\omega \models \varphi \text{ if and only if }
			\msf{Neg}(\varphi) \text{ and } \msf{TrueNbPr}_{\bar{\alpha}}(\varphi) \text{ have the same parity}
		\end{equation*}
	\end{corollary}
	\begin{proof}
		This is a direct consequence of Lemma~\ref{lem:equiv_satis_prop} and the fact that $\msf{TrueNbPr}_{\bar{\alpha}}(\varphi)$ and $\msf{NbPr}_{\bar{\alpha}}(\varphi)$ have the same parity.
	\end{proof}

	This corollary suggests that the $\LTL_{\msf{Learn}}(\{\neg\},\{\Leftrightarrow,\prescript{\neg}{}{\Leftrightarrow}\},\emptyset,\infty)$ learning problem is linked to modulo 2-calculus. In fact, there exists an $\msf{NP}$-hard decision problem dealing with modulo 2-calculus, which we define below. The definition of this problem, and the proof that it is $\msf{NP}$-complete can be found in \cite{DBLP:journals/tit/BerlekampMT78}.
	\begin{definition}[Coset Weight]
		\label{def:coset_weight}
		We denote by $\msf{CW}$ the following decision problem:
		\begin{itemize}
			\item Input: $(A,k,y)$ where $A$ is $n \times l$-matrix on $\Z/2\Z$, $k \leq l$ is an integer and $y$ is an $n$-vector in $\Z/2\Z$. 
			\item Output: yes iff there is an $l$-vector $x$ in $\Z/2\Z$ with at most $k$ components of value 1 such that $A \cdot x = y$ in $\Z/2\Z$. 
		\end{itemize}
		In terms of representation, the integer $k$ may be given in unary. 
	\end{definition}
	
	Since we are going to manipulate matrices, we introduce below the relevant notations. 
	\begin{definition}
		Given $n,l \in \N_1$, an $n \times l$-matrix $A$ has $n$ rows and $l$ columns. Furthermore, for all $1 \leq i \leq n$ and $1 \leq j \leq l$, $A_{i,j}$ refers to the coefficient of the matrix $A$ at the intersection of the $i$-th row and $j$-th column.
	\end{definition}

	Let us now define the reduction that we consider. 
	\begin{definition}
		\label{def:reduction_equiv_not_equiv}
		Consider an instance $(A,k,y)$ of the coset weight problem $\msf{CW}$. If $y = 0$ (i.e. it is the null vector), then $(A,k,y)$ is trivially a positive instance of $\msf{CW}$, hence in that case $\msf{In}^{\Leftrightarrow}_{(A,k,y)}$ is defined as any positive instance of the $\LTL_{\msf{Learn}}(\{\neg\},\{\Leftrightarrow,\prescript{\neg}{}{\Leftrightarrow}\},\emptyset,\infty)$ decision problem. Otherwise, we define:
		\begin{itemize}
			\item $\prop := \{ a_j \mid 1 \leq j \leq l \}$ to be the set of propositions;
			\item For all $1 \leq i \leq n$, we set $\alpha_i := \{ a_j \mid 1 \leq j \leq l,\; A_{i,j} = 0 \}^\omega$;
			\item $\mathcal{P} := \{ \beta \} \cup \{ \alpha_i \mid 1 \leq i \leq n,\; y[i] = 0 \}$ with $\beta := \prop^\omega$;
			\item $\mathcal{N} := \{ \alpha_i \mid 1 \leq i \leq n,\; y[i] = 1 \}$;
			\item $B = 2k$.
		\end{itemize}
		Then, we define the input $\msf{In}^{\Leftrightarrow}_{(A,k,y)} := (\prop,\mathcal{P},\mathcal{N},B)$ of the $\LTL_{\msf{Learn}}(\{\neg\},\{\Leftrightarrow,\prescript{\neg}{}{\Leftrightarrow}\},\emptyset,\infty)$ decision problem.
	\end{definition}

	The definition above satisfies the lemma below.
	\begin{lemma}
		\label{lem:reduc_equiv}
		Consider a set $\Bl$ of logical binary operators and assume that $\emptyset \neq \Bl \subseteq \{ \Leftrightarrow,\prescript{\neg}{}{\Leftrightarrow} \} $. Then, for all $\Ut \subseteq \Op{Un}{}$ and $\Bt \subseteq \{ \lU,\lW \}$, $(A,k,y)$ is a positive instance of the coset weight problem $\msf{CW}$ if and only if $\msf{In}^{\Leftrightarrow}_{(A,k,y)}$ is a positive instance of the the $\LTL_\msf{Learn}(\Ut,\Bt,\Bl,\infty)$ decision problem.
	\end{lemma}
	\begin{proof}
		If $y = 0$, the equivalence is straightforward. Let us now assume in the following that $y \neq 0$. In that case, we have $\mathcal{P} \neq \emptyset$ and $\mathcal{N} \neq \emptyset$. 
		
		Assume that $(A,k,y)$ is a positive instance of the coset weight problem $\msf{CW}$ and consider an $l$-vector $x$ in $\Z/2\Z$ with at most $k$ components of value 1 such that $Ax = y$. Since $y \neq 0$, $x$ has at least one component of value 1. We let $H := \{ a_j \mid 1 \leq j \leq l,\; x[j] = 1 \}$. We write $H = \{ a_{j_1},\ldots,a_{j_r} \}$ with $1 \leq r \leq k$. We then consider the formulas:
		\begin{itemize}
			\item $\varphi_{\Leftrightarrow} := a_{j_1} \Leftrightarrow a_{j_2} \Leftrightarrow \ldots \Leftrightarrow a_{j_r}$
			\item We let $\varphi_{\prescript{\neg}{}{\Leftrightarrow}}^\msf{even} := a_{j_1} \prescript{\neg}{}{\Leftrightarrow} \; a_{j_2} \prescript{\neg}{}{\Leftrightarrow} \ldots \prescript{\neg}{}{\Leftrightarrow} \; a_{j_r}$. We also define:
			\begin{equation*}
				\varphi_{\prescript{\neg}{}{\Leftrightarrow}} := \begin{cases}
					\varphi_{\prescript{\neg}{}{\Leftrightarrow}}^\msf{even} & \text{ if }r\text{ is even} \\
					\neg \varphi_{\prescript{\neg}{}{\Leftrightarrow}}^\msf{even} & \text{ if }r\text{ is odd}
				\end{cases}
			\end{equation*}
		\end{itemize}
		Both of these formulas have size at most $2r \leq 2k = B$. 
		
		By definition, we have both $\msf{Neg}(\varphi_{\Leftrightarrow})$ and $\msf{Neg}(\varphi_{\prescript{\neg}{}{\Leftrightarrow}})$ even. Therefore, by Corollary~\ref{coro:equiv_satis_prop_useful}, both $\varphi_{\Leftrightarrow}$ and $\varphi_{\prescript{\neg}{}{\Leftrightarrow}}$ accept the word $\beta$. Consider now some $1 \leq i \leq n$. 
		\begin{itemize}
			\item Assume that $y[i] = 0$. Then, they are evenly many indices $j \in [1,\ldots,l]$ such that $x[j] = 1 = A_{i,j}$. That is, $|H \setminus \alpha_i|$ is even. Hence, by Corollary~\ref{coro:equiv_satis_prop_useful}, $\varphi_{\Leftrightarrow},\varphi_{\prescript{\neg}{}{\Leftrightarrow}}$ accept $\alpha_i \in \mathcal{P}$. 
			\item Assume that $y[i] = 1$. Then, they are oddly many indices $j \in [1,\ldots,l]$ such that $x[j] = 1 = A_{i,j}$. That is, $|H \setminus \alpha_i|$ is odd. Hence, by Corollary~\ref{coro:equiv_satis_prop_useful}, $\varphi_{\Leftrightarrow},\varphi_{\prescript{\neg}{}{\Leftrightarrow}}$ reject $\alpha_i \in \mathcal{N}$. 
		\end{itemize}
		Hence, $\varphi_{\Leftrightarrow},\varphi_{\prescript{\neg}{}{\Leftrightarrow}}$ accept $\mathcal{P}$ and reject $\mathcal{N}$. Thus, $\msf{In}^{\Leftrightarrow}_{(A,k,y)}$ is a positive instance of the $\LTL_\msf{Learn}(\Ut,\Bt,\Bl,\infty)$ decision problem.
		
		Assume now that $\msf{In}^{\Leftrightarrow}_{(A,k,y)}$ is a positive instance of the $\LTL_\msf{Learn}(\Ut,\Bt,\Bl,\infty)$ decision problem. Consider an $\LTL(\Ut,\Bt,\Bl)$-formula $\varphi$ of size at most $B$ that accepts $\mathcal{P}$ and rejects $\mathcal{N}$. By Lemma~\ref{lem:temp_into_logic}, there is a $\Leftrightarrow$-formula $\varphi'$ of size at most $B$ that is equivalent to the formula $\varphi$ on size-1 infinite words. By Corollary~\ref{coro:equiv_satis_prop_useful}, since $\varphi'$ accepts $\beta$, it follows that $\msf{Neg}(\varphi')$ is even. We let $H := \{ i \in [1,\ldots,l] \mid a_i \text{ occurs oddly many times in }\varphi'\}$. By Lemma~\ref{lem:at_least_that_many_subformulas}, we have $H \leq k$. We define the $l$-vector $x$ in $\Z/2\Z$ by, for all $1 \leq j \leq l$, $x[j] := 1$ if and only if $j \in H$. Consider now any $1 \leq i \leq n$. 
		\begin{itemize}
			\item If $\alpha_i \in \mathcal{P}$, we have $|H \setminus \alpha_i[1]|$ even, by Corollary~\ref{coro:equiv_satis_prop_useful}. In addition, for all $j \in H \setminus \alpha_i[1]$, we have $x[j] = 1 = A_{i,j}$ whereas, for all $j \notin H \setminus \alpha_i$, we have $x[j] = 0$ or $A_{i,j} = 0$. Therefore, $A x[i] = 0 = y[i]$. 
			\item If $\alpha_i \in \mathcal{N}$, we have $|H \setminus \alpha_i|$ odd, by Corollary~\ref{coro:equiv_satis_prop_useful}. Therefore, $A x[i] = 1 = y[i]$. 
		\end{itemize} 
		Hence, $(A,k,y)$ is a positive instance of the coset weight problem $\msf{CW}$. 
	\end{proof}
	
	We deduce the corollary below.
	\begin{corollary}
		\label{coro:equiv}
		Consider a set $\Bl \subseteq \Op{Bin}{lg}$ of binary logical operators with $\Bl \cap \{ \Leftrightarrow,\prescript{\neg}{}{\Leftrightarrow} \} \neq \emptyset$. For all $\Ut \subseteq \Op{Un}{}$ and $\Bt \subseteq \Op{Bin}{tp}$, the $\LTL_\msf{Learn}(\Ut,\Bt,\Bl,\infty)$ decision problem is $\msf{NP}$-hard.
	\end{corollary}
	\begin{proof}
		If we have $\Bl \cap \{ \lor,\Rightarrow,\Leftarrow,\wedge,\prescript{\neg}{}{\Rightarrow},\prescript{\neg}{}{\Leftarrow},\prescript{\neg}{}{\lor},\prescript{\neg}{}{\wedge} \} \neq \emptyset$, Corollary~\ref{coro:or} or Corollary~\ref{coro:not_or} gives that the $\LTL_\msf{Learn}(\Ut,\Bt,\Bl,\infty)$ decision problem is $\msf{NP}$-hard. Furthermore, if $\Bt \cap \{ \lW,\lM \} \neq \emptyset$, then Corollary~\ref{coro:W_M} gives that the $\LTL_\msf{Learn}(\Ut,\Bt,\Bl,\infty)$ decision problem is $\msf{NP}$-hard. If it is not the case, then we have both $\emptyset \neq \Bl \subseteq \{ \Leftrightarrow,\prescript{\neg}{}{\Leftrightarrow} \} $ and $\Bt \subseteq \{ \lU,\lW \}$. The result then follows from  Lemma~\ref{lem:reduc_equiv}, the fact that the Coset Weight decision problem is $\msf{NP}$-hard and the fact that the instance $\msf{In}^\Leftrightarrow_{(l,C,k)}$ can be computed in logarithmic space from $(l,C,k)$.
	\end{proof}
	
	\subsubsection{Proof of Theorem~\ref{thm:non_unary_binary_NP_hard}: with the temporal operators $\lU$ and $\lR$}
	We conclude with the temporal operators $\lU$ and $\lR$. We will consider a reduction from the hitting set problem again. However, as can be seen in Lemma~\ref{lem:equiv_temp}, with these operators, a reduction with size-1 infinite words will not work since they simplify too much on them. 
	
	We define below a way to build infinite words or $\LTL$-formulas from subsets of integers.
	\begin{definition}
		Let $l \in \N$. We consider a set of propositions $\prop^l := \{ a_{i,j},b_{i,j} \mid 1 \leq i \leq j \leq l \}$. Consider any set subset $C \subseteq [1,\ldots,l]$. For all $1 \leq i \leq l$, we let $\msf{S}_i(C) := \{ (k,j) \mid 1 \leq k \leq i,\; i \leq j, j \in C\}$. Then, we let:
		\begin{equation*}
			w^{\lR}(C) := w_1(C) \cdot \ldots \cdot w_l(C) \cdot \emptyset^\omega \in (2^{\prop^l})^\omega
		\end{equation*}
		and
		\begin{equation*}
			w^{\lU}(C) := w_1(C) \cdot \ldots \cdot w_l(C) \cdot (\prop^l)^\omega \in (2^{\prop^l})^\omega
		\end{equation*} 
		where, for all $1 \leq i \leq l$, we have:
		\begin{equation*}
			w_i(C) := \{ a_{k,j} \mid (k,j) \in \msf{S}_i(C) \} \cup \{ b_{k,j} \mid (k,j) \notin \msf{S}_i(C) \}
		\end{equation*}
		
		Furthermore, consider any $H = \{ i_1,\ldots,i_r \} \subseteq [1,\ldots,l]$ with $i_1 < i_2 < i_3 < \ldots < i_r$. We let $i_0 := 1$ and we define:
		\begin{itemize}
			\item $\varphi_{\lR}(H,r) := a_{i_{r-1},i_r}$, and	$\varphi_{\lU}(H,r) := b_{i_{r-1},i_r}$;
			\item for all $2 \leq p \leq r$, $\varphi_{\lR}(H,p-1) := \varphi_{\lR}(H,p) \; \lR \;  a_{i_{p-2},i_{p-1}}$, and $\varphi_{\lU}(H,p-1) := \varphi_{\lU}(H,p) \; \lR \;  b_{i_{p-2},i_{p-1}}$.
		\end{itemize}
		We let $\varphi_{\lR}(H) := \varphi_{\lR}(H,1)$ and $\varphi_{\lU}(H) := \varphi_{\lU}(H,1)$. That way, we have:
		\begin{equation*}
			\varphi_{\lR}(H) = ((a_{i_{r-1},i_r} \; \lR \; a_{i_{r-2},i_{r-1}}) \; \lR \; \ldots \;) \lR \; a_{1,i_1}
		\end{equation*}
		and 
		\begin{equation*}
			\varphi_{\lU}(H) = ((b_{i_{r-1},i_r} \; \lU \; b_{i_{r-2},i_{r-1}}) \; \lU \; \ldots \;) \lU \; b_{1,i_1}
		\end{equation*}
	\end{definition}
	
	We first establish below what the $\lR$-formula satisfies.
	\begin{lemma}
		\label{lem:criterion_accept_formula_r}
		Consider any $H = \{ i_1,\ldots,i_r \} \subseteq [1,\ldots,l]$ with $i_1 < i_2 < i_3 < \ldots < i_r$ and some $C \subseteq [1,\ldots,l]$. We have the following equivalence:
		\begin{equation*}
			H \subseteq C \text{ if and only if } w^{\lR}(C) \models \varphi_{\lR}(H)
		\end{equation*}
	\end{lemma}
	\begin{proof}
		For all $r \geq p \geq 1$, we let $H_p := \{ i_p,\ldots,i_r \} \subseteq H$. Let us show by induction on $r \geq p \geq 1$ the property $\mathcal{P}(p)$:  
		\begin{center}
			For all $1 \leq i \leq i_{p-1}-1$: $w^{\lR}(C)[i:] \not\models \varphi_{\lR}(H,p)$ and $w^{\lR}(C)[i_{p-1}:] \models \varphi_{\lR}(H,p)$ iff $H_p \subseteq C$
		\end{center}
		
		We start with the base case $\mathcal{P}(r)$. For all $1 \leq i \leq i_{r-1}-1$, we have $a_{i_{r-1},i_r} \notin w^{\lR}(C)[i]$. Furthermore, we have $a_{i_{r-1},i_r} \in w^{\lR}(C)[i_{r-1}]$ if and only if $i_r \in C$. Since $\varphi_{\lR}(H,r) = a_{i_{r-1},i_r}$, $\mathcal{P}(r)$ follows. 
		
		Let us now assume that $\mathcal{P}(p)$ holds for some $2 \leq p \leq r$. We have: 
		\begin{equation*}
			\varphi_{\lR}(H,p-1) = \varphi_{\lR}(H,p) \; \lR \; a_{i_{p-2},i_{p-1}}
		\end{equation*}
		For all $1 \leq i \leq i_{p-2}-1$, we have $a_{i_{p-2},i_{p-1}} \notin w^{\lR}(C)[i]$, thus $w^{\lR}(C)[i:] \not\models a_{i_{p-2},i_{p-1}}$. In addition, by $\mathcal{P}(p)$, for all $1 \leq i \leq i_{p-1}-1$: $w^{\lR}(C)[i:] \not\models \varphi_{\lR}(H,p)$ (in particular, this holds for $i = i_{p-2}$). It follows that, for all $1 \leq i \leq i_{p-2}-1$, we have $w^{\lR}(C)[i:] \not\models \varphi_{\lR}(H,p-1)$. 
		\begin{itemize}
			\item Assume that $H_{p-1} \subseteq C$. Then, we also have $H_{p} \subseteq C$, thus, by  $\mathcal{P}(p)$, we have $w^{\lR}(C)[i_{p-1}:] \models \varphi_{\lR}(H,p)$. Furthermore, for all $i_{p-2} \leq i \leq i_{p-1}$, we have $w^{\lR}(C)[i] \models a_{i_{p-2},i_{p-1}}$ since $i_{p-1} \in C$. It follows that:
			\begin{equation*}
				w^{\lR}(C)[i_{p-2}:] \models \varphi_{\lR}(H,p-1)
			\end{equation*}
			\item Assume that $H_{p-1} \not\subseteq C$. Then, there are two cases, since $H_{p-1} = H_p \cup \{i_{p-1}\}$.
			\begin{itemize}
				\item Assume that $H_{p} \not\subseteq C$. By $\mathcal{P}(p)$, for all $1 \leq i \leq i_{p-1}$, we have $w^{\lR}(C)[i:] \not\models \varphi_{\lR}(H,p)$. Furthermore, $w^{\lR}(C)[i_{p-1}+1:] \not\models a_{i_{p-2},i_{p-1}}$. Therefore, $w^{\lR}(C)[i_{p-2}:] \not\models \varphi_{\lR}(H,p-1)$. 
				\item Assume that $i_{p-1} \notin C$. Then, for all $1 \leq i \leq l$, we have $w^{\lR}(C)[i:] \not\models a_{i_{p-2},i_{p-1}}$. Furthermore, by $\mathcal{P}(p)$, for all $1 \leq i \leq i_{p-1}-1$: $w^{\lR}(C)[i:] \not\models \varphi_{\lR}(H,p)$. It follows that $w^{\lR}(C)[i_{p-2}:] \not\models \varphi_{\lR}(H,p-1)$.  
			\end{itemize}
		\end{itemize}
		Hence, the property $\mathcal{P}(p-1)$ holds. Therefore, $\mathcal{P}(p)$ holds for all $1 \leq p \leq r$. The lemma is then given by $\mathcal{P}(1)$, since $H_1 = H$, $w^{\lR}(C)[1:] = w^{\lR}(C)$ and $\varphi_{\lR}(H,1) = \varphi_{\lR}(H)$.
	\end{proof}
	
	Let us now consider the case of $\lU$-formulas.
	\begin{lemma}
		\label{lem:criterion_accept_formula_u}
		Consider any $H = \{ i_1,\ldots,i_r \} \subseteq [1,\ldots,l]$ with $i_1 < i_2 < i_3 < \ldots < i_r$ and some $C \subseteq [1,\ldots,l]$. We have the following equivalence:
		\begin{equation*}
			H \subseteq C \text{ if and only if } w^{\lU}(C) \not\models \varphi_{\lU}(H)
		\end{equation*}
	\end{lemma}
	\begin{proof}
		Let us show by induction on $r \geq p \geq 1$ the property $\mathcal{P}(p)$: for all $k \in \N_1$, we have $w^{\lU}(C)[k:] \models \varphi_{\lU}(H,p)$ if and only if $w^{\lR}(C)[k:] \not\models \varphi_{\lR}(H,p)$.
		
		Let us start with the base case. We have $\varphi_{\lU}(H,r) = b_{i_{r-1},i_r}$  and $\varphi_{\lR}(H,r) = a_{i_{r-1},i_r}$. Furthermore, for all $k \in \N_1$, we have $a_{i_{r-1},i_r} \in w^{\lU}(C)[k]$ if and only if $b_{i_{r-1},i_r} \notin w^{\lR}(C)[k]$. Thus, $\mathcal{P}(r)$ follows.
		
		Assume now that $\mathcal{P}(p)$ holds for some $2 \leq p \leq r$. We have:
		\begin{equation*}
			\varphi_{\lU}(H,p-1) = \varphi_{\lU}(H,p) \; \lU \; b_{i_{p-2},i_{p-1}}
		\end{equation*}
		and
		\begin{equation*}
			\varphi_{\lR}(H,p-1) = \varphi_{\lR}(H,p) \; \lR \; a_{i_{p-2},i_{p-1}}
		\end{equation*}
		Thus, by definition of the operator $\lR$, we have $\varphi_{\lR}(H,p-1)$ equivalent to $\neg (\neg \varphi_{\lR}(H,p) \lU \neg a_{i_{p-2},i_{p-1}})$. In addition, for all $i \in \N_1$, we have $b_{i_{p-2},i_{p-1}} \in w^{\lU}(C)[i]$ iff $a_{i_{p-2},i_{p-1}} \in w^{\lR}(C)[i]$. Thus, by $\mathcal{P}(p)$, for all $k \in \N_1$, we have:
		\begin{align*}
			w^{\lU}(C)[k:] \models \varphi_{\lU}(H,p-1) & \text{ iff } \exists j \geq k,\; b_{i_{p-2},i_{p-1}} \in w^{\lU}(C)[j] \text{ and }\forall k \leq i < j,\; w^{\lU}(C)[i:] \models \varphi_{\lU}(H,p) \\ 
			& \text{ iff } \exists j \geq k,\; a_{i_{p-2},i_{p-1}} \notin w^{\lR}(C)[j] \text{ and }\forall k \leq i < j,\; w^{\lR}(C)[i:] \not\models \varphi_{\lR}(H,p) \\ 
			& \text{ iff } w^{\lR}(C)[k:] \models \neg \varphi_{\lR}(H,p) \lU \neg a_{i_{p-2},i_{p-1}} \\
			& \text{ iff } w^{\lR}(C)[k:] \models \neg \varphi_{\lR}(H,p-1)
		\end{align*}
		Thus, the property $\mathcal{P}(p-1)$ holds. In fact, it holds for all $1 \leq p \leq r$. This lemma is then a direct consequence of the property $\mathcal{P}(1)$ and Lemma~\ref{lem:criterion_accept_formula_u}.
	\end{proof}

	We can now define the reduction that we consider.
	\begin{definition}
		\label{def:reduction_U_R}
		Consider an instance $(l,C,k)$ of the hitting set problem $\msf{Hit}$. 
		We define:
		\begin{itemize}
			\item $\prop := \prop^l = \{ a_{i,j},b_{i,j} \mid 1 \leq i \leq j \leq l \}$ as set of propositions;
			\item $\mathcal{P}^{\lU} := \{ v_1^{\lU},\ldots,v_n^{\lU} \}$ 
			and  $\mathcal{P}^{\lR} := \{ v^{\lR} \}$
			;
			\item $\mathcal{N}^{\lU} := \{ v^{\lU} \}$ and $\mathcal{N}^{\lR} := \{ v_1^{\lR},\ldots,v_n^{\lR} \}$;
			\item $B = 2 k - 1$.
		\end{itemize}
		with, for $\bullet \in \{\lR,\lU\}$, we let $v^{\bullet} := w^{\bullet}([1,\ldots,l]) \in (2^\prop)$ and for all $1 \leq i \leq n$, we let $v_i^{\bullet} := w^{\bullet}([1,\ldots,l] \setminus C_i) \in (2^\prop)$. 
		
		Then, we define the inputs $\msf{In}^{\lR}_{(l,C,k)} := (\prop,\mathcal{P}^{\lR},\mathcal{N}^{\lR},B)$ and $\msf{In}^{\lU}_{(l,C,k)} := (\prop,\mathcal{P}^{\lU},\mathcal{N}^{\lU},B)$ of the $\LTL_\msf{Learn}(\Ut,\Bt,\Bl,\infty)$ decision problem.
	\end{definition}

	This definition satisfies the lemma below.
	\begin{lemma}
		\label{lem:reduc_U_R}
		Let $\bullet \in \{\lR,\lU\}$. Consider a set $\Bt$ of operators and assume that $\bullet \in \Bt \neq \emptyset$. Then, for all $\Ut \subseteq \Op{Un}{}$ and $\Bl \subseteq \Op{Bin}{lg}$, $(l,C,k)$ is a positive instance of the hitting set problem $\msf{Hit}$ if and only if $\msf{In}^{\bullet}_{(l,C,k)}$ is a positive instance of the the $\LTL_\msf{Learn}(\Ut,\Bt,\Bl,\infty)$ decision problem.
	\end{lemma}
	\begin{proof}
		Assume that $(l,C,k)$ is a positive instance of the hitting set problem $\msf{Hit}$. Consider a hitting set $H \subseteq [1,\ldots,l]$. We consider the formulas $\varphi_{\lR} := \varphi_{\lR}(H)$ and $\varphi_{\lU} := \varphi_{\lU}(H)$. Clearly, these formulas have size $2k-1 = B$. In addition, Lemmas~\ref{lem:criterion_accept_formula_r} and~\ref{lem:criterion_accept_formula_u} gives that $v^{\lU} \not\models \varphi_{\lU}$ and $v^{\lR} \models \varphi_{\lR}$. Consider now some $1 \leq i \leq n$. Since $H \cap C_i \neq \emptyset$, it follows that $H \not\subseteq [1,\ldots,l] \setminus C_i$, hence, by Lemmas~\ref{lem:criterion_accept_formula_r} and~\ref{lem:criterion_accept_formula_u}, we have $v_i^{\lU} \models \varphi_{\lU}$ and $v_i^{\lR} \not\models \varphi_{\lR}$. 
		Therefore, for $\bullet \in \{\lR,\lU\}$, we have $\msf{In}^{\bullet}_{(l,C,k)}$ a positive instance of the  $\LTL_\msf{Learn}(\Ut,\Bt,\Bl,\infty)$ decision problem. 
		
		Assume now that $\msf{In}^{\bullet}_{(l,C,k)}$ is a positive instance of the $\LTL_\msf{Learn}(\Ut,\Bt,\Bl,\infty)$ decision problem. Consider an $\LTL$-formula $\varphi$ of size at most $B = 2k-1$ that accepts $\mathcal{P}^{\bullet}$ and rejects $\mathcal{N}^{\bullet}$. We let $H := \{ \alpha \in [1,\ldots,l] \mid \exists 1 \leq i \leq l,\; \{ a_{i,\alpha},b_{i,\alpha} \} \cap \prop(\varphi) \neq \emptyset \}$. By Lemma~\ref{lem:at_least_that_many_subformulas}, we have $|H| \leq k$. Let us show that it is a hitting set. Consider some $1 \leq p \leq n$. 
		Given tow sets $A$ and $B$, we let $A \Delta B$ denote the symmetric difference: $A \; \Delta \; B := A \setminus B \cup B \setminus A$. Then, we have that:
		\begin{equation*}
			\forall 1 \leq i \leq l,\; v_p^{\bullet}[i] \; \Delta \; v^{\bullet}[i] = \{ a_{k,j},b_{k,j} \mid 1 \leq k \leq i,\; i \leq j, j \in C_p \}
		\end{equation*}
		and 
		\begin{equation*}
			\forall i \geq l,\;
			v_p^{\bullet}[i] = v^{\bullet}[i]
		\end{equation*}
		Hence, by Lemma~\ref{lem:distinguish}, it follows that $\prop(\varphi)$ must contain at least a variable $a_{k,j}$ or $b_{k,j}$ with $j \in C_p$. That is, $C_p \cap H \neq \emptyset$. Since this holds for all $1 \leq p \leq n$, it follows that $H$ is a hitting set and $(l,C,k)$ is a positive instance of the hitting set problem $\msf{Hit}$.
	\end{proof}

	We obtain the corollary below.
	\begin{corollary}
		\label{coro:U_R}
		Consider a set $\Bt \subseteq \Op{Bin}{tp}$ of binary temporal operators with $\Bt \cap \{ \lU,\lR \} \neq \emptyset$. For all $\Ut \subseteq \Op{Un}{}$ and $\Bl \subseteq \Op{Bin}{lg}$, the $\LTL_\msf{Learn}(\Ut,\Bt,\Bl,\infty)$ decision problem is $\msf{NP}$-hard.
	\end{corollary}
	\begin{proof}
		This is a direct consequence of Lemma~\ref{lem:reduc_U_R} and the fact that the instances $\msf{In}^{\lU}_{(l,C,k)}$ and $\msf{In}^{\lR}_{(l,C,k)}$ can be computed in logarithmic space from $(l,C,k)$.
	\end{proof}

	The proof of Theorem~\ref{thm:non_unary_binary_NP_hard} follows.
	\begin{proof}
		It is a direct consequence of Corollaries~\ref{coro:or},~\ref{coro:not_or},~\ref{coro:equiv},~\ref{coro:W_M} and~\ref{coro:U_R}.
	\end{proof}
	
	\subsection{$\CTL$ learning is at least as hard as $\LTL$ learning}
	Let us now turn to $\CTL$ learning. Our goal is to show that $\CTL$ learning is at least as hard as $\LTL$ learning, under logarithmic space reductions, regardless of the operators allowed. This is formally stated in the theorem below. 
	\begin{theorem}
		\label{thm:CTL_hard_than_LTL}
		For all $\Ut \subseteq \Op{Un}{}$, $\Bt \subseteq \Op{Bin}{tp}$, $\Bl \subseteq \Op{Bin}{lg}$ and $n \in \N \cup \{\infty\}$, the decision problem $\CTL_\msf{Learn}(\Ut,\Bt,\Bl,n)$ is at least as hard as the decision problem $\LTL_\msf{Learn}(\Ut,\Bt,\Bl)$, under logarithmic space reduction.
		
		Therefore, it is also the case of the decision problems $\ATL_\msf{Learn}^k(\Ut,\Bt,\Bl,n)$, for $k \in \N_1$.
	\end{theorem}

	The proof of this theorem consists in translating ultimately periodic words into Kripke structures, translating $\LTL$-formulas into $\CTL$-formulas, and vice versa, and relating all these translations together. 
	
	Let first translate ultimately periodic words into Kripke structure of the same size.
	\begin{definition}
		\label{def:kripke_from_words}
		Consider a set of propositions $\prop$ and an ultimately periodic word $w  = u \cdot v^\omega \in (2^\prop)^\omega$. We $k := |u|$ and $n := |v|$. Then, we define the Kripke structure $K_w = \langle Q,I,\prop, \pi, \msf{Succ} \rangle$ where:
		\begin{itemize}
			\item $Q := \{ q^u_1, \ldots, q^u_k,q^v_{1},\ldots,q^v_{n} \}$;
			\item $I := \{ q_1^u \}$;
			\item for all $1 \leq i \leq k-1$, we have $\msf{Succ}(q^u_i) := \{ q^u_{i+1} \}$, $\msf{Succ}(q^u_{k}) := \{ q^v_{1} \}$ and for all $1 \leq i \leq n-1$, we have $\msf{Succ}(q^v_i) := \{ q^v_{i+1} \}$, $\msf{Succ}(q^v_{n}) := \{ q^v_{1} \}$;
			\item for all $1 \leq i \leq k$, $\pi(q^u_i) := u[i] \subseteq \prop$ and for all $1 \leq i \leq n$, $\pi(q^v_{i}) := v[i] \subseteq \prop$.
		\end{itemize}
	\end{definition}
	
	We then define below how to translate $\CTL$-formulas into $\LTL$-formulas and vice versa, while staying in the same fragment.
	\begin{definition}
		\label{def:translate_ltl_ctl}
		Consider a set of propositions $\prop$ and  $\Ut \subseteq \Op{Un}{}$, $\Bt \subseteq \Op{Bin}{tp}$ and $\Bl \subseteq \Op{Bin}{lg}$. Then, we let $\msf{f}_{\CL(\Ut,\Bt,\Bl)}: \CTL(\Ut,\Bt,\Bl) \rightarrow \LTL(\Ut,\Bt,\Bl)$ be such that it removes all $\exists$ or $\forall$ quantifiers from a $\CTL$-formula, while keeping the same operators. More formally:
		\begin{itemize}
			\item $\msf{f}_{\CL(\Ut,\Bt,\Bl)}(p) = p$ for all $p \in \prop$;
			\item $\msf{f}_{\CL(\Ut,\Bt,\Bl)}(\neg \phi) = \neg \msf{f}_{\CL(\Ut,\Bt,\Bl)}(\phi)$;
			\item $\msf{f}_{\CL(\Ut,\Bt,\Bl)}(\diamond (\bullet \msf{f}_{\CL(\Ut,\Bt,\Bl)}(\phi)) = \bullet \msf{f}_{\CL(\Ut,\Bt,\Bl)}(\phi))$ for all $\bullet \in \Ut \setminus \{ \neg \}$ and $\diamond \in \{\exists,\forall\}$;
			\item $\msf{f}_{\CL(\Ut,\Bt,\Bl)}(\phi_1 \bullet \phi_2) = \msf{f}_{\CL(\Ut,\Bt,\Bl)}(\phi_1) \bullet \msf{f}_{\CL(\Ut,\Bt,\Bl)}( \phi_2)$ for all $\bullet \in \Bl$;
			\item $\msf{f}_{\CL(\Ut,\Bt,\Bl)}(\diamond (\phi_1 \bullet \phi_2)) = \msf{f}_{\CL(\Ut,\Bt,\Bl)}(\phi_1) \bullet \msf{f}_{\CL(\Ut,\Bt,\Bl)}( \phi_2)$ for all $\bullet \in \Bt$ and $\diamond \in \{\exists,\forall\}$.
		\end{itemize}
	
		We also define the function $\msf{f}_{\LC(\Ut,\Bt,\Bl)}: \LTL(\Ut,\Bt,\Bl) \rightarrow \CTL(\Ut,\Bt,\Bl)$ that adds an $\exists$ quantifier before every temporal operator of an $\LTL$-formula, thus transforming it into a $\CTL$-formula. More formally:
		\begin{itemize}
			\item $\msf{f}_{\LC(\Ut,\Bt,\Bl)}(p) = p$ for all $p \in \prop$;
			\item $\msf{f}_{\LC(\Ut,\Bt,\Bl)}(\neg \varphi) = \neg \msf{f}_{\LC(\Ut,\Bt,\Bl)}(\varphi)$;
			\item $\msf{f}_{\LC(\Ut,\Bt,\Bl)}(\bullet \msf{f}_{\LC(\Ut,\Bt,\Bl)}(\varphi) = \exists (\bullet \msf{f}_{\LC(\Ut,\Bt,\Bl)}(\varphi)))$ for all $\bullet \in \Ut \setminus \{\neg \}$;
			\item $\msf{f}_{\LC(\Ut,\Bt,\Bl)}(\varphi_1 \bullet \varphi_2) = \msf{f}_{\LC(\Ut,\Bt,\Bl)}(\varphi_1) \bullet \msf{f}_{\LC(\Ut,\Bt,\Bl\mathsf{O}_U,\mathsf{O}_B)}( \varphi_2)$ for all $\bullet \in \Bl$;
			\item $\msf{f}_{\LC(\Ut,\Bt,\Bl)}(\varphi_1 \bullet \varphi_2) = \exists (\msf{f}_{\LC(\Ut,\Bt,\Bl)}(\varphi_1) \bullet \msf{f}_{\LC(\Ut,\Bt,\Bl)}( \varphi_2))$ for all $\bullet \in \Bt$.
		\end{itemize}
	\end{definition}
	
	A direct proof by induction shows that the above definition satisfies the proposition below:
	\begin{proposition}
		\label{prop:translate_ltl_ctl_simple}
		Consider a set of propositions $\prop$ and  $\Ut \subseteq \Op{Un}{}$, $\Bt \subseteq \Op{Bin}{tp}$ and $\Bl \subseteq \Op{Bin}{lg}$. We have:
		\begin{itemize}
			\item for all $\CTL$-formulas $\phi \in \CTL(\Ut,\Bt,\Bl)$, $\Size{\phi} \geq \Size{\msf{f}_{\CL(\Ut,\Bt,\Bl)}(\phi)}$ (the size may decrease for instance for the $\CTL$-formula $\phi = \exists \lX p \lor \forall \lX p$);
			\item for all $\LTL$-formulas $\varphi \in \LTL(\Ut,\Bt,\Bl)$, $\Size{\varphi} = \Size{\msf{f}_{\LC(\Ut,\Bt,\Bl)}(\varphi)}$;
			\item for all $\LTL$-formulas $\varphi \in \LTL(\Ut,\Bt,\Bl)$, we have $\msf{f}_{\CL(\Ut,\Bt,\Bl)}(\msf{f}_{\LC(\Ut,\Bt,\Bl)}(\varphi)) = \varphi$.
		\end{itemize}
	\end{proposition}
	
	Interestingly for us, Definition~\ref{def:translate_ltl_ctl} also satisfies the lemma below, which is slightly less direct to show that Proposition~\ref{prop:translate_ltl_ctl_simple} above.
	\begin{lemma}
		\label{lem:equiv_word_ltl_ctl}
		Consider a set of propositions $\prop$ and  $\Ut \subseteq \Op{Un}{}$, $\Bt \subseteq \Op{Bin}{tp}$ and $\Bl \subseteq \Op{Bin}{lg}$.  Let $\phi \in \CTL(\Ut,\Bt,\Bl)$ be a $\CTL$-formula. For all ultimately periodic words $w = u \cdot v^\omega \in 2^\prop$, we have $K_{w} \models \phi$ if and only if  $w \models \msf{f}_{\CL(\Ut,\Bt,\Bl)}(\phi)$.
	\end{lemma}
	\begin{proof}
		Let $k := |u|$ and $n := |v|$. In the Kripke structure $K_w$, all states have exactly one successor. Hence, for all states $q \in Q$, we have $|\msf{Out}^Q(q)| = 1$ and we let $\rho_q \in Q^\omega$ be the infinite path such that $\msf{Out}^Q(q) = \{\rho_q\}$. 

		Let $g: \N_1 \rightarrow Q$ be such that for all $i \in \N_1$ we have:
		\begin{align*}
			g(i) := \begin{cases}
				q^u_i & \text{ if }i \leq k \\
				q^v_{1 + (i - k -1 \mod n)} & \text{ if }i \geq k+1
			\end{cases}
		\end{align*}
		
		For all $i \in \N_1$, we have $\msf{Succ}(g(i)) = \{ g(i+1)\}$. Indeed, for all $i \in \N_1$, we have:
		\begin{itemize}
			\item if $i \leq k-1$, then $g(i) = q^u_i$ and $g(i+1) = q^u_{i+1}$;
			\item $g(k) = q^u_k$ and $g(k+1) = q^v_{1}$;
			\item if $i \geq k+1$ and $l := (i - k - 1) \mod n \leq n-2$, then $g(i) = q^v_{l+1}$ and $g(i+1) = q^v_{l+2}$;
			\item if $i \geq k+1$ and $(i - k - 1) \mod n = n-1$, then $g(i) = q^v_{n}$ and $g(i+1) = q^v_{1}$.
		\end{itemize}
		
		It follows that, for all $i \in \N_1$, we have $w[i:] = \pi(\rho_{g(i)}) \in (2^\prop)^\omega$. 
		
		Now, for all $i \in \N_1$, we denote by $K_{w}^i$ the Kripke structure that is equal to the Kripke structure $K_{w}$ except that the initial state is now $g(i)$, i.e. $I = \{ g(i)\}$. Let us show by induction on $\CTL$-formulas $\phi \in \CTL(\Ut,\Bt,\Bl)$ the property $\mathcal{P}(\phi)$: for all $i \in \N_1$, we have $K_{w}^i \models \phi$ if and only if $w[i:] \models \msf{f}_{\CL(\Ut,\Bt,\Bl)}(\phi)$. Consider some $\phi \in \CTL(\Ut,\Bt,\Bl)$. We have:
		\begin{itemize}
			\item Assume that $\phi = x$ for any $x \in \prop$. In that case, $\msf{f}_{\CL(\Ut,\Bt,\Bl)}(\phi) = x$. Let $i \in \N_1$. We have 
			$K_{w}^i \models \phi$ iff $x \in \pi(q_{g(i)}) = w[i]$ iff $w[i:] \models \phi$. Hence, $\mathcal{P}(\phi)$ holds;
			\item Assume that $\phi = \neg \phi'$ for some $\phi' \in \CTL(\Ut,\Bt,\Bl)$
			. In that case, we have $\msf{f}_{\CL(\Ut,\Bt,\Bl)}(\phi) = \neg \msf{f}_{\CL(\Ut,\Bt,\Bl)}(\phi')$. Hence, $\mathcal{P}(\phi)$ is a straightforward consequence of $\mathcal{P}(\phi')$. 
			\item For all $\bullet \in \Ut$, assume that $\phi = \diamond (\bullet \phi')$ for some $\diamond \in \{ \exists,\forall \}$ and that $\mathcal{P}(\phi')$ holds. In that case, we have $\varphi := \msf{f}_{\CL(\Ut,\Bt,\Bl)}(\phi) = \bullet \varphi'$ with $\varphi' := \msf{f}_{\CL(\Ut,\Bt,\Bl)}(\phi')$. Let $i \in \N_1$. Since $|\msf{Out}^Q(g(i))| = 1$, it follows that $K_{w}^i \models \exists (\bullet \phi')$ iff $K_{w}^i \models \forall (\bullet \phi')$ iff:
			\begin{itemize}
				\item If $\bullet = \lX$: $K_{w}^{i+1} \models \phi'$ iff $w[i+1:] \models \varphi'$ (by $\mathcal{P}(\phi')$) iff $w[i:] \models \varphi$;
				\item If $\bullet = \lF$: there is some $j \in \N$, such that $K_{w}^{i+j} \models \phi'$ iff there is some $j \in \N$, such that $w[i+j:] \models \varphi'$ (by $\mathcal{P}(\phi')$) iff $w[i:] \models \bullet \varphi$;
				\item If $\bullet = \lG$: for all $j \in \N$, we have $K_{w}^{i+j} \models \phi'$ iff for all $j \in \N$, we have $w[i+j:] \models \varphi'$ (by $\mathcal{P}(\phi')$) iff $w[i:] \models \bullet \varphi$.
			\end{itemize}
			Hence, $\mathcal{P}(\phi)$ holds.
			\item The case of binary operators is similar.
		\end{itemize}
		In fact, $\mathcal{P}(\phi)$ holds for all $\CTL$-formulas $\phi \in \CTL(\Ut,\Bt,\Bl)$. The lemma follows.
	\end{proof}
	
	We can now define the reduction that we consider. 
	\begin{definition}
		\label{def:reduc_ctl}
		Consider an instance $\msf{In}_{\LTL} = (\prop,\mathcal{P},\mathcal{N},B)$ of the $\LTL_\msf{Learn}$ decision problem. We define the input $\msf{In}_{\CTL} = (\prop,\mathcal{P}',\mathcal{N}',B)$ with: 
		\begin{equation*}
			\mathcal{P}' := \{ M_w \mid w \in \mathcal{P} \}
		\end{equation*}
		and 
		\begin{equation*}
			\mathcal{N}' := \{ M_w \mid w \in \mathcal{N} \}
		\end{equation*}
	\end{definition}
	Clearly this reduction can be computed in logarithmic space. Let us now show that it satisfies the desired property. 
	
	\begin{lemma}
		\label{lem:equiv_ctl}
		Consider a set of propositions $\prop$, sets of operators $\Ut \subseteq \Op{Un}{}$, $\Bt \subseteq \Op{Bin}{tp}$ and $\Bl \subseteq \Op{Bin}{lg}$, and $k \in \N \cup \{\infty\}$. Then, the input $\msf{In}_{\LTL}$ is a positive instance of the $\LTL_\msf{Learn}(\Ut,\Bt,\Bl,k)$ decision problem if and only if $\msf{In}_{\CTL}$ is a positive instance of the $\CTL_\msf{Learn}(\Ut,\Bt,\Bl,k)$ decision problem.
	\end{lemma}
	\begin{proof}
		Assume that $\msf{In}_{\LTL}$ is a positive instance of the $\LTL_\msf{Learn}(\Ut,\Bt,\Bl)$ decision problem. Let $\varphi_\LTL$ be an $\LTL$-formula in $\LTL(\Ut,\Bt,\Bl,k)$ separating $\mathcal{P}$ and $\mathcal{N}$ of size at most $B$. Consider the $\CTL$-formula $\varphi_\CTL := \msf{f}_{\LC(\Ut,\Bt,\Bl)}(\varphi_\LTL)$. By Proposition~\ref{prop:translate_ltl_ctl_simple}, we have $\Size{\varphi_\CTL} = \Size{\varphi_\LTL} \leq B$ and  $\msf{f}_{\CL(\Ut,\Bt,\Bl)}(\varphi_\CTL) = \varphi_\LTL$.
		Consider now any $M_w \in \mathcal{P}'$ with $w \in \mathcal{P}$. By Lemma~\ref{lem:equiv_word_ltl_ctl}, we have $M_{w} \models \varphi_\CTL$ if and only if $w \models \msf{f}_{\CL(\Ut,\Bt,\Bl)}(\varphi_\CTL) = \varphi_\LTL$. Since $w \in \mathcal{P}$, it follows that $w \models \varphi_\LTL$ and $M_{w} \models \varphi_\CTL$. This is similar for any $M_w \in \mathcal{N}'$ with $w \in \mathcal{N}$. Hence, $\msf{In}_{\CTL}$ is a positive instance of the $\CTL_\msf{Learn}(\Ut,\Bt,\Bl,k)$ decision problem. 
		
		Assume now that $\msf{In}_{\CTL}$ is a positive instance of the $\CTL_\msf{Learn}(\Ut,\Bt,\Bl,k)$ decision problem. Let $\varphi_\CTL$ be a $\CTL$-formula separating $\mathcal{P}'$ and $\mathcal{N}'$ of size at most $B$. Consider the $\LTL$-formula $\varphi_\LTL := \msf{f}_{\CL(\Ut,\Bt,\Bl)}(\varphi_\CTL)$. By Proposition~\ref{prop:translate_ltl_ctl_simple}, we have $\Size{\varphi_\LTL} \leq \Size{\varphi_\CTL} \leq B$. Consider now any $w \in \mathcal{P}$. By Lemma~\ref{lem:equiv_word_ltl_ctl}, we have $M_{w} \models \varphi_\CTL$ if and only if $w \models \varphi_\LTL$. Since $M_{w} \in \mathcal{P}'$, it follows that $M_{w} \models \varphi_\CTL$ and $w \models \varphi_\LTL$. This is similar for any $w \in \mathcal{N}$. Hence, $\msf{In}_{\LTL}$ is a positive instance of the $\LTL_\msf{Learn}(\Ut,\Bt,\Bl,k)$ decision problem.
	\end{proof}
	
	The proof of Theorem~\ref{thm:CTL_hard_than_LTL} is now direct.
	\begin{proof}
		It is straightforward consequence of Lemma~\ref{lem:equiv_ctl} and the fact that the reduction from Definition~\ref{def:reduc_ctl} can be computed in logarithmic space.
		
		For all $k \in \N_1$, one can straightforwardly simulate a Kripke structure by a concurrent game structure with $k$ agents by making $k-1$ agents having only one action. Hence, $\ATL_\msf{Learn}^k(\Ut,\Bt,\Bl,n)$ is at least as hard as $\CTL_\msf{Learn}(\Ut,\Bt,\Bl,n)$, under logarithmic space reduction.
	\end{proof} 
	
	\section{Learning formulas using only unary operators}
	\label{sec:unary}
	We have seen that, regardless of the operators considered, $\CTL$ learning (and therefore $\ATL$ learning as well) is at least as hard as $\LTL$ learning. Furthermore, we have also seen that $\LTL$ learning is $\msf{NP}$-hard as soon as any non-unary operator is allowed. In this section, we focus on the learning problems without non-unary operators in order to be able to distinguish the complexity of $\LTL$, $\CTL$ and different kinds of $\ATL$ learning. In this setting, we show that:
	\begin{itemize}
		\item $\LTL$ learning can now be decided in logarithmic space. To establish this result, we use simple results on equivalences of $\LTL$-formulas. Some of these equivalences were already proven in \cite{arXivFijalkow}.
		\item $\CTL$ learning remains $\msf{NP}$-complete, this is proved again via a reduction from the hitting set problem. The reduction relies heavily on the use of the next operator $\lX$. On the other hand, $\CTL$ learning without the next operator $\lX$ is equivalent to the $\CTL$ learning with formulas of size at most 5, and is in $\msf{NL}$ (i.e. it can be decided in non-deterministic logarithmic space). 
		\item On the other hand, $\ATL$ learning with at least two agents and at least two operators in $\{ \lF,\lG,\neg \}$ 
		is still $\msf{NP}$-complete. The reduction is an adaptation of the reduction for the $\CTL$ case that makes use of the additional players to mimic the behavior of the next operator $\lX$ with both the eventually and globally operators $\lF$ and $\lG$. This is the most technical reduction of the paper. However, if one only allows the use of either $\lF$ or $\lG$ without negations, then the $\ATL$ learning with at most two agents can be decided in polynomial time as, given a bound $k$ and a set of propositions $\prop$, the number of formulas to check is polynomial in $k$ and $|\prop|$.
		\item Finally, $\ATL$ learning with at least three agents remains $\msf{NP}$-complete even if only of the three operators $\lX,\lF,\lG$ is allowed. The reduction is an adaptation of the previous one where the third player is used to replace one of the operators $\lF$ or $\lG$. Hence, in this setting, we do not distinguish, complexity-wise, the 
		$\ATL$ learning problems with a fixed number of agents (at least 3) and a number of agents as part of the input. This will be done in the next section.
	\end{itemize}	
	
	
	\subsection{$\LTL$ learning}
	We first focus on the case of $\LTL$ learning. The goal of this subsection is to show the proposition below.
	\begin{proposition}
		\label{prop:ltl_unary_in_l}
		For all sets $\Ut \subseteq \Op{Un}{}$ and $\Bl \subseteq \Op{Bin}{lg}$, and $n \in \N$, the decision problem $\LTL_\msf{Learn}(\Ut,\emptyset,\Bl,n)$ is in $\msf{L}$.
	\end{proposition}
	
	To establish this proposition, we first consider $\LTL$-formulas that do not use any binary operators. First of all, since we consider ultimately periodic words, we have the (well-known, see for instance \cite[Proposition 8]{arXivFijalkow}) equivalences below.
	\begin{observation}
		\label{obser:equiv_ltl}
		Consider a non-empty set of propositions $\prop$ and some $k \in \N$. For all $\LTL$-formulas $\varphi$ on $\prop$, we have: 
		\begin{itemize}
			\item[1.] $\lF \lX^k \varphi \equiv \lX^k \lF \varphi $
			\item[2.] $\lG \lX^k \varphi \equiv \lX^k \lG \varphi $
			\item[3.] $\lF \lF \varphi \equiv \lF \varphi $
			\item[4.] $\lG \lG \varphi \equiv \lG \varphi $
			\item[5.] $\lF \lG \lF \varphi \equiv \lG \lF \varphi $
			\item[6.] $\lG \lF \lG \varphi \equiv \lF \lG \varphi $
		\end{itemize}
	\end{observation}
	\begin{proof}
		Consider an ultimately periodic word $w = u \cdot v^\omega \in (2^\prop)^\omega$. We prove the first, third and fifth items, the other ones are obtained by duality. 
		\begin{itemize}
			\item[1.] We have $w \models \lF \lX^k \varphi$ iff there is some $i \geq 1$ such that $w[i:] \models \lX^k \varphi$ iff there is some $i \geq 1$ such that $w[i+k:] \models \varphi$ iff there is some $j \geq k+1$ such that $w[j:] \models \varphi$ iff $w[k+1:] \models \lF\varphi$ iff $w \models \lX^k \lF \varphi$.
			\item[3.] We have $w \models \lF \lF \varphi$ iff there is some $i \geq 1$ such that $w[i:] \models \lF \varphi$ iff there is some $i \geq 1$ and some $j \geq i$ such that $w[i+j:] \models \varphi$ iff there is some $j \geq 1$ such that $w[j:] \models \varphi$ iff $w \models \lF\varphi$.
			\item[5.] Straightforwardly, we have $\lG \lF \varphi \implies \lF\lG \lF \varphi$. On the other hand, assume that $w \models \lF\lG \lF \varphi$, i.e. that there is some $i \geq 1$ such that $w[i:] \models \lG \lF \varphi$. Thus, for all $j\geq i$, we have $w[j:] \models \lF \varphi$, i.e. there is some $k_j \geq j$ such that $w[k_j:] \models \varphi$. Let us show that $w \models \lG \lF \varphi$. Let $j \geq 1$ and $j' := \max(i,j) \geq i$. We have $w[k_{j'}:] \models \varphi$, hence $w[j:] \models \lF\varphi$ since $k_{j'} \geq j' \geq j$. Since this holds for all $j \geq 1$, we have $w \models \lG \lF \varphi$.
		\end{itemize}
	\end{proof}
	
	In turn, let us consider the definition below of sequences of $\LTL$-operators that we will consider. 
	\begin{definition}
		\label{def:ltl_seq}
		Consider some $\Ut \subseteq \{ \lX,\lF,\lG,\neg \}$.
		We let:
		\begin{equation*}
			\msf{Qt}_{\lX}(\Ut) := \{ \msf{Y}^k \mid k \in \N,\; \msf{Y} \in \{ \epsilon \} \cup (\Ut \cap \{\lX\}) \}
		\end{equation*}
		and
		\begin{align*}
			\msf{Qt}_{\lF,\lG}(\Ut) := \begin{cases}
				\{ \lF,\lG,\lF \lG,\lG \lF \} & \text{ if }\lF,\lG \in \Ut \\
				\{ \lF \} & \text{ if }\lF \in \Ut, \neg \notin \Ut \\
				\{ \lF,\neg \lF \neg,\lF \neg \lF,\neg \lF \neg \lF \} & \text{ if }\lF,\neg \in \Ut\\
				\{ \lG \}& \text{ if }\lG \in \Ut, \neg \notin \Ut \\
				\{ \lG,\neg \lG \neg,\neg \lG \neg \lG,\lG \neg \lG \neg \} & \text{ if }\lG,\neg \in \Ut \\
				\{ \epsilon \} & \text{ if }\lF,\lG \in \Ut
			\end{cases}
		\end{align*}
		and 
		\begin{equation*}
			\msf{Qt}_{\neg}(\Ut) := \{\epsilon\} \cup \Ut \cap \{\neg\}
		\end{equation*}
		
		Then, we let:
		\begin{equation*}
			\msf{SeqQt}_{\LTL}(\Ut) := \{ \msf{Q}_{\lX} \cdot \msf{Q}_{\lF,\lG} \cdot \msf{Q}_{\neg}
			\mid \msf{Q}_{\lX} \in \msf{Qt}_{\lX}(\Ut),\; \msf{Q}_{\lF,\lG} \in \msf{Qt}_{\lF,\lG}(\Ut),\; \msf{Q}_{\neg} \in \msf{Qt}_{\neg}(\Ut) 
			\}
		\end{equation*}
	\end{definition}

	We deduce the corollary below
	.
	\begin{corollary}
		\label{coro:all_possible_cases_ltl_unary}
		Consider a non-empty set of propositions $\prop$, and some $\Ut \subseteq \{ \lX,\lF,\lG,\neg \}$. For any $\LTL$-formula $\varphi = \msf{Qt} \cdot \varphi' \in \LTL(\prop,\Ut,\emptyset,\emptyset,0)$, where $\msf{Qt}$ is a sequence of operators and $\varphi' \in \LTL(\prop,\Ut,\emptyset,\emptyset,0)$ is an $\LTL$-formula, there is a sequence of operators $\msf{Qt}' \in \msf{SeqQt}_{\LTL}(\Ut)$ such that, for $\varphi'' := \msf{Qt}' \cdot \varphi' \in \LTL(\prop,\Ut,\emptyset,\emptyset,0)$, we have $\varphi \equiv \varphi''$ and $\Size{\varphi''} \leq \Size{\varphi}$.
	\end{corollary}
	\begin{proof}
		This is a direct consequence of Observation~\ref{obser:equiv_ltl} and the fact that, for all $\LTL$-formulas $\varphi$, we have $\varphi \equiv \neg \neg \varphi$, $\lF \varphi \equiv \neg \lG \neg \varphi$ and $\lG \varphi \equiv \neg \lF \neg \varphi$.
	\end{proof}
	
	Let us now consider $\LTL$-formulas with (a bounded amount of occurrences of) binary operators.
	\begin{definition}
		Consider a non-empty set of propositions $\prop$, and some $\Ut \subseteq \Op{Un}{}$ and $\Bl \subseteq \Op{Bin}{lg}$. We define inductively the set $\LTL_{\msf{Unif}}(\prop,\Ut,\emptyset,\Bl,\infty) \subseteq \LTL(\prop,\Ut,\emptyset,\Bl,\infty)$:
		\begin{itemize}
			\item $p \in \LTL_{\msf{Unif}}(\prop,\Ut,\emptyset,\Bl,\infty)$, for all $p \in \prop$;
			\item $\msf{Qt} \cdot \varphi \in \LTL_{\msf{Unif}}(\prop,\Ut,\emptyset,\Bl,\infty)$, for all $\msf{Qt} \in \msf{SeqQt}_{\LTL}(\Ut)$ and $\varphi \in \LTL_{\msf{Unif}}(\prop,\Ut,\emptyset,\Bl,\infty)$;
			\item $\varphi_1 \bullet \varphi_2 \in \LTL_{\msf{Unif}}(\prop,\Ut,\emptyset,\Bl,\infty)$, for all $\bullet \in \Bl$ and $\varphi_1,\varphi_2 \in \LTL_{\msf{Unif}}(\prop,\Ut,\emptyset,\Bl,\infty)$.
		\end{itemize}
		Then, for all $n \in \N$, the set of $\LTL$-formulas $\LTL_{\msf{Unif}}(\prop,\Ut,\emptyset,\Bl,n)$ is defined by $\LTL_{\msf{Unif}}(\prop,\Ut,\emptyset,\Bl,n) := \LTL_{\msf{Unif}}(\prop,\Ut,\emptyset,\Bl,\infty) \cap \LTL(\prop,\Ut,\emptyset,\Bl,n)$.
	\end{definition}

	To conclude and prove Proposition~\ref{prop:ltl_unary_in_l}, we establish the three following facts, for sets $\Ut \subseteq \Op{Un}{}$ and $\Bl \subseteq \Op{Bin}{lg}$: 1) for all instances $I = (\prop,\mathcal{P},\mathcal{N},k)$ of the decision problem $\LTL_\msf{Learn}(\Ut,\emptyset,\Bl,n)$, $I$ is a positive instance if and only if there is an $\LTL_{\msf{Unif}}(\prop,\Ut,\emptyset,\Bl,n)$-formula of size at most $k$ accepting $\mathcal{P}$ and rejecting $\mathcal{N}$. Furthermore, let us fix a bound $n \in \N$ on the number of occurrences of binary operators, then: 2) the number of $\LTL_{\msf{Unif}}(\prop,\Ut,\emptyset,\Bl,n)$-formulas of size at most $k$ is polynomial in $k$ and $|\prop|$; and 3) there is a logarithmic-space algorithm that decides if an ultimately-periodic word satisfies $\LTL_{\msf{Unif}}(\prop,\Ut,\emptyset,\Bl,n)$-formulas.
	
	Before we argue that these facts hold, let us show that they directly imply Proposition~\ref{prop:ltl_unary_in_l}.
	\begin{proof}
		Consider the decision problem $\LTL_\msf{Learn}(\Ut,\emptyset,\Bl,n)$. With fact 1, deciding if an instance $(\prop,\mathcal{P},\mathcal{N},k)$ is positive amounts to deciding the existence of an $\LTL_{\msf{Unif}}(\prop,\Ut,\emptyset,\Bl,n)$-formula of size at most $k$ accepting $\mathcal{P}$ and rejecting $\mathcal{N}$. With facts 2) and 3), we can design a logarithmic space algorithmic solving that problem: it suffices to have a counter with which we enumerate all the polynomially-many $\LTL_{\msf{Unif}}(\prop,\Ut,\emptyset,\Bl,n)$-formulas to consider, and then check if one them does accept $\mathcal{P}$ and reject $\mathcal{N}$. 
	\end{proof}

	Let us now argue that these three facts hold. The first two facts are rather straightforward. Indeed, Fact 1) is direct consequence of Corollary~\ref{coro:all_possible_cases_ltl_unary}. In addition, Fact 2) is a direct consequence (which can be proved straightforwardly by induction on $n$) of the fact that the number of sequences of operators in $ \msf{SeqQt}_{\LTL}(\Ut)$ of size at most $k$ is polynomial (in fact, linear) in $k$. As it is more involved, we state Fact 3 in a lemma below. Its proof concludes this subsection.
	\begin{lemma}
		\label{lem:log_space_algo_ltl}
		Consider a bound $n \in \N$. The following decision problem can be decided in logarithmic space:
		\begin{itemize}
			\item Input: an $\LTL_{\msf{Unif}}(\prop,\Ut,\emptyset,\Bl,n)$-formula $\varphi$, and an ultimately periodic word $w = u \cdot v^\omega$;
			\item Output: Yes iff $w \models \varphi$.
		\end{itemize}
	\end{lemma}
	\begin{proof}
		The recursive algorithm depicted in Figure~\ref{fig:algo_LTL} (in which $\bullet$ refers to a binary operator) straightforwardly solves the decision problem (by simply following the $\LTL$ semantics). We have to argue that it can be implemented in logarithmic space. To execute the algorithmic, it is sufficient to keep a pointer to the current position in the word, plus additional pointers:
		\begin{itemize}
			\item To keep track of the sub-formulas currently being evaluated, and the intermediary results already computed, which is necessary with binary operators. Since there are at most $n$ occurrences of binary operators, the total number of intermediary results to keep track of is bounded by $2^n$.
			\item To keep track of the indices being evaluated, which is necessary with the $\lF$ and $\lG$ operators. However, since we consider $\LTL_{\msf{Unif}}(\prop,\Ut,\emptyset,\Bl,n)$-formulas, between binary operators, there is at most two $\lF$ operators and $\lG$ operators. Thus, again, the amount of pointers sufficient is bounded by $2n$.
		\end{itemize}
		Overall, the total number of pointers sufficient to keep track of everything is bounded, independently of the input. Thus, Algorithm  $\mathsf{DecideLTL}_{\msf{Unif}}$ can be implemented in logarithmic space.
	\end{proof}

	\begin{figure}
		\centering
		\includegraphics[scale=1]{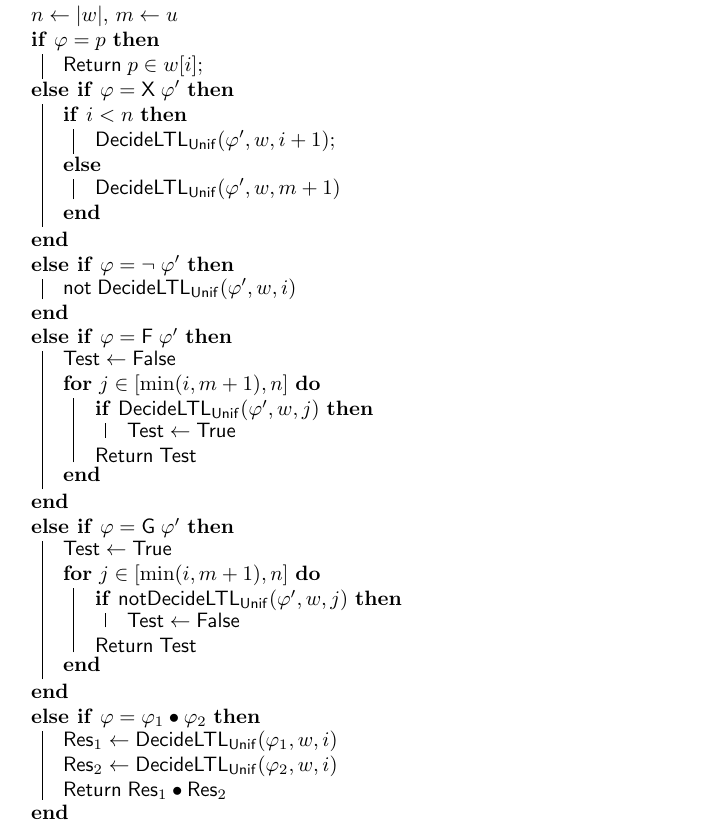}
		\caption{Algorithm $\mathsf{DecideLTL}_{\msf{Unif}}(\varphi,w,i)$.}
		\label{fig:algo_LTL}
	\end{figure}

	\subsection{Abstract recipe for $\msf{NP}$-hardness proofs and binary operators}
	As can be seen in Table~\ref{tab:summary}, we are going to establish three $\msf{NP}$-hardness results in the following cases: $\CTL$-learning with the operator $\lX$, $\ATL$-learning with two agents and both operators $\lF$ and $\lG$, and $\ATL$-learning with three agents and the operator $\lF$, or the operator $\lG$. Although these reductions differ, they all share a similar structure. The goal of this subsection is: first, to give the abstract recipe that we follow for all these reductions; and second, to state and prove the relevant lemmas to handle binary operators in our reductions. 

	\subsubsection{Abstract recipe}
	\label{subsubsec:abstract_recipe}
	Since the formulas to learn can only contain a bounded amount of binary operators, we cannot define a hitting set from a (small enough) separating formula by looking at the propositions that it uses, as we have done for LTL in Section~\ref{sec:hard_binary_op}. Instead, our reduction focuses on formulas using only unary operators (and therefore a single proposition). In that context,we create a sample (and a bound $B$) such that all separating (unary) formulas have a specific shape, and there is bijection between subsets $H \subseteq [1,\ldots,l]$ and formulas $\phi(l,H)$ of that specific shape. This correspondence allows us to
	extract a hitting set. Note that, however, although the the formulas that we consider can only use a bounded amount of binary operators, they can still use some binary operators. Therefore, to be able to only focus on unary operators, we first need to control how the (separating) formulas use binary operators. More specifically, we follow the abstract recipe below, from an $(l,C,k)$ of the hitting set problem: 
	\begin{enumerate}[label={\alph*)}]
		\item We first handle binary operators. That is, given any $n \in \N$ and some binary operators $\bullet_2 \in \Bl$, we consider $n$ propositions $p_1,\ldots,p_n$ and several positive $\mathcal{P}_{n,\bullet}$ and negative $\mathcal{N}_{n,\bullet}$ structures such that a separating formula has to use binary operators to express properties with the $n$ propositions $p_1,\ldots,p_n$.
		\label{stepa}
		\item We can now focus on unary formulas. We define the bound $B$ and additional positive $\mathcal{P}_{\msf{Un}}$ and negative $\mathcal{N}_{\msf{Un}}$ structures (using only the proposition $p$) that \textquotedblleft eliminate\textquotedblright{} certain (unary) operators or pattern of (unary) operators. This way we ensure that any unary formula separating $\mathcal{P}_{\msf{Un}}$ and $\mathcal{N}_{\msf{Un}}$ will be of the form $\phi(l,H)$, for some $H \subseteq [1,\ldots,l]$. 
		\label{stepb} 
		\item We can finally encode the hitting set problem itself. We define a negative structure (on $p$) that a unary formula $\phi(l,H)$ accepts if and only if $|H| \geq k+1$.
		\label{stepc} 
		\item For all $1 \leq i \leq n$, we define a positive structure (on $p$) that a unary formula $\phi(l,H)$ accepts if and only if $H \cap C_i \neq \emptyset$.
		\label{stepd} 
	\end{enumerate}
	By construction, the instance of the learning decision problem that we obtain is a positive instance if and only if $(l,C,k)$ also is.

	\subsubsection{Handling binary operators}
	\label{subsubsec:handling_binary_operators}
	The statement of this section deal with $\ATL$-structures and formulas because they will be used for the next three $\msf{NP}$-hardness proof. 
	
	\paragraph{Notations and definitions}
	First of all, we will use notations akin to that of regular languages to describe the formulas that we will consider.
	\begin{notation}
		\label{nota:atl_formulas}
		For any $\ATL$-formula $\phi$ and set of operators $O$, we denote by $O^* \; \phi$ the set of $\ATL$-formulas beginning with finitely many operators in $O$ followed by the $\ATL$-formula $\phi$. Furthermore, when $O$ is not a singleton, its elements may be enumerated with commas.
	\end{notation}
	
	For the $\ATL$-reductions, we will use turn-based game structures, where, at each state, only one agent is choosing the next state. Note that Kripke structure can be seen as turn-based structures, with only one player.
	\begin{definition}
		\label{def:turn_based_game_structure}
		Given any coalition of agents $A
		$, an $A$-turn-based game structure $T$ is defined by a tuple $T = \langle Q,I,\prop,\pi,\msf{AgSt},\msf{Succ} \rangle$ where $\msf{AgSt}: Q \rightarrow A$ maps every state to an agent in $A$ and $\msf{Succ}: Q \rightarrow 2^Q$ maps every state $q \in Q$ its set of successor states where the agent $\msf{AgSt}(q)$ can choose to go. Note that when a state has only one successor, i.e. one outgoing edge, the identity of the agent owning the state is irrelevant. The coalitions $A$ of agents that we consider are always such that 
		$A \subseteq \{ 1,2,3 \}$.
	\end{definition}
	
	Interestingly for us, when evaluated on turn-based structures, $\ATL$-formulas satisfy the classical equivalences w.r.t. negations. Let us introduce below a notation that refers to the dual of the operators considered.
	\begin{definition}
		\label{def:ATL_2_dual_op}
		For all $k \in \N_1$, $A \subseteq [1,\ldots,k]$, and $\msf{H} \in \{ \lX,\lF,\lG\}$, we let:
		\begin{equation*}
			\overline{\fanBr{A} \msf{H}}^k := \fanBr{[1,\ldots,k] \setminus A} \overline{\msf{H}}
		\end{equation*}
		where $\overline{\lX} := \lX$, $\overline{\lF} := \lG$ and $\overline{\lG} := \lF$.
	\end{definition}
	This definition satisfies the proposition below.
	\begin{proposition}
		\label{prop:equiv_negation_ATL_turn_based}
		Let $k \in \N_1$ and $\Ag := [1,\ldots,k]$. For all $A \subseteq \Ag$, $\msf{H} \in \{ \lX,\lF,\lG\}$ and any $\ATL$-formula $\phi$, when evaluated on $\Ag$-turn-based game structures, for all coalitions of agents $A \subseteq [1,\ldots,k]$, we have the equivalence $\neg \fanBr{A} \msf{H} \phi \equiv \overline{\fanBr{A} \msf{H}}^k \neg \phi$. 
	\end{proposition}
	\begin{proof}
		This is a consequence of the fact that two-player (with one player representing the coalition $A$ and the other the coalition $\Ag \setminus A$) turn-based reachability/safety games (respectively for the $\lF/\lG$ operator) are determined: from any starting state, either one of the players has winning (positional) strategies. 
	\end{proof}
	
	We can use this property to define a way to remove negations from sequences of unary operators, while maintaining the semantics of the formula (and not increasing the size). This is done in the definition.
	\begin{definition}
		\label{def:unnegate_unary}
		Let $k \in \N_1$, and $\Ag := [1,\ldots,k]$. For $\Ut \subseteq \Op{Un}{}$, we let $\msf{Op}(k,\Ut) := \{ \fanBr{A} \; \msf{H} \mid A \subseteq \Ag,\; \msf{H} \in \Ut \} \cup \{ \neg \} \cap \Ut$. We define inductively on $(\msf{Op}(k,\Ut))^+$ a function $\msf{UnNeg}: (\msf{Op}(k,\Ut))^* \times \{ 0,1\} \rightarrow (\msf{Op}(k,\{\lX,\lF,\lG,\neg\}))^* \times \{0,1\}$ as follows:
		\begin{itemize}
			\item For all $x \in \{0,1\}$, we let $\msf{UnNeg}(\epsilon,x) := (\epsilon,x)$.
			\item For all $\msf{O} \in \msf{Op}(k,\Ut)$, $\msf{Qt} \in (\msf{Op}(k,\Ut))^*$, and $x \in \{0,1\}$, we let:
			\begin{equation*}
				\msf{UnNeg}(\msf{O} \cdot \msf{Qt},x) := \begin{cases}
					(\msf{O} \cdot \msf{Qt}',x') 
					\text{ if }x = 0 \\
					(\overline{\msf{O}}^k \cdot \msf{Qt}',x') 
					\text{ if }x = 1
				\end{cases}
			\end{equation*}
			for $(\msf{Qt}',x') := \msf{UnNeg}( \msf{Qt},x)$.
			\item For all $x \in \{0,1\}$ and $\msf{Qt} \in (\msf{Op}(k,\Ut))^*$, we let $\msf{UnNeg}(\neg \msf{Qt},x) := \msf{UnNeg}(\msf{Qt},1-x)$.
		\end{itemize}  
	\end{definition}
	
	This definition satisfies the lemma below.
	\begin{lemma}
		\label{lem:unnegate_unary}
		Let $k \in \N_1$, $\Ag := [1,\ldots,k]$, $\Ut \subseteq \Op{Un}{}$, and $\msf{Qt} \in (\msf{Op}(k,\Ut))^*$. We have:
		\begin{itemize}
			\item For all $x \in \{0,1\}$, letting $(\msf{Qt},y) := \msf{UnNeg}(\msf{Qt},x)$, we have $\msf{Qt}' \in (\msf{Op}(k,\{\lX,\lF,\lG\}))^*$ and if $\lF \in \Ut$ if and only if $\lG \in \Ut$, then $\msf{Qt}' \in (\msf{Op}(k,\Ut \setminus \{\neg\}))^*$;
			\item For all $x \in \{0,1\}$, $(\msf{Qt},y) := \msf{UnNeg}(\msf{Qt},x)$, we have $|\msf{Qt}'| \leq |\msf{Qt}|$;
			\item Consider an $\ATL$-formula $\phi$. For all $i \in \{0,1\}$, we let $(\msf{Qt}_i,x_i) := \msf{UnNeg}(\msf{Qt},i)$ and $\psi_i := \phi$ if $x_i = 0$, and $\psi_i := \neg \phi$ otherwise. Then, we have $\msf{Qt} \cdot \phi \equiv \msf{Qt}_0 \cdot \psi_0$ and $\neg \msf{Qt} \cdot \phi \equiv \msf{Qt}_1 \cdot \psi_1$.
		\end{itemize}
	\end{lemma}
	\begin{proof}
		The first item follows from Definition~\ref{def:ATL_2_dual_op}. The second item follows from Definition~\ref{def:unnegate_unary}. The third item can be obtained by induction on $\msf{Qt}$, by successively applying Proposition~\ref{prop:equiv_negation_ATL_turn_based}.
	\end{proof}
	
	Now, let us define the kind of turn-based game structures that we will consider. Note that there is a slight difference with the informal explanations of Section~\ref{subsubsec:abstract_recipe}: in addition to the $n$ propositions $p_1,\ldots,p_n$, we do not consider a single proposition $p$ but rather two propositions $p$ and $\bar{p}$. This is done to negate formulas without using negations (i.e. informally, $\bar{p}$ will be equivalent to $\neg p$ on the structures that we will consider).
	\begin{definition}
		\label{def:proper_turn_based_structures}
		Let $n \in \N$. We let $\prop_n := \{ p,\bar{p} \} \cup \prop_n'$ where $\prop_n' := \{ p_1,\ldots,p_n \}$ be two sets of propositions. Then, for all $S \subseteq \prop_n'$, a turn-based structure $T$ is \emph{$(n,S)$-proper} if, for all states $q$ in $T$, we have $\pi(q) \in \{ S \cup \{p\},S \cup \{\bar{p}\} \}$. 
		
		
		Let $k \in \N_1$ and $\Ag := [1,\ldots,k]$. Given any two $\ATL$-formulas $\phi,\phi'$, we denote by $\phi \equiv_{k,n,S} \phi'$ the fact that $\phi$ and $\phi'$ are equivalent on $(n,S)$-proper $\Ag$-turn-based structures, i.e. for all proper turn-based structures $T$, we have $T \models \phi$ if and only if $T \models \phi'$. 
		%
		
		A structure $T$ is trivial if it contains a single self-looping state. It can be seen as a turn-based structure with any number of agents. In addition, such a structure is entirely defined by the set $S$ of propositions labeling the unique state of the structure. It is denoted $T(S)$.
	\end{definition}
	
	We can now state the two lemmas that will let us properly handle binary operators in our $\msf{NP}$-hardness proofs. The first one, Lemma~\ref{lem:n_proper}, states that from an $\ATL$-formula $\phi$ using $k$ binary operators and featuring $k$ propositions in $\prop_n'$ (for some $k \leq n$), for all $S \subseteq \prop_n'$, we can extract a unary $\ATL$-formula $\phi'$ equivalent to $\phi$ on $S$-proper structures. This is stated below.
	\begin{lemma}
		\label{lem:n_proper}
		Let $k \in \N_1$. Consider a set of unary operators $\Ut \in \Op{Un}{}$ such that $\lG \in \Ut$ if and only if $\lF \in \Ut$, and some set of binary logical operator $\Bl \subseteq \Op{Bin}{lg}$. Let $n \in \N$ and $i \leq n$. For all $\ATL$-formulas $\phi \in \ATL^k(\prop,\Ut,\emptyset,\Bl,i)$ such that $|\prop(\phi) \cap \prop_n'| = i$ and for all $S \subseteq \prop_n'$, there are two $\ATL$-formulas $\psi,\widehat{\psi} \in \ATL^k(\{p,\bar{p}\},\Ut \setminus \{\neg\},\emptyset,\Bl,0)$ such that:
		\begin{itemize}
			\item $\Size{\psi},\Size{\widehat{\psi}} \leq \Size{\phi}-2i$; and
			\item $\phi \equiv_{k,n,S} \psi$ and $\neg \phi \equiv_{k,n,S} \widehat{\psi}$.
		\end{itemize}
	\end{lemma}

	The second lemma, Lemma~\ref{lem:set_for_all_binary}, shows the existence, for all $n \in \N$ and binary operators $\bullet \in \Bl$, of the two sets of structures $\mathcal{P}_{n,\bullet}$ and $\mathcal{N}_{n,\bullet}$ mentioned in Step~\ref{stepa} in the abstract recipe described in Section~\ref{subsubsec:abstract_recipe}. It is stated below.
	\begin{lemma}
		\label{lem:set_for_all_binary}
		Let $k \in \N_1$. Consider a binary operator $\bullet \in \Op{Bin}{lg}$ and some $n \in \N_1$. There is some $S_{n,\bullet} \subseteq \prop_n'$ and two sets trivial structures $\mathcal{A}_{n,\bullet}$ and  $\mathcal{B}_{n,\bullet}$ such that:
		\begin{itemize}
			\item If an $\ATL$-formula $\phi$  distinguishes $\mathcal{A}_{n,\bullet}$ and $\mathcal{B}_{n,\bullet}$, then $\prop_n' \subseteq \prop(\phi)$.
			\item There is an $\ATL$-formula $\phi_{n,\bullet} \in \ATL^k(\prop^n,\emptyset,\emptyset,\{\bullet\},n)$ of size $2n-1$ and such that, for all positive and negative sets $\mathcal{P},\mathcal{N}$ of $(n,S_{n,\bullet})$-proper structures, there is $(\mathcal{P}',\mathcal{N}') \in \{ (\mathcal{P} \cup \mathcal{A}_{n,\bullet},\mathcal{N} \cup \mathcal{B}_{n,\bullet}),(\mathcal{N} \cup \mathcal{B}_{n,\bullet},\mathcal{P} \cup \mathcal{A}_{n,\bullet}) \}$ such that, for all $\ATL$-formulas $\phi \in \ATL^k(\prop^n,\Ut \setminus \{\neg\},\emptyset,\Bl,0)$, we have: 
			\begin{center}
				$\phi$ accepts $\mathcal{P}$ and rejects $\mathcal{N}$ if and only if $\psi := \phi \bullet \phi_{n,\bullet}$ accepts $\mathcal{P}'$ and rejects $\mathcal{N}'$
			\end{center}
		\end{itemize}
	\end{lemma}
	
	With both of those lemmas, we are able to show the theorem below that allows to prove the $\msf{NP}$-hardness of the learning decision problems that we consider, with an arbitrary bound on the number of occurrences of binary operators, by showing the $\msf{NP}$-hardness of the learning problem without binary operator.
	\begin{theorem}
		\label{thm:unary_is_sufficient}
		Let $k \in \N_1$, $\Ag := [1,\ldots,k]$, and $\Ut \in \Op{Un}{}$. 
		Assume that there is a function $f_k$ computable in logarithmic space that takes an input an instance $(l,C,k')$ of the hitting set problem $\msf{Hit}$ and returns an instance of the learning decision problem $\ATL^k_\msf{Learn}(\msf{Un},\emptyset,\emptyset,0)$ such that for all instances $(l,C,k')$:
		\begin{itemize}
			\item The set of propositions in $f_k((l,C,k'))$ is $\prop_0$;
			\item All structures in $f_k((l,C,k'))$ are $(0,\emptyset)$-proper $\Ag$-turn-based structures;
			\item The three statements below are equivalent:
			\begin{itemize}
				\item $(l,C,k')$ is a positive instance of $\msf{Hit}$;
				\item $f_k((l,C,k'))$ is a positive instance of $\ATL^k_\msf{Learn}(\Ut \setminus \{\neg\},\emptyset,\emptyset,0)$;
				\item $f_k((l,C,k'))$ is a positive instance of $\ATL^k_\msf{Learn}(\Ut \cup \{\lF,\lG\},\emptyset,\emptyset,0)$.
			\end{itemize}
		\end{itemize}
		Then, for all $\Bl \subseteq \Op{Bin}{lg}$ and $n \in \N$, the decision problem $\ATL^k_\msf{Learn}(\Ut,\emptyset,\Bl,n)$ is $\msf{NP}$-complete.
	\end{theorem}
	\begin{proof}
		Let $\Bl \subseteq \Op{Bin}{lg}$ and $n \in \N$. As mentioned in Proposition~\ref{prop:easy}, the decision problem $\ATL^k_\msf{Learn}(\Ut,\emptyset,\Bl,n)$ is in $\msf{NP}$. Let us now show that it in $\msf{NP}$-hard. 
		
		If $n = 0$ or $\Bl = \emptyset$, the two problems $\ATL^k_\msf{Learn}(\Ut,\emptyset,\emptyset,0)$ and $\ATL^k_\msf{Learn}(\Ut,\emptyset,\Bl,n)$ are the same. Furthermore, or all instances $(l,C,k')$ of $\msf{Hit}$, we have that if $f_k((l,C,k'))$ is a positive instance of $\ATL^k_\msf{Learn}(\Ut \setminus \{\neg\},\emptyset,\emptyset,0)$, then it is also a positive instance of $\ATL^k_\msf{Learn}(\Ut,\emptyset,\emptyset,0)$, and similarly if $f_k((l,C,k'))$ is a positive instance of $\ATL^k_\msf{Learn}(\Ut,\emptyset,\emptyset,0)$, then it is also a positive instance of $\ATL^k_\msf{Learn}(\Ut \cup \{\lF,\lG\},\emptyset,\emptyset,0)$. Therefore, $f_k((l,C,k'))$ is a positive instance of $\ATL^k_\msf{Learn}(\Ut \setminus \{\neg\},\emptyset,\emptyset,0)$ if and only if it is a positive instance $\ATL^k_\msf{Learn}(\Ut,\emptyset,\emptyset,0)$. We can conclude that the decision problem $\ATL^k_\msf{Learn}(\Ut,\emptyset,\Bl,n)$ is $\msf{NP}$-hard. 		
		
		Assume now that we have $n \in \N_1$ and $\Bl \neq \emptyset$. We let $\bullet \in \Bl$ and we consider the set of propositions $S_{n,\bullet} \subseteq \prop_n'$ and the two sets of trivial structures $\mathcal{A}_{n,\bullet}$ and  $\mathcal{B}_{n,\bullet}$ from Lemma~\ref{lem:set_for_all_binary}, along with the formula $\phi_{n,\bullet} \in \ATL^k(\prop^n,\emptyset,\emptyset,\{\bullet\},n)$ of size at most $2n-1$. Now, consider an instance $(l,C,k')$ of the hitting set problem $\msf{Hit}$. Let $\msf{In}_0 := f_k((l,C,k')) =  (\prop_0,\mathcal{P},\mathcal{N},B)$. 
		Note that all structures in $\mathcal{P} \cup \mathcal{N}$ are $(0,\emptyset)$-proper $\Ag$-turn-based structures. We let $\widehat{\mathcal{P}},\widehat{\mathcal{N}}$ be two sets of $\Ag$-turn-based structures equal to $\mathcal{P},\mathcal{N}$ respectively except that the labels of all states is changed from $x \in \prop_0 = \{p,\bar{p}\}$ to $S_{n,\bullet} \cup \{x\} \subseteq \prop_n$. That way, all the turn-based structures in $\widehat{\mathcal{P}}$ and $\widehat{\mathcal{N}}$ are $(n,S_{n,\bullet})$-proper structures. Then, we let $(\mathcal{P}',\mathcal{N}') \in \{ (\widehat{\mathcal{P}} \cup \mathcal{A}_{n,\bullet},\widehat{\mathcal{N}} \cup \mathcal{B}_{n,\bullet}),(\widehat{\mathcal{N}} \cup \mathcal{B}_{n,\bullet},\widehat{\mathcal{P}} \cup \mathcal{A}_{n,\bullet})\}$ be as in the second point of Lemma~\ref{lem:set_for_all_binary} and we define the input $\msf{In}_n := (\prop_n,\mathcal{P}',\mathcal{N}',B+2n)$ of the decision problem $\ATL^k_\msf{Learn}(\Ut,\emptyset,\Bl,n)$. Note that, the structures in $\mathcal{A}_{n,\bullet}$ and $\mathcal{B}_{n,\bullet}$ are defined independently of the input $(l,C,k')$ and, by assumption, the input $\msf{In}_0$ can be computed in logarithmic space from the instance $(l,C,k')$. Therefore, the input $\msf{In}_n$ can also be computed in logarithmic space from the instance $(l,C,k')$.
		
		Let us show that $\msf{In}_n$ is a positive instance of the decision problem $\ATL^k_\msf{Learn}(\Ut,\emptyset,\Bl,n)$ if and only if $\msf{In}_0$ is a positive instance decision problem $\ATL^k_\msf{Learn}(\Ut \setminus \{\neg\},\emptyset,\emptyset,0)$. 
		
		Assume that $\msf{In}_n$ is a positive instance of the decision problem $\ATL^k_\msf{Learn}(\Ut,\emptyset,\Bl,n)$. Consider a formula $\phi \in \ATL^k_\msf{Learn}(\prop_n,\Ut,\emptyset,\Bl,n)\subseteq \ATL^k_\msf{Learn}(\prop_n,\Ut \cup \{\lF,\lG\},\emptyset,\Bl,n)$ of size at most $B+2n$ that accepts $\mathcal{P}'$ and rejects $\mathcal{N}'$. 
		In that case, the formula $\phi$ distinguishes $\mathcal{A}_{n,\bullet}$ and $\mathcal{B}_{n,\bullet}$,  and therefore, by Lemma~\ref{lem:set_for_all_binary}, we have $\prop_n' \subseteq \prop(\phi)$. Hence, by Lemma~\ref{lem:n_proper}, there are two formulas $\psi,\widehat{\psi} \in \ATL^k_\msf{Learn}(\prop_0,(\Ut \cup\{\lF,\lG\}) \setminus \{\neg\},\emptyset,\emptyset,0)$ such that $\Size{\psi},\Size{\widehat{\psi}} \leq \Size{\phi} - 2n \leq B$ and $\phi \equiv_{k,n,S_{n,\bullet}} \psi$ and $\neg \phi \equiv_{k,n,S_{n,\bullet}} \widehat{\psi}$. Let $\phi ' \in \{\psi,\widehat{\psi}\}$ be such that $\phi' = \psi$ if and only if $(\mathcal{P}',\mathcal{N}') = (\widehat{\mathcal{P}} \cup \mathcal{A}_{n,\bullet},\widehat{\mathcal{N}} \cup \mathcal{B}_{n,\bullet})$. That way, we have that $\phi' \in \ATL^k_\msf{Learn}(\prop_n,\Ut \cup \{\lF,\lG\},\emptyset,\emptyset,0)$ accepts the set $\widehat{\mathcal{P}}$ and rejects the set $\widehat{\mathcal{N}}$. Since $\prop(\phi') \subseteq \prop_0 = \{p,\bar{p}\}$, and by definition of $\widehat{\mathcal{P}}$ and $\widehat{\mathcal{N}}$, it follows that $\phi'$ also accepts $\mathcal{P}$ and rejects $\mathcal{N}$. Thus, $\msf{In}_0$ is a positive instance of the decision problem $\ATL^k_\msf{Learn}(\Ut \cup \{\lF,\lG\},\emptyset,\emptyset,0)$, and is therefore also a positive instance of the decision problem $\ATL^k_\msf{Learn}(\Ut \setminus \{\neg\},\emptyset,\emptyset,0)$.
		
		Assume now that $\msf{In}_0$ is a positive instance of the decision problem $\ATL^k_\msf{Learn}(\Ut \setminus \{\neg\},\emptyset,\emptyset,0)$. Consider an $\ATL$-formula $\phi \in \ATL^k_\msf{Learn}(\prop_0,\Ut \setminus \{\neg\},\emptyset,\emptyset,0)$ of size at most $B$ that accepts $\mathcal{P}$ and rejects $\mathcal{N}$. 
		We consider the $\ATL$ formula $\phi' := \phi \bullet \phi_{n,\bullet} \in \ATL^k_\msf{Learn}(\prop_0,\Ut \setminus \{\neg\},\emptyset,\emptyset,0)$. We have $\Size{\phi'} \leq \Size{\phi} + \Size{\phi_{n,\bullet}} + 1 \leq B + 2n-1+1 = B+2n$. Furthermore, by definition of the sets $\mathcal{P}',\mathcal{N}'$, we have that $\phi'$ accepts $\mathcal{P}'$ and rejects $\mathcal{N}'$ since $\phi$ accepts $\mathcal{P}$ and rejects $\mathcal{N}$\footnote{Note that, to apply Lemma~\ref{lem:set_for_all_binary}, it is crucial that $\phi$ does not feature negations.}. Thus, $\msf{In}_n$ is a positive instance of the decision problem $\ATL^k_\msf{Learn}(\Ut,\emptyset,\Bl,n)$.
		
		Overall, we obtain that $\msf{In}_n$ is a positive instance of the decision problem $\ATL^k_\msf{Learn}(\Ut,\emptyset,\Bl,n)$ if and only if $(l,C,k')$ is a positive instance of the hitting set problem $\msf{Hit}$. Hence, the decision problem $\ATL^k_\msf{Learn}(\Ut,\emptyset,\Bl,n)$ is $\msf{NP}$-hard.
	\end{proof}
	
	Let us now prove these two Lemmas~\ref{lem:n_proper} and~\ref{lem:set_for_all_binary}. We start with the proof of Lemma~\ref{lem:n_proper}, but before we proceed to it, we state and prove below a crucial lemma regarding the shape of formulas that use a bounded amount of binary operators and feature many propositions.	
	\begin{lemma}
		\label{lem:non_intersect_binary}
		For all $k \in \N$, for all $\ATL$-formulas 
		$\phi$, if $|\phi|_{\msf{Bin}} = k$, then $|\prop(\phi)| \leq k+1$. In addition, if 
		$|\prop(\phi)| = k+1$, and $\phi = \phi_1 \bullet \phi_2$ for some binary operator $\bullet \in \Bl$, and $\ATL$-formulas $\phi_1,\phi_2$, then $\prop(\phi_1) \cap \prop(\phi_2) = \emptyset$.
	\end{lemma}
	\begin{proof}
		We prove this lemma by induction on $k$. It straightforwardly holds for $k=0$. Assume now that it holds for all $i \leq k$, for some $k \in \N$. Then, consider an $\ATL$-formula $\phi$ for which $|\phi|_{\msf{Bin}} = k+1$. Let us first assume that $\phi = \phi_1 \bullet \phi_2$, for some with binary operator $\bullet$, and $\ATL$-formulas $\phi_1,\phi_2$. We are going to modify the formula $\phi_1$. To do so, we consider a fresh proposition $x \notin \prop(\phi)$. Then, for all $\ATL$-formulas $\psi,\psi'$, we let $\msf{ch}_x(\psi,\psi') := x$ if $\psi \in \msf{SubF}(\phi_2)$, and 
		$\msf{ch}_x(\psi,\psi') := \psi'$ otherwise. Now, we define inductively a function $\msf{tr}_x$ on $\ATL$-formulas as follows: 
		\begin{itemize}
			\item for all propositions $p$, we let $\msf{tr}_x(p) : = \msf{ch}_x(p,x)$;
			\item for all unary operators $\bullet_1$ and $\ATL$-formulas $\psi$, we let $\msf{tr}_x(\bullet_1 \psi) := \bullet_1 \msf{ch}_x(\psi,\msf{tr}_x(\psi))$;
			\item for all binary operators $\bullet_2$, and $\psi,\psi'$ two $\ATL$-formulas, we let  $\msf{tr}_x(\psi \bullet_2 \psi') := \msf{ch}_x(\psi,\msf{tr}_x(\psi)) \bullet_2 \msf{ch}_x(\psi',\msf{tr}_x(\psi'))$.
		\end{itemize}
		We are going to apply the function $\msf{tr}_x$ to the formula $\phi_1$, and then use the properties satisfied by the obtained formula. Thus, 
		we show by induction on $\ATL$-formulas $\psi$ the property $\mathcal{P}(\psi)$: $\msf{SubF}(\msf{ch}_x(\psi,\msf{tr}_x(\psi))) \cap \msf{SubF}(\phi_2) = \emptyset$ 
		and 
		$\msf{SubBin}(\msf{ch}_x(\psi,\msf{tr}_x(\psi))) \subseteq \msf{tr}_x(\msf{SubBin}(\psi) \setminus \msf{SubBin}(\phi_2))$. First of all, for all $\ATL$-formulas, if $\psi \in \msf{SubF}(\phi_2)$, then $\msf{ch}_x(\psi,\msf{tr}_x(\psi)) = x$, in which case $\msf{SubF}(\msf{ch}_x(\psi,\msf{tr}_x(\psi))) \cap \msf{SubF}(\phi_2) = \emptyset$, and $\msf{SubBin}(\msf{ch}_x(\psi,\msf{tr}_x(\psi))) = \emptyset$. Hence, $\mathcal{P}(\psi)$ always holds in that case. Thus, we will focus on the cases where $\psi \notin \msf{SubF}(\phi_2)$ and $\msf{ch}_x(\psi,\msf{tr}_x(\psi)) = \msf{tr}_x(\psi)$.
		
		Furthermore, $\mathcal{P}(p)$ straightforwardly holds for all propositions $p$. Assume now that $\mathcal{P}(\psi)$ holds for some $\ATL$-formula $\psi$. Consider any unary operator $\bullet_1$. We have $\msf{tr}_x(\bullet_1 \; \psi) = \bullet_1 \; \msf{ch}_x(\psi,\msf{tr}_x(\psi))$, hence $\msf{SubF}(\msf{tr}_x(\bullet_1 \; \psi)) = \{ \bullet_1 \; \msf{ch}_x(\psi,\msf{tr}_x(\psi))\} \cup \msf{SubF}(\msf{ch}_x(\psi,\msf{tr}_x(\psi)))$ and $\msf{SubBin}(\bullet_1 \; \msf{tr}_x(\psi)) = \msf{SubBin}(\msf{tr}_x(\psi))
		$. Thus, $\mathcal{P}(\bullet_1 \; \psi)$ follows directly from $\mathcal{P}(\psi)$. 
		
		Assume now that $\mathcal{P}(\psi_1),\mathcal{P}(\psi_2)$ hold for some $\ATL$-formulas $\psi_1,\psi_2$. Consider any binary operator $\bullet_2$. We have $\msf{tr}_x(\psi \bullet_2 \psi') := \msf{ch}_x(\psi,\msf{tr}_x(\psi))  \bullet_2 \msf{ch}_x(\psi',\msf{tr}_x(\psi'))$, hence:
		\begin{equation*}
			\msf{SubF}(\msf{tr}_x(\psi \bullet_2 \psi')) = \{ 
			\msf{tr}_x(\psi \bullet_2 \psi') \} \cup \msf{SubF}(\msf{ch}_x(\psi,\msf{tr}_x(\psi))) \cup \msf{SubF}(\msf{ch}_x(\psi',\msf{tr}_x(\psi')))
		\end{equation*}
		Since $\msf{SubF}(\msf{ch}_x(\psi,\msf{tr}_x(\psi))) \cap \msf{SubF}(\phi_2) = \emptyset$ (by $\mathcal{P}(\psi)$) and $\msf{SubF}(\msf{ch}_x(\psi',\msf{tr}_x(\psi'))) \cap \msf{SubF}(\phi_2) = \emptyset$ (by $\mathcal{P}(\psi')$), it follows that $\msf{SubF}(\msf{tr}_x(\psi \bullet_2 \psi')) \cap \msf{SubF}(\phi_2) = \emptyset$. Furthermore, we have:
		\begin{equation*}
			\msf{SubBin}(\msf{tr}_x(\psi \bullet_2 \psi')) = \{ \msf{tr}_x(\psi \bullet_2 \psi') \} \cup \msf{SubBin}(\msf{ch}_x(\psi,\msf{tr}_x(\psi))) \cup \msf{SubBin}(\msf{ch}_x(\psi',\msf{tr}_x(\psi')))
		\end{equation*}
		with $\{ \msf{tr}_x(\psi \bullet_2 \psi') \} = \msf{tr}_x(\{ \psi \bullet_2 \psi' \})$, and for any $\tilde{\psi} \in \{\psi,\psi'\}$, we have $\msf{SubBin}(\msf{ch}_x(\tilde{\psi},\msf{tr}_x(\tilde{\psi}))) \subseteq \msf{tr}_x(\msf{SubBin}(\tilde{\psi}) \setminus \msf{SubBin}(\phi_2))$ (by $\mathcal{P}(\tilde{\psi})$). Thus, $\mathcal{P}(\psi \bullet_2 \psi')$ follows from the fact that, if $\psi \bullet_2 \psi' \notin \msf{SubF}(\phi_2)$, then $\msf{SubBin}(\psi \bullet_2 \psi') \setminus \msf{SubBin}(\phi_2) = \{\psi \bullet_2 \psi'\} \cup (\msf{SubBin}(\psi) \setminus \msf{SubBin}(\phi_2)) \cup (\msf{SubBin}(\psi') \setminus \msf{SubBin}(\phi_2))$. In fact, $\mathcal{P}(\psi)$ holds for all $\ATL$-formulas $\psi$.
		
		Now, let $\phi_1^x := \msf{ch}_x(\phi_1,\msf{tr}_x(\phi_1))$ and $\phi^x : \phi_1^x \bullet \phi_2$. By $\mathcal{P}(\phi_1^x)$, we have the following facts:
		\begin{itemize}
			\item $\msf{SubF}(\phi_1^x) \cap \msf{SubF}(\phi_2) = \emptyset$;
			\item $\msf{SubBin}(\phi^x) = \{ \phi^x \} \cup \msf{SubBin}(\phi_1^x) \cup \msf{SubBin}(\phi_2)$ and $\msf{SubBin}(\phi) = \{ \phi \} \cup \msf{SubBin}(\phi_1) \cup \msf{SubBin}(\phi_2)$. Since $\msf{SubBin}(\phi_1^x) \subseteq \msf{tr}_x(\msf{SubBin}(\phi_1) \setminus \msf{SubBin}(\phi_2))$, we have:
			\begin{align*}
				|\phi^x|_{\msf{Bin}} = |\msf{SubBin}(\phi^x)| & = 1 + |\msf{SubBin}(\phi_1^x)| + |\msf{SubBin}(\phi_2)| \\ 
				& \leq 1 + |\msf{SubBin}(\phi_1) \setminus \msf{SubBin}(\phi_2)| + |\msf{SubBin}(\phi_2)| \\ 
				& = |\msf{SubBin}(\phi)| = |\phi|_{\msf{Bin}}
			\end{align*}
			That is, $|\phi^x|_{\msf{Bin}} \leq |\phi|_{\msf{Bin}}$;
			\item If $\msf{SubF}(\phi_1) \cap \msf{SubF}(\phi_2) = \emptyset$, then $\phi^x = \phi$ and $\prop(\phi) = \prop(\phi^x)$. Otherwise, we have $\prop(\phi) \cup \{x\} = \prop(\phi^x)$ (it can be proved straightforwardly by induction).
		\end{itemize}
		Therefore, we have $|\phi^x|_{\msf{Bin}} = |\phi_1^x|_{\msf{Bin}} + |\phi_2|_{\msf{Bin}} + 1 \leq k+1$. Hence, we can apply our induction hypotheses to $|\phi_1^x|_{\msf{Bin}} \leq k$ and $|\phi_2^x|_{\msf{Bin}} \leq k$ to obtain that: $|\prop(\phi_1^x)| \leq |\phi_1^x|_{\msf{Bin}} + 1$, $|\prop(\phi_2)| \leq |\phi_2|_{\msf{Bin}} + 1$. Hence, $|\prop(\phi)| \leq |\prop(\phi^x)| = |\prop(\phi_1^x)| + |\prop(\phi_2)| \leq k+2$. In addition, if $\msf{SubF}(\phi_1) \cap \msf{SubF}(\phi_2) \neq \emptyset$, then $|\prop(\phi)| < |\prop(\phi^x)|$, and thus $|\prop(\phi)| < k+2$.
		
		If instead there is an $\ATL$-formula $\phi' = \phi_1 \bullet \phi_2$, for some $\ATL$-formulas $\phi_1,\phi_2$ and binary operator $\bullet$, such that $\phi = \msf{Qt} \cdot \phi'$ where $\msf{Qt}$ is a non-empty sequence of unary operators, we can apply the above arguments to the formula $\phi'$. Thus, the property holds also for $k+1$. In fact, it holds for all $k \in \N$. The lemma follows.
	\end{proof}
	We can now proceed to the proof of Lemma~\ref{lem:n_proper}. 
	\begin{proof}
		First of all, note that for all $\ATL$-formulas $\phi$ such that $\prop(\phi) \subseteq \prop_n'$, and for all $S \subseteq \prop_n'$, there is $\psi \in \{\msf{True},\msf{False}\}$ such that we have $\phi \equiv_{k,n,S} \psi$. This can be straightforwardly established by induction on $\ATL$-formulas $\phi$
		.
		
		Now, let us show by induction on $0 \leq i \leq n$ the property $\mathcal{P}(i)$: for all $\ATL$-formulas $\phi \in \ATL^k(\prop,\Ut,\emptyset,\Bl,i)$ such that $|\prop(\phi) \cap \prop_n'| = i$, and for all $S \subseteq \prop_n'$, there are two $\ATL$-formulas $\psi,\widehat{\psi} \in \ATL^k(\{p,\bar{p}\},\Ut \setminus \{\neg\},\emptyset,\Bl,0)$ such that: $\Size{\psi},\Size{\widehat{\psi}} \leq \Size{\phi}-2i$, $\phi \equiv_{k,n,S} \psi$, and $\neg \phi \equiv_{k,n,S} \widehat{\psi}$.
		
		Consider first the case $i = 0$. Let $\phi \in \ATL^k(\prop,\Ut,\emptyset,\Bl,0)$. We have $\phi := \msf{Qt} \cdot r$ for some $r \in \{p,\bar{p}\}$ and $\msf{Qt} \in (\msf{Op}(k,\Ut))^*$. For $j \in \{0,1\}$, we let $(\msf{Qt}_j,x_j) := \msf{UnNeg}(\msf{Qt},j)$ and $r_j \in \{p,\bar{p}\}$ be such that $r_j = r$ if and only if $x_j = 0$. Then, we let $\psi := \msf{Qt}_0 \cdot r_0$ and $\widehat{\psi} := \msf{Qt}_1 \cdot r_1$. By Lemma~\ref{lem:unnegate_unary}, we have: $\psi,\widehat{\psi} \in \ATL^k(\prop,\Ut \setminus \{\neg\},\emptyset,\Bl,0)$ and $\Size{\psi},\Size{\widehat{\psi}} \leq \Size{\psi}$. Furthermore, we also have, for all $S \subseteq \prop_n'$, $\phi \equiv_{k,n,S} \psi$, and $\neg \phi \equiv_{k,n,S} \widehat{\psi}$, again by Lemma~\ref{lem:unnegate_unary}, and since $p \equiv_{k,n,S} \neg \bar{p}$ (by definition of $(n,S)$-proper structures). Thus, $\mathcal{P}(0)$ holds. 

		Assume now that it $\mathcal{P}(j)$ holds for all $j \leq i$, for some $i \leq n-1 \in \N$. Consider an $\ATL$-formula $\phi \in \ATL^k(\prop,\Ut,\emptyset,\Bl,i+1)$ such that $|\prop(\phi) \cap \prop_n'| = i+1$. Let $S \subseteq \prop_n'$. If $|\prop(\phi)| = i+1$, then we have $\Size{\phi} \geq 2i+3$ by Lemma~\ref{lem:at_least_that_many_subformulas}\footnote{This lemma is established for $\LTL$-formulas, but its result can be applied to $\ATL$-formulas as well}. 
		Furthermore, there is $\psi,\widehat{\psi} \in \{ \msf{True},\msf{False}\}$ of size 1 such that, $\phi \equiv_{k,n,S} \psi$ and $\neg\phi \equiv_{k,n,S} \widehat{\psi}$. Assume now that $|\prop(\phi)| > i+1$. Since $\prop(\phi) \leq |\phi|_{\msf{Bin}} + 1 \leq i+2$ (by Lemma~\ref{lem:non_intersect_binary}), it follows that $|\prop(\phi)| = i+2$. Let $x \in \{p,\bar{p}\}$ be such that $x \in \prop(\phi)$. Then, there is some sequence of unary operators $\msf{Qt} \in (\msf{Op}(k,\Ut))^*$, a binary operator $\bullet$ and two $\ATL$-formulas $\phi_1,\phi_2 \in \ATL^k(\prop,\Ut,\emptyset,\Bl,i)$ such that $\phi = \msf{Qt} \cdot (\phi_1 \bullet \phi_2)$. Thus, $\prop(\phi) = \prop(\phi_1 \bullet \phi_2)$. Since we have $\phi_1 \bullet \phi_2 \in \ATL^k(\prop,\Ut,\emptyset,\Bl,i+1)$ and $|\prop(\phi_1 \bullet \phi_2)| = i+2$, it follows that $\prop(\phi_1) \cap \prop(\phi_2) = \emptyset$, by Lemma~\ref{lem:non_intersect_binary}. Without loss of generality, let us assume that $x \in \prop(\phi_1)$. In that case, we have $\prop(\phi_2) \subseteq \prop_n'$. Therefore, as mentioned at the beginning of this proof, there is some $\psi_{\msf{tr}} \in \{\msf{True},\msf{False}\}$ such that $\phi_2 \equiv_{k,S} \psi_{\msf{tr}}$. Hence, $\phi_1 \bullet \phi_2 \equiv_{k,S} \phi_1 \bullet \psi_{\msf{tr}}$. It follows, since $\bullet$ is a binary logical operator, that 
		there is some $\Psi,\widehat{\Psi} \in \{\msf{True},\msf{False},\phi_1,\neg \phi_1\}$ such that $\phi_1 \bullet \phi_2 \equiv_{k,n,S} \Psi$ and $\neg (\phi_1 \bullet \phi_2) \equiv_{k,n,S} \widehat{\Psi}$. Now, let $m_1 := |\prop(\phi_1) \cap \prop_n'| \leq i$. We have $m_2 := |\prop(\phi_2)|$ such that $m_1 + m_2 = |\prop(\phi)| = i+1$. In addition, by $\mathcal{P}(m_1)$, $\Psi$ and $\widehat{\Psi}$ can be chosen such that $\Psi,\widehat{\Psi} \in \ATL^k(\{p,\bar{p}\},\Ut \setminus \{\neg\},\emptyset,\Bl,0)$ and $\Size{\Psi},\Size{\widehat{\Psi}} \leq \Size{\phi_1} - 2m_1$. In addition, by Lemma~\ref{lem:at_least_that_many_subformulas}, we have $\Size{\phi_2} \geq 2m_2-1$. Overall, we obtain that $\Psi,\widehat{\Psi}$ can be chosen such that (note that since $\prop(\phi_1) \cap \prop(\phi_2) = \emptyset$, we have $\msf{SubF}(\phi_1) \cap \msf{SubF}(\phi_2) = \emptyset$):
		\begin{align*}
			\Size{\Psi},\Size{\widehat{\Psi}} & \leq \Size{\phi_1} - 2m_1 = \Size{\phi_1 \bullet \phi_2} - (\Size{\phi_2} + 1) - 2m_1 \\ 
			& \leq \Size{\phi_1 \bullet \phi_2} - 2m_2 - 2m_1 = \Size{\phi_1 \bullet \phi_2} - 2 (i+1)
		\end{align*}
		Then, for $j \in \{0,1\}$, we let $(\msf{Qt}_j,x_j) := \msf{UnNeg}(\msf{Qt},j)$ and $\psi_j \in \{\Psi,\widehat{\Psi}\}$ be such that $\psi_j \equiv_{k,n,S} \Psi$ if and only if $x_j = 0$. Then, we let $\psi := \msf{Qt}_0 \cdot \psi_0$ and $\widehat{\psi} := \msf{Qt}_1 \cdot \psi_1$. By Lemma~\ref{lem:unnegate_unary}, we have: $\psi,\widehat{\psi} \in \ATL^k(\prop,\Ut \setminus \{\neg\},\emptyset,\Bl,0)$, $\phi \equiv_{k,n,S} \psi$ and $\neg \phi \equiv_{k,n,S} \widehat{\psi}$, and $|\msf{Qt}_0|,|\msf{Qt}_1| \leq |\msf{Qt}|$. Thus, we have $\Size{\psi} = |\msf{Qt}_0| + \Size{\psi_0} \leq |\msf{Qt}| + \Size{\phi_1 \bullet \phi_2} - 2(i+1) = \Size{\phi} - 2(i+1)$, and similarly for $\widehat{\psi}$. Thus, $\mathcal{P}(i+1)$ follows. In fact, $\mathcal{P}(i)$ holds for all $0 \leq i \leq n$. The lemma is then given by $\mathcal{P}(n)$. 
	\end{proof}

	Let us consider Lemma~\ref{lem:set_for_all_binary}. Before we proceed to its proof, we state below a useful lemma analogous to Lemma~\ref{lem:distinguish}
	.
	\begin{lemma}
		\label{lem:distinguish_ATL}
		For all sets of propositions $\prop$, for all $Y \subseteq \prop$, and $S,S' \subseteq \prop$, if, for all $x  \in Y$, we have $x \in S$ if and only if $x \in S'$, then an $\ATL$-formula $\phi$ such that $\prop(\phi) \subseteq Y$ cannot distinguish the trivial structures $T(S)$ and $T(S')$. 
	\end{lemma}
	\begin{proof}
		A straightforward proof by induction on $\ATL$-formulas establishes the lemma.
	\end{proof}

	We can now proceed to the proof of Lemma~\ref{lem:set_for_all_binary}.
	\begin{proof}
		Let $n \in \N$.
		\begin{itemize}
			\item Assume that $\bullet = \lor$. The cases $\bullet = \Rightarrow$ and $\bullet = \Leftarrow$ are analogous.
			We let $S_{n,\bullet} := \emptyset$ and, for all $1 \leq i \leq n$, we let $S_i := \{p_i\}$. Then, we define: $\mathcal{A}_{n,\bullet} := \{ T(S_i) \mid 1 \leq i \leq n\}$, $\mathcal{B}_{n,\bullet} := \{ T(S_{n,\bullet}) \}$, and  $\phi_{n,\bullet} := p_1 \lor \ldots \lor p_n$.
			
			That way, for all $1 \leq i \leq n$, to distinguish $T(S_i)$ and $T(S)$, an $\ATL$-formula $\phi$ needs to be such that $p_i \in \prop(\phi)$, by Lemma~\ref{lem:distinguish_ATL}. Hence, if an $\ATL$-formula $\phi$ separates $\mathcal{A}_{n,\bullet}$ and $\mathcal{B}_{n,\bullet}$, we have $\prop_n' \subseteq \prop(\phi)$. Consider now any two positive and negative sets $\mathcal{P},\mathcal{N}$ of $(n,S_{n,\bullet})$-proper structures. Let $\mathcal{P}' := \mathcal{P} \cup \mathcal{A}_{n,\bullet}$ and $\mathcal{N}' := \mathcal{N} \cup \mathcal{B}_{n,\bullet}$. Consider any formula $\phi \in \ATL(\prop^n,\Ut,\emptyset,\Bl,0)$ and $\psi := \phi \lor \phi_{n,\bullet}$.  For all $1 \leq i \leq n$, we have $T(S_i) \models \phi_{n,\bullet}$ and $T(S_{n,\bullet}) \not\models \phi_{n,\bullet}$. Furthermore, since $\phi$ does not use negations, we also have $T(S_{n,\bullet}) \not\models \phi$. In addition, by definition of $S_{n,\bullet}$, for all $(n,S_{n,\bullet})$-proper structures $T$, we have $T \not\models \phi_{n,\bullet}$. Thus, we have that $\phi$ accepts $\mathcal{P}$ and rejects $\mathcal{N}$ if and only if $\psi$ accepts $\mathcal{P}'$ and rejects $\mathcal{N}'$.
			\item Assume now that $\bullet = \wedge$. The cases $\bullet = \prescript{\neg}{}{\Rightarrow}$ and $\bullet = \prescript{\neg}{}{\Leftarrow}$ are analogous.
			We let $S_{n,\bullet} := \prop_n'$, $S := \prop_n$ and, for all $1 \leq i \leq n$, we let $S_i := \prop_n \setminus \{p_i\}$. Then, we define: $\mathcal{A}_{n,\bullet} := \{ T(S) \}$, $\mathcal{B}_{n,\bullet} := \{ T(S_i) \mid 1 \leq i \leq n \}$, and  $\phi_{n,\bullet} := p_1 \wedge \ldots \wedge p_n$.
			
			That way, for all $1 \leq i \leq n$, to distinguish $T(S_i)$ and $T(S)$, an $\ATL$-formula $\phi$ needs to be such that $p_i \in \prop(\phi)$, by Lemma~\ref{lem:distinguish_ATL}. Hence, if an $\ATL$-formula $\phi$ separates $\mathcal{A}_{n,\bullet}$ and $\mathcal{B}_{n,\bullet}$, we have $\prop_n' \subseteq \prop(\phi)$. Consider now any two positive and negative sets $\mathcal{P},\mathcal{N}$ of $(n,S_{n,\bullet})$-proper structures. Let $\mathcal{P}' := \mathcal{P} \cup \mathcal{A}_{n,\bullet}$ and $\mathcal{N}' := \mathcal{N} \cup \mathcal{B}_{n,\bullet}$. Consider any formula $\phi \in \ATL(\prop^n,\Ut,\emptyset,\Bl,0)$ and $\psi := \phi \wedge \phi_{n,\bullet}$. We have $T(S) \models \phi_{n,\bullet}$ and for all $1 \leq i \leq n$, we have $T(S_i) \not\models \phi_{n,\bullet}$. Furthermore, since $\phi$ does not use negations, we also have $T(S) \models \phi$. In addition, by definition of $S_{n,\bullet}$, for all $(n,S_{n,\bullet})$-proper structures $T$, we have $T \models \phi_{n,\bullet}$. Overall, we obtain that $\phi$ accepts $\mathcal{P}$ and rejects $\mathcal{N}$ if and only if $\psi$ accepts $\mathcal{P}'$ and rejects $\mathcal{N}'$.
			\item Assume now that $\bullet = \prescript{\neg}{}{\wedge}$, the case $\bullet = \prescript{\neg}{}{\lor}$ is analogous. We let $\phi^1 := p_1$ and, for all $2 \leq i \leq n$, we let $\phi^i := \phi^{i-1} \prescript{\neg}{}{\wedge} \; p_i$. Let us first prove by induction on $1 \leq i \leq n$ that, for all $S \subseteq \{p_1,\ldots,p_i\}$, letting $\tilde{S}^i := S \cup \{p_{i+1},\ldots,p_n\}$, we have that $T(S) \models \phi^i$ and $T(\tilde{S}^i) \models \phi^n$ have the same truth value if and only if $i$ and $n$ have the same parity. This obviously holds for $i = n$. Assume now that it holds for some $2 \leq i \leq n$. Let $S \subseteq \{p_1,\ldots,p_{i-1}\}$. We have $\phi^i := \phi^{i-1} \prescript{\neg}{}{\wedge} \; p_i \equiv \neg \phi^{i-1} \lor \neg p_i$. Since $p_i \in \tilde{S}^{i-1}$, and $\prop(\phi^{i-1}) \subseteq \{p_1,\ldots,p_{i-1}\}$ we have $T(\tilde{S}^{i-1}) \models \phi^i $ if and only if $T(S \cup \{p_i\}) \models \phi^{i}$ if and only if $T(S) \not\models \phi^{i-1}$. By our induction hypothesis, we have that $T(S \cup \{p_i\}) \models \phi^{i}$ and $T(\widetilde{S \cup \{p_i\}}^i) \models \phi^{n}$ have the same truth value if and only if $i$ and $n$ have the same parity. We deduce that $T(S) \models \phi^{i-1}$ and $T(\widetilde{S}^{i-1}) \models \phi^{n}$ have the same truth value if and only if $i-1$ and $n$ have the same parity, since $\widetilde{S \cup \{p_i\}}^i = \tilde{S}^{i-1}$. In fact, the property holds for all $1 \leq i \leq n$. 
			
			Let us now show by induction on $1 \leq i \leq n$ that there is some $S_i,S_i' \subseteq \{p_1,\ldots,p_i\}$ such that $S_i \cup \{p_i\} = S_i' \cup \{p_i\}$, and $T(\tilde{S_i}^i) \models \phi^n$ while $T(\tilde{S_i'}^i) \not\models \phi^n$. We have $T(\{p_1\}) \models \phi^1$ and $T(\emptyset) \not\models \phi^1$. Hence, we can conclude that this property holds for $i = 1$ with what we have proved above. Assume now that it holds for some $1 \leq i-1 \leq n-1$. We have $\phi^i := \phi^{i-1} \prescript{\neg}{}{\wedge} \; p_i$. We let $S := S_{i-1}$ if $n$ and $i-1$ have the same parity, and $S := S_{i-1}'$ otherwise. That way, with our above result and the induction hypothesis, we know that $T(S) \models \phi^{i-1}$. Then, for all $S_i,S_i' \in \{ S \cup \{p_i\},S\}$ such that $S_i,\neq S_i'$, we have that the truth value of $T(S_i) \models \phi^{i}$ and $T(S_i') \models \phi^{i}$ are different. We can then conclude with our above result. In fact, this property holds for all $1 \leq i \leq n$.
			
			We can now finally define the formula and structures that we consider.  We let $\phi_{n,\bullet} := \phi^n$ and $S_{n,\bullet} := S_n \subseteq \{p_1,\ldots,p_n\}$. Furthermore, we let $\mathcal{A}_{n,\bullet} := \{ T(\tilde{S_i}^i \cup \{p,\bar{p}\})  \mid 1 \leq i \leq n \}$, and $\mathcal{B}_{n,\bullet} := \{ T(\tilde{S_i'}^i \cup \{p,\bar{p}\}) \mid 1 \leq i \leq n \}$. 
			
			That way, for all $1 \leq i \leq n$, to distinguish $T(\tilde{S_i}^i \cup \{p,\bar{p}\})$ and $T(\tilde{S_i'}^i \cup \{p,\bar{p}\})$, an $\ATL$-formula $\phi$ needs to be such that $p_i \in \prop(\phi)$, since $S_i \cup \{p_i\} = S_i' \cup \{p_i\}$ and by Lemma~\ref{lem:distinguish_ATL}. Hence, if an $\ATL$-formula $\phi$ separates $\mathcal{A}_{n,\bullet}$ and $\mathcal{B}_{n,\bullet}$, we have $\prop_n' \subseteq \prop(\phi)$. Consider now any two positive and negative sets $\mathcal{P},\mathcal{N}$ of $(n,S_{n,\bullet})$-proper structures. Let $\mathcal{P}' := \mathcal{N} \cup \mathcal{B}_{n,\bullet}$ and $\mathcal{N}' := \mathcal{P} \cup \mathcal{A}_{n,\bullet}$. Consider any formula $\phi \in \ATL(\prop^n,\Ut,\emptyset,\Bl,0)$ and $\psi := \phi \prescript{\neg}{}{\wedge} \; \phi_{n,\bullet} \equiv \neg \phi \lor \neg \phi_{n,\bullet}$. Let $1 \leq l \leq n$. We have $\phi_{n,\bullet} \models T(\tilde{S_i}^i)$ and $\phi \models T(\tilde{S_i}^i)$. Therefore, $\psi \not\models T(\tilde{S_i}^i)$. However, $\phi_{n,\bullet} \not\models T(\tilde{S_i'}^i)$. Therefore, $\psi \models T(\tilde{S_i'}^i)$. Hence, $\psi$ accepts $\mathcal{B}_{n,\bullet}$ and rejects $\mathcal{A}_{n,\bullet}$. In addition, by definition of $S_{n,\bullet}$, for all $(n,S_{n,\bullet})$-proper structures $T$, we have $T \models \phi_{n,\bullet}$. Overall, we obtain that $\phi$ accepts $\mathcal{P}$ and rejects $\mathcal{N}$ if and only if $\psi$ accepts $\mathcal{P}'$ and rejects $\mathcal{N}'$.
			\item Assume now that $\bullet = \; \Leftrightarrow$. The case $\bullet = \prescript{\neg}{}{\Leftrightarrow}$ is analogous.
			We let $S_{n,\bullet} := \prop_n'$, $S := \prop_n$ and, for all $1 \leq i \leq n$, we let $S_i := \prop_n \setminus \{p_i\}$. Then, we define: $\mathcal{A}_{n,\bullet} := \{ T(S) \}$, $\mathcal{B}_{n,\bullet} := \{ T(S_i) \mid 1 \leq i \leq n \}$, and  $\phi_{n,\bullet} := p_1 \Leftrightarrow \ldots \Leftrightarrow p_n$.
			
			That way, for all $1 \leq i \leq n$, to distinguish $T(S_i)$ and $T(S)$, an $\ATL$-formula $\phi$ needs to be such that $p_i \in \prop(\phi)$, by Lemma~\ref{lem:distinguish_ATL}. Hence, if an $\ATL$-formula $\phi$ separates $\mathcal{A}_{n,\bullet}$ and $\mathcal{B}_{n,\bullet}$, we have $\prop_n' \subseteq \prop(\phi)$. Consider now any two positive and negative sets $\mathcal{P},\mathcal{N}$ of $(n,S_{n,\bullet})$-proper structures. Let $\mathcal{P}' := \mathcal{P} \cup \mathcal{A}_{n,\bullet}$ and $\mathcal{N}' := \mathcal{N} \cup \mathcal{B}_{n,\bullet}$. Consider any formula $\phi \in \ATL(\prop^n,\Ut,\emptyset,\Bl,0)$ and $\psi := \phi \Leftrightarrow \phi_{n,\bullet}$. We have $T(S) \models \phi_{n,\bullet}$ and for all $1 \leq i \leq n$, we have $T(S_i) \not\models \phi_{n,\bullet}$. Furthermore, since $\phi$ does not use negations, we also have $T(S) \models \phi$, and $T(S_i) \models \phi$ for all $1 \leq i \leq n$. Thus, $\psi$ accepts $\mathcal{A}_{n,\bullet}$ and rejects $\mathcal{B}_{n,\bullet}$. In addition, by definition of $S_{n,\bullet}$, for all $(n,S_{n,\bullet})$-proper structures $T$, we have $T \models \phi_{n,\bullet}$. Thus, we have that $\phi$ accepts $\mathcal{P}$ and rejects $\mathcal{N}$ if and only if $\psi$ accepts $\mathcal{P}'$ and rejects $\mathcal{N}'$.
		\end{itemize}
	\end{proof}

	\subsection{$\CTL$ learning}
	We start with $\CTL$ learning. We consider two cases: with and without the next operator $\lX$. In the former case, the learning problem is $\msf{NP}$-complete, in the latter it is $\msf{NL}$-complete.
	
	\subsubsection{With the next operator $\lX$}
	The goal of this subsection is to show the theorem below.
	\begin{theorem}
		\label{thm:ctl_unary_X_NP_hard}
		Consider a set $\Ut \subseteq \Op{Un}{}$ of unary temporal operators and assume that $\lX \in \Ut$. Then, for all sets $\Bl \subseteq \Op{Bin}{lg}$ and $n \in \N$, the decision problem $\CTL_\msf{Learn}(\Ut,\emptyset,\Bl,n)$ is $\msf{NP}$-complete.
	\end{theorem}
	For the remainder of this subsection, we consider a set $\Ut \subseteq \Op{Un}{}$ of unary temporal operators and assume that $\lX \in \Ut$
	.
	
	\paragraph{Overview of the reduction.}
	We follow the steps described in Section~\ref{subsubsec:abstract_recipe}. However, Step~\ref{stepa} was already taken care of in the previous section. This lets us focus on $\CTL$-formulas using only unary operators (and a single proposition). First of all, we define Kripke structures that prevent the use of the $\lF$ and $\lG$ operators, as well as negations. Hence, the only operators remaining are $\forall \lX$ and $\exists \lX$. Our idea is that now a $\CTL$-formula of size $l+1$ is entirely defined by a subset $H \subseteq [1,\ldots,l]$: such a subset 
	defines the $\CTL$-formula $\phi(l,H)$ in which, for all $i \in [1,\ldots,l]$, the $i$-th operator of $\phi$ is $\exists \lX$ if and only if $i \in H$. 
	
	On the other hand, a subset $C \subseteq [1,\ldots,l]$ defines a positive Kripke structure $K^{l,C}$ that is defined as a sequence of $l+1$ states
	$\{ q_1^{\msf{lose}},\ldots,q_{l+1}^{\msf{lose}} \}$ such that no $\CTL$-formula of the shape $\phi(l,H)$ can satisfy the state $q_{l+1}^{\msf{lose}}$. Furthermore, all states $q_i^{\msf{lose}}$ are such that $q_{i+1}^{\msf{lose}} \in \msf{Succ}(q_i^{\msf{lose}})$. However, some states $q_i^{\msf{lose}}$ can branch out to the winning state $q^{\msf{win}}$ (that all $\CTL$-formulas of the shape $\phi(l,H)$ satisfy). This occurs for those indices $i \in [1\ldots,l]$ such that $i \in C$. In fact, with such a definition, we obtain that $K^{l,C} \models \phi(l,H)$ if and only if $H \cap C \neq \emptyset$. 
	
	The final step that we take is to define, for $k \leq l$, a negative Kripke structure $K_{\exists > k}^l$ that is satisfied by the $\CTL$-formula $\phi(l,H)$ 
	if and only if $\phi(l,H)$ uses at least $k+1$ times the operator $\exists \lX$, which is equivalent to $|H| > k$. The Kripke structure $K_{\exists > k}^l$ has $k+2$ levels, from the bottom up:
	\begin{itemize}
		\item with a single starting state at the bottommost level;
		\item where the formula $p$ is satisfied in no state of the bottom $k+1$ levels, but it is satisfied in 
		the only (self-looping) state $q^{\msf{win}}$ of the topmost level;
		\item where every state of the bottom $k+1$ levels has a successor one level higher.
	\end{itemize}
	
	\paragraph{Formal definitions and proofs.}
	Let us now formally define the Kripke structures that we will use in the reduction. We first define the reduction for $n = 0$ (i.e. when no binary operator is allowed) and prove its correctness. We then use it to exhibit a reduction for arbitrary $n \in \N$.
	
	We start with the Kripke structures that ensure that the candidate formulas are of a specific shape. 
	\begin{definition}
		Given $l \in \N_1$ we define the two Kripke structures $K^{\mathcal{P}}_l$ and $K^{\mathcal{N}}_l$ that can be seen in Figures~\ref{fig:simple_pos} and~\ref{fig:simple_neg}.
	\end{definition}
	
	\begin{figure}[t]
		\centering
		\includegraphics[scale=1]{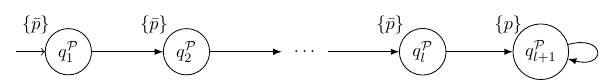}
		\caption{The Kripke structure $K^{\mathcal{P}}_{l}$.}
		\label{fig:simple_pos}
	\end{figure}
	\begin{figure}[t]
		\centering
		\includegraphics[scale=1]{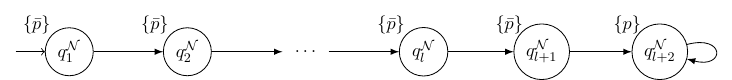}
		\caption{The Kripke structure $K^{\mathcal{N}}_{l}$.}
		\label{fig:simple_neg}
	\end{figure}
	
	Let us now consider the Kripke structures $K^{l,C}$ that encode a subset of $C \subseteq [1,\ldots,l]$.
	\begin{definition}
		Given some $l \in \N_1$ and $C \subseteq [1,\ldots,l]$, we define the Kripke structure $K_{(l,C)} = \langle Q,I,\{ p,\bar{p} \}, \pi, \msf{Succ} \rangle$ where:
		\begin{itemize}
			\item $Q := \{ q^{\msf{lose}}_1, \ldots, q^{\msf{lose}}_{l+1},q^{\msf{win}} \}$;
			\item $I := \{ q^{\msf{lose}}_1 \}$;
			\item for all $1 \leq i \leq l$, we have:
			\begin{align*}
				\msf{Succ}(q^{\msf{lose}}_i) := \begin{cases}
					\{ q^{\msf{lose}}_{i+1},q^{\msf{win}} \} & \text{ if }i \in C\\
					\{ q^{\msf{lose}}_{i+1} \} & \text{ if }i \notin C
				\end{cases}
			\end{align*}
			Furthermore, $\msf{Succ}(q^{\msf{win}}) := 	\{ q^{\msf{win}} \}$ and $\msf{Succ}(q^{\msf{lose}}_{l+1}) := 	\{ q^{\msf{lose}}_{l+1} \}$. 
			\item $\pi(q^{\msf{win}}) := \{ p \}$ and, for all $q \in Q \setminus \{ q^{\msf{win}} \}$, we have $\pi(q) := \{\bar{p}\}$.
		\end{itemize}
	\end{definition}
	An example of such a construction is depicted in Figure~\ref{fig:example_Kripke_C}. 
	\begin{figure}
		\centering
		\includegraphics[scale=1]{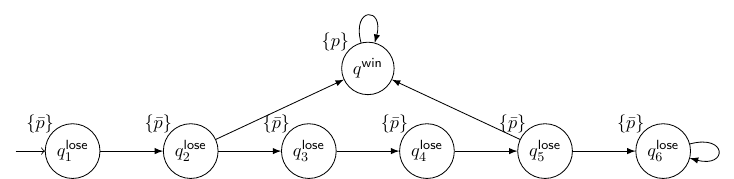}
		\caption{The Kripke structure $K_{(5,\{ 2,5 \})}$.}
		\label{fig:example_Kripke_C}
	\end{figure}

	Finally, we also consider a Kripke structure that prevents from using too many $\exists$ operators.
	\begin{definition}
		Given some $l \in \N_1$ and $k \leq l$, we define the Kripke structure $K_{\exists > k}^l = \langle Q,I,\{ p,\bar{p} \}, \pi, \msf{Succ} \rangle$ where:
		\begin{itemize}
			\item $Q := \{ q^{\msf{win}}\} \cup \{ q^{\msf{lose}}_{i,j} \mid 0 \leq i \leq k,\; i+1\leq j \leq l+1 \}$;
			\item $I := \{ q^{\msf{lose}}_{0,1} \}$;
			\item For all $0 \leq i \leq k,\; i+1\leq j \leq l$, we have:
			\begin{align*}
				\msf{Succ}(q^{\msf{lose}}_{i,j}) := \begin{cases}
					\{ q^{\msf{lose}}_{i,j+1},q^{\msf{lose}}_{i+1,j+1} \} & \text{ if }i \leq k-1 \\
					\{ q^{\msf{lose}}_{i,j+1},q^{\msf{win}} \} & \text{ if }i = k
				\end{cases}
			\end{align*}
			For all $0 \leq i \leq k$:
			\begin{align*}
				\msf{Succ}(q^{\msf{lose}}_{i,l+1}) := \{ q^{\msf{lose}}_{i,l+1} \}
			\end{align*}
			and $\msf{Succ}(q^{\msf{win}}) := \{ q^{\msf{win}} \}$. 
			\item $\pi(q^{\msf{win}}) := \{ p \}$ and, for all $q \in Q \setminus \{ q^{\msf{win}} \}$, we have $\pi(q) := \bar{p}$.
		\end{itemize}
	\end{definition}
	An example of such a Kripke structure is depicted in Figure~\ref{fig:example_Kripke_more_than_k}. 
	\begin{figure}
		\centering
		\includegraphics[scale=1]{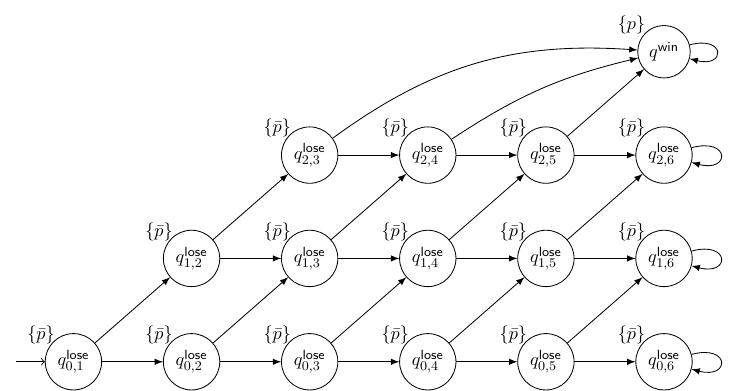}
		\caption{The Kripke structure $K^5_{\exists > 2}$.}
		\label{fig:example_Kripke_more_than_k}
	\end{figure}
	
	We can finally define the reduction from the hitting set problem that we consider. 
	\begin{definition}
		\label{def:reduction_CTL_X}
		Consider an instance $(l,C,k)$ of the hitting set problem $\msf{Hit}$. We consider:
		\begin{itemize}
			\item $\prop_0 = \{ p,\bar{p} \}$ as set of propositions;
			\item $\mathcal{P} := \{ K_{(l,C_i)} \mid 1 \leq i \leq n \} \cup \{ K^{\mathcal{P}}_{l} \}$;
			\item $\mathcal{N} := \{ K^l_{\exists > k},K_{l}^{\mathcal{N}} \}$;
			\item $B := l+1$.
		\end{itemize}
		Then, we define the input $\msf{In}^{\CTL(\lX),0}_{(l,C,k)} := (\prop,\mathcal{P},\mathcal{N},B)$
		.
	\end{definition}

	This reduction satisfies the lemma below. 
	\begin{lemma}
		\label{lem:reduc_CTL_X}
		The input $(l,C,k)$ is a positive instance of the hitting set problem $\msf{Hit}$ if and only if $\msf{In}^{\CTL(\lX),0}_{(l,C,k)}$ is a positive instance of $\CTL_\msf{Learn}(\Ut \setminus \{\neg\},\emptyset,\emptyset,0)$ if and only if $\msf{In}^{\CTL(\lX),0}_{(l,C,k)}$ is a positive instance of $\CTL_\msf{Learn}(\Ut \cup \{\lF,\lG\},\emptyset,\emptyset,0)$.
	\end{lemma}
	
	We start with the formal definition below $\phi(l,H)$-formulas. 
	\begin{definition}
		Let $l \in \N$. A $\CTL$-formula $\phi \in \CTL(\{p,\bar{p}\},\Ut \cup \{\lF,\lG\},\emptyset,\emptyset,0)$ is an $\lX^l$-formula if $\phi \in \CTL(\{p,\bar{p}\},\{\lX\},\emptyset,\emptyset,0)$ and $\Size{\phi} = l+1$.
	
		The $\lX^l$-formulas differ by the indices where $\exists$ and $\forall$ quantifiers appear. For all $H \subseteq [1,\ldots,l]$, we let $\phi(l,H)$ denote the $\lX^l$-formula:
		\begin{equation*}
			\phi(l,H) := Q_1 \lX \ldots Q_l \lX p
		\end{equation*}
		where, for all $i \in [1,\ldots,l]$, we have:
		\begin{align*}
			Q_i := \begin{cases}
				\exists & \text{ if }i \in H\\
				\forall & \text{ if }i \in [1,\ldots,l] \setminus H
			\end{cases}
		\end{align*}
	\end{definition}

	Let us show that we can restrict ourselves to $\lX^l$-formulas by considering the Kripke structures $K_{l}^{\mathcal{P}},K_{l}^{\mathcal{N}}$. 
	\begin{lemma}
		\label{lem:not_use_F_and_G}
		Consider some $l \in \N_1$ and a $\CTL$ formula $\phi \in \CTL(\{p,\bar{p}\},\Ut,\emptyset,\emptyset,0)$ of size at most $l+1$. The formula $\phi$ accepts $K_{l}^{\mathcal{P}}$ and rejects $K_{l}^{\mathcal{N}}$ if and only if $\phi$ is an $\lX^l$-formula.
	\end{lemma}
	\begin{proof}
		First of all, note that, for all $1 \leq i \leq l+1$, the states $q_{i}^{\mathcal{P}}$ and $q_{i+1}^{\mathcal{N}}$ satisfy exactly the same $\CTL$-formulas. Then, we show by induction on $0 \leq i \leq l$ the property $\mathcal{P}(i)$: 
		\begin{itemize}
			\item for all $1 \leq j \leq l-i$, any $\CTL(\{p,\bar{p}\},\Ut,\emptyset,\emptyset,0)$-formula of size at most $i+1$ cannot distinguish the states $q_j^{\mathcal{P}}$ and $q_j^{\mathcal{N}}$;
			\item a $\CTL(\{p,\bar{p}\},\Ut,\emptyset,\emptyset,0)$-formula of size at most $i+1$ distinguishes the states $q_{l+1-i}^{\mathcal{P}}$ and $q_{l+1-i}^{\mathcal{N}}$ if and only if it is an $\lX^{i}$-formula (in which case, it accepts $q_{l+1-i}^{\mathcal{P}}$ and rejects $q_{l+1-i}^{\mathcal{N}}$).
		\end{itemize}
		The property $\mathcal{P}(0)$ straightforwardly holds.
		
		Assume now that the property $\mathcal{P}(i)$ holds for some $0 \leq i \leq l-1$. Consider a $\CTL(\{p,\bar{p}\},\Ut,\emptyset,\emptyset,0)$-formula $\phi$ of size at most $i+2$. First of all, if $\phi$ is of size at most $i+1$, then $\mathcal{P}(i)$ gives that for all $1 \leq j \leq l+1-(i+1) = l-i$, the formula $\phi$ does not distinguish the states $q_j^{\mathcal{P}}$ and $q_j^{\mathcal{N}}$. Assume now that $\Size{\phi} = i+2$. Consider any $1 \leq j \leq l-i$. There are several cases.
		\begin{itemize}
			\item Assume that $\phi = \neg \phi'$, in which case $\Size{\phi'} = i+1$. By $\mathcal{P}(i)$, $\phi'$ does not distinguish the states $q_j^{\mathcal{P}}$ and $q_j^{\mathcal{N}}$, thus $\phi$ does not either. 
			\item Assume that $\phi = \msf{Q} \; \lF \phi'$, with $\msf{Q} \in \{\exists,\forall\}$, in which case $\Size{\phi'} = i+1$. Assume that $q_j^{\mathcal{P}} \models \phi$. Then, there is some $j \leq r \leq l+1$ such that $q_r^{\mathcal{P}} \models \phi'$, which is equivalent to $q_{r+1}^{\mathcal{N}} \models \phi'$, and thus $q_j^{\mathcal{N}} \models \phi$. On the other hand, assume that $q_j^{\mathcal{N}} \models \phi$, in which case there is some $j \leq r \leq l+2$ such that $q_r^{\mathcal{N}} \models \phi'$. If $r \geq j+1$, then we have $r-1 \geq j$ and $q_{r-1}^{\mathcal{P}} \models \phi'$, and thus $q_j^{\mathcal{P}} \models \phi$. Otherwise, $r = j$. Since, by $\mathcal{P}(i)$, we have that $\phi'$ does not distinguish $q_j^{\mathcal{P}}$ and $q_j^{\mathcal{N}}$, then it follows that we also have $q_j^{\mathcal{P}} \models \phi'$, and thus $q_j^{\mathcal{P}} \models \phi$. Hence, in any case the formula $\phi$ does not distinguish $q_j^{\mathcal{P}}$ and $q_j^{\mathcal{N}}$.
			\item Assume that $\phi = \msf{Q} \; \lG \phi'$, with $\msf{Q} \in \{\exists,\forall\}$. This case is analogous to the previous one.
			\item Assume finally that $\phi = \msf{Q} \; \lX \phi'$, with $\msf{Q} \in \{\exists,\forall\}$, in which case $\Size{\phi'} = i+1$. If $j < l - i$, then $j+1 \leq l-i$, and thus, by $\mathcal{P}(i)$, $\phi'$ does not distinguish the two states $q_{j+1}^{\mathcal{P}}$ and $q_{j+1}^{\mathcal{N}}$. Hence, $\phi$ does not distinguish the two states $q_{j}^{\mathcal{P}}$ and $q_{j}^{\mathcal{N}}$. If $j = l-i$, then $j+1 = l+1-i$, hence $\phi'$ distinguishes the states $q_{j+1}^{\mathcal{P}}$ and $q_{j+1}^{\mathcal{N}}$ if and only if $\phi'$ is an $\lX^i$-formula (in which case it accepts $q_{j+1}^{\mathcal{P}}$ and rejects $q_{j+1}^{\mathcal{N}}$), which is equivalent to $\phi$ being an $\lX^{i+1}$-formula (in which case it accepts $q_{j}^{\mathcal{P}}$ and rejects $q_{j}^{\mathcal{N}}$).
		\end{itemize}
		Therefore, the property $\mathcal{P}(i+1) $ holds. In fact, it holds for all $0 \leq i \leq l$. The lemma is then given by $\mathcal{P}(l)$.
	\end{proof}
	
	We deduce as a corollary.
	\begin{corollary}
		\label{coro:reduc_CTL_only_X}
		The input $\msf{In}^{\CTL(\lX),0}_{(l,C,k)}$ is a positive instance of $\CTL_\msf{Learn}(\Ut \setminus \{\neg\},\emptyset,\emptyset,0)$ if and only if $\msf{In}^{\CTL(\lX),0}_{(l,C,k)}$ is a positive instance of $\CTL_\msf{Learn}(\Ut \cup \{\lF,\lG\},\emptyset,\emptyset,0)$.
	\end{corollary}
	\begin{proof}
		This is straightforward consequence of Lemma~\ref{lem:reduc_CTL_X}: if a formula $\phi$ accepts $\mathcal{P}$ and rejects $\mathcal{N}$, then it accepts $K_l^\mathcal{P}$ and rejects $K_l^\mathcal{N}$. Thus, if its size is at most $l+1$, it only uses the operator $\lX \in \Ut$.
	\end{proof}
	
	Let us now consider at which conditions the $\CTL$ $\lX^l$-formulas $\phi(l,H)$ accepts the Kripke structure $K_{(l,C)}$.
	\begin{lemma}
		\label{lem:reduc_ctl_intersect}
		Consider some $l \in \N_1$ and $H,C \subseteq [1,\ldots,l]$. The $\CTL$ $\lX^l$-formula $\phi(l,H)$ accepts the Kripke structure $K_{(l,C)}$ if and only if $H \cap C \neq \emptyset$. 
	\end{lemma}
	\begin{proof}
		Let $H \subseteq [1,\ldots,l]$. We have $\phi(l,H) = Q_1 \lX \ldots Q_l \lX p$. For all $1 \leq i \leq l+1$, we let $H_i := H \cap [i,\ldots,l]$, $\phi_i(H) := Q_i \lX  \ldots Q_l \lX p$. In particular, $\phi_{l+1}(H) = p$ and $\phi_{1}(H) = \phi(l,H)$. 
		
		Let us show by induction on $l+1 \geq i \geq 1$ the property $\mathcal{P}(i)$: $q_i^{\msf{lose}} \models \phi_i(H)$ if and only if $H_i \cap C \neq \emptyset$.  The property $\mathcal{P}(l+1)$ straightforwardly holds since $\pi(q^{\msf{lose}}_{l+1}) = \emptyset$ and $H_{l+1} = \emptyset$. 
		
		Assume now that $\mathcal{P}(i)$ holds for some $l+1 \geq i \geq 2$. We have $\phi_{i-1}(H) := Q_{i-1} \lX \phi_{i}(H)$ with $Q_{i-1} = \exists$ if and only if $i-1 \in H$. Then:
		\begin{itemize}
			\item Assume that $H_{i-1} \cap C \neq \emptyset$. If we have $H_i \cap C \neq \emptyset$, then by $\mathcal{P}(i)$, $q_i^{\msf{lose}} \models \phi_i(H)$. Since $q^{\msf{win}} \models \phi_i(H)$ and $\msf{Succ}(q_{i-1}^{\msf{lose}}) \subseteq \{ q_i^{\msf{lose}},q^{\msf{win}} \}$, it follows that $q_{i-1}^{\msf{lose}} \models \phi_{i-1}(H)$ (regardless of whether $Q_{i-1} = \exists$ or $Q_{i-1} = \forall$). Otherwise, it must be that $i-1 \in H \cap C$. Thus, we have $Q_{i-1} = \exists$ and $q^{\msf{win}} \in \msf{Succ}(q_{i-1}^{\msf{lose}})$. Since $q^{\msf{win}} \models \phi_i(H)$, it follows that $q_{i-1}^{\msf{lose}} \models \phi_{i-1}(H)$.
			\item Assume now that $H_{i-1} \cap C  = \emptyset$. It follows that $q_i^{\msf{lose}} \not\models \phi_i(H)$ by $\mathcal{P}(i)$, since $H_i \cap C = \emptyset$. There are two cases:
			\begin{itemize}
				\item Assume that $i-1 \notin H$. In that case, $Q_{i-1} = \forall$. Furthermore, $q^{\msf{lose}}_i \in \msf{Succ}(q_{i-1}^{\msf{lose}})$ and $q_i^{\msf{lose}} \not\models \phi_i(H)$. Hence, $q_{i-1}^{\msf{lose}} \not\models \phi_{i-1}(H)$.
				\item Assume that $i-1 \notin C$. In that case, $\msf{Succ}(q_{i-1}^{\msf{lose}}) = \{ q^{\msf{lose}}_i \}$ and $q_i^{\msf{lose}} \not\models \phi_i(H)$. Hence, $q_{i-1}^{\msf{lose}} \not\models \phi_{i-1}(H)$.
			\end{itemize}
		\end{itemize}
		Therefore, the property $\mathcal{P}(i-1)$ holds. In fact, $\mathcal{P}(i)$ holds for all $1\leq i \leq l+1$. The lemma is then given by $\mathcal{P}(1)$.
	\end{proof}
	
	In addition, for a $\CTL$ $\lX^l$-formula $\phi(l,H)$ not to accept the Kripke structure $K^l_{\exists > k}$, it must have not too much existential quantifiers, as stated below.
	\begin{lemma}
		\label{lem:redcu_ctl_small_hitting_set}
		Consider some $l \in \N_1$, $k \leq l$ and $H \subseteq [1,\ldots,l]$. The $\CTL$ $\lX^l$-formula $\phi(l,H)$ accepts the Kripke structure $K^l_{\exists > k}$ if and only if $|H| > k$. 
	\end{lemma}
	\begin{proof}
		Let $H \subseteq [1,\ldots,l]$. As in the proof of Lemma~\ref{lem:reduc_ctl_intersect}, we have $\phi(l,H) = Q_1 \lX \ldots Q_l \lX p$ and, for all $1 \leq i \leq l+1$, we let $H_i := H \cap [i,\ldots,l]$ and $\phi_i(H) := Q_i \lX  \ldots Q_l \lX p$. Thus, $\phi_{l+1}(H) = p$ and $\phi_{1}(H) = \phi(H)$. Furthermore, for all $k+1 \geq i \geq 0$, $l+1 \geq j \geq i+1$, we let:
		\begin{align*}
			q(i,j) := 
			\begin{cases}
				q_{(i,j)}^{\msf{lose}} & \text{ if }i < k+1 \\
				q^{\msf{win}} & \text{ otherwise }
			\end{cases}
		\end{align*}
		Note that, for all $k \geq i \geq 0$, $l \geq j \geq i+1$, we have $\msf{Succ}(q(i,j)) = \{ q(i,j+1),q(i+1,j+1) \}$. 
		
		Let us show by induction on $k+1 \geq i \geq 0$ the property $\mathcal{P}(i)$: for all $l+1 \geq j \geq i+1$, $q(i,j) \models \phi_j(H)$ if and only if $|H_j| > k-i$. 
		
		The property $\mathcal{P}(k+1)$ states that for all $k+2 \leq j \leq l+1$, $q(k+1,j) = q^{\msf{win}} \models \phi_j(H)$ if and only if $|\emptyset| > -1$. This straightforwardly holds since $\phi_j(H)$ is a $\neg$-free formula. 
		
		Assume now that the property $\mathcal{P}(i)$ holds for some $k+1 \geq i \geq 1$. Let us show by induction on $l+1 \geq j \geq i$ the property $\mathcal{Q}(j)$: $q(i-1,j) \models \phi_j(H)$ if and only if $|H_j| > k-i+1$. The property $\mathcal{Q}(l+1)$ straightforwardly holds since $q(i-1,l+1) = q_{i-1,l+1}^{\msf{lose}} \not\models \phi_{l+1}(H)$ and $H_{l+1} = \emptyset$. Assume now that the property $\mathcal{Q}(j)$ holds for some $l+1 \geq j \geq i+1$. We have $\phi_{j-1}(H) := Q_{j-1} \lX \phi_{j}(H)$ with $Q_{j-1} = \exists$ if and only if $j-1 \in H$. 
		\begin{itemize}
			\item Assume that $|H_{j-1}| > k-i+1$. If we have $j-1 \in H$, then $Q_{j-1} = \exists$ and $|H_j| > k-i$. Therefore, by $\mathcal{P}(i)$, we have $q(i,j) \models \phi_j(H)$. Since  $q(i,j) \in \msf{Succ}(q(i-1,j-1))$, it follows that $q(i-1,j-1) \models \phi_{j-1}(H)$. 
			
			On the other hand, if $j-1 \notin H$, then $|H_j| > k-i+1 > k-i$. Hence, by $\mathcal{P}(i)$, we have $q(i,j) \models \phi_j(H)$ and by $\mathcal{Q}(j)$, we have $q(i-1,j) \models \phi_j(H)$. Since $\msf{Succ}(q(i-1,j-1)) = \{ q(i,j),q(i-1,j)\}$, it follows that $q(i-1,j-1) \models \phi_{j-1}(H)$.
			\item Assume now that $|H_{j-1}| \leq k-i+1$. If we have $j-1 \in H$, then $|H_j| \leq k-i < k-i+1$. Hence, by $\mathcal{P}(i)$, we have $q(i,j) \not\models \phi_j(H)$ and by $\mathcal{Q}(j)$, we have $q(i-1,j) \not\models \phi_j(H)$. Since $\msf{Succ}(q(i-1,j-1)) = \{ q(i,j),q(i-1,j)\}$, it follows that $q(i-1,j-1) \not\models \phi_{j-1}(H)$. 
			
			On the other hand, if $j-1 \notin H$, then $Q_{j-1} = \forall$ and $|H_j| \leq k-i+1$. Therefore, by $\mathcal{Q}(j)$, we have $q(i-1,j) \not\models \phi_j(H)$. Since $q(i-1,j) \in \msf{Succ}(q(i-1,j-1))$, it follows that $q(i-1,j-1) \not\models \phi_{j-1}(H)$.
		\end{itemize}
		Hence, $\mathcal{Q}(j-1)$ holds. In fact, for all $l+1 \geq j \geq i$, $\mathcal{Q}(j-1)$ holds and therefore $\mathcal{P}(i-1)$ holds. In fact, $\mathcal{P}(i)$ holds for all $0\leq i \leq k+1$. The lemma is then given by $\mathcal{P}(0)$ applied with $j = 1$.
	\end{proof}
	
	The proof of Lemma~\ref{lem:reduc_CTL_X} is now direct.
	\begin{proof}		
		Assume that $(l,C,k)$ is a positive instance of the hitting set problem $\msf{Hit}$. Consider a hitting set $H \subseteq [1,\ldots,l]$ with $|H| \leq k$. We let $\phi := \phi(l,H) \in \CTL(\{p,\bar{p}\},\Ut \setminus \{\neg\},\emptyset,\emptyset,0)$ since $\lX \in \Ut$. We have $\Size{\phi} = l+1$. By Lemma~\ref{lem:not_use_F_and_G}, $\varphi$ accepts $K_{l}^{\mathcal{P}}$ and rejects both $K_{l}^{\mathcal{N}}$. By Lemma~\ref{lem:redcu_ctl_small_hitting_set}, $\phi$ rejects $K^l_{\exists > k}$. Consider now some $1 \leq i \leq n$. Since $H$ is a hitting set, we have $H \cap C_i \neq \emptyset$. Hence, by Lemma~\ref{lem:reduc_ctl_intersect}, $\phi(l,H)$ accepts the Kripke structure $K_{(l,C_i)}$. In follows that the $\CTL$-formula $\phi$ accepts $\mathcal{P}$ and rejects $\mathcal{N}$. Hence, $\msf{In}^{\CTL(\lX),0}_{(l,C,k)}$ is a positive instance of the $\CTL_\msf{Learn}(\Ut \setminus \{\neg\},\emptyset,\emptyset,n)$ decision problem.
		
		Assume now that $\msf{In}^{\CTL(\lX),0}_{(l,C,k)}$ is a positive instance of the decision problem $\CTL_\msf{Learn}(\Ut \setminus \{\neg\},\emptyset,\emptyset,0)$. Consider a $\CTL$-formula $\phi \in \CTL(\{p,\bar{p}\},\Ut\setminus \{\neg\},\emptyset,\emptyset,0)$ of size at most $l+1$ that accepts $\mathcal{P}$ and rejects $\mathcal{N}$. Since $\phi$ accepts $K_{l}^{\mathcal{P}} \in \mathcal{P}$ and rejects $K_{l}^{\mathcal{N}}\in \mathcal{N}$, it follows by Lemma~\ref{lem:not_use_F_and_G} that $\phi$ is an $\lX^l$-formula. Let $H \subseteq [1,\ldots,l]$ be such that $\phi = \phi(l,H)$. Since $\phi$ rejects $K^l_{\exists > k} \in \mathcal{N}$, it follows, by Lemma~\ref{lem:redcu_ctl_small_hitting_set}, that $|H| \leq k$. Consider some $1 \leq i \leq n$. Since $\phi$ accepts $K_{(l,C_i)} \in \mathcal{P}$, it follows, by Lemma~\ref{lem:reduc_ctl_intersect}, that $H \cap C_i \neq \emptyset$. Therefore, $H$ is a hitting set and $(l,C,k)$ is a positive instance of the hitting set problem $\msf{Hit}$.
		
		We conclude with Corollary~\ref{coro:reduc_CTL_only_X}.
	\end{proof} 

	Theorem~\ref{thm:ctl_unary_X_NP_hard} follows.
	\begin{proof}
		
		This is direct consequence of Lemma~\ref{lem:reduc_CTL_X}, Corollary~\ref{coro:reduc_CTL_only_X}, of the fact that the instance $\msf{In}^{\CTL(\lX),0}_{(l,C,k)}$ can be computed in logarithmic space from $(l,C,k)$, and of Theorem~\ref{thm:unary_is_sufficient}
		.
	\end{proof}

	\subsubsection{Without the next operator $\lX$}
	In the previous subsection, we have seen that the $\CTL$ learning problem with a bounded amount of binary operator is $\msf{NP}$-complete. However, as can be seen, the proof of $\msf{NP}$-hardness heavily relies on the use of the operator $\lX$. In this subsection, we focus on the $\CTL$ learning problem where the operator $\lX$ is not allowed anymore. We show that this decision problem 
	in $\msf{NL}$, as stated in the theorem below. 
	\begin{theorem}
		\label{thm:ctl_unary_no_X_P}
		For all $\Ut \subseteq \{\lF,\lG,\neg \}$, $\Bl \subseteq \Op{Bin}{lg}$ and $n \in \N$, the $\CTL_\msf{Learn}(\Ut,\emptyset,\Bl,n)$ decision problem is in $\msf{NL}$. If $\lF \in \Ut$ or $\lG \in \Ut$, then the $\CTL_\msf{Learn}(\Ut,\emptyset,\Bl,n)$ decision problem is $\msf{NL}$-complete.
	\end{theorem}
	
	Proving that the decision problem is $\msf{NL}$-hard is rather straightforward and will handled as a second step. Let us first show that this decision problem is in $\msf{NL}$. 
	
	To prove this, we are going to proceed similarly to the $\LTL$ case (except that this case is more involved), i.e. we first consider $\CTL$-formulas that do not use binary operators at all, and then consider the case of $\CTL$-formulas using binary operators. 
	
	Let us first tackle the case of $\CTL$-formulas using no binary operators. Our goal is to successively restrict the set of $\CTL$-formulas that is sufficient to consider. More precisely, we show that sequentially using several operators in a row is useless. We start by showing that using twice in a row either of the operators $\lF$ or $\lG$ is useless. The formal statement that we give below is in the context of $\ATL$-formulas, more general than $\CTL$-formulas, because we will use this statement later on in this paper. 
	\begin{lemma}
		\label{lem:equiv_atl_dominating_quantifiers}
		Consider some $k \in \N_1$, $A \subseteq A' \subseteq [1,\ldots,k]$. Let $\phi$ be any $\ATL$-formula with $k$ agents. We have:
		\begin{equation*}
			\fanBr{A} \lG \fanBr{A'} \lG \phi \equiv \fanBr{A'} \lG \fanBr{A} \lG \phi \equiv \fanBr{A} \lG \phi
		\end{equation*}
		and dually 
		\begin{equation*}
			\fanBr{A} \lF \fanBr{A'} \lF \phi \equiv \fanBr{A'} \lF \fanBr{A} \lF \phi \equiv \fanBr{A'} \lF \phi
		\end{equation*}
	\end{lemma}
	\begin{proof}
		We start with the operator $\lG$. 
		\begin{itemize}
			\item By definition of the globally operator $\lG$, we have $\fanBr{A'} \lG \phi \implies \phi$, hence $\fanBr{A} \lG \fanBr{A'} \lG \phi \implies \fanBr{A} \lG \phi$. Also, we have $\fanBr{A'} \lG \fanBr{A} \lG \phi \implies \fanBr{A} \lG \phi$. 
			\item 
			Let us show that $\fanBr{A} \lG \phi \implies \fanBr{A} \lG \fanBr{A} \lG \phi$. Consider a concurrent game structure $C$ and any state $q \in Q$. Assume that $q \models \fanBr{A} \lG \phi$. Consider a strategy profile $s^A \in \msf{S}_A$ for the coalition of agents $A$ such that, for all $\rho \in \msf{Out}^Q(q,s^A)$, we have $\rho \models \lG \phi$. 
			We claim that we have for all $\rho \in \msf{Out}^Q(q,s^{A})$, $\rho \models \lG \fanBr{A} \lG \phi$. Indeed, consider any $\rho \in \msf{Out}^Q(q,s^{A})$ and $i \in \N_1$. Let $\tilde{s}^{A}$ be a strategy profile for the coalition $A$ that coincides with $s^A$ after $\rho[:i-1]$. That is, for all $a \in A$ and $\theta \in Q^+$, we have $\tilde{s}_a^{A}(\theta) := s_a^{A}(\rho[:i-1] \cdot \theta)$. Consider then any $\theta \in \msf{Out}^Q(\rho[i],\tilde{s}^{A})$. By definition of the strategy profiles $\tilde{s}^{A}$ and $s^{A}$, we have that $\rho[:i-1] \cdot \theta \in \msf{Out}^Q(q,s^A)$. Therefore, for all $j \in \N_1$, we have $(\rho[:i-1] \cdot \theta)[j] \models \phi$. Thus, for all $j \in \N_1$, we have $\theta[j] \models \phi$. That is, $\theta \models \lG \phi$. In fact, $\rho[i] \models \fanBr{A} \lG \phi$, which holds for all $i \in \N_1$. Thus, $\rho \models \lG \fanBr{A} \lG \phi$. We obtain that $\fanBr{A} \lG \phi \implies \fanBr{A} \lG \fanBr{A} \lG \phi$.
			
			Since $A \subseteq A'$, we have, for all $\ATL$-formulas $\phi'$, $\fanBr{A} \lG \phi' \implies \fanBr{A'} \lG \phi$. We can conclude that $\fanBr{A} \lG \phi \implies \fanBr{A'} \lG \fanBr{A} \lG \phi$ and $\fanBr{A} \lG \phi \implies \fanBr{A} \lG \fanBr{A'} \lG \phi$.
		\end{itemize}
		We now turn to the operator $\lF$, with dual arguments.
		\begin{itemize}
			\item By definition of the eventually operator $\lF$, we have $\phi \implies \fanBr{A} \lF \phi$, hence $\fanBr{A'} \lF \phi \implies \fanBr{A'} \lF \fanBr{A} \lF \phi$. Also, we have $\fanBr{A'} \lF \phi \implies \fanBr{A} \lF \fanBr{A'} \lF \phi$. 
			\item Consider now two coalitions $A_1,A_2 \subseteq A'$, a concurrent game structure $C$ and any state $q \in Q$. Assume that $q \models \fanBr{A_1} \lF \fanBr{A_2} \lF \phi$. Let us show that $q \models \fanBr{A'} \lF \phi$. Consider a strategy profile $s^{A_1} \in \msf{S}_{A_1}$ for the coalition of agents $A_1$ such that, for all $\rho \in \msf{Out}^Q(q,s^{A_1})$, we have $\rho \models \lF \fanBr{A_2} \lF \phi$. We let $\msf{Already}(\fanBr{A_2} \lF \phi) := \{  \rho \in Q^{+} \cup Q^{\omega} \mid \exists i_\rho \in \N_1,\; i_\rho \leq |\rho|,\; \rho[i_\rho] \models \fanBr{A_2} \lF \phi \}$. By definition of the strategy $s^{A_1}$, we have $\msf{Out}^Q(q,s^{A_1}) \subseteq \msf{Already}(\fanBr{A_2}\lF \phi)$. Then, for all such 
			$\rho \in \msf{Already}(\fanBr{A_2}\lF \phi)$, we let $i_\rho \in \N_1$ be the least index such that $\rho[i_\rho] \models \fanBr{A_2} \lF \phi$ and $\theta_\rho \in Q^* \cup Q^\omega$ be such that $\rho = \rho[:i_\rho-1] \cdot \theta_\rho$. We also let 
			 $s^{\rho,A_2} \in \msf{S}_{A_2}$ be a strategy profile for the coalition $A_2$ such that for all $\theta \in \msf{Out}^Q(\rho[i_\rho],s^{\rho,A_2})$, we have $\theta \models \lF \phi$.  
			
			We now define the strategy profile $s^{A'}$ such that, for all $\rho \in Q^+$:
			\begin{align*}
				s^{A'}(\rho) \text{ coincides with } \begin{cases}
					s^{A_1}(\rho) \text{ on }A_1& \text{ if }\rho \notin \msf{Already}(\fanBr{A_2}\lF \phi) \\
					s^{\rho,A_2}(\theta_\rho) \text{ on }A_2 & \text{ if } \rho \in \msf{Already}(\fanBr{A_2}\lF \phi)
				\end{cases}
			\end{align*}
			where, for all coalitions of agents $X,X_1,X_2$, we say that a strategy $s \in \msf{S}_{X_1}$ coincides with another strategy $s' \in \msf{S}_{X_2}$ on $X \subseteq X_1 \cap X_2$ if, for all $a \in X$, we have $s_a = s'_a$. 
			
			We claim that for all $\rho \in \msf{Out}^Q(q,s^{A'})$, $\rho \models \lF \phi$. Indeed, consider any $\rho \in \msf{Out}^Q(q,s^{A'})$. By definition, on finite paths not in $\msf{Already}(\fanBr{A_2}\lF \phi)$, the strategy $s^{A'}$ coincides with the strategy $s^{A_1}$ on $A_1$. Since $\msf{Out}^Q(q,s^{A_1}) \subseteq \msf{Already}(\fanBr{A_2}\lF \phi)$, it follows that $\rho \in \msf{Already}(\fanBr{A_2}\lF \phi)$. Therefore, we have $\rho = \rho[:i_\rho-1] \cdot \theta_\rho$. In addition, by definition of the strategy $s^{A'}$, we have $\theta_\rho \in \msf{Out}^Q(\rho[i_\rho],s^{\rho,A_2})$. By choice of the strategy $s^{\rho,A_2}$, this implies $\theta_\rho \models \lF \phi$. It follows that $\rho \models \lF \phi$. Therefore, for all $\rho \in \msf{Out}^Q(q,s^{A'})$, we have $\rho \models \lF \phi$. Thus, $q \models \fanBr{A'} \lF \phi$. Hence, we have proved that $\fanBr{A_1} \lF \fanBr{A_2} \lF \phi \implies \fanBr{A'} \lF \phi$. 
			
			We obtain the desired implications since $A,A' \subseteq A'$. 
		\end{itemize}
		We obtain the equivalences for both operators $\lF$ and $\lG$.
	\end{proof}
	
	Let us now come back to the more restrictive context of $\CTL$-formulas where there are only two different strategic quantifiers: $\fanBr{\emptyset}$ (i.e. $\forall$) and $\fanBr{\{1\}}$ (i.e. $\exists$). To properly express the above lemma in this context, we define below the notion of dominating quantifiers (between $\exists$ and $\forall$) used with the operators $\lF$ and $\lG$.
	\begin{definition}
		\label{def:dominating_quantifiers}
		For all $Q_1,Q_2 \in \{ \exists,\forall\}$, we define $\msf{Dom}_{\lF}(Q_1,Q_2) \in \{ Q_1,Q_2 \}$ and $\msf{Dom}_{\lG}(Q_1,Q_2) \in \{ Q_1,Q_2 \}$ as follows:
		\begin{align*}
			\msf{Dom}_{\lF}(Q_1,Q_2) := \begin{cases}
				\exists & \text{ if }Q_1 = \exists \text{ or }Q_2 = \exists \\
				\forall & \text{ otherwise }
			\end{cases}
		\end{align*}
		and 
		\begin{align*}
			\msf{Dom}_{\lG}(Q_1,Q_2) := \begin{cases}
				\forall & \text{ if }Q_1 = \forall \text{ or }Q_2 = \forall \\
				\exists & \text{ otherwise }
			\end{cases}
		\end{align*}
	\end{definition}

	Note that $\msf{Dom}_{\lG}$ is the dual of the operator $\msf{Dom}_{\lF}$, as stated in the observation below.
	\begin{observation}
		\label{obser:dominating_quantifiers}
		For all $\CTL$-formula $\phi$ and $Q,Q_1,Q_2 \in \{ \exists,\forall\}$, we have: 
		\begin{align}
			\neg \msf{Dom}_{\lF}(Q_1,Q_2) \lF & \phi \equiv \msf{Dom}_{\lG}(\neg Q_1,\neg Q_2) \lG \neg \phi \label{eqn:equiv_neg_F_G} \\
			\phi & \implies Q \lF \phi \label{eqn:simple_F}\\
			Q \lG \phi & \implies \phi \label{eqn:simple_G}
		\end{align} 
		where $\neg \forall := \exists$ and $\neg \exists := \forall$.
	\end{observation}
	\begin{proof}
		The equivalence $(\ref{eqn:equiv_neg_F_G})$ is a straightforward consequence of the following equivalences, for $\phi'$ any $\CTL$-formula: $\neg \forall \lG \phi' \equiv \exists \lF \neg \phi'$, $\neg \exists \lF \phi' \equiv \forall \lG \neg\phi'$, $\neg \exists \lG \phi' \equiv \forall \lF \neg\phi'$ and $\neg \forall \lF \phi' \equiv \exists \lG \neg\phi'$.
		%
		%
		The implications $(\ref{eqn:simple_F})$ and $(\ref{eqn:simple_G})$ come from the definition of the operators $\lF$ and $\lG$ and the fact that, in all Kripke structures $K$, for all states $q \in Q$ and for all $\rho \in \msf{Out}^Q(q)$, we have $\rho[1] = q$.
	\end{proof}
	
	We can now state below the corollary of Lemma~\ref{lem:equiv_atl_dominating_quantifiers} with $\CTL$-formulas. (Which justifies the terminology defined above of \emph{dominating} quantifiers.)
	\begin{corollary}
		\label{coro:equiv_ctl_dominating_quantifiers}
		Let $\phi$ be any $\CTL$-formula and $Q_1,Q_2 \in \{ \exists,\forall\}$. We have:
		\begin{equation*}
			Q_1 \lG Q_2 \lG \phi \equiv Q_2 \lG Q_1 \lG \phi \equiv \msf{Dom}_{\lG}(Q_1,Q_2) \lG \phi
		\end{equation*}
		and dually 
		\begin{equation*}
			Q_1 \lF Q_2 \lF \phi \equiv Q_2 \lF Q_1 \lF \phi \equiv \msf{Dom}_{\lF}(Q_1,Q_2) \lF \phi
		\end{equation*}
	\end{corollary}
	\begin{proof}
		This is a direct consequence of Lemma~\ref{lem:equiv_atl_dominating_quantifiers} and the fact that $\forall$ stands for $\fanBr{\emptyset}$ and $\exists$ stands for $\fanBr{\{1\}}$.
	\end{proof}
	
	We deduce the corollary below stating equivalences over $\CTL$-formulas alternating the $\lF$ and $\lG$ operators. 
	\begin{corollary}
		\label{coro:equiv_ctl}
		Let $\phi$ be any $\CTL$-formula and $Q_1,Q_2,Q_3,Q_4 \in \{ \exists,\forall\}$. Let $Q_{\lF} := \msf{Dom}_{\lF}(Q_1,Q_3)$ and $Q_{\lG} := \msf{Dom}_{\lG}(Q_2,Q_4)$. Assume that $Q_1 = Q_{\lF}$ and $Q_4 = Q_{\lG}$. Then:
		\begin{equation*}
			Q_1 \lF Q_2 \lG Q_3 \lF Q_4 \lG \phi \equiv Q_{\lF} \lF Q_{\lG} \lG \phi
		\end{equation*}
		Dually, letting $Q_{\lG} := \msf{Dom}_{\lG}(Q_1,Q_3)$ and $Q_{\lF} := \msf{Dom}_{\lF}(Q_2,Q_4)$, if $Q_1 = Q_{\lG}$ and $Q_4 = Q_{\lF}$, then:
		\begin{equation*}
			Q_1 \lG Q_2 \lF Q_3 \lG Q_4 \lF \phi \equiv Q_{\lG} \lG Q_{\lF} \lF \phi
		\end{equation*}
	\end{corollary}
	\begin{proof}
		We prove the first equivalence, the second is then given by Observation~\ref{obser:dominating_quantifiers} (Eq. (\ref{eqn:equiv_neg_F_G})). 
		
		By Corollary~\ref{coro:equiv_ctl_dominating_quantifiers} and Observation~\ref{obser:dominating_quantifiers} (Eq. (\ref{eqn:simple_F})), we have:
		\begin{align*}
			Q_{\lG} \lG \phi \equiv Q_{2} \lG Q_{4} \lG \phi \implies Q_{2} \lG Q_3 \lF Q_{4} \lG \phi
		\end{align*}
		Since $Q_{\lF} = Q_1$, we have:
		\begin{equation*}
			 Q_{\lF} \lF Q_{\lG} \lG \phi \implies Q_1 \lF Q_{2} \lG Q_3 \lF Q_{4} \lG \phi
		\end{equation*}
		On the other hand, we have by Observation~\ref{obser:dominating_quantifiers} (Eq. (\ref{eqn:simple_G})):
		\begin{equation*}
			Q_{2} \lG Q_3 \lF Q_{4} \lG \phi \implies Q_3 \lF Q_{4} \lG \phi
		\end{equation*}
		Therefore, by Corollary~\ref{coro:equiv_ctl_dominating_quantifiers} and 
		since $Q_{\lG} = Q_4$: 
		\begin{equation*}
			Q_1 \lF Q_{2} \lG Q_3 \lF Q_{4} \lG \phi \implies Q_1 \lF Q_3 \lF Q_{4} \lG \phi \equiv Q_{\lF} \lF Q_{4} \lG \phi \implies Q_{\lF} \lF Q_{\lG} \lG \phi
		\end{equation*}
		We obtain the desired equivalence.
	\end{proof}
	
	We deduce that it is useless to use quantifiers between $\exists \lF$ and $\forall \lG$. 
	\begin{lemma}
		\label{coro:equiv_ctl_exist_F_forall_G}
		Let $\phi$ be any $\CTL$-formula and $\msf{Qt} \in ( \exists \lF,\forall \lF,\exists \lG,\forall \lG)^*$ be a sequence of quantifiers. We have:
		\begin{equation*}
			\exists \lF \; \msf{Qt} \; \forall \lG \phi \equiv \exists \lF \forall \lG \phi
		\end{equation*}
		and 
		\begin{equation*}
			\forall \lG \; \msf{Qt} \; \exists \lF \phi \equiv \forall \lG \exists \lF \phi
		\end{equation*}
	\end{lemma}
	\begin{proof}
		We prove the result by induction on the size of $\msf{Qt}$. If $\msf{Qt}$ if the empty sequence, both equivalences are straightforward. Assume now that both equivalences hold for all sequences $\msf{Qt} \in ( \exists \lF,\forall \lF,\exists \lG,\forall \lG)^*$ of size at most $k$, for some $k \in \N$. Consider a sequence of quantifiers $\msf{Qt} \in ( \exists \lF,\forall \lF,\exists \lG,\forall \lG)^{k+1}$. Let us consider the sequence $\exists \lF \; \msf{Qt} \; \forall \lG$, the arguments are similar for the other one. If $k+1$ is odd, then the sequences $\exists \lF \; \msf{Qt} \; \forall \lG$ features an operator $\lF$ or $\lG$ used twice in a row. By Corollary~\ref{coro:equiv_ctl_dominating_quantifiers}, there is some $\msf{Qt}' \in ( \exists \lF,\forall \lF,\exists \lG,\forall \lG)^{*}$ with $|\msf{Qt}'| \leq k$ such that $\exists \lF \; \msf{Qt} \; \forall \lG \phi \equiv \exists \lF \; \msf{Qt}' \; \forall \lG \phi$. We can then apply our induction hypothesis. 
		
		Assume now that $k+1$ is even and that the sequence $\exists \lF \; \msf{Qt} \; \forall \lG$ does not feature an operator $\lF$ or $\lG$ used twice in a row. If $|\msf{Qt}| = 2$, we can apply Corollary~\ref{coro:equiv_ctl}. Assume now that $|\msf{Qt}| \geq 4$. Let us us write $\msf{Qt}$ as $\msf{Qt} = \msf{T}_1 \cdots \msf{T}_{2n}$ where $\msf{T}_i \in \{\exists \lF,\forall \lF,\exists \lG,\forall \lG\}$ for all $1 \leq i \leq 2n$. By assumption, for all even $1 \leq i \leq 2n$, $\msf{T}_i$ is an $\lF$ operator, whereas for all odd $1 \leq i \leq 2n$, $\msf{T}_i$ is an $\lG$ operator. We denote by $Q_i$ 
		the $\exists$ or $\forall$ quantifier associated with $\msf{T}_i$. There are three cases:
		\begin{itemize}
			\item Assume that $Q_1 = \exists$. Then, we necessarily have $\msf{Dom}_{\lG}(Q_1,Q_3) = Q_3$, and $\msf{Dom}_{\lF}(\exists,Q_2) = \exists$.  Therefore, by Corollary~\ref{coro:equiv_ctl}, we have 
			$\exists \lF \; \msf{Qt} \; \forall \lG \phi 
			\equiv \exists \lF \; \msf{T}_3 \cdots \msf{T}_{2n} \; \forall \lG \phi$. We can then apply our induction hypothesis.
			\item Similarly, assume that $Q_{2n} = \forall$. Then, we necessarily have $\msf{Dom}_{\lF}(Q_{2(n-1)},Q_{2n}) = Q_{2(n-1)}$, and $\msf{Dom}_{\lG}(Q_{2n-1},\forall) = \forall$.  Therefore, by Corollary~\ref{coro:equiv_ctl}, we have 
			$\exists \lF \; \msf{Qt} \; \forall \lG \phi 
			\equiv \exists \lF \; \msf{T}_1 \cdots \msf{T}_{2(n-1)} \; \forall \lG \phi$. We can then apply our induction hypothesis.
			\item Otherwise, we have $\msf{T}_1 = \forall \lG$ and $\msf{T}_{2n} = \exists \lF$. Therefore, by our induction hypothesis and Corollary~\ref{coro:equiv_ctl}, $\exists \lF \; \msf{Qt} \; \forall \lG \phi = \exists \lF \; \msf{T}_1 \cdots \msf{T}_{2n} \; \forall \lG \phi \equiv \exists \lF \; \forall \lG \; \exists \lF \; \forall \lG \phi \equiv \exists \lF \forall \lG \phi$.
		\end{itemize}
		Thus the equivalence holds for all sequences of operators $\msf{Qt} \in ( \exists \lF,\forall \lF,\exists \lG,\forall \lG)^*$ of size $k+1$. The lemma follows.
	\end{proof}

	We also deduce that it is useless to use too long sequences of operators next to $\exists \lF \; \forall \lG$.
	\begin{lemma}
		\label{coro:small_sequence_outside_F_G}
		Let $\phi$ be any $\CTL$-formula. For all $\msf{Qt},\msf{Qt}' \in ( \exists \lF,\forall \lF,\exists \lG,\forall \lG)^*$, there is $\msf{Qt}_s,\msf{Qt}_s' \in ( \exists \lF,\forall \lF,\exists \lG,\forall \lG)^{*}$ of size at most $3$ such that:
		\begin{equation*}
			\msf{Qt} \; \exists \lF \forall \lG \; \msf{Qt}' \phi \equiv \msf{Qt}_s \; \exists \lF \forall \lG \; \msf{Qt}_s' \; \phi
		\end{equation*}
		Similarly, for all $\msf{Qt},\msf{Qt}' \in ( \exists \lF,\forall \lF,\exists \lG,\forall \lG)^*$, there is $\msf{Qt}_s,\msf{Qt}_s' \in ( \exists \lF,\forall \lF,\exists \lG,\forall \lG)^{*}$ of size at most 3 such that:
		\begin{equation*}
			\msf{Qt} \; \forall \lG \exists \lF \; \msf{Qt}' \phi \equiv \msf{Qt}_s \; \forall \lG \exists \lF \; \msf{Qt}_s' \; \phi
		\end{equation*}
	\end{lemma}
	\begin{proof}
		We prove the result for $\msf{Qt} \in ( \exists \lF,\forall \lF,\exists \lG,\forall \lG)^*$ with the sequence $\exists \lF \; \forall \lG$, the three other cases are analogous.
		
		We prove by induction on the size of $\msf{Qt}
		$ that there is $\msf{Qt}_s \in ( \exists \lF,\forall \lF,\exists \lG,\forall \lG)^{*}$ of size at most $3$ such that: $\msf{Qt} \; \exists \lF \forall \lG \; \phi \equiv \msf{Qt}_s \; \exists \lF \forall \lG \; \; \phi$. This obviously holds if $\msf{Qt}$ is of size at most 3. Assume now that it holds for all sizes $i \leq k$ for some $k \geq 3$. Consider a sequence $\msf{Qt} \in ( \exists \lF,\forall \lF,\exists \lG,\forall \lG)^{k+1}$. We let $\msf{Qt}$ be equal to $\msf{Qt} = \msf{Qt}' \cdot \msf{T}_1 \cdot \msf{T}_2 \cdot \msf{T}_3 \cdot \msf{T_4}$ where, for all $1 \leq i \leq n$, we have $\msf{T}_i \in \{ \exists \lF,\forall \lF,\exists \lG, \forall \lG\}$. We assume that in the sequence $\msf{Qt} \; \exists \lF \forall \lG$, no $\lF$ or $\lG$ operator appears twice in a row, otherwise we can apply Corollary~\ref{coro:equiv_ctl_dominating_quantifiers} and our induction hypothesis. For $1 \leq i \leq 4$, we let $Q_i$ denote the $\exists$ or $\forall$ quantifier associated with the operator $\msf{T}_i$. There are two cases:
		\begin{itemize}
			\item If $\msf{T}_2 = \forall \lG$ or $\msf{T}_3 = \exists \lF$, Corollary~\ref{coro:equiv_ctl}  gives that $\msf{Qt} \cdot \exists \lF \forall \lG \phi \equiv \msf{Qt}' \cdot \msf{T}_1 \cdot \msf{T}_2 \cdot \exists \lF \forall \lG \phi$. We can then apply our induction hypothesis to  $\msf{Qt}' \cdot \msf{T}_1 \cdot \msf{T}_2 \cdot \exists \lF \forall \lG \phi$ since $|\msf{Qt}' \cdot \msf{T}_1 \cdot \msf{T}_2| \leq k$. 
			\item Otherwise, we have both $\msf{Dom}_{\lF}(Q_{1},Q_{3}) = Q_{1}$ and $\msf{Dom}_{\lG}(Q_{2},Q_{4}) = Q_{4}$. Therefore, by Corollary~\ref{coro:equiv_ctl}, we have 
			$\msf{Qt} \cdot \exists \lF \forall \lG \phi \equiv \msf{Qt}' \cdot \msf{T}_1 \cdot \msf{T}_4 \cdot \exists \lF \forall \lG \phi$. We can then apply our induction hypothesis to  $\msf{Qt}' \cdot \msf{T}_1 \cdot \msf{T}_4 \exists \lF \forall \lG \phi$ since $|\msf{Qt}' \cdot \msf{T}_1 \cdot \msf{T}_4| \leq k$. 
		\end{itemize}
		Overall, the equivalences also holds for all sequences $\msf{Qt} \in ( \exists \lF,\forall \lF,\exists \lG,\forall \lG)^{k+1}$. The lemma follows.
	\end{proof}

	Finally, let us consider the case where at most one of the two operators $\exists \lF$ or $\forall \lG$ is used. 
	\begin{lemma}
		\label{lem:one_F_or_one_G}
		Let $\phi$ be any $\CTL$-formula. For all $\msf{Qt} \in (\exists \lF,\forall \lF,\exists \lG)^*$, there is $\msf{Qt}_s \in (\exists \lF,\forall \lF,\exists \lG)^{*}$ of size at most $5$ such that $\msf{Qt} \; \phi \equiv \msf{Qt}_s \; \phi$. 

		Similarly, for all $\msf{Qt} \in ( \exists \lG,\forall \lG,\exists \lF)^*$, there is $\msf{Qt}_s \in (\exists \lF,\forall \lF,\exists \lG)^{*}$ of size at most $5$ such that $\msf{Qt} \; \phi \equiv \msf{Qt}_s \; \phi$.
	\end{lemma}
	\begin{proof}
		We prove the result for the first case, the second one is analogous. We prove the result by induction on the size of $\msf{Qt} \in ( \exists \lF,\forall \lF,\exists \lG)^*$ that there is $\msf{Qt}_s \in ( \exists \lF,\forall \lF,\exists \lG)^{*}$ of size at most $5$ such that $\msf{Qt} \; \phi \equiv \msf{Qt}_s \; \phi$. This obviously holds if $\msf{Qt}$ is of size at most 5. Assume now that it holds for all sizes $i \leq k$ of $\msf{Qt}$ for some $k \geq 5$. Consider a sequence $\msf{Qt} \in ( \exists \lF,\forall \lF,\exists \lG)^{k+1}$. We denote $\msf{Qt}$ as follows $\msf{Qt} = \msf{T}_1 \cdot \msf{T}_2 \cdot \msf{T}_3 \cdot \msf{T}_4 \cdot \msf{T}_5 \cdot \msf{T}_6 \cdot \msf{Qt}'$. We assume that in $\msf{Qt}$, no $\lF$ or $\lG$ operator appears twice in a row, otherwise we can apply Corollary~\ref{coro:equiv_ctl_dominating_quantifiers} and our induction hypothesis. For $1 \leq i \leq 6$, we let $Q_i$ denote the $\exists$ or $\forall$ quantifier associated with the operator $\msf{T}_i$. There are two cases. 
		\begin{itemize}
			\item Assume that there is some $1 \leq i \leq 6$ such that $\msf{T}_i = \exists \lF$. If $i \leq 3$, 
			then we have $\msf{Dom}_{\lF}(Q_i,Q_{i+2}) = Q_{i}$ and $\msf{Dom}_{\lG}(Q_{i+1},Q_{i+3}) = Q_{i+3}$ since $Q_{i+1} = Q_{i+3} = \exists$. Thus, by Corollary~\ref{coro:equiv_ctl}, for all CTL-formulas $\phi'$, we have $\msf{T}_i \cdot \msf{T}_{i+1} \cdot \msf{T}_{i+2} \cdot \msf{T}_{i+3} \cdot \phi' \equiv \msf{T}_i \cdot \msf{T}_{i+3} \cdot \phi'$. We can then apply our induction hypothesis to conclude. If $i \geq 4$, then we have $\msf{Dom}_{\lF}(Q_{i-2},Q_{i}) = Q_{i}$ and $\msf{Dom}_{\lG}(Q_{i-3},Q_{i-1}) = Q_{i-3}$ since $Q_{i-3} = Q_{i-1} = \exists$. Thus, by Corollary~\ref{coro:equiv_ctl}, for all CTL-formulas $\phi'$, we have $\msf{T}_{i-3} \cdot \msf{T}_{i-2} \cdot \msf{T}_{i-1} \cdot \msf{T}_{i} \cdot \phi' \equiv \msf{T}_{i-3} \cdot \msf{T}_{i} \cdot \phi'$. We can then apply our induction hypothesis to conclude.
			\item Otherwise, we have $\msf{Qt} \in (\forall \lF,\exists \lG)^{k+1}$ and thus we can apply Corollary~\ref{coro:equiv_ctl} and our induction hypothesis to conclude.
		\end{itemize}
		Therefore, the result holds for sequences $\msf{Qt} \in ( \exists \lF,\forall \lF,\exists \lG)^{k+1}$. The lemma follows.
	\end{proof}

	Overall, we obtain the lemma below, bounding the size of CTL-formulas (using only unary operators) that is sufficient to consider.
	\begin{lemma}
		\label{lem:bound_size_CTL_formula_no_X}
		There is a bound $B \in \N$ such that, 
		for all $\Ut \subseteq \{ \neg,\lF,\lG \}$, for all CTL-formulas $\phi = \msf{Qt} \cdot \phi' \in \CTL(\prop,\Ut,\emptyset,\emptyset,0)$, there is a sequence of $\CTL$-operators $\msf{Qt}'$ such that $\phi'' = \msf{Qt}' \cdot \phi' \in \CTL(\prop,\Ut,\emptyset,\emptyset,0)$, $|\msf{Qt}| \leq B$, and $\phi \equiv \phi''$.
	\end{lemma}
	\begin{proof}
		Let us first consider any $\phi = \msf{Qt} \cdot \phi' \in \CTL(\prop,\Ut,\emptyset,\emptyset,0)$ with $\msf{Qt}$ as follows $\msf{Qt} = \msf{T}_1 \cdots \msf{T}_n$ where, for all $1 \leq i \leq n$, we have $\msf{T}_i \in \{ \exists \lF,\forall \lF,\exists \lG,\forall \lG \}$. There are two cases:
		\begin{itemize}
			\item Assume that $\phi$ uses both operators $\exists \lF$ and $\forall \lG$. 
			Consider $1 \leq i < j \leq n$ such that $\msf{T}_i,\msf{T}_j \in \{ \exists \lF,\forall \lG\}$ and $\msf{T}_i \neq \msf{T}_j$. By Lemma~\ref{coro:equiv_ctl_exist_F_forall_G}, we have $\phi \equiv \msf{T}_1 \cdots \msf{T}_i \cdot \msf{T}_j \cdots \msf{T}_n \; \phi'$. Furthermore, by Lemma~\ref{coro:small_sequence_outside_F_G}, there is $\msf{Qt}_s,\msf{Qt}_s' \in ( \exists \lF,\forall \lF,\exists \lG,\forall \lG)^{*}$ of size at most $3$ such that: $\msf{T}_1 \cdots \msf{T}_i \cdot \msf{T}_j \cdots \msf{T}_n \; \phi' \equiv \msf{Qt}_s \cdot \msf{T}_i \cdot \msf{T}_j \cdot \msf{Qt}'_s \; \phi$. We let $\msf{Qt}' := \msf{Qt}_s \cdot \msf{T}_i \cdot \msf{T}_j \cdot \msf{Qt}'_s \; \phi' \in \CTL(\prop,\Ut,\emptyset,\emptyset)$. We have $|\msf{Qt}'| \leq 8$ and $\msf{Qt}' \cdot \phi' \equiv \msf{Qt} \cdot \phi'$. 
			\item Assume $\phi$ does not use both operators $\exists \lF$ or $\forall \lG$. Then, by Lemma~\ref{lem:one_F_or_one_G}, there is $\msf{Qt}'$ of size at most 5 such that $\msf{Qt}' \cdot \phi' \equiv \msf{Qt} \cdot \phi'$. 
		\end{itemize}
		Thus, in both cases, there is a sequence of $\CTL$-operators $\msf{Qt}'$ of size at most 8such that $\msf{Qt}' \cdot \phi' \equiv \msf{Qt} \cdot \phi'$ with at most 8 quantifiers. 
	
		Then, it is straightforward to handle the cases where the sequence of quantifiers $\msf{Qt}$ uses negations, since 
		%
		we have the equivalence, for all $\CTL$-formulas $\phi$, $\lG \phi \equiv \neg \lF \neg \phi$ and $\lF \phi \equiv \neg \lG \neg \phi$
		. 
		%
		%
	\end{proof}

	We obtain a similar statement with formulas that can use (a bounded amount of) binary operators.	
	\begin{lemma}
		\label{lem:bound_size_CTL_formula_no_X_with_binary}
		For all $\Ut \subseteq \{ \neg,\lF,\lG \}$, $\Bl \subseteq \Op{Bin}{lg}$ and $n \in \N$, there is a bound $B_n \in \N$ such that, for all CTL-formulas $\phi \in \CTL(\prop,\Ut,\emptyset,\Bl,n)$, there is a $\CTL$-formula $\phi' \in \CTL(\prop,\Ut,\emptyset,\Bl,n)$ such that  $\Size{\phi'} \leq B_n$ and $\phi' \equiv \phi$.
	\end{lemma}
	\begin{proof}
		Let $B$ denote the bound from Lemma~\ref{lem:bound_size_CTL_formula_no_X}. 
		We proceed by induction on $n$ (note that the bound $B_n$ depends on $n$). For the case $n = 0$, it suffices to consider $B_0 := B + 1$. Assume now that it holds for some $n \in \N$. 
		We let $B_{n+1} := 2 B_n + B + 1 \geq B_n$. Consider a formula $\phi \in \CTL(\prop,\Ut,\emptyset,\Bl,n+1) \setminus \CTL(\prop,\Ut,\emptyset,\Bl,n)$. This formula can be written as $\phi = \msf{Qt} \cdot (\phi_1 \bullet \phi_2)$ where $\msf{Qt}$ is a sequence of unary operators, $\bullet \in \Bl$ is a binary operator and $\phi_1,\phi_2 \in \CTL(\prop,\Ut,\emptyset,\Bl,n)$. By our induction hypothesis, there are formulas $\phi_1',\phi_2' \in \CTL(\prop,\Ut,\emptyset,\Bl,n)$ such that $\Size{\phi_1'},\Size{\phi_2'} \leq B_n$ and $\phi_1 \equiv \phi_1'$,$\phi_2 \equiv \phi_2'$. In addition, by Lemma~\ref{lem:bound_size_CTL_formula_no_X}, there is a sequence of operators $\msf{Qt}'$ of size at most $B$ such that, for all $\CTL$-formulas $\phi'$, we have $\msf{Qt} \cdot \phi' \equiv \msf{Qt}' \cdot \phi'$ and $\msf{Qt}' \cdot (\phi_1 \bullet \phi_2) \in \CTL(\prop,\Ut,\emptyset,\Bl,n+1)$. Overall, we have $\msf{Qt}' \cdot (\phi_1' \bullet \phi_2') \in \CTL(\prop,\Ut,\emptyset,\Bl,n+1)$ with $\Size{\msf{Qt}' \cdot (\phi_1' \bullet \phi_2')} \leq B_{n+1}$ and $\msf{Qt}' \cdot (\phi_1' \bullet \phi_2') \equiv \msf{Qt} \cdot (\phi_1 \bullet \phi_2)$. Hence, our inductive property holds also for $n+1$. The lemmas follows.
	\end{proof}	
	
	We can establish that deciding the learning $\CTL$ decision problem without the operator $\lX$ can be done in non-deterministic logarithmic space.
	\begin{lemma}
		\label{lem:ctl_unary_no_X_nl}
		For all $\Ut \subseteq \{\lF,\lG,\neg \}$, $\Bl \subseteq \Op{Bin}{lg}$ and $n \in \N$, the problem $\CTL_\msf{Learn}(\Ut,\emptyset,\Bl,n)$ can be decided in non-deterministic logarithmic space.  
	\end{lemma}
	\begin{proof}
		By Immerman-Szelepcsényi's theorem \cite{immerman1988nondeterministic}, we have $\msf{NL} = \msf{coNL}$. In other words, any problem that can decided by a logarithmic-space Turing machine using only existential (TM) states can also be decided by a logarithmic-space Turing machine using only universal (TM) states.
	
		
		Now, consider the bound $B_n$ from Lemma~\ref{lem:bound_size_CTL_formula_no_X_with_binary}. We let $\widetilde{\CTL}(\Ut,\emptyset,\Bl,n)$ denote a set of $\CTL$-formula structures, i.e. $\CTL$-formulas where the propositions are left unspecified. Each one of these formula structure can be seen, for some $k \in \N$, as functions $\prop^k \rightarrow \CTL(\prop,\Ut,\emptyset,\Bl,n)$, for all sets of propositions $\prop$. The set $\widetilde{\CTL}(\Ut,\emptyset,\Bl,n)$ corresponds to the set of all $\CTL$-formula structures such that, for all non-empty sets of propositions $\prop$, there are propositions in $\prop$ specifying them such that the obtained formula is in $\CTL(\prop,\Ut,\emptyset,\Bl,n)$ and of size at most $B_n$. Note that the set $\widetilde{\CTL}(\Ut,\emptyset,\Bl,n)$ is finite. For all $\widetilde{\phi} \in \widetilde{\CTL}(\Ut,\emptyset,\Bl,n)$, let us exhibit an $\msf{NL}$-algorithm $\msf{SatKripke}_{\widetilde{\phi}}$ that decides, given the propositions specifying the formula structure $\widetilde{\phi}$ and a state in a Kripke structure, whether that state satisfies the obtained $\CTL$-formula. 
		
		First, checking that a state satisfies a proposition is straightforward. Assume now that we have designed an $\msf{NL}$-algorithm $\msf{SatKripke}_{\widetilde{\phi}}$ for a $\CTL$-formula structure $\widetilde{\phi}$. The case of the formula structure $\widetilde{\psi} = \neg \widetilde{\phi}$ is straightforward since $\msf{NL} = \msf{coNL}$. Consider now the formula structure $\widetilde{\psi} = \exists \lF \widetilde{\phi}$. Consider any propositions specifying $\widetilde{\psi}$ into a $\CTL$-formula $\psi$ (and therefore $\widetilde{\phi}$ into a $\CTL$-formula $\phi$). Checking that a state $q$ satisfies $\psi$ amounts to guessing a path from $q$ of size at most $|Q|$ and checking, with a logarithmic space Turing machine using only existential (TM) states, that it satisfies $\phi$ (by calling the algorithm $\msf{SatKripke}_{\widetilde{\phi}}$). This induces an $\msf{NL}$-algorithm. Consider now the formula $\widetilde{\psi} = \forall \lF \widetilde{\phi}$. Since $\msf{NL} = \msf{coNL}$, there is a $\msf{coNL}$-algorithm $\msf{SatKripke}_{\widetilde{\phi}}'$ for the formula structure $\widetilde{\phi}$
		. As above, consider any propositions specifying $\widetilde{\psi}$ into a $\CTL$-formula $\psi$ (and therefore $\widetilde{\phi}$ into a $\CTL$-formula $\phi$). Then, checking that a state $q$ satisfies the formula $\psi$ amounts to exploring, with universal (TM) states, paths of lengths at most $|Q|$ from $q$ and check that we encounter a state that satisfies $\varphi$ (by calling $\msf{SatKripke}_{\widetilde{\phi}}'$). This induces a $\msf{coNL}$-algorithm, and therefore an $\msf{NL}$-algorithm as well. The arguments are similar for the operators $\exists \lG$ and $\forall \lG$. Furthermore, consider a formula structure $\widetilde{\psi} = \widetilde{\phi_1} \bullet \widetilde{\phi_2}$ for some binary operator $\bullet \in \Bl$, and assume that we have designed $\msf{NL}$-algorithms $\msf{SatKripke}_{\widetilde{\phi_i}},\msf{SatKripke}_{\neg \widetilde{\phi_i}}$ for the $\CTL$-formula structures $\widetilde{\phi_i},\neg\widetilde{\phi_i}$, for $i \in \{1,2\}$. In that case, we may rewrite $\bullet$ in disjunctive normal form (for instance $x_1 \Leftrightarrow x_2 
		\equiv (x_1 \wedge x_2) \lor (\neg x_1 \wedge \neg x_2)$). Then, an $\msf{NL}$-algorithm  $\msf{SatKripke}_{\widetilde{\psi}}$ could guess which clause to satisfy and run at most two of the algorithms $\msf{SatKripke}_{\widetilde{\phi_1}},\msf{SatKripke}_{\widetilde{\phi_2}},\msf{SatKripke}_{\neg\widetilde{\phi_1}},\msf{SatKripke}_{\neg\widetilde{\phi_2}}$. That way, we obtain an $\msf{NL}$-algorithm $\msf{SatKripke}_{\widetilde{\psi}}$. 
		
		Overall, we have an $\msf{NL}$-algorithm for all $\CTL$-formula structures (and their negations since $\msf{NL} = \msf{coNL}$) in $\widetilde{\CTL}(\Ut,\emptyset,\Bl,n)$. Let us now design an $\msf{NL}$-algorithm for the $\CTL_\msf{Learn}(\Ut,\emptyset,\Bl,n)$ decision problem. That algorithm could do the following, on an input $(\prop,\mathcal{P},\mathcal{N},B)$:
		\begin{enumerate}
			\item Loop over all $\CTL$-formula structures $\widetilde{\phi}$ in $\widetilde{\CTL}(\Ut,\emptyset,\Bl,n)$ 
			\item Letting $k$ be the number of unspecified propositions in $\widetilde{\phi}$, loop over all tuples $t$ in $\prop^k$ for which the $\CTL$-formula obtained from $\widetilde{\phi}$ with $t$ has size at most $\min(B_n,B)$
			\item Loop over:
			\begin{itemize}
				\item all starting states $q$ of all positive structures and run the $\msf{NL}$-algorithm $\msf{SatKripke}_{\widetilde{\phi}}$ on $t$ and $q$
				\item all negative structures $K \in \mathcal{N}$, guess a starting state $q$ in $K$ and run the $\msf{NL}$-algorithm $\msf{SatKripke}_{\neg \widetilde{\phi}}$ on $t$ and $q$
			\end{itemize}
			Accept if all calls return positive answers
		\end{enumerate}
		If the algorithm does not accept the input, then it rejects it. The first loop is entered a bounded number of times (independent of the input). Since of formula structures in $\widetilde{\CTL}(\Ut,\emptyset,\Bl,n)$ use at most $n$ occurrences of binary operators, the number $k$ of unspecified variables in $\widetilde{\phi}$ is at most $2^n$, thus the second loop is entered at most $|\prop|^{2^n}$ times, which is polynomial in $|\prop|$. Therefore, the algorithm that we have designed above runs in non-deterministic logarithmic space. Furthermore, it decides the problem $\CTL_\msf{Learn}(\Ut,\emptyset,\Bl,n)$ by Lemma~\ref{lem:bound_size_CTL_formula_no_X_with_binary}.
		
	\end{proof}
	
	\paragraph{$\msf{NL}$-hardness}
	To establish $\msf{NL}$-hardness, we are going to exhibit a reduction from the problem of reachability in a graph. We introduce in the definition below that decision problem as a sub-case of a more general decision problem of reachability in a two-player game, that we will use in the next section.
	\begin{definition}[Reachability Game]
		\label{def:reach_decision_problem}
		We denote by $\msf{Reach}$ the following decision problem:
		\begin{itemize}
			\item Input: two propositions $p,\bar{p}$ and a $\{1,2\}$-game structure $G$ on $\{ p,\bar{p} \}$, such that all state in $G$ satisfy exactly one of the two propositions $p$ and $\bar{p}$, and with a single starting state;
			\item Output: yes iff $q \models \fanBr{\{1\}} \lF p$. 
		\end{itemize}
		
		Similarly, we consider the decision problem 1-$\msf{Reach}$, where the game structure taken as input is in fact a $\{1\}$-game structure, i.e. it is a Kripke structure. 
	\end{definition}
	
	\begin{theorem}[\cite{ImmermanReachPComplete},\cite{papadimitriou1994computational}]
		The decision problem $\msf{Reach}$ is $\msf{P}$-complete and the decision problem 1-$\msf{Reach}$ is $\msf{NL}$-complete under logarithmic space reductions. 
	\end{theorem}
	\begin{proof}
		The decision problem $\msf{Reach}$ is equivalent to solving two-player reachability games, which is $\msf{P}$-complete \cite{ImmermanReachPComplete}. The decision problem 1-$\msf{Reach}$ is equivalent to solving reachability in a graph (\textquotedblleft s-t connectivity\textquotedblright{}), which is $\msf{NL}$-complete \cite[Theorem 16.2]{papadimitriou1994computational}.
	\end{proof}
	
	Let us define the reduction that we consider. 
	\begin{definition}
		\label{def:reduction_ctl_2_nl_complete}
		Consider any input $p,\bar{p},K$ of the decision problem 1-$\msf{Reach}$. We let $K_{\bar{p}}$ be a single-state Kripke structure whose only state is labeled by $\{\bar{p}\}$. 

		We define the inputs $\msf{In}^{\CTL,\lF}_{(p,\bar{p},K)} := (\{p,\bar{p}\},\{K\},\{K_{\bar{p}}\},2)$ and  $\msf{In}^{\CTL,\lG}_{(p,\bar{p},K)} := (\{p,\bar{p}\},\{K_{\bar{p}}\},\{K\},2)$ of a $\CTL$ learning problem. 
	\end{definition}
	
	The definition above satisfies the lemma below.
	\begin{lemma}
		\label{lem:ctl_2_reduction}
		Consider some set of unary operators $\Ut \subseteq \{ \lX,\lF,\lG,\neg \}$, $\Bl \subseteq \Op{Bin}{lg}$,  $n \in \N$ and an input $(p,\bar{p},K)$ of the 1-$\msf{Reach}$ decision problem. Let $\msf{H} \in \{\lF,\lG\}$. If $\msf{H} \in \Ut$, then the input $(p,\bar{p},K)$ is a positive instance of the 1-$\msf{Reach}$ decision problem if and only if $\msf{In}^{\CTL,\msf{H}}_{(p,\bar{p},K)}$ is a positive instance of the  $\CTL_\msf{Learn}(\Ut,\emptyset,\Bl,n)$ decision problem. 
	\end{lemma}
	\begin{proof}
		Let us first consider the case $\msf{H} = \lF$. Assume that $(p,\bar{p},K)$ is a positive instance of 1-$\msf{Reach}$. We let $\varphi := \exists \lF p$. We have $\Size{\varphi} = 2$. Furthermore, $K_{\bar{p}} \not\models \varphi$ and, by assumption, $K \models \varphi$. Hence, $\msf{In}^{\CTL^,\msf{F}}_{(p,\bar{p},K)}$ is a positive instance of the  $\CTL_\msf{Learn}^2(\Ut,\emptyset,\Bl,n)$ decision problem. 
		
		On the other hand, assume that $\msf{In}^{\CTL,\msf{F}}_{(p,\bar{p},K)}$ is a positive instance of the  $\CTL_\msf{Learn}(\Ut,\emptyset,\Bl,n)$ decision problem. Consider a separating formula $\varphi$ of size at most 2 that accepts the Kripke structure $K$ and rejects the Kripke structure $K_{\bar{p}}$. The structures $K$ and $K_{\bar{p}}$ satisfy the following property: for all states $q$, $q \models p$ if and only if $q \models \neg \bar{p}$. Let us show that on these structures, we have $\varphi \implies \exists \lF p$. If $\varphi$ uses a negation, then $\varphi = \neg \bar{p}$ and $\varphi \equiv p \implies \exists \lF p$ on such structures. If $\varphi$ uses a binary operator, then it is equivalent to either $p,\bar{p},\msf{True},\msf{False}$. The only possibility is that $\varphi$ is equivalent to $p$, in which case $\varphi \implies \exists \lF p$. Otherwise, if $\varphi$ does not use a negation, then necessarily it uses the proposition $p$. Therefore, $\varphi \in \{ p,\exists \lF p,\forall \lF p,\exists \lG p,\forall \lG p,\exists \lX p,\forall\lX p\}$. One can then check that, in all these cases, we have $\varphi \implies \exists \lF p$. Therefore, since $K \models \varphi$, we also have $K \models \exists \lF p$. Hence, $(p,\bar{p},K)$ is a positive instance of 1-$\msf{Reach}$.
		
		The case $\msf{H} = \lG$ is dual: in that case, we consider the formula $\varphi := \forall \lG \bar{p}$. 
	\end{proof}
	
	With the two above lemmas, Theorem~\ref{thm:ctl_unary_no_X_P} follows.
	\begin{proof}
		The fact that the decision problem $\CTL_\msf{Learn}(\Ut,\emptyset,\Bl,n)$ is in $\msf{NL}$ comes from Lemma~\ref{lem:ctl_unary_no_X_nl}. The fact that, if $\lF \in \Ut$ or $\lG \in \Ut$, then it is also $\msf{NL}$-hard comes from Lemma~\ref{lem:ctl_2_reduction} and the fact that the reduction given in Definition~\ref{def:reduction_ctl_2_nl_complete} can be computed in logarithmic space.
	\end{proof}

	\subsection{$\ATL$ learning without the operator $\lX$}
	We have seen in the previous section that $\CTL$ learning with the operator $\lX$ is $\msf{NP}$-complete, while it can be solved in non-deterministic logarithmic space 
	if this operator is not allowed anymore. In this section, we study the learning problem for $\ATL$-formulas, that do not use the operator $\lX$, and concurrent game structures, with two or three agents. 
	
	The cases of $\ATL$ learning with two or three agents are different. However, we start by giving some central definitions and establishing central lemmas that will be used in both cases for the $\msf{NP}$-hardness proof. 
	
	\subsubsection{Alternating $\ATL$-formulas and turn-based structures} 	
	
	We have introduced the notion of turn-based structures in Section~\ref{subsubsec:handling_binary_operators}. We will also use them in this subsection. As exemplified in Figure~\ref{fig:example_1_2_3}, whenever we draw turn-based game structures, we will use the following conventions:
	\begin{itemize}
		\item Blue diamond-shaped states are Agent-1 states;
		\item Red rectangle-shaped states are Agent-2 states;
		\item Violet octagon-shaped states are Agent-3 states.
	\end{itemize}
	
	\begin{figure}
		\centering
		\begin{minipage}{0.2\linewidth}
			\includegraphics[scale=1.4]{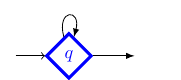}
		\end{minipage}
		\begin{minipage}{0.05\linewidth}
			\phantom{a}
		\end{minipage}
		\begin{minipage}{0.2\linewidth}
			\includegraphics[scale=1.2]{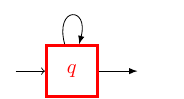}
		\end{minipage}
		\begin{minipage}{0.05\linewidth}
			\phantom{a}
		\end{minipage}
		\begin{minipage}{0.2\linewidth}
			\includegraphics[scale=1.3]{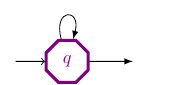}
		\end{minipage}
		\caption{A Agent-1 state on the left, an Agent-2 state in the middle, and an Agent-3 state on the right.}
		\label{fig:example_1_2_3}
	\end{figure}

	Let us now introduce several useful definition on turn-based game structures. First of all, we consider self-looping turn-based structures.
	\begin{definition}
		\label{def:self-looping} Let $\alpha \in \{2,3\}$, $\Ag := [1\,\ldots,\alpha]$ and consider a $(0,\emptyset)$-proper $\Ag$-turn-based structure $T$. This structure is \emph{self-looping} if, for all states $q$ in $T$, we have that $q \in \msf{Succ}(q)$. 
		
		Furthermore, we will often use the states $q^{\msf{win}}$ and $q^{\msf{lose}}$ which will always satisfy the following: $\msf{Succ}(q^{\msf{win}}) := \{q^{\msf{win}}\}$,  $\msf{Succ}(q^{\msf{lose}}) := \{q^{\msf{lose}}\}$, $\pi(q^{\msf{win}}) := \{p\}$, and $\pi(q^{\msf{lose}}) := \{\bar{p}\}$.
	\end{definition}

	Such structures satisfy the lemma below.
	\begin{lemma}
		\label{lem:useless_op}
		Let $k \in \{2,3\}$ and $\Ag := [1\,\ldots,k]$. For all $\ATL^k$-formulas $\phi$, the three $\ATL$-formulas $\phi$, $\fanBr{\emptyset} \lF \phi$, and $\fanBr{\Ag} \lG \phi$ are equivalent on all $\Ag$-turn-based self-looping structures. 
	\end{lemma}
	\begin{proof}
		By definition of the operators $\lF$ and $\lG$, we have $\fanBr{\Ag}\lG \phi \implies \phi \implies \fanBr{\emptyset} \lF \phi$. Furthermore, in any $\Ag$-turn-based structure, for any state $q$ satisfying $q \in \msf{Succ}(q)$, there is a strategy $s$ for the coalition of agents $\Ag$ such that $\msf{Out}^Q(q,s) = \{ q^\omega \}$. It follows that $q \models \emptyset \lF \phi$ implies $q \models \phi$, which itself implies $q \models 1,2 \lG \phi$.
	\end{proof}
	
	Let us now introduce below another notion on turn-based structures: (alternating) winning path. In a subsequent definition, we will also introduce the notion of alternating formulas and relate them with alternating winning paths.
	\begin{definition}
		\label{def:winning_path}
		Let $\alpha \in \{2,3\}$, $\Ag := [1\,\ldots,\alpha]$ and consider a $(0,\emptyset)$-proper self-looping $\Ag$-turn-based structure $T$. Consider a state $q \in Q$. A \emph{winning path} in $T$ from $q$ is a sequence of states $\rho \in Q^{n+1}$, for some $n \in \N$, such that: $\rho[1] = q$, $\pi(\rho[n+1]) = \{p\}$, and for all $1 \leq i \leq n$, we have $\pi(\rho[n+1]) = \{\bar{p}\}$ and $\rho_{i+1} \in \msf{Succ}(\rho_i)$. We let $\msf{WinPath}_T(q)$ denote the set of all winning paths from $q$. 
		
		The winning path $\rho \in \msf{WinPath}_T(q)$ is \emph{safe} if: for all $1 \leq i \leq n$, we have $\msf{Succ}(\rho[i]) = \{\rho[i],\rho[i+1]\}$ and $\msf{Succ}(\rho[n+1]) = \{\rho[n+1]\}$. Note that, in that case, we have $|\msf{WinPath}_T(q)| = 1$. 
		
		For any two coalitions $A,A' \subseteq \Ag$ such that $A \cap A' = \emptyset$, we say that the winning path $\rho \in \msf{WinPath}_T(q)$ is \emph{$(A,A',n)$-alternating} if: for all $1 \leq i \leq n$, we have $\msf{AgSt}(\rho[i]) \in A \cup A'$ and for all $1 \leq i \leq n-1$, we have $\msf{AgSt}(\rho[i]) \in A$ if and only if $\msf{AgSt}(\rho[i+1]) \in A'$. 
	\end{definition} 

	Let us now define the notion of alternating $\ATL$-formulas.
	\begin{definition}
		\label{def:alternating_formulas}
		Let $\alpha \in \{2,3\}$ and $\Ag := [1\,\ldots,\alpha]$. Consider two coalitions $A,A' \subseteq \Ag$ such that $A \cap A' = \emptyset$ (the union is disjoint). For all $\msf{O} \subseteq \{\lF,\lG\}$, we let $\msf{Op}_{A,A'}(\msf{O})$ denote the set $\msf{Op}_{A,A'}(\msf{O}) := \{ \fanBr{B} \msf{H} \mid B \cap A' = \emptyset \text{ or } B \cap A = \emptyset,\; \msf{H} \in \msf{O} \}$. Then, $\msf{Op}_{A,A'}(\msf{O})$-formulas refer to the set of $\ATL(\{p\},\msf{O},\emptyset,\emptyset,0)$-formulas $\phi$ that only use operators in $\msf{Op}_{A,A'}(\msf{O})$. 
		
		Then, we say that an $\msf{Op}_{A,A'}(\msf{O})$-formula $\phi$ is $(A,A',n)$-alternating if (recall Notation~\ref{nota:atl_formulas}):
		\begin{equation*}
			\phi \in \msf{Op}_{A,A'}(\msf{O})^* \cdot \fanBr{B_1} \lF \cdot \; \msf{Op}_{A,A'}(\msf{O})^* 
			\cdots \; \msf{Op}_{A,A'}(\msf{O})^* \cdot \fanBr{B_n} \lF \cdot \; \msf{Op}_{A,A'}(\msf{O})^* \cdot p
		\end{equation*}
		where, for all odd $i \leq n$, we have $A' \cap B_i \neq \emptyset$, and for all even $i \leq n$, we have $A \cap B_i \neq \emptyset$. 
	\end{definition}

	We state below a lemma relating alternating formulas and turn-based structures with alternating winning paths.
	\begin{lemma}
		\label{lem:ATL_formula_necessary_sufficient}
		Let $k \in \{2,3\}$ and $\Ag := [1\,\ldots,k]$. Consider a self-looping structure $T$ and let $n \in \N$. We have, for all states $q \in Q$, and two non-empty coalitions $A,A' \subseteq \Ag$ such that $A \cap A' = \emptyset$:
		\begin{itemize}
			\item[a)] If all winning paths $\rho \in \msf{WinPath}_T(q)$ are $(A,A',n)$-alternating, then, all $\msf{Op}_{A,A'}(\{\lF,\lG\})$-formulas $\phi$ that accept the state $q$ are $(A,A',n)$-alternating formulas.
			\item[b)] If $|A| = |A'| = 1$ and there is a safe winning path from $q$ that is $(A,A',n)$-alternating, then any $\msf{Op}_{A,A'}(\{\lF,\lG\})$-formula that is $(A,A',n)$-alternating accepts the state $q$.
		\end{itemize}
	\end{lemma}
	\begin{proof}
		Let us first argue the following: if an $\msf{Op}_{A,A'}(\{\lF,\lG\})$-formula $\phi$ accepts a state reachable from some state $q' \in Q$, then $\msf{WinPath}_T(q') \neq \emptyset$. Indeed, if $\msf{WinPath}_T(q') = \emptyset$, then all states $q''$ reachable from $q'$ are such that $\pi(q'') = \{\bar{p}\}$. Therefore, the $\msf{Op}_{A,A'}(\{\lF,\lG\})$-formula $\phi$ --- that does not use negations and that is such that $\prop(\phi) = p$ --- does not accept any state reachable from $q'$.
		
		Now, we prove both items of the lemma by induction on $n \in \N$
		. Item a) holds when $n = 0$ since all $\msf{Op}_{A,A'}(\{\lF,\lG\})$-formulas $\phi$ are $(A,A',0)$-alternating. Furthermore, if there is a safe winning path from $q$ that is $(A,A',n)$-alternating, it means that $\pi(q) = \{p\}$ and $\msf{Succ}(q) = \{q\}$. Thus, since any $\msf{Op}_{A,A'}(\{\lF,\lG\})$-formula $\phi$ is such that $\prop(\phi) = \{p\}$ and does not use any negations, Item b) follows.
		
		Assume now that Items a) and b) hold for some $n \in \N$. Let us first consider Item a). Assume that all winning paths $\rho \in \msf{WinPath}_T(q)$ are $(A,A',n+1)$-alternating. Let us assume that $\msf{WinPath}_T(q) \neq \emptyset$, otherwise no $\msf{Op}_{A,A'}(\{\lF,\lG\})$-formula accepts $q$. In particular, it must be that $\msf{AgSt}(q) \in A$ and $\pi(q) = \{\bar{p}\}$. Let us show by induction on  $\msf{Op}_{A,A'}(\{\lF,\lG\})$-formulas $\phi$ the following property $\mathcal{P}(\phi)$:  if $q \models \phi$, then there is a state $q' \in \msf{Succ}(q) \setminus \{q\}$ and some $\ATL$-formula $\phi'$ such that $q' \models \phi'$ and $\phi \in \msf{Op}_{A,A'}(\{\lF,\lG\})^* \cdot \fanBr{B} \lF \cdot \phi'$ such that $A \cap B \neq \emptyset$. The property $\mathcal{P}(p)$ holds since $q \not\models p$. Assume now that $\mathcal{P}(\phi)$ holds for some $\msf{Op}_{A,A'}(\{\lF,\lG\})$-formula $\phi$. Let $\psi := \msf{O} \cdot \phi$ for some $\msf{O} \in \msf{Op}_{A,A'}(\{\lF,\lG\})$. Assume that $q \models \psi$. Then, there are two cases:
		\begin{itemize}
			\item Assume that $\msf{O} = \fanBr{B} \lF$ for some coalition $B \subseteq \Ag$. If $q \models \phi$, we can deduce $\mathcal{P}(\psi)$ from $\mathcal{P}(\phi)$. Assume now that it is not the case, i.e. $q \not\models \phi$. If $\msf{AgSt}(q) \notin B$, for all strategies $s_B$ of the coalition $B$, we have $q^\omega \in \msf{Out}^Q(q,s_B)$. This is not possible since $q \models \psi$ and $q \not \models \phi$. In fact, $\msf{AgSt}(q) \in B$. Since $\msf{AgSt}(q) \in A$, this implies $A \cap B \neq \emptyset$. Now, since $q \models \psi$, there is a strategy $s_B$ for the coalition $B$ such that, for all $\rho \in \msf{Out}^Q(q,s_B)$, we have $\rho \models \lF \phi$. Consider some $\rho \in \msf{Out}^Q(q,s_B)$ and let $i \in \N_1$ be the least index such that $\rho[i] \models \phi$. We have $\rho[i] \in \msf{Succ}(q) \setminus \{q\}$. Since there is some $j \geq i$ such that $\rho[j] \models \phi$ (and $\rho[j]$ is reachable from $\rho[i]$), then we have $\msf{WinPath}_T(\rho[i]) \neq \emptyset$. 
			Therefore, since all winning paths from $q$ are $(A,A',n+1)$-alternating, it follows that $\msf{AgSt}(\rho[i]) \in A'$. Since $A \cap B \neq \emptyset$, we have $A' \cap B = \emptyset$ (since $\fanBr{B} \lF \in \msf{Op}_{A,A'}(\{\lF,\lG\})$), and thus $\msf{AgSt}(\rho[i]) \notin B$. Furthermore, since $\rho \in \msf{Out}^Q(q,s_B)$, it follows that the strategy $s_B$ chooses to go to $\rho[i]$ after looping $i-1$ times on $q$. Thus, since $\rho[i] \in \msf{Succ}(\rho[i])$, we have $q^{i-1} \cdot (\rho[i])^\omega \in \msf{Out}^Q(q,s_B)$ and $q^{i-1} \cdot (\rho[i])^\omega \models \lF \phi$. Hence, since $q \not \models \phi$, it follows that $\rho[i] \models \phi$. That is, the property $\mathcal{P}(\psi)$ holds.
			\item Assume that $\msf{O} = \fanBr{B} \lG$ for some coalition $B \subseteq \Ag$. Then, $\psi \implies \phi$, thus $q \models \phi$ and we can the deduce $\mathcal{P}(\psi)$ from $\mathcal{P}(\phi)$. 
		\end{itemize}
		We deduce that $\mathcal{P}(\phi)$ holds for all $\msf{Op}_{A,A'}(\{\lF,\lG\})$-formulas $\phi$. Therefore, for all $\msf{Op}_{A,A'}(\{\lF,\lG\})$-formulas $\phi$ such that $q \models \phi$, we have that there is a state $q' \in \msf{Succ}(q) \setminus \{q\}$ and some $\ATL$-formula $\phi'$ such that $q' \models \phi'$ and $\phi \in \msf{Op}_{A,A'}(\{\lF,\lG\})^* \cdot \fanBr{B} \lF \cdot \phi'$ such that $A \cap B \neq \emptyset$. Since all the winning paths from $q$ are $(A,A',n)$-alternating, it follows that all the winning paths from $q'$ are $(A',A,n)$-alternating. Thus, by our induction hypothesis, we have that $\phi'$ is an $(A',A,n)$-alternating formula, and therefore $\phi$ is an $(A,A',n+1)$-alternating formula. Hence, Item a) also holds at index $n+1$. 

		Consider now Item b). Assume that $|A| = |A'| = 1$ and that there is a safe winning path from $q$ that is $(A,A',n+1)$-alternating. We let $\rho \in \msf{WinPath}_T(q)$. Note that $\rho[1] = q$ with $\msf{AgSt}(q) \in A$. Consider any $\ATL$-formula $\phi$ that is $(A',A,n)$-alternating. Let $\psi := \fanBr{B} \lF \cdot \phi$ for some coalition $B$ such that $A \cap B \neq \emptyset$. Since $|A| = 1$, this implies $\msf{AgSt}(q) \in B$. Let us show that $q \models \psi$. The coalition of agents $B$ has a strategy $s_B$ to ensure that, for all $\rho' \in \msf{Out}^Q(q,s_B)$, we have $\rho'[2] = \rho[2]$. Therefore, since the winning path $\rho[2:]$ is safe and is $(A',A,n)$-alternating, it follows by our induction hypothesis that $\rho[2] \models \phi$. Hence, $q = \rho[1] \models \psi$. Furthermore, for all $2 \leq i \leq n+2$, there is a safe winning path from $\rho[i]$ that is $(A,A',n+2-i)$-alternating or $(A',A,n+2-i)$-alternating. Hence, since $\psi$ is both $(A,A',n+2-i)$-alternating and $(A',A,n+2-i)$-alternating, by our induction hypothesis, we have $\rho[i] \models \psi$. Furthermore, since the winning path $\rho \in \msf{WinPath}_T(q)$ is safe, it follows that the set of states reachable from $q$ in $T$ is equal to $\{\rho[i] \mid 1 \leq i \leq n+2\}$. Hence, from all states $q'$ reachable from $q$, we have $q' \models \psi$. We can then deduce that, for all $\msf{O} \in \msf{Op}_{A,A'}(\{\lF,\lG\})^*$, we have $q \models \msf{O} \cdot \psi$. Therefore, Item b) also holds at index $n+1$. The lemma follows.
	\end{proof}
	
	We conclude this section with a definition of alternating turn-based structures that we will use in the both proof of $\msf{NP}$-hardness for $\ATL$ learning. We give the formal definition below, it is illustrated on Figures~\ref{fig:two_pl_alt_2} and~\ref{fig:two_pl_alt_1}.
	\begin{definition}
		\label{def:turn_based_q_alt}
		Let $i \neq j \in \{1,2,3\}$. We let: 
		\begin{itemize}
			\item 
			For all $l \in \N_1$, $Q^{l: \; i,j} := \{ q_{h}^{i,j} \mid 1 \leq h \leq l \} \cup \{q^{\msf{win}}\}$;
			\item For all $h \in \N_1$, $\msf{AgSt}(q^{i,j}_{2h-1}) := j$ and $\msf{AgSt}(q^{i,j}_{2h}) := i$;
			\item For all $h \in \N_1$, we have:
			\begin{align*}
				\msf{Succ}(q_{h}^{i,j}) := \begin{cases}
					\{ q^{i,j}_{h},q^{i,j}_{h-1} \} & \text{ if }h > 1 \\
					\{ q^{i,j}_{1},q^{\msf{win}} \} & \text{ if }h = 1 \\
				\end{cases}
			\end{align*}
		\end{itemize}
		Then, for all $l \in \N_1$, we define the turn-based structure $T^{l: \; i,j} = \langle Q^{l: \; i,j},I_{i,j,l},2,\{p\},\pi,\msf{AgSt},\msf{Succ} \rangle$ where $I_{i,j,l} := \{ q_{l}^{i,j} \}$.
	\end{definition} 
	
	\begin{figure}
		\centering
		\includegraphics[scale=0.8]{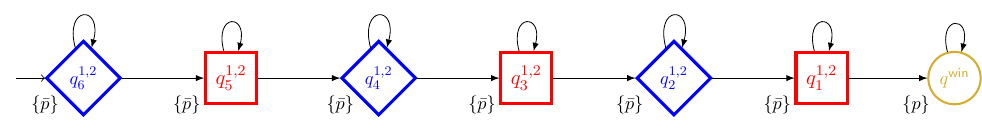}
		\caption{The turn based game structure $T^{6: \; 1,2}$.}
		\label{fig:two_pl_alt_2}
	\end{figure}
	
	\begin{figure}
		\centering
		\includegraphics[scale=0.8]{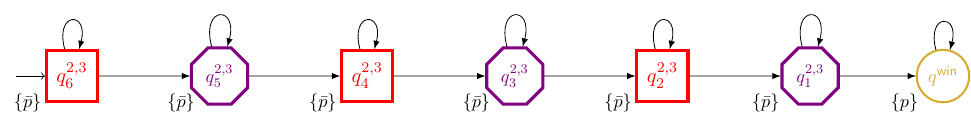}
		\caption{The turn based game structure $T^{6: \; 2,3}$.}
		\label{fig:two_pl_alt_1}
	\end{figure}
	
	\subsubsection{$\ATL$ learning with two agents and operators $\lF$ and $\lG$}

	In this section, we focus on the case of $\ATL$ learning with two agents (Agent 1 and Agent 2). The goal of this subsection is to show the theorem below. 
	
	\begin{theorem}
		\label{thm:atl_2_NP_complete}
		Consider a set $\Ut \subseteq \Op{Un}{}$ of unary temporal operators such that $\{\lF,\lG\} \subseteq \Ut \subseteq \{\lF,\lG,\neg\}$. Then, for all sets $\Bl \subseteq \Op{Bin}{lg}$ and $n \in \N$, the decision problem $\ATL^2_\msf{Learn}(\Ut,\emptyset,\Bl,n)$ is $\msf{NP}$-complete.
	\end{theorem}

	\paragraph{Overview of the reduction.}
	As for the $\CTL$ reduction, we follow the steps described in Section~\ref{subsubsec:abstract_recipe}, with Step~\ref{stepa} already taken care of in Section~\ref{subsubsec:handling_binary_operators}. Thus, we focus on $\ATL^2$-formulas using only unary operators (and a single proposition). First, we define turn-based structures ensuring that: the proposition used is $p$, and the operators $\fanBr{\{1,2\}} \lF,\fanBr{\emptyset} \lG,\fanBr{\{2\}} \lG$ are not used. Since all the structures that we use are self-looping, Lemma~\ref{lem:useless_op} gives that the operators $\fanBr{\{1,2\}} \lG,\fanBr{\emptyset} \lF$ are useless. Hence, we can focus on formulas using only the operators $\fanBr{\{1\}} \lF,\fanBr{\{2\}} \lF,\fanBr{\{1\}} \lG$ and the proposition $p$, which are $\msf{Op}_{\{1\},\{2\}}(\{\lF,\lG\})$-formulas. 
	
	Fix an instance $(l,C,k)$ of the hitting set problem. We consider the bound $B := 3l+1-k$. Our idea is to focus on $(\{1\},\{2\},2l)$-alternating formulas. To do so, we consider $T^{2l: \; 1,2}$ as positive structure and use Lemma~\ref{lem:ATL_formula_necessary_sufficient} (Item a). Note that these $(\{1\},\{2\},2l)$-alternating formulas feature at least $l$ occurrences of the operator $\fanBr{\{1\}} \lF$ and $l$ occurrences of the operator $\fanBr{\{2\}} \lF$. With the proposition occurring in the formula of size at most $B$, there remains $l-k$ operators to use. In fact, we define a negative structure $T_{\msf{no} \; \fanBr{\{1\}} \lG\geq k+1}$ (see Figure~\ref{fig:example_turn_based_more_than_k}) that is accepted by any $\msf{Op}_{\{1\},\{2\}}(\{\lF,\lG\})$-formula featuring at least $k+1$ sequences $\fanBr{\{1\}} \lF \cdot \fanBr{\{2\}} \lF$ where the operators $\fanBr{\{1\}} \lF$ and $\fanBr{\{2\}} \lF$ are not separated by an operator $\fanBr{\{1\}} \lG$. That way, we ensure that the remaining $l-k$ operators are  $\fanBr{\{1\}} \lG$ operators separating $\fanBr{\{1\}} \lF$ and  $\fanBr{\{2\}} \lF$. Hence, 
	the formulas $\phi^{\ATL(2)}(l,H)$ that we consider are $(\{1\},\{2\},2l)$-alternating, features exactly $l-k$ operators $\fanBr{\{1\}} \lG$. The exact positions of these operators are given by the subset $H \subseteq [1,\ldots,l]$. Note that, for this reduction, we use the fact that if there is a hitting set of size at most $k$, then there is one of size exactly $k$.
	
	Then, there remains to define, given a subset $C \subseteq [1,\ldots,l]$, a positive turn-based structure $T_{l,C,2}$ such that $\phi^{\ATL(2)}(l,H)$ accepts $T_{l,C,2}$ if and only if $H \cap C \neq \emptyset$. The structure $T_{l,C,2}$ (see Figure~\ref{fig:example_turn_based_from_C}) is analogous to the structure $T^{2l: \; 1,2}$ except that the final state reached is $q^{\msf{lose}}$ instead of $q^{\msf{win}}$. However, the Agent-1 states corresponding to the indices $i \in C$ may not only continue towards the state $q^{\msf{lose}}$, but also branch to a Agent-2 testing state that can branch to the losing state $q^{\msf{lose}}$, or that can branch to the actual structure $T^{2(i-1): \; 1,2}$. That way, these testing states are rejected by all $\msf{Op}_{\{1\},\{2\}}(\{\lF,\lG\})$-formulas starting with the operator $\fanBr{\{1\}} \lG$. Overall, we do obtain the desired equivalence.
	
	\textbf{Formal definitions and proofs.} For readability, we will use the notations below.	
	\begin{notation}
		For the coalition of agent $A = \{ i \}$ with $i \in \{ 1,2\}$, the operators $\fanBr{A} \lF$ and $\fanBr{A} \lG$ will be denoted $i \lF$ and $i \lG$ respectively.
	\end{notation}

	Now, the first step that we take is to define two simple turn-based game structures that will restrict the set of operators that we need to consider. 
	\begin{definition}
		\label{def:simple_games_F_G}		
		We define the trivial game structure $T_p := \langle \{ q^{\msf{win}} \},\{ q^{\msf{win}} \},2,\{p\},\pi,\msf{Ag},\msf{Succ} \rangle$.
		
		
		We also define the two-state turn-based game $T_{\msf{no} \; \emptyset \lG,2 \lG} := \langle \{ q,q^{\msf{lose}} \},\{ q \},2,\{p,\bar{p}\},\pi,\msf{Ag},\msf{Succ} \rangle$ where $\msf{AgSt}(q) := 1$, $\msf{AgSt}(q^{\msf{lose}}) := 1$,  $\msf{Succ}(q) := \{ q,q^{\msf{lose}} \}$, $\pi(q) := \{ p \}$.
	\end{definition}

	\begin{figure}
		\centering
		\includegraphics[scale=1.2]{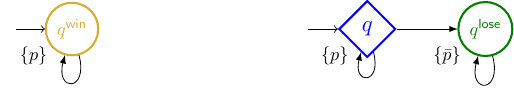}
		\caption{The game $T_p$ on the left and the game $T_{\msf{no} \; \emptyset \lG,2 \lG}$ on the right.}
		\label{fig:simpleGames}
	\end{figure}
	These two games are depicted in Figure~\ref{fig:simpleGames}, and they satisfy the property below.
	\begin{lemma}
		\label{lem:ATL_no_use_bad_lG}
		Consider any $\ATL$-formula $\phi \in \ATL^2(\prop_0,\{ \lF,\lG \},\emptyset,\emptyset,0)$. It accepts both $T_p$ and $T_{\msf{no} \; \emptyset \lG,2 \lG}$ if and only if $\prop(\phi) = \{p\}$ and it does not use the operators $\emptyset \lG$ or $2 \lG$.
	\end{lemma}
	\begin{proof}
		If the proposition $\prop(\phi) = \{\bar{p}\}$, then $\phi \not\models T_p$. Hence, $\prop(\phi) = \{\bar{p}\}$ (note that $|\prop(\phi)| = 1$ because $\phi$ only uses unary operators).
		
		Let us now show by induction the following property: an $\ATL$-formula $\phi \in \ATL^2(\{p\},\{ \lF,\lG \},\emptyset,\emptyset,0)$ accepts $T_{\msf{no} \; \emptyset \lG,2 \lG}$ if and only if it does not use the operators $\emptyset \lG$ or $2 \lG$. This holds straightforwardly for $\phi = p$. Assume now that it holds for some formula $\phi \in \ATL^2(\{p\},\{ \lF,\lG \},\emptyset,\emptyset,0)$ accepts $T_{\msf{no} \; \emptyset \lG,2 \lG}$. Let $\phi' = \fanBr{A} \textbf{H} \phi$ with $A \subseteq \{ 1,2 \}$ and $\textbf{H} \in \{ \lF,\lG \}$. 
		\begin{itemize}
			\item Assume that either $\emptyset \lG$ or $2 \lG$ occurs in $\phi'$. If $\emptyset \lG$ or $2 \lG$ occurs in $\phi$, then by our induction hypothesis, we have $q \not\models \phi$. Since we also have $q^{\msf{lose}} \not\models \phi$ (since $\prop(\phi) = \{p\}$) it follows that $q \not \models \phi'$. Otherwise, we have $\phi' = \fanBr{A} \lG \phi$ with $1 \notin A$. Since $\msf{AgSt}(q) = 1$, for all strategies $s$ for the coalition $A$, we have $q \cdot (q^{\msf{lose}})^\omega \in \msf{Out}^Q(q,s)$ with $q^{\msf{lose}} \not\models \phi$. Therefore, $q \not\models \phi'$.
			\item Assume that neither $\emptyset \lG$ nor $2 \lG$ occurs in $\phi'$. By our induction hypothesis, we have $q \models \phi$. Hence, if $\textbf{H} = \lF$, we have $q \models \phi'$. Otherwise, $1 \in A$. Since the Agent-1 strategy $s_1$ that always loops on $q$ is such that $\msf{Out}^Q(q,s) = \{ q^\omega \}$, 
			it follows that $q \models \phi'$.
		\end{itemize}
		Thus, the property also holds for $\phi'$. In fact it holds for all formulas in $\ATL^2(\{p\},\{ \lF,\lG \},\emptyset,\emptyset,0)$. The lemma follows.
	\end{proof}

	Furthermore, we have seen in Lemma~\ref{lem:useless_op} that the operators $\emptyset \lF$ and $1,2 \lG$ are useless on self-looping structures. Thus, we restrict ourselves to promising formulas, i.e. formulas that only use the operators $1 \lF,2\lF$ and $1\lG$. Note that we have not handled the operator $1,2 \lF$\footnote{That is, we have not shown that we can restrict ourselves to formulas that do not use this operator.} yet. It will be done on the fly later.  
	\begin{definition}
		An $\ATL^2(\prop_0,\{ \lF,\lG \},\emptyset,\emptyset,0)$ formula is \emph{promising} if $\prop(\phi) = \{p\}$ and it only uses the operators $1 \lF,2\lF$, and $1\lG$. 
	\end{definition}	

	Let us now consider the exact kinds of formulas that we consider, along with the structures that we define to encode the hitting set problem.
	
	\paragraph{Bounding the size of hitting sets}
	As mentioned at the beginning of this subsection, the way we encode a hitting set is by considering when the operator $1\lG$ is used between the operators $1\lF$ and $2\lF$. More precisely, we define the notion of concise formulas, a special kind of alternating formulas
	.
	\begin{definition}
		Let $l \in \N_1$. For all $H \subseteq [1,\ldots,l]$, we denote by $\phi^{\ATL(2)}(l,H)$ the promising $\ATL$-formula defined by:
		\begin{equation*}
			\phi^{\ATL(2)}(l,H) := 1 \lF \; \msf{Qt}_l \; 2 \lF \; 1 \lF \msf{Qt}_{l-1} \; 2 \lF \ldots \; 1 \lF \msf{Qt}_1 \; 2 \lF p
		\end{equation*}
		%
		where, for all $i \in [1,\ldots,l]$, we have $\msf{Qt}_k \in \{ \epsilon,1 \lG \}$ and $\msf{Qt}_i = \epsilon$ if and only if $i \in H$. 
		
		For all $0 \leq k \leq l$, we say that a promising $\ATL$-formula $\phi$ is $(l,k)$-\emph{concise} if it is equal to $\phi^{\ATL(2)}(l,H)$ for some $H \subseteq [1,\ldots,l]$ with $|H| = k$. (In which case, $\Size{\phi} = 3l+1-k$.)
	\end{definition}
	
	Let us now define the turn-based structure that will force the use of a minimal number of $1 \lG$ operators, an example of which is depicted in Figure~\ref{fig:example_turn_based_more_than_k}. This corresponds to the fact that, in the reduction, hitting sets have a bounded size
	. 
	\begin{definition}
		Let $i \in \{ 1,2 \}$. For all $k \in \N_1$, we let $T_{\msf{no} \; 1 \lG\geq k} := \langle Q^{1\lG}_k,I_k,2,\{p\},\pi,\msf{Ag},\msf{Succ} \rangle$ where:
		\begin{itemize}
			\item $Q^{1\lG}_k := \{ q_{h}^{1\lG} \mid 1 \leq h \leq 2k \} \cup \{q^{\msf{win}},q^{\msf{lose}} \}$;
			\item $I_k := \{ q_{2k}^{1\lG} \}$;
			\item For all $1 \leq h \leq k$, $\msf{AgSt}(q^{1\lG}_{2h-1}) := 2$ and $\msf{AgSt}(q^{1\lG}_{2h}) := 1$;
			\item For all $1 \leq h \leq k$, we have:
			\begin{equation*}
				\msf{Succ}(q_{2h}^{1\lG}) :=
					\{ q^{1\lG}_{2h},q^{1\lG}_{2h-1} \}
			\end{equation*}
			\begin{align*}
				\msf{Succ}(q_{2h-1}^{1\lG}) := \begin{cases}
					\{ q_{2h-1}^{1\lG},q_{2h-2}^{1\lG},q^{\msf{lose}} \} & \text{ if }h > 1 \\
					\{ q_{2h-1}^{1\lG},q^{\msf{win}},q^{\msf{lose}} \} & \text{ if }h = 1 \\
				\end{cases}
			\end{align*}
			\item For all $q \in Q^{1\lG}_k \setminus \{q^{\msf{win}}\}$, $\pi(q) := \{\bar{p}\}$.
		\end{itemize}
	\end{definition}
	
	\begin{figure}
		\hspace*{-0.75cm}
		\centering
		\includegraphics[scale=0.8]{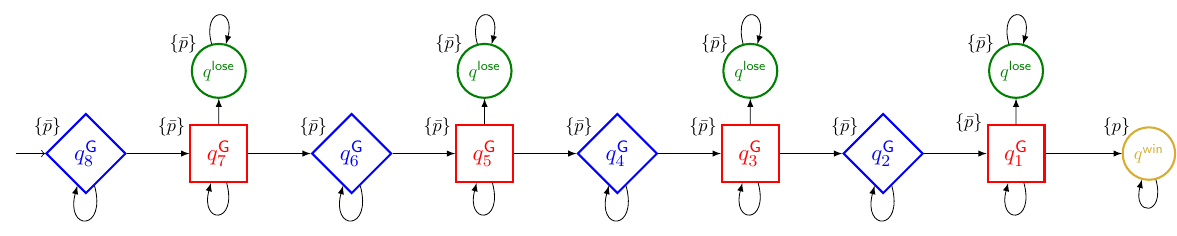}
		\caption{The turn-based game structure $T_{\msf{no} \; 1\lG \geq 4}$.}
		\label{fig:example_turn_based_more_than_k}
	\end{figure}

	Let us consider the lemma linking concise formulas and the turn-based structures $T_{\msf{no}\; 1\lG \geq k}$.
	\begin{lemma}
		\label{lem:ATL_reject_G_concise}
		Let $l \in \N_1$. Consider a promising $\ATL$-formula $\phi$ that is $(\{1\},\{2\},2l)$-alternating. Let $k \leq l$. If $\Size{\phi} \leq 3l+1-k$, we have the following equivalence:
		\begin{center}
			$q^{1\lG}_{2(k+1)} \not \models \phi$ if and only if $\phi$ is $(l,k)$-concise
		\end{center}
	\end{lemma}
	\begin{proof}
		We prove by induction on $n \in \N_1$ the property $\mathcal{P}(n)$ stating the lemma for pairs $(l,k)$ such that $l \in \N_1$, $0 \leq k \leq l$, and $l+k \leq n$. We start with the case $n = 0$. The only possible pair is $l = 1$, and $k = 0$. Consider a promising $\ATL$-formula $\phi$ that is $(\{1\},\{2\},2)$-alternating and such that $\Size{\phi} \leq 4$. We have $\phi = \msf{Qt}_1 \; 1\lF \; \msf{Qt}_2 \; 2\lF \; \msf{Qt}_3 \; p$, with $\msf{Qt}_1,\msf{Qt}_2,\msf{Qt}_3 \in \{\epsilon,1\lF,1\lG,2\lF\}$ and at most one of them not equal to $\epsilon$. One can then check that $q_2^{\lG} \not \models\phi$ if and only $\msf{Qt}_2 = 1\lG$, and thus $\mathcal{P}(1)$ follows.

		Consider now some $n \in \N_1$ and assume that $\mathcal{P}(n)$ holds. Let us show that $\mathcal{P}(n+1)$ holds. Consider any pair $(l,k)$ such that $l \in \N_1$, $0 \leq k \leq l$, and $l+k = n+1$. Consider a promising $\ATL$-formula $\phi$ that is $(\{1\},\{2\},2l)$-alternating and such that $\Size{\phi} \leq 3l+1-k$. We let:
		\begin{equation*}
			\phi = \msf{Qt}_1 \; 1\lF \; \msf{Qt}_2 \; 2\lF \; \phi'
		\end{equation*}
		for some $\msf{Qt}_1 \in (1\lG,2\lF)^*$, $\msf{Qt}_2 \in (1\lF,1\lG)^*$, and sub-formula $\phi' \in \msf{SubF}(\phi)$. By definition, $\phi'$ is $(\{1\},\{2\},2(l-1))$-alternating. First of all, note that, since $\msf{AgSt}(q_{2(k+2)}^{\lG}) = 1$, letting $\psi := 1\lF \; \msf{Qt}_2 \; 2\lF \; \phi'$, we have that $q_{2(k+1)}^{\lG} \models \phi$ if and only if $q_{2(k+1)}^{\lG} \models \psi$. There are two cases. 
		\begin{itemize}
			\item Assume that $\msf{Qt}_2 \in (1\lF)^*$. If $q_{2(k+1)}^{\lG} \not\models \psi$, then we have $k \geq 1$ (since $q^{\msf{win}} \models \phi'$) and it must be that $q_{2k}^{\lG} \not\models \phi'$. Furthermore, $\phi'$ is $(\{1\},\{2\},2(l-1))$-alternating, and $\Size{\phi'} \leq \Size{\phi} - 2 \leq 3l+1-k = 3(l-1)+1-(k-1)$. Therefore, by $\mathcal{P}(n)$ (applied to the pair $(l-1,k-1)$), we have that $\phi'$ is $(l-1,k-1)$-concise, and of size $\Size{\phi'} = 3(l-1)+1-(k-1)$. Since  $\Size{\phi} \leq 3l+1-k$, it follows that $\msf{Qt}_1 = \msf{Qt}_2 = \epsilon$, and $\phi = 1\lF 2\lF \phi'$ is $(l,k)$-concise. 
			
			On the other hand, if $\psi$ is $(l,k)$-concise, since $\msf{Qt}_2 \neq 1\lG$, it follows that $k \geq 1$, $\phi = 1\lF 2\lF \phi'$ with $\phi'$ a formula that is $(l-1,k-1)$-concise. Hence, by $\mathcal{P}(n)$ (applied to the pair $(l-1,k-1)$), we have $q_{2k}^{\lG} \not\models \phi'$. This implies that all sub-formulas $\phi'' \in \msf{SubF}(\phi')$ are also such that $q_{2k}^{\lG} \not\models \phi''$ (because $\msf{AgSt}(q_{2k}^{\lG}) = 1$, and $\phi' \in (1\lF,1\lG,2\lF)^* \cdot \phi''$). It follows that it is also the case of all the sub-formulas of $\phi$. Therefore, since $\pi(q_{2(k+1)}^{\lG}) = \pi(q_{2k+1}^{\lG}) = \{\bar{p}\}$, and $\prop(\phi) = \{p\}$, we have $q_{2(k+1)}^{\lG} \not\models \phi$.
			\item Assume now $\msf{Qt}_2 \notin (1\lF)^*$. First of all, this implies $k \leq l-1$. Indeed, if $k=l$, we have $\Size{\phi} \leq 2l+1$, and thus, since it is a $(\{1\},\{2\},2l)$-alternating formula, it is equal to $\phi = (1\lF 2\lF)^* p$, and thus $\msf{Qt}_2 = \epsilon$. In fact, we do have $k \leq l-1$. Furthermore, since $\msf{AgSt}(q_{2(k+1)}^{\lG}) = 1$ and $\msf{AgSt}(q_{2k+1}^{\lG}) = 2$ and $q^{\msf{lose}} \not\models \phi'$, we have $q_{2(k+1)}^{\lG} \models \psi$ if and only if $q_{2(k+1)}^{\lG} \not\models \phi'$. By $\mathcal{P}(n)$ (applied to the pair $(l-1,k)$), we have $q_{2(k+1)}^{\lG} \not\models \phi'$ if and only if $\phi'$ is $(l-1,k)$-concise. 
			Furthermore, if $\phi'$ is $(l-1,k)$-concise, we have $\Size{\phi'} = 3(l-1)+1-k$. Since $\Size{\phi} \leq 3l+1-k$ and $\msf{Qt}_2 \notin (1\lF)^*$, it follows that $\phi = 1\lF 1\lG 2\lF \phi'$. In fact, $\phi'$ is $(l-1,k)$-concise if and only if $\phi$ is $(l,k)$-concise. Overall, we obtain that $q_{2(k+1)}^{\lG} \not\models \phi$ if and only if $\phi$ is $(l,k)$-concise.  
		\end{itemize}
		Hence, we have established $\mathcal{P}(n+1)$. In fact, $\mathcal{P}(n)$ holds for all $n \in \N$, and the lemma follows.
		
		%
	\end{proof}

	\paragraph{Hitting sets should intersect the sets $C_i$}
	Let us now define the positive turn-based structures that encode the fact that the positions where there is a lack of operators $1\lG$ should match the subsets of integers $C_i \subseteq [1,\ldots,l]$ from the hitting set problem. We give the formal definition below, it is illustrated in Figure~\ref{fig:example_turn_based_from_C}. 
	\begin{definition}
		\label{def:turn_based_game_atl_2_intersect}
		Let $l \in \N_1$ and $C \subseteq [1,\ldots,l]$. We let $T_{l,C,2} := \langle Q^{l,C},I_{l},2,\{p\},\pi,\msf{Ag},\msf{Succ} \rangle$ where (recall that the states $q^{1,2}_{h}$ come from the turn-based structure $T_{2l: \; 1,2}$ from Definition~\ref{def:turn_based_q_alt}):
		\begin{itemize}
			\item $Q^{l,C,2} := \{ q_{i} \mid 1 \leq i \leq 2l \} \cup \{ q^{1,2}_{h} \mid 1 \leq h \leq 2l \} \cup \{ q^\msf{Test}_{2i-1} \mid i \in C \} \cup \{q^{\msf{lose}},q^{\msf{win}}\}$;
			\item $I_l := \{ q_{2l} \}$;
			\item For all $1 \leq i \leq l$, $\msf{AgSt}(q_{2i}) := 1$ and $\msf{AgSt}(q_{2i-1}) := 2$. 
			For all $i \in C$, we have $\msf{AgSt}(q^{\msf{Test}}_{2i-1}) := 2$
			\item For all $1 \leq i \leq l$, we have:
			\begin{equation*}
				\msf{Succ}(q_{2i}) :=
				\begin{cases}
					\{ q_{2i},q_{2i-1} \} & \text{ if }i \notin C \\
					\{ q_{2i},q_{2i-1},q^{\msf{Test}}_{2i-1} \} & \text{ if }i \in C \\
				\end{cases}
			\end{equation*}
			and
			\begin{equation*}
				\msf{Succ}(q_{2i-1}) :=
				\begin{cases}
					\{ q_{2i-1},q_{2(i-1)} \} & \text{ if }i > 1 \\
					\{ q_{2i-1},q^{\msf{lose}} \} & \text{ if }i = 1 \\
				\end{cases}
			\end{equation*}
			and, for all $i \in C$:
			\begin{align*}
				\msf{Succ}(q^{\msf{Test}}_{2i-1}) :=
				\begin{cases}
					\{ q^{\msf{Test}}_{2i-1},q^{\msf{lose}},q_{2(i-1)}^{1,2} \} & \text{ if }i > 1 \\
					\{ q^{\msf{Test}}_{2i-1},q^{\msf{lose}},q^{\msf{win}} \} & \text{ if }i = 1 \\
				\end{cases}
			\end{align*}
			\item For all $q \in Q^{l,C,2} \setminus \{q^{\msf{win}}\}$, $\pi(q) := \{\bar{p}\}$.
		\end{itemize}
	\end{definition}

	\begin{figure}
		\centering
		\includegraphics[scale=0.8]{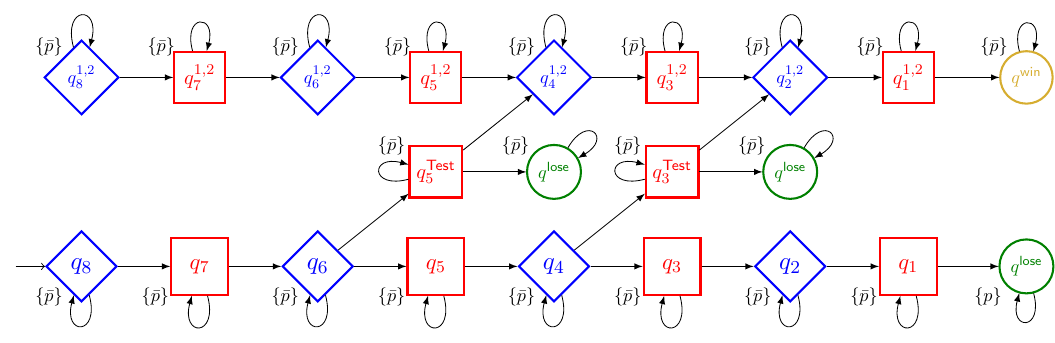}
		\caption{The turn-based game structure $T_{4,\{2,3\},2}$.}
		\label{fig:example_turn_based_from_C}
	\end{figure}

	The above definition satisfies the lemma below.
	\begin{lemma}
		\label{lem:ATL_2_intersect}
		Consider any $l \in \N_1$ and $C,H \subseteq [1,\ldots,l]$. We have:
		\begin{equation*}
			T_{l,C,2} \models \phi^{\ATL(2)}(l,H) \text{ if and only if }H \cap C \neq \emptyset
		\end{equation*}
	\end{lemma}
	\begin{proof}
		For all $1 \leq i \leq l$, we let $H_i := H \cap [1,\ldots,l]$. 
		Then, in the turn-based structure $T_{l,C,2}$, we prove by induction on $1 \leq i \leq l$ the property $\mathcal{P}(i)$: $q_{2i} \models \phi^{\ATL(2)}(i,H_i)$ if and only if $H_i \cap C \neq \emptyset$. We start with the case $i = 1$. There are two cases.
		\begin{itemize}
			\item Assume that $H_1 \cap C = \{ 1 \}$. Then, we have $\phi^{\ATL(2)}(1,H_1) = 1\lF 2\lF p$. Furthermore, $q^{\msf{Test}}_{1} \in Q^{l,C,2}$ with $q^{\msf{Test}}_{1} \models 2\lF p$ since $q^{\msf{win}} \in \msf{Succ}(q^{\msf{Test}}_{1})$. Therefore, $q_2 \models \phi^{\ATL(2)}(1,H_1)$ since $q^{\msf{Test}}_{1} \in \msf{Succ}(q_2)$.
			\item Assume now that $H_1 \cap C = \emptyset$. If $1 \in C$, we have $1 \notin H$, thus $\phi^{\ATL(2)}(1,H_1) = 1\lF 1\lG 2\lF p$. Since $q^{\msf{lose}} \in \msf{Succ}(q^{\msf{Test}}_{1})$, we have $q^{\msf{Test}}_{1} \not \models 1\lG 2\lF p$ and $q_{1} \not \models 1\lG 2\lF p$. Hence, $q_2 \not \models \phi^{\ATL(2)}(1,H_1)$. On the other hand, if $1 \notin C$, there is no winning path from $q_2$, thus $q_2 \not \models \phi^{\ATL(2)}(1,H_1)$ by Lemma~\ref{lem:ATL_formula_necessary_sufficient} (Item a). 
		\end{itemize}
		Hence, the property $\mathcal{P}(1)$ holds. Assume now that $\mathcal{P}(i)$ holds for some $1 \leq i \leq l-1$. 
	%
		We have:
		\begin{equation*}
			\phi^{\ATL(2)}(i+1,H_{i+1}) = 1\lF \; \msf{Qt}_{i+1} \; 2\lF \phi^{\ATL(2)}(i,H_{i})
		\end{equation*}
		with $\msf{Qt}_{i+1} = 1 \lG$ if $i+1 \notin H$ and $\msf{Qt}_{i+1} = \epsilon$ otherwise. As above, there are two cases.
		\begin{itemize}
			\item Assume that $i+1 \in H \cap C$. The only winning path from $q_{2i}^{1,2}$ is $(\{1\},\{2\},2i)$-alternating and safe. Therefore, by Lemma~\ref{lem:ATL_formula_necessary_sufficient} (Item b), we have $q_{2i}^{1,2} \models \phi^{\ATL(2)}(i,H_{i})$
			. Hence, $q^{\msf{Test}}_{2i+1} \in Q^{l,C}$ with $q^{\msf{Test}}_{2i+1} \models 2\lF \phi^{\ATL(2)}(i,H_{i})$. Therefore, $q_{2(i+1)} \models \phi^{\ATL(2)}(i+1,H_{(i+1)})$
			.
			\item Assume now that $i+1 \notin H \cap C$. Let us show that $q_{2(i+1)} \models \phi^{\ATL(2)}(i+1,H_{i+1})$ if and only if $q_{2i} \models \phi^{\ATL(2)}(i,H_{i})$. First, if $q_{2i} \models \phi^{\ATL(2)}(i,H_{i})$, then $q_{2i+1} \models 2\lF \phi^{\ATL(2)}(i,H_{i})$ and $q_{2i+1} \models 1\lG 2\lF \phi^{\ATL(2)}(i,H_{i})$. Thus, we have $q_{2(i+1)} \models \phi^{\ATL(2)}(i+1,H_{i+1})$. Assume now that $q_{2(i+1)} \models \phi^{\ATL(2)}(i+1,H_{i+1})$. Note that, the winning paths from $q_{2(i+1)}$ are all $(\{1\},\{2\},2(i+1))$-alternating, therefore, by Lemma~\ref{lem:ATL_formula_necessary_sufficient} (Item a), no strict sub-formula of $\phi^{\ATL(2)}(i+1,H_{i+1})$ accept the state $q_{2(i+1)}$. Similarly, the winning paths from $q_{2i+1}$ are all $(\{2\},\{1\},2i+1)$-alternating, therefore, by Lemma~\ref{lem:ATL_formula_necessary_sufficient} (Item a), no strict sub-formula of $2\lF\phi^{\ATL(2)}(i,H_{i})$ accept the state $q_{2i+1}$.
			Then, there are two cases.
			\begin{itemize}
				\item If $i+1 \in C$, we have $i+1 \notin H$, thus $\phi^{\ATL(2)}(i+1,H_{i+1}) = 1\lF 1\lG 2\lF \phi^{\ATL(2)}(i,H_{i})$. Since $q^{\msf{lose}} \in \msf{Succ}(q^{\msf{Test}}_{2i+1})$, we have $\msf{Succ}(q^{\msf{Test}}_{2i+1}) \not \models 1\lG 2\lF \phi^{\ATL(2)}(i,H_{i})$. Hence, with what we have argued above, we have $q_{2i} \models \phi^{\ATL(2)}(i,H_{i})$. 
				\item If $i+1 \notin C$, then $\msf{Succ}(q_{2(i+1)}) = \{ q_{2(i+1)},q_{2i+1} \}$, hence as for the previous item, we have that $q_{2i} \models \phi^{\ATL(2)}(i,H_{i})$. 
			\end{itemize}
			We have established that $q_{2(i+1)} \models \phi^{\ATL(2)}(i+1,H_{i+1})$ if and only if $q_{2i} \models \phi^{\ATL(2)}(i,H_{i})$. Furthermore, by $\mathcal{P}(i)$, we have $q_{2i} \models \phi^{\ATL(2)}(i,H_{i})$ if and only if $H_i \cap C \neq \emptyset$. Since $i+1 \notin H \cap C$, it follows that $H_{i+1} \cap C = H_i \cap C$. Hence, we do obtain that $q_{2(i+1)} \models \phi^{\ATL(2)}(i+1,H_{i+1})$ if and only if $H_{i+1} \cap C \neq \emptyset$.
		\end{itemize}
		Hence, $\mathcal{P}(i+1)$ holds. In fact, $\mathcal{P}(i)$ holds for all $1 \leq i \leq l$. The lemma follows.
	\end{proof}

	\paragraph{Definition of the reduction}
	We can finally define the reduction that we consider. 
	\begin{definition}
		\label{def:reduction_ATL_2}
		Consider an instance $(l,C,k)$ of the hitting set problem $\msf{Hit}$. We define:
		\begin{itemize}
			\item $\mathcal{P} := \{ T_p,T_{\msf{no} \; \emptyset \lG,2\lG} \} \cup \{ T_{2l: 1,2} \} \cup \{ T_{(l,C_i,2)} \mid 1 \leq i \leq n \}$;
			\item $\mathcal{N} := \{ T_{\msf{no} \; 1\lG \geq k+1} \} \cup \{ T_{2l+1: 2,1} \}$.
		\end{itemize}
		
		Then, we define the input $\msf{In}^{\ATL(2),\lF,\lG}_{(l,C,k)} := (\prop_0,\mathcal{P},\mathcal{N},3l+1-k)$
		.
	\end{definition}
	
	This definition satisfies the lemma below.
	\begin{lemma}
		\label{lem:reduc_ATL_2_equiv}
		Let $\Ut \subseteq \Op{Un}{}$ be a set of unary temporal operators such that $\{\lF,\lG\} \subseteq \Ut \subseteq \{\lF,\lG,\neg\}$. 
		Consider an input $(l,C,k)$ of the hitting set problem $\msf{Hit}$ and the corresponding input $\msf{In}^{\ATL(2)}_{(l,C,k)}$. Then, $(l,C,k)$ is a positive instance of the decision problem $\msf{Hit}$ if and only if $\msf{In}^{\ATL(2)}_{(l,C,k)}$ is a positive instance of the decision problem $\ATL^2_\msf{Learn}(\Ut \setminus \{\neg\},\emptyset,\emptyset,0)$ if and only if $\msf{In}^{\ATL(2)}_{(l,C,k)}$ is a positive instance of the decision problem $\ATL^2_\msf{Learn}(\Ut,\emptyset,\emptyset,0)$.
	\end{lemma}
	\begin{proof}
		Assume that $(l,C,k)$ is a positive instance of the decision problem $\msf{Hit}$. Consider a hitting set $H \subseteq [1,\ldots,l]$ of size at most $k$. We let $H \subseteq H' \subseteq [1,\ldots,l]$ be another hitting set of size exactly $k$. We let $\phi := \phi^{\ATL(2)}(l,H')$. We have $\phi \in \ATL^2(\prop_0,\Ut \setminus \{\neg\},\emptyset,\emptyset,0)$ and $\Size{\phi} = 3l+1-k$. Furthermore, since $\prop(\phi) = \{p\}$ and $\phi$ does not use the operators $\emptyset \lG,2\lG$, by Lemma~\ref{lem:ATL_no_use_bad_lG}, it accepts  both structures $T_p,T_{\msf{no} \; \emptyset \lG, 2 \lG} \in \mathcal{P}$. In addition, $\phi$ is $(\{1\},\{2\},2l)$-alternating, hence by Lemma~\ref{lem:ATL_formula_necessary_sufficient} (Item b), it accepts $T_{2l: 1,2} \in \mathcal{P}$ (since there is a safe winning path from $q_{2l}^{1,2}$ --- in $T_{2l: 1,2}$ --- that is $(\{1\},\{2\},2l)$-alternating). On the other hand, all winning paths from $q_{2l+1}^{2,1}$ are $(\{2\},\{1\},2l+1)$-alternating, while $\phi$ is not $(\{2\},\{1\},2l+1)$-alternating, hence by Lemma~\ref{lem:ATL_formula_necessary_sufficient} (Item a), $\phi$ rejects the structure $T_{2l+1: 2,1} \in \mathcal{N}$. In addition, by Lemma~\ref{lem:ATL_reject_G_concise}, since $\phi$ is $(l,k)$-concise, it rejects $T_{\msf{no} \; 1\lG \geq k+1} \in \mathcal{N}$. Finally, consider any $1 \leq i \leq n$. Since we have $C_i \cap H' \neq \emptyset$, it follows, by Lemma~\ref{lem:ATL_2_intersect}, that $\phi$ accepts $T_{(l,C_i,2)} \in \mathcal{P}$. Overall, the $\ATL$-formula $\phi$ accepts $\mathcal{P}$ and rejects $\mathcal{N}$. Hence, $\msf{In}^{\ATL(2)}_{(l,C,k)}$ is a positive instance of the decision problem $\ATL^2_\msf{Learn}(\Ut \setminus \{\neg\},\emptyset,\emptyset,2)$. 
		
		Furthermore, clearly, if $\msf{In}^{\ATL(2)}_{(l,C,k)}$ is a positive instance of the decision problem $\ATL^2_\msf{Learn}(\Ut \setminus \{\neg\},\emptyset,\emptyset,0)$, then it is also a positive instance of $\msf{In}^{\ATL(2)}_{(l,C,k)}$ is a positive instance of the decision problem $\ATL^2_\msf{Learn}(\Ut,\emptyset,\emptyset,0)$.
		
		Assume now that $\msf{In}^{\ATL(2)}_{(l,C,k)}$ is a positive instance of the decision problem $\ATL^2_\msf{Learn}(\Ut,\emptyset,\emptyset,0)$. Consider a formula $\phi \in \ATL^2(\prop_0,\Ut,\emptyset,\emptyset,0)$ with $\Size{\phi} \leq 3l+1-k$ that accepts $\mathcal{P}$ and rejects $\mathcal{N}$. We let $\phi = \msf{Qt} \cdot x$ for some $\msf{Qt} \in (\msf{Op}(2,\Ut))^*$ and $x \in \prop_0$. Let $(\msf{Qt}',x') := \msf{UnNeg}(\msf{Qt},0)$. We let $y \in \prop_0$ be such that $y = x$ if and only if $x' = 0$. Finally, we let $\phi' := \msf{Qt}' \cdot y$. By Lemma~\ref{lem:unnegate_unary}, we have:
		\begin{itemize}
			\item $|\msf{Qt}'| \leq |\msf{Qt}|$, therefore $\Size{\phi'} \leq \Size{\phi} \leq 3l+1-k$;
			\item $\msf{Qt}' \in \msf{Op}(2,\Ut \setminus \{\neg\})$, therefore $\phi' \in \ATL^2(\prop_0,\Ut \setminus \{\neg\},\emptyset,\emptyset,0)$;
			\item Since all the structures in $\mathcal{P}$ and $\mathcal{N}$ are $(0,\emptyset)$-proper structures (and therefore, on those structures, $p$ and $\neg \bar{p}$ are equivalent), $\phi'$ also accepts $\mathcal{P}$ and rejects $\mathcal{N}$. 
		\end{itemize}
		
		In addition, all the turn-based structures in $\mathcal{P}$ and $\mathcal{N}$ are self-looping, hence by Lemma~\ref{lem:useless_op}, we may assume that neither of the operators $\emptyset \lF$ or $1,2 \lG$ occur in $\phi'$. Furthermore, since $\phi'$ accepts the games $T_p \in \mathcal{P}$ and $T_{\msf{no} \; \emptyset \lG,2 \lG} \in \mathcal{P}$, by Lemma~\ref{lem:ATL_no_use_bad_lG}, we have $\prop(\phi') = \{p\}$ and $\phi'$ does not use the operators $\emptyset \lG,2\lG$. Furthermore, since $\phi'$ rejects the game structure $T_{2l+1: 1,2} \in \mathcal{N}$, it does not use the operator $1,2 \lF$ (since all sub-formulas of $\phi'$ accept the state $q^{\msf{win}}$). In fact, the formula $\phi'$ is promising. 
		
		Since $\phi'$ accepts $T_{2l: 1,2} \in \mathcal{P}$, and all the winning paths from the state $q_{2l}^{1,2}$ are $(\{1\},\{2\},2l)$-alternating, the formula $\phi'$ is $(\{1\},\{2\},2l)$-alternating by Lemma~\ref{lem:ATL_formula_necessary_sufficient} (Item a). Furthermore, we have $\Size{\phi'} \leq 3l+1-k$. Hence, since $\phi'$ rejects $T_{\msf{no} \; 1\lG \geq k+1} \in \mathcal{N}$, by Lemma~\ref{lem:ATL_reject_G_concise}, $\phi'$ is necessarily $(l,k)$-concise. We let $H \subseteq [1,\ldots,l]$ be such that $|H| = k$ and $\phi' = \phi^{\msf{ATL}(2)}(l,H)$. Then, consider some $1 \leq i \leq n$. Since $\phi'$ accepts $T_{(l,C_i,2)} \in \mathcal{P}$, it follows, by Lemma~\ref{lem:ATL_2_intersect}, that $H \cap C_i \neq \emptyset$. This holds for all $1 \leq i \leq n$. Therefore, $H$ is a hitting set and $(l,C,k)$ is a positive instance of the hitting set problem $\msf{Hit}$.
	\end{proof}

	Theorem~\ref{thm:atl_2_NP_complete} follows.
	\begin{proof}
		This is direct consequence of Lemma~\ref{lem:reduc_ATL_2_equiv}, of the fact that the instance $\msf{In}^{\ATL^2(\lX)}_{(l,C,k)}$ can be computed in logarithmic space from $(l,C,k)$, and of Theorem~\ref{thm:unary_is_sufficient}.
	\end{proof}

	We believe that the result of Theorem~\ref{thm:atl_2_NP_complete} would also hold with $\Ut = \{\lF,\neg\}$ or $\Ut = \{\lG,\neg\}$, but handling these cases would induce many additional technical difficulties. In particular, Theorem~\ref{thm:unary_is_sufficient} can be not applied as is.

	\subsubsection{$\ATL$ learning with two agents and with one of the two operators $\lF$ and $\lG$}
	Let us now focus to the case of $\ATL$ learning with two agents where only one of the two operators $\lF$ ot $\lG$ is allowed, without any negations. The goal of this subsection is to show the theorem below. 
	\begin{theorem}
		\label{thm:atl_2_F_and_G_P_complete}
		For all $\Bl \subseteq \Op{Bin}{lg}$ and $n \in \N$, both decision problems $\ATL^2_\msf{Learn}(\{\lF\},\emptyset,\Bl,n)$ and $\ATL^2_\msf{Learn}(\{\lG\},\emptyset,\Bl,n)$ are $\msf{P}$-complete.
	\end{theorem}

	The idea is that with only the operator $\lF$ (it is the same with only the operator $\lG$), given a bound $B \in \N_1$ and a set of propositions $\prop$, it is sufficient to consider only polynomially many $\ATL$-formulas. The reason why comes from Lemma~\ref{lem:equiv_atl_dominating_quantifiers} (that we have stated when considering $\CTL$-formulas). Indeed, this lemma states that, for all $\ATL$-formulas $\phi$, for two coalitions of agents $A_1,A_2$, as soon as $A_1 \subseteq A_2$ or $A_2 \subseteq A_1$, then $\fanBr{A_1} \lF \fanBr{A_2} \lF \phi \equiv \fanBr{A} \lF \phi$, with $A \in \{ A_1,A_2 \}$ the largest of the two coalitions. However, since there are only two agents, there are only four coalitions: $\emptyset,\{1\},\{2\},\{1,2\}$. Furthermore, we can make the following observations, for formulas without binary operators:
	\begin{itemize}
		\item If the operator $1,2 \lF$ is used, then all other $\lF$-operators are useless;
		\item If at least one $\lF$-operator is used, using additional $\emptyset \lF$ is useless;
		\item Using twice in a row the operator $1\lF$ or the operator $2\lF$ is useless.
	\end{itemize}
	
	This implies that it is sufficient to consider only polynomially many $\ATL$-formulas. We define below the set of $\ATL$-formulas to consider.
	\begin{definition}
		Consider a set of propositions $\prop$. We define the sets $\ATL_{\lF}(\prop,2)$ and $\ATL_{\lG}(\prop,2)$ below. For $\msf{H} \in \{ \lF,\lG\}$, we define:
		\begin{equation*}
			\msf{Quant}_{\msf{Alt}}^{\lF} := \{ \epsilon,\emptyset \lF,1,2 \lF, (1\lF \cdot 2\lF)^*,(1\lF \cdot 2\lF)^* \cdot 1\lF,(2\lF \cdot 1\lF)^*,(2\lF \cdot 1\lF)^* \cdot 2\lF \}
		\end{equation*}
		and
		\begin{equation*}
			\msf{Quant}_{\msf{Alt}}^{\lG} := \{ \epsilon,1,2 \lG,\emptyset \lG, (1\lG \cdot 2\lG)^*,(1\lG \cdot 2\lG)^* \cdot 1\lG,(2\lG \cdot 1\lG)^*,(2\lG \cdot 1\lG)^* \cdot 2\lG \}
		\end{equation*}
	\end{definition}
	
	This definition satisfies the lemma below.
	\begin{lemma}
		\label{lem:ATL_2_only_F_or_G}
		Consider any set of propositions $\prop$, $\msf{H} \in \{ \lF,\lG\}$. For all sequences of quantifiers $\msf{Qt} \in \msf{Op}(\{\msf{H}\},2)^*$, there is some $\msf{Qt}' \in 	\msf{Quant}_{\msf{Alt}}^{\msf{H}}$ such that: 
		\begin{itemize}
			\item $|\msf{Qt}| \leq |\msf{Qt}'|$;
			\item for all $\ATL^2$-formulas $\phi$, $\msf{Qt} \cdot \phi \equiv \msf{Qt}' \cdot \phi$.
		\end{itemize} 
	\end{lemma}
	\begin{proof}
		We prove the result for $\msf{H} = \lF$, the case $\msf{H} = \lG$ is analogous, and also relies entirely on Lemma~\ref{lem:equiv_atl_dominating_quantifiers}. 
		There are several cases:
		\begin{itemize}
			\item If $\msf{Qt} = \epsilon$, then $\msf{Qt} \in \msf{Quant}_{\msf{Alt}}^{\msf{H}}$. 
			\item Otherwise, assume that $\msf{Qt}$ only features the operator $\emptyset \lF$. In that case, for all $\ATL^2$-formulas $\phi$, we have $\msf{Qt} \cdot \phi \equiv \emptyset \lF \cdot \phi$.
			\item Otherwise, assume that $\msf{Qt}$ features the operator $1,2 \lF$. In that case, for all $\ATL^2$-formulas $\phi$, we have $\msf{Qt} \cdot \phi \equiv 1,2 \lF \cdot \phi$.
			\item Otherwise, $\msf{Qt}$ does not feature the operator $1,2 \lF$ and features the operator $1\lF$ or the operator $2\lF$. In that case, consider the sequence $\msf{Qt}'$ obtained from $\msf{Qt}$ by:
			\begin{itemize}
				\item removing all $\emptyset \lF$ operators;
				\item shrinking all sequences in $(1 \lF)^+$ into $1 \lF$ and shrinking all sequences in $(2 \lF)^+$ into $2 \lF$.
			\end{itemize}
			By Lemma~\ref{lem:equiv_atl_dominating_quantifiers}, for all $\ATL^2$-formulas $\phi$, the formula $\msf{Qt} \cdot \phi$ is equivalent to $\msf{Qt}' \cdot \phi$, and $|\msf{Qt}| \leq |\msf{Qt}'|$.
		\end{itemize}
	\end{proof}
	
	We handle in the definition below the case of formulas that may use a bounded amount of binary operators. 
	\begin{definition}
		\label{def:ATL_2_sketch}
		Consider some $\Bl \subseteq \Op{Bin}{lg}$ and some $\msf{H} \in \{\lF,\lG\}$. We define inductively on $n \in \N$ the set $\msf{Rel}^{\ATL(2)}(\Bl,\msf{H},n)$ as follows:
		\begin{itemize}
			\item $\msf{Rel}^{\ATL(2)}(\Bl,\msf{H},0) := \{ \msf{Qt} \cdot ? \mid \msf{Qt} \in 	\msf{Quant}_{\msf{Alt}}^{\msf{H}} \}$;
			\item For all $n \in \N$, $\msf{Rel}^{\ATL(2)}(\Bl,\msf{H},n+1) := \{ \msf{Qt} \cdot (\gamma_1 \bullet \gamma_2) \mid \msf{Qt} \in 	\msf{Quant}_{\msf{Alt}}^{\msf{H}},\; \bullet \in \Bl,\; \gamma_1,\gamma_2 \in \msf{Rel}^{\ATL(2)}(\Bl,\msf{H},n) \}$.
		\end{itemize}
		For all sets of propositions $\prop$, $n \in \N$, and $\gamma \in \msf{Rel}^{\ATL(2)}(\Bl,\msf{H},n)$, we say that an $\ATL^2$-formula $\phi$ is \emph{$\prop$-from} $\gamma$ if it is equal to $\gamma$, up to replacing every $?$ with some proposition in $\prop$. 
		
		For all bounds $B \in \N$, we let $\msf{Rel}^{\ATL(2)}(\Bl,\msf{H},n,B)$ denote the set of elements in $\msf{Rel}^{\ATL(2)}(\Bl,\msf{H},n)$ from which we can obtain a formula of size at most $B$. 
	\end{definition}
	
	This definition satisfies the lemma below. 
	\begin{lemma}
		\label{lem:ATL_2_binary_operators_polynomial}
		Consider some $\Bl \subseteq \Op{Bin}{lg}$  and some $\msf{H} \in \{\lF,\lG\}$. For all $n \in \N$, there is a polynom $q_n$ such that, for all $B \in \N$, $|\msf{Rel}^{\ATL(2)}(\Bl,\msf{H},n,B)| \leq q_n(n)$. 
		
		Furthermore, for all non-empty sets of propositions $\prop$, for all $\ATL^2$-formulas $\phi$ of size at most $B$, there is some $\gamma \in \msf{Rel}^{\ATL(2)}(\Bl,\msf{H},n,B)$ and a formula $\phi'$ of size at most $B$ that is $\prop$-obtained from $\gamma$ and such that $\phi \equiv \phi'$. 
	\end{lemma}
	\begin{proof}
		The first part of the lemma is a direct consequence of the fact that the number of sequences in $\msf{Quant}_{\msf{Alt}}^{\msf{H}}$ of size at most $B \in \N$ is bounded by a polynomial in $B$. 
		
		The second part of the lemma is a direct consequence of Lemma~\ref{lem:ATL_2_only_F_or_G}. 
	\end{proof}
	
	We deduce that the learning problem for ATL with two agents and only one of the operators $\lF,\lG$ can be decided in polynomial time. 
	\begin{lemma}
		\label{lem:atl_2_F_and_G_in_P}
		For all $\Bl \subseteq \Op{Bin}{lg}$ and $n \in \N$, both decision problems $\ATL^2_\msf{Learn}(\{\lF\},\emptyset,\Bl,n)$ and $\ATL^2_\msf{Learn}(\{\lG\},\emptyset,\Bl,n)$ can be decided in polynomial time.
	\end{lemma}
	\begin{proof}
		Let $\msf{H} \in \{\lF,\lG\}$, $\Bl \subseteq \Op{Bin}{lg}$ and $n \in \N$. Given an instance $\msf{In} = (\prop,\mathcal{P},\mathcal{N},B)$ of the decision problem $\ATL^2_\msf{Learn}(\{\msf{H}\},\emptyset,\Bl,n)$, one can follow the following steps to check if $\msf{In}$ is a positive instance:
		\begin{enumerate}
			\item Loop over all elements $\gamma \in \msf{Rel}^{\ATL(2)}(\Bl,\msf{H},n,B)$;
			\item Loop over all formulas $\phi$ that can be $\prop$-obtained from $\gamma$ of size at most $B$;
			\item Check whether or not this formula $\phi$ accepts $\mathcal{P}$ and rejects $\mathcal{N}$.
		\end{enumerate} 
		Following these steps is enough to decide if $\msf{In}$ is a positive instance of $\ATL^2_\msf{Learn}(\{\msf{H}\},\emptyset,\Bl,n)$ by Lemma~\ref{lem:ATL_2_binary_operators_polynomial}. Furthermore, they can be executed in polynomial time. Indeed, by Lemma~\ref{lem:ATL_2_binary_operators_polynomial}, the loop of step 1) is entered polynomially many times. Furthermore, all formulas with at most $n$ occurrences of binary operators use at most $2^n$ propositions. Hence, the loop of step 2) is entered at most $|\prop|^{2^n}$ times, which is polynomial in $\prop$ since $n$ is fixed. Finally, the last step can be done in polynomial time as well. Thus, we do obtain a polynomial-time procedure deciding the problem $\ATL^2_\msf{Learn}(\{\msf{H}\},\emptyset,\Bl,n)$.
	\end{proof}
	
	\textbf{$\msf{P}$-hardness} As mentioned for the $\msf{NL}$-hardness proof of the $\CTL$ learning problem without $\lX$, to establish the $\msf{P}$-hardness, we are going to exhibit a reduction from the reachability problem in two-player games $\msf{Reach}$, introduced in Definition~\ref{def:reach_decision_problem}. We define the reduction that we consider. Note the three turn-based structures that we define below are depicted in Figure~\ref{fig:simpleGamesATL2}.
	\begin{definition}
		\label{def:reduc_atl_2_p_hard}
		Consider two propositions $p,\bar{p}$ and a proper $\{1,2\}$-turn-based structure $T$ on $\{p,\bar{p}\}$. We let: 
		\begin{itemize}
			\item $T_{\bar{p}}$ be a trivial structure whose only state is labeled by the proposition $\bar{p}$; 
			\item $T_{\msf{no} \lG}$ be a two-state $(0,\emptyset)$-proper $\{1,2\}$-turn-based structure whose only starting state, labeled by $\bar{p}$, belongs to Agent 1, with two outgoing edges, one that loops, and one that goes to a self-looping sink labeled by $p$; 
			\item $T_{\msf{no} \; \fanBr{\{2\}}}$ be a two-state $(0,\emptyset)$-proper $\{1,2\}$-turn-based structure similar to $T_{\msf{no} \lG}$, except that the starting state belongs to Agent 2, instead of Agent 1.
			\item $\mathcal{P} := \{T,T_{\msf{no} \lG}\}$ and $\mathcal{N} := \{T_{\bar{p}},T_{\msf{no} \; \fanBr{\{2\}}}\}$.
		\end{itemize}
		
		Then, we define the inputs $\msf{In}^{\ATL^2(2),\lF}_{(p,\bar{p},T)} := (\{p,\bar{p}\},\mathcal{P},\mathcal{N},2)$ and $\msf{In}^{\ATL^2(2),\lG}_{(p,\bar{p},T)} := (\{p,\bar{p}\},\mathcal{N},\mathcal{P},2)$.
	\end{definition}
	
	\begin{figure}
		\centering
		\includegraphics[scale=1.2]{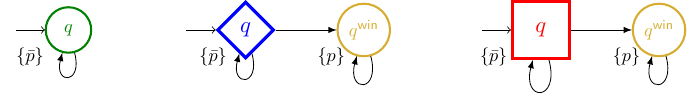}
		\caption{The game $T_{\bar{p}}$ on the left, the game $T_{\msf{no} \lG}$ in the middle and the game $T_{\msf{no} \; \fanBr{\{2\}}}$ on the right.}
		\label{fig:simpleGamesATL2}
	\end{figure}
	
	The definition above satisfies the lemma below.
	\begin{lemma}
		\label{lem:atl_2_reduc}
		Consider some set of unary operators $\Ut \subseteq \{\lX,\lF,\lG,\neg\}$, some set of binary operators $\Bl \subseteq \Op{Bin}{lg}$, some $n \in \N$ and an input $(p,\bar{p},T)$ of the decision problem $\msf{Reach}$. Let $\msf{H} \in \{\lF,\lG\}$. If $\msf{H} \in \Ut$, the input $(p,\bar{p},T)$ is a positive instance of the $\msf{Reach}$ decision problem if and only if $\msf{In}^{\ATL^2(2),\msf{H}}_{(p,\bar{p},T)}$ is a positive instance of the  $\ATL_\msf{Learn}^2(\Ut,\emptyset,\Bl,n)$ decision problem. 
	\end{lemma}
	\begin{proof}
		First assume that $\msf{H} = \lF$. Assume that $(p,\bar{p},T)$ is a positive instance of $\msf{Reach}$. We let $\varphi := \fanBr{\{1\}} \lF p$. We have $\Size{\varphi} = 2$. By assumption, we have $T \models \varphi$. In addition, $T_{\msf{no} \lG} \models \varphi$, $T_{\bar{p}} \not \models \varphi$ and $T_{\msf{no} \; \fanBr{\{2\}}} \not\models \varphi$. Hence, $\msf{In}^{\ATL^2(2),\msf{F}}_{(p,\bar{p},T)}$ is a positive instance of the  $\ATL_\msf{Learn}^2(\Ut,\emptyset,\Bl,n)$ decision problem. 
		
		On the other hand, assume that $\msf{In}^{\ATL^2(2),\msf{F}}_{(p,\bar{p},T)}$ is a positive instance of the $\ATL_\msf{Learn}^2(\Ut,\emptyset,\Bl,n)$ decision problem. Consider $\phi$ a separating formula of size at most 2 that accepts $\mathcal{P}$ and rejects $\mathcal{N}$. Let us show that, on $(0,\emptyset)$-proper structures, we have $\phi \implies \fanBr{\{1\}} \lF p$. If $\phi$ is a proposition, uses a negation or a binary operator, then it is equivalent to either $p,\bar{p},\msf{True},\msf{False}$ on $(0,\emptyset)$-proper structures. Since $\phi$ accepts $T_{\bar{p}}$ and rejects some negative structures, it necessarily is equivalent to $p$, and therefore it implies $\fanBr{\{1\}} \lF p$. Assume now that $\phi$ uses an operator in $\{\lX,\lF,\lG\}$ (in which we necessarily have $\prop(\phi) = \{p\}$). Since $\phi$ accepts the structure $T_{\msf{no} \lG}$, it does not use the operator $\lG$. Therefore, if it does use an operator, it is not an $\lG$-operator, and therefore the coalition of agents used cannot contain Agent 2 since $\varphi$ rejects the structure $T_{\msf{no} \; \fanBr{\{2\}}}$. Therefore, we have $\varphi \in \{\fanBr{\{1\}} \lX p,\fanBr{\{1\}} \lF p,\fanBr{\emptyset} \lX p,\fanBr{\emptyset} \lF p \}$. One can then check that, in all these cases, we have $\varphi \implies \fanBr{\{1\}} \lF p$. Therefore, since $T \models \varphi$, we also have $T \models \fanBr{\{1\}} \lF p$. Hence, $(p,\bar{p},T)$ is a positive instance of $\msf{Reach}$.
		
		The case $\msf{H} = \lG$ is dual: in that case, we consider the formula $\varphi := \fanBr{\{2\}} \lG \bar{p}$ (recall Proposition~\ref{prop:equiv_negation_ATL_turn_based}).
	\end{proof}
	
	The proof of Theorem~\ref{thm:atl_2_F_and_G_P_complete} is now direct.
	\begin{proof}
		Let $\msf{H} \in \{\lF,\lG\}$. The decision problem $\ATL^2_\msf{Learn}(\{\msf{H}\},\emptyset,\Bl,n)$ can be decided in polynomial time by Lemma~\ref{lem:atl_2_F_and_G_in_P}.
		
		Furthermore, the decision problem $\ATL^2_\msf{Learn}(\{\msf{H}\},\emptyset,\Bl,n)$ is $\msf{P}$-hard by Lemma~\ref{lem:atl_2_reduc} and the fact that the reduction given in Definition~\ref{def:reduc_atl_2_p_hard} can be computed in logarithmic space. 
	\end{proof}

	\subsubsection{$\ATL$ learning with three agents}
	Let us now consider the case of $\ATL$ learning with three agents. The goal of this subsection is to show the lemma below. 
	\begin{theorem}
		\label{thm:atl_3_F_or_G}
		For all sets $\Bl \subseteq \Op{Bin}{lg}$ and $n \in \N$, both decision problems $\ATL^3_\msf{Learn}(\{\lF\},\emptyset,\Bl,n)$ and $\ATL^3_\msf{Learn}(\{\lG\},\emptyset,\Bl,n)$ are $\msf{NP}$-complete.
	\end{theorem}
	
	\paragraph{Overview of the reduction.}
	We present the reduction with the operator $\lF$. The reduction with the operator $\lG$ is the same, up to reversing the sets of positive and negative structures. First, as for the $\CTL$ and $\ATL^2$ reductions, we follow the steps described in Section~\ref{subsubsec:abstract_recipe}, with Step~\ref{stepa} already taken care of in Section~\ref{subsubsec:handling_binary_operators}. Thus, we focus on $\ATL^3$-formulas using only the operator $\lF$ (and a single proposition). First, we define turn-based structures ensuring that: the proposition used is $p$ (with a trivial positive structure), and the operators $\fanBr{A} \lF$ such that $1 \in A$ and $\{2,3\} \cap A \neq \emptyset$ are not used (with two negative structures). Furthermore, since all the structures that we use are self-looping, Lemma~\ref{lem:useless_op} gives that the operator $\fanBr{\emptyset} \lF$ is useless. Hence, we can focus on formulas using only the operators $\fanBr{\{1\}} \lF,\fanBr{\{2\}} \lF,\fanBr{\{3\}} \lF,\fanBr{\{2,3\}} \lF$ and the proposition $p$, which are $\msf{Op}_{\{1\},\{2,3\}}(\{\lF\})$-formulas. 
	
	Fix an instance $(l,C,k)$ of the hitting set problem. We consider the bound $B := 2l+1$. Our idea is to focus on $(\{1\},\{2\},2l)$-alternating formulas. As for the $\ATL^2$ case, we consider $T^{2l: \; 1,2}$ as positive structure and use Lemma~\ref{lem:ATL_formula_necessary_sufficient} (Item a). Then, these $(\{1\},\{2\},2l)$-alternating formulas feature at least $l$ occurrences of the operator $\fanBr{\{1\}} \lF$ and $l$ occurrences of operators $\fanBr{A} \lF$ with $2 \in A$. Therefore the operator $\fanBr{\{3\}} \lF$ is not used. 
	
	Overall, the only thing that differs between the formulas that we consider is when the operators are $\fanBr{\{2\}} \lF$ and when the operators are $\fanBr{\{2,3\}} \lF$. Thus, the formulas $\phi^{\ATL(3)}(l,H)$ that we consider are such that the set $H$ entirely determines at which index $i$ an operator $\fanBr{\{2,3\}} \lF$ occurs (when $i \in H$), and at which index $i$ an operator $\fanBr{\{2\}} \lF$ occurs (when $i \notin H$). To ensure that $|H| \leq k$, we consider $T^{2(k+1): \; 1,3}$ as negative structure and we use Lemma~\ref{lem:ATL_formula_necessary_sufficient} (Item b). 
	
	Then, there remains to define, given a subset $C \subseteq [1,\ldots,l]$, a positive turn-based structure $T_{l,C,3}$ such that $\phi^{\ATL(3)}(l,H)$ accepts $T_{l,C,3}$ if and only if $H \cap C \neq \emptyset$. The structure $T_{l,C,3}$ (see Figure~\ref{fig:example_turn_based_3_from_C}) is similar to the structure $T_{l,C,2}$, except that the testing states are Agent-3 states.

	\textbf{Formal definitions and proofs.} With the alternating turn-based structures that we have already defined (recall Definition~\ref{def:turn_based_q_alt}) it is actually sufficient for the reduction define one additional type of turn-based structures to encode the fact that hitting set intersect all sets. Before we define it, let us first define the shape of the $\ATL$-formulas that we will consider. 
	\begin{definition}
		Let $l \in \N$. For all $H \subseteq [1,\ldots,l]$, an $\ATL$-formula is 
		an $\phi^{\ATL(3)}(l,H)$-formula if:
		\begin{equation*}
			\phi = 1 \lF \; \fanBr{A_l} \lF \cdots 1\lF \; \fanBr{A_1} \lF p
		\end{equation*}
		where for all $i \in [1,\ldots,l]$, we have $A_i \in \{ \{2\}.\{2,3\} \}$ and $A_i = \{2,3\}$ if and only if $i \in H$
		. 
	\end{definition}
	
	Let us now define the turn-based structure of interest whose definition is illustrated in Figure~\ref{fig:example_turn_based_3_from_C}. This definition, and the subsequent lemma and proof are very similar to what we did with the turn-based structure $T_{l,C,2}$ from Definition~\ref{def:turn_based_game_atl_2_intersect}.
	\begin{definition}
		Let $l \in \N_1$ and $C \subseteq [1,\ldots,l]$. We let $T_{l,C,3} := \langle Q^{l,C,3},I_{l,3},2,\{p\},\pi,\msf{AgSt},\msf{Succ} \rangle$ where (recall that the states $q^{1,2}_{h}$ come from the turn-based structure $T_{2l: \; 1,2}$ from Definition~\ref{def:turn_based_q_alt}):
		\begin{itemize}
			\item $Q^{l,C,3} := \{ q_{i} \mid 1 \leq i \leq 2l \} \cup \{ q^{1,2}_{h} \mid 1 \leq h \leq 2l \} \cup \{ q^\msf{Test}_{2i-1} \mid i \in C \} \cup \{q^{\msf{lose}},q^{\msf{win}}\}$;
			\item $I_l := \{ q_{2l} \}$;
			\item For all $2 \leq i \leq l+1$, $\msf{AgSt}(q_{2i}) := 1$ and $\msf{AgSt}(q_{2i-1}) := 2$. For all $i \in C$, we have $\msf{AgSt}(q^{\msf{Test}}_{2i-1}) := 3$.
			\item For all $1 \leq i \leq l$, we have:
			\begin{equation*}
				\msf{Succ}(q_{2i}) :=
				\begin{cases}
					\{ q_{2i},q_{2i-1} \} & \text{ if }i \notin C \\
					\{ q_{2i},q_{2i-1},q^{\msf{Test}}_{2i-1} \} & \text{ if }i \in C \\
				\end{cases}
			\end{equation*}
			and
			\begin{equation*}
				\msf{Succ}(q_{2i-1}) :=
				\begin{cases}
					\{ q_{2i-1},q_{2(i-1)} \} & \text{ if }i > 1 \\
					\{ q_{2i-1},q^{\msf{lose}} \} & \text{ if }i = 1 \\
				\end{cases}
			\end{equation*}
			and, for all $i \in C$:
			\begin{equation*}
				\msf{Succ}(q^{\msf{Test}}_{2i-1}) :=
				\begin{cases}
					\{ q^{\msf{Test}}_{2i-1},q^{1,2}_{2(i-1)} \} & \text{ if }i > 1 \\
					\{ q^{\msf{Test}}_{2i-1},q^{\msf{win}} \} & \text{ if }i =1 \\
				\end{cases}
			\end{equation*}
			\item For all $q \in Q^{l,C,3} \setminus \{q^{\msf{win}}\}$, we have $\pi(q) := \{\bar{p}\}$. 
		\end{itemize}
	\end{definition}
	
	\begin{figure}
		\centering
		\includegraphics[scale=0.8]{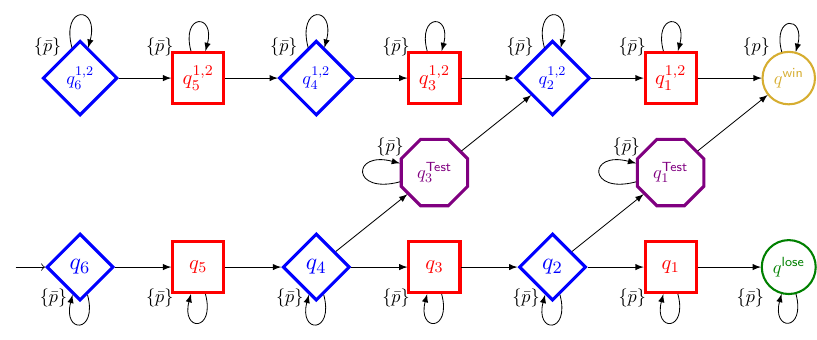}
		\caption{The turn-based game structure $T_{3,\{1,2\},3}$.}
		\label{fig:example_turn_based_3_from_C}
	\end{figure}
	
	The above definition satisfies the lemma below.
	\begin{lemma}
		\label{lem:ATL_3_intersect}
		Consider any $l \in \N_1$ and $C,H \subseteq [1,\ldots,l]$ and an $\phi^{\ATL(3)}(l,H)$-formula $\phi$. We have:
		\begin{equation*}
			T_{l,C,3} \models \phi \text{ if and only if }H \cap C \neq \emptyset
		\end{equation*}
	\end{lemma}
	\begin{proof}
		For all $1 \leq i \leq l$, we let $H_i := H \cap [1,\ldots,l]$. Then, in the turn-based structure $T_{l,C,3}$, we prove by induction on $1 \leq i \leq l$ the property $\mathcal{P}(i)$:  $q_{2i} \models \phi^{\ATL(3)}(i,H_i)$ if and only if $H_i \cap C \neq \emptyset$. We first handle the case $i = 1$. Consider the $\phi^{\ATL(3)}(1,H_1)$-formula $\phi = 1\lF \fanBr{A_1} \lF p$. There are two cases.
		\begin{itemize}
			\item Assume that $H_1 \cap C = \{ 1 \}$. Then, we have $A_1 = \{2,3\}$. Furthermore, $q^{\msf{Test}}_{1} \in Q^{l,C,3}$ with 
			$q^{\msf{Test}}_{1} \models 2,3\lF p$ since $\msf{AgSt}(q^{\msf{Test}}_{1}) = 3$ and $q^{\msf{win}} \in \Succc(q^{\msf{Test}}_{1})$. Thus, $q_2 \models \phi$ since $\msf{AgSt}(q_2) = 1$. 
			\item Assume now that $H_1 \cap C = \emptyset$. If $1 \in C$, we have $1 \notin H$, thus $A_1 = \{2\}$. Since $\msf{AgSt}(q^{\msf{Test}}_{1}) = 3$ and $\pi(q^{\msf{Test}}_{1}) = \{\bar{p}\}$, we have $q^{\msf{Test}}_{1} \not\models \fanBr{A_1} \lF p$. Hence, $q_2 \not \models \phi$. On the other hand, if $1 \notin C$, we have $\msf{Succ}(q_2) = \{q_2,q_1\}$ and $\msf{Succ}(q_1) = \{q_1,q^{\msf{lose}}\}$. Therefore, $q_2 \not \models \phi$.
		\end{itemize}
		Hence, the property $\mathcal{P}(1)$ holds.  Assume now that $\mathcal{P}(i)$ holds for some $1 \leq i \leq l-1$. Consider the $\phi^{\ATL(3)}(i+1,H_{i+1})$-formula $\phi$ defined by:
		\begin{equation*}
			\phi := 1\lF \; \fanBr{A_{i+1}}\lF \phi'
		\end{equation*}
		with $A_{i+1} = \{2,3\}$ if $i+1 \in H$ and $A_{i+1} = \{2\}$ otherwise, and $\phi' = \phi^{\ATL(3)}(i,H_{i})$. As above, there are two cases.
		\begin{itemize}
			\item Assume that $i+1 \in H \cap C$. Since there is a safe winning path from $q_{2i}^{1,2}$ that is $(\{1\},\{2\},2i)$-alternating and the formula $\phi'$ is $(\{1\},\{2\},2i)$-alternating, it follows that $q_{2i}^{1,2} \models \phi'$, by Lemma~\ref{lem:ATL_formula_necessary_sufficient} (Item b). Hence, we have $q^{\msf{Test}}_{2i+1} \in Q^{l,C,3}$ with $q^{\msf{Test}}_{2i+1} \models 2,3 \lF \phi'$. Therefore, $q_{2(i+1)} \models 1\lF \fanBr{A_{i+1}}\lF \phi'$ since $\msf{AgSt}(q_{2(i+1)}) = 1$.
			\item Assume now that $i+1 \notin H \cap C$. Let us show that $q_{2(i+1)} \models \phi$ if and only if $q_{2i} \models \phi'$. First, if $q_{2i} \models \phi'$, then $q_{2i+1} \models \fanBr{A_{i+1}}\lF \phi'$ since $\msf{AgSt}(q_{2i+1}) = 2 \in A_{i+1}$. Thus, we have $q_{2(i+1)} \models \phi$. Assume now that $q_{2(i+1)} \models \phi$. Note that, the winning paths from $q_{2(i+1)}$ are all $(\{1\},\{2,3\},2(i+1))$-alternating, therefore, by Lemma~\ref{lem:ATL_formula_necessary_sufficient} (Item a), no strict sub-formula of $\phi$ accept the state $q_{2(i+1)}$. Similarly, the winning paths from $q_{2i+1}$ are all $(\{2,3\},\{1\},2i+1)$-alternating, therefore, by Lemma~\ref{lem:ATL_formula_necessary_sufficient} (Item a), no strict sub-formula of $\fanBr{A_{i+1}}\lF\phi^{\ATL(2)}(i,H_{i})$ accept the state $q_{2i+1}$. Then, there are two cases.
			\begin{itemize}
				\item If $i+1 \in C$, we have $i+1 \notin H$, thus $\phi = 1\lF 2\lF \phi'$. Furthermore, the winning paths from $q_{2i}^{1,2}$ are all $(\{1\},\{2,3\},2(i+1))$-alternating, therefore, by Lemma~\ref{lem:ATL_formula_necessary_sufficient} (Item a), no strict sub-formula of $\phi'$ accept the state $q_{2i}^{1,2}$. Thus, since $\msf{AgSt}(q^{\msf{Test}}_{2i+1}) = 3$, it follows that $q^{\msf{Test}}_{2i+1} \not\models 2\lF \phi'$. Since $\msf{Succ}(q_{2(i+1)}) = \{q_{2(i+1)},q_{2i+1},q^{\msf{Test}}_{2i+1}\}$ and $\msf{Succ}(q_{2i+1}) = \{q_{2i+1},q_{2i} \}$, we have $q_{2i} \models \phi'$. 
				\item If $i+1 \notin C$, then $\msf{Succ}(q_{2(i+1)}) = \{ q_{2(i+1)},q_{2i+1} \}$, we have that $q_{2i} \models \phi'$. 
			\end{itemize}
			We have established that $q_{2(i+1)} \models \phi$ if and only if $q_{2i} \models \phi'$, with $\phi' = \phi^{\ATL(3)}(i,H_{i})$. Furthermore, by $\mathcal{P}(i)$, we have $q_{2i} \models \phi'$ if and only if $H_i \cap C \neq \emptyset$. Since $i+1 \notin H \cap C$, it follows that $H_{i+1} \cap C = H_i \cap C$. Hence, we do obtain that $q_{2(i+1)} \models \phi$ if and only if $H_{i+1} \cap C \neq \emptyset$.
		\end{itemize}
		Hence, $\mathcal{P}(i+1)$ holds. In fact, $\mathcal{P}(i)$ holds for all $1 \leq i \leq l$. The lemma follows.
	\end{proof}

	\paragraph{Definition of the reduction}
	We can now define the reductions that we consider for the two cases $\Ut = \{ \lF \}$ and $\Ut = \{ \lG \}$. 
	\begin{definition}
		\label{def:reduction_ATL_3}
		Consider an instance $(l,C,k)$ of the hitting set problem $\msf{Hit}$. We define:
		\begin{itemize}
			\item $\mathcal{P} := \{ T_p, T_{2l: 1,2} \} \cup \{ T_{(l,C_i,3)} \mid 1 \leq i \leq n \}$;
			\item $\mathcal{N} := \{ T_{2(l+1): 1,2}, T_{2(k+1): 1,3} \}$.
		\end{itemize}
		
		Then, we define the inputs $\msf{In}^{\ATL(3),\lF}_{(l,C,k)} := (\prop_0,\mathcal{P},\mathcal{N},2l+3)$ and $\msf{In}^{\ATL(3),\lG}_{(l,C,k)} := (\prop_0,\mathcal{N},\mathcal{P},2l+3)$.
	\end{definition}
	
	This definition satisfies the lemma below.
	\begin{lemma}
		\label{lem:reduc_ATL_3_F}
		Let $\msf{H} \in \{\lF,\lG\}$. An instance $(l,C,k)$ of the hitting set problem is positive if and only if $\msf{In}^{\ATL(3),\msf{H}}_{(l,C,k)}$ is a positive instance of the $\ATL^3_\msf{Learn}(\{ \msf{H} \},\emptyset,\emptyset,0)$ decision problem if and only if $\msf{In}^{\ATL(3),\msf{H}}_{(l,C,k)}$ is a positive instance of the $\ATL^3_\msf{Learn}(\{ \lF,\lG \},\emptyset,\emptyset,0)$ decision problem.
	\end{lemma}
	\begin{proof}
		We consider the case where $\msf{H} = \lF$, the case $\msf{H} = \lG$ is analogous. 
		
		Assume that $(l,C,k)$ is a positive instance of the hitting set problem $\msf{Hit}$. Consider a hitting set $H \subseteq [1,\ldots,l]$ and the $\phi := \phi^{\ATL(3),\lF}(l,H) \in\ATL^3(\{p\},\{\lF\},\emptyset,\Bl,0)$. We have $\Size{\phi} = 2l+1$. Furthermore:
		\begin{itemize}
			\item the proposition used in $\phi$ is $p$, therefore $\phi$ accepts the structure $T_p$;
			\item $\phi$ is $(\{1\},\{2\},2l)$-alternating, thus it accepts the structure $T_{2l: 1,2}$ by Lemma~\ref{lem:ATL_formula_necessary_sufficient} (Item b);
			\item for $1 \leq i \leq n$, we have $C_i \cap H \neq \emptyset$, hence by Lemma~\ref{lem:ATL_3_intersect}, $\phi$ accepts the structure $T_{(l,C_i,3)}$;
			\item $|H| \leq k$, hence there are at most $k$ times an operator $\fanBr{A} \lF$ used in $\phi$ with $3 \in A$. Hence, $\phi$ is not $(\{1\},\{3\},2(k+1))$-alternating. Thus, it rejects the structure $T_{2(k+1): 1,3}$ by Lemma~\ref{lem:ATL_formula_necessary_sufficient} (Item a).
		\end{itemize}
		Therefore, $\msf{In}^{\ATL(3),\lF}_{(l,C,k)}$ is a positive instance of the $\ATL^3_\msf{Learn}(\{ \lF\},\emptyset,\emptyset,0)$ decision problem.
		
		Straightforwardly, if $\msf{In}^{\ATL(3),\lF}_{(l,C,k)}$ is a positive instance of  $\ATL^3_\msf{Learn}(\{ \lF\},\emptyset,\emptyset,0)$, then it is also a positive instance of $\ATL^3_\msf{Learn}(\{ \lF,\lG\},\emptyset,\emptyset,0)$.
		
		Assume now that $\msf{In}^{\ATL(3),\lF}_{(l,C,k)}$ is a positive instance of $\ATL^3_\msf{Learn}(\{ \lF,\lG \},\emptyset,\emptyset,0)$. Consider a formula $\phi \in \ATL^3(\prop_0,\{\lF,\lG\},\emptyset,\emptyset,0)$ with $\Size{\phi} \leq 2l+1$ that accepts $\mathcal{P}$ and rejects $\mathcal{N}$. We let $\phi = \msf{Qt} \cdot x$ for some $\msf{Qt} \in (\msf{Op}(3,\{\lF,\lG\}))^*$ and $x \in \prop_0$. Let $(\msf{Qt}',x') := \msf{UnNeg}(\msf{Qt},0)$. We let $y \in \prop_0$ be such that $y = x$ if and only if $x' = 0$. Finally, we let $\phi' := \msf{Qt}' \cdot y$. By Lemma~\ref{lem:unnegate_unary}, we have:
		\begin{itemize}
			\item $|\msf{Qt}'| \leq |\msf{Qt}|$, therefore $\Size{\phi'} \leq \Size{\phi} \leq 2l+1$;
			\item $\msf{Qt}' \in \msf{Op}(3,\{\lF,\lG\})$, therefore $\phi' \in \ATL^3(\prop_0,\{\lF,\lG\},\emptyset,\emptyset,0)$;
			\item Since all the structures in $\mathcal{P}$ and $\mathcal{N}$ are $(0,\emptyset)$-proper structures (and therefore, on those structures, $p$ and $\neg \bar{p}$ are equivalent), $\phi'$ also accepts $\mathcal{P}$ and rejects $\mathcal{N}$. 
		\end{itemize}
		
		The formula $\phi'$ accepts the structure $T_p$, therefore $\prop(\phi') = \{p\}$. Furthermore, the formula $\phi$ rejects both structures $T_{2(l+1): 1,2}$ and $T_{2(k+1): 1,3}$. Hence, for all operators $\fanBr{A} \lF$ used in $\phi$, if $1 \in A$, then $\{2,3\} \cap A = \emptyset$. Therefore, since $\phi$ accepts the structure $T_{2l: 1,2}$, by Lemma~\ref{lem:ATL_formula_necessary_sufficient} (Item a), the formula $\phi$ is $(\{1\},\{2\},2l)$-alternating. Since we have $\Size{\phi} \leq 2l+1$, this implies that 
		\begin{equation*}
			\phi = \fanBr{A^1_1} \lF \fanBr{A^2_1} \lF  \ldots \fanBr{A^1_{l}} \lF \fanBr{A^2_{l}} \lF p
		\end{equation*}
		where, for all $1 \leq i \leq l$, we have $A^1_i = \{1\}$ and $2 \in A^2_i, 1 \notin A^2_i$. However, since $\phi$ rejects the structure $T_{2(k+1): 1,3}$, by Lemma~\ref{lem:ATL_formula_necessary_sufficient} (Item b), we have that $\phi$ is not $(\{1\},\{3\},2(k+1))$-alternating. Hence, there are at most $k$ indices $1 \leq i \leq l$ such that $3 \in A^2_i$. This implies that $\phi= \phi^{\ATL(3)}(l,H)$ for a set $H \subseteq [1,\ldots,l]$ such that $|H| \leq k$. Consider then any $1 \leq i \leq n$. Since the formula $\phi$ accepts the structure $T_{(l,C_i,3)}$, it follows by Lemma~\ref{lem:ATL_3_intersect} that $C_i \cap H \neq \emptyset$. In fact, $H$ is a hitting set and $(l,C,k)$ is a positive instance of the hitting set problem $\msf{Hit}$.
		
		The case $\msf{H} = \lG$ is analogous. (It suffices to consider the negation of the separating formula.)
	\end{proof}

	The proof of Theorem~\ref{thm:atl_3_F_or_G} is now direct.
	\begin{proof}
		This is a direct consequence of Lemmas~\ref{lem:reduc_ATL_3_F}, the fact that the reductions from Definition~\ref{def:reduction_ATL_3} can be computed in logarithmic space and Theorem~\ref{thm:unary_is_sufficient}.
	\end{proof}
	
	\section{Conclusion and future Work}
	\label{sec:FutureWork}
	In this work, we undertake an in-depth complexity analysis of the passive learning problems for $\LTL$, $\CTL$ and $\ATL$. Our results are gathered in Table~\ref{tab:summary}, and could be roughly summarized as follows. When the number of occurrences of binary operators is unbounded, all the learning problems are $\msf{NP}$-complete. On the other hand, when the number of occurrences of binary operators is bounded, discrepancies between the behaviors of $\LTL$, $\CTL$, and $\ATL$ learning appear: there are subsets of operators for which the learning problem is tractable with some number of agents, while it becomes untractable with more agents.	
	
	Overall, this paper essentially tackles reductions and hardness proofs, while the arguments that specific problems are in $\msf{L},\msf{NL},\msf{P}$ are (relatively) more straightforward. However, this is made possible by the fact that the bound in the size of the formula is given in unary. We have argued in Section~\ref{subsubsec:discussion} why we believe that it makes sense to consider such a setting. Nonetheless, the decision problems that would arise with a bound given in binary would certainly be interesting and challenging research questions (just like it was in \cite{arXivFijalkow}). Another interesting direction, which can be combined with the above one, could be, as is done in \cite{arXivFijalkow}, to study the existence of tractable approximation algorithms.
	
	\bibliographystyle{splncs04}
	\bibliography{paper}
\end{document}